\documentclass[10pt]{article}

\usepackage{tikz}
\usetikzlibrary{matrix,arrows,decorations.pathmorphing,shapes.geometric}

\usepackage{etex}
\usepackage{subcaption}
\usepackage[utf8]{inputenc}
\usepackage{a4wide}
\usepackage{graphicx}
\usepackage{wrapfig}
\usepackage[hmarginratio={1:1},     
vmarginratio={1:1},     
textwidth=16cm,        
textheight=22cm,
heightrounded,]{geometry}

\usepackage{amsmath}
\usepackage{amsthm}
\usepackage{amssymb}

\usepackage{mathrsfs, mathtools}
\usepackage{overpic}
\usepackage[all,cmtip]{xy}

\let\emph\undefined
\newcommand{\emph}[1]{\textsl{#1}}

\usepackage{hyperref}

\numberwithin{equation}{section}
\usepackage{mathtools}
\mathtoolsset{showonlyrefs}
\numberwithin{equation}{section}

\newtheoremstyle{style1}
{13pt}
{13pt}
{}
{}
{\normalfont\bfseries}
{.}
{.5em}
{}

\theoremstyle{style1}

\newtheorem{definition}{Definition}[section]
\newtheorem{example}[definition]{Example}
\newtheorem{remark}[definition]{Remark}

\newcommand{\catf}[1]{{\mathsf{#1}}}

\usepackage{tocloft}

\newtheoremstyle{style2}
{13pt}
{13pt}
{\slshape}
{}
{\normalfont\bfseries}
{.}
{.5em}
{}

\theoremstyle{style2}

\newtheorem{lemma}[definition]{Lemma}
\newtheorem{theorem}[definition]{Theorem}
\newtheorem{proposition}[definition]{Proposition}
\newtheorem{corollary}[definition]{Corollary}

\usepackage{multicol}
\usepackage{enumitem}

\newcommand{\R}{\mathbb{R}}
\newcommand{\C}{\mathbb{C}}
\newcommand{\Z}{\mathbb{Z}}

\usepackage{needspace}

\newcommand{\SO}{\operatorname{SO}}

\newcommand{\ob}{\operatorname{ob}}

\newcommand{\colim}{\operatorname{colim}}

\newcommand{\TGrpd}{2\catf{Grpd}}
\newcommand{\TGrp}{2\catf{Grp}}

\newcommand{\Top}{\catf{Top}}

\newcommand{\cat}[1]{\mathcal{#1}}

\newcommand{\Spec}{\operatorname{Spec}}
\newcommand{\Aut}{\operatorname{Aut}}

\newcommand{\End}{\operatorname{End}}

\newcommand{\Hom}{\operatorname{Hom}}

\newcommand{\id}{\operatorname{id}}

\newcommand{\pr}{\operatorname{pr}}

\newcommand{\ev}{\operatorname{ev}}
\newcommand{\coev}{\operatorname{coev}}
\newcommand{\fd}{\operatorname{fd}}

\newcommand{\Fun}{\operatorname{Fun}}

\newcommand{\Obj}{\operatorname{Obj}}

\newcommand{\BiCat}{{\catf{BiCat}}}

\newcommand{\Cat}{\catf{Cat}}

\newcommand{\DS}{\text{/\hspace{-0.1cm}/}}

\let\to\undefined
\newcommand{\to}{\longrightarrow}
\let\mapsto\undefined
\newcommand{\mapsto}{\longmapsto}

\newcommand{\Alg}{\catf{Alg}}

\newcommand{\hsVect}{\catf{hsVect}}
\newcommand{\sAlg}{\catf{sAlg}}
\newcommand{\stAlg}{\catf{stAlg}}
\newcommand{\sVect}{\catf{sVect}}
\newcommand{\Vect}{\catf{Vect}}
\newcommand{\Bord}{\catf{Bord}}

\newcommand{\op}{\text{op}}

\newcommand{\Spin}{\operatorname{Spin}}
\newcommand{\fr}{\operatorname{fr}}
\newcommand{\Pin}{\operatorname{Pin}}
\newcommand{\Cl}{\operatorname{Cl}}

\let\O\undefined\newcommand{\O}{\operatorname{O}}
\newcommand{\stFrob}{\catf{stFrob}}

\DeclareMathSymbol{\Phiit}{\mathalpha}{letters}{"08} 
\DeclareMathSymbol{\Psiit}{\mathalpha}{letters}{"09}
\DeclareMathSymbol{\Sigmait}{\mathalpha}{letters}{"06}
\DeclareMathSymbol{\Xiit}{\mathalpha}{letters}{"04}
\DeclareMathSymbol{\Piit}{\mathalpha}{letters}{"05}\let\Pi\undefined\newcommand{\Pi}{\Piit}
\DeclareMathSymbol{\Gammait}{\mathalpha}{letters}{"00}
\DeclareMathSymbol{\Omegait}{\mathalpha}{letters}{"0A}\let\Omega\undefined\newcommand{\Omega}{\Omegait}
\DeclareMathSymbol{\Upsilonit}{\mathalpha}{letters}{"07}
\DeclareMathSymbol{\Thetait}{\mathalpha}{letters}{"02}
\DeclareMathSymbol{\Lambdait}{\mathalpha}{letters}{"03}\let\Lambda\undefined\newcommand{\Lambda}{\Lambdait}

\let\Phi\undefined\newcommand{\Phi}{\Phiit}
\let\Sigma\undefined\newcommand{\Sigma}{\Sigmait}
\let\Psi\undefined\newcommand{\Psi}{\Psiit}
\let\Gamma\undefined\newcommand{\Gamma}{\Gammait}

\definecolor{Blue}  {rgb} {0.282352,0.239215,0.803921}
\definecolor{Green} {rgb} {0.133333,0.545098,0.133333}
\definecolor{Red}   {rgb} {0.803921,0.000000,0.000000}
\definecolor{Violet}{rgb} {0.580392,0.000000,0.827450}

\newcounter{jfc}

\usepackage{tikz}
\usetikzlibrary{matrix,arrows,decorations.pathmorphing,shapes.geometric, decorations.pathreplacing,decorations.markings}

\usepackage{environ, etoolbox}
\usetikzlibrary{cd, external}

\def\temp{&} \catcode`&=\active \let&=\temp

\tikzset{
	on each segment/.style={
		decorate,
		decoration={
			show path construction,
			moveto code={},
			lineto code={
				\path [#1]
				(\tikzinputsegmentfirst) -- (\tikzinputsegmentlast);
			},
			curveto code={
				\path [#1] (\tikzinputsegmentfirst)
				.. controls
				(\tikzinputsegmentsupporta) and (\tikzinputsegmentsupportb)
				..
				(\tikzinputsegmentlast);
			},
			closepath code={
				\path [#1]
				(\tikzinputsegmentfirst) -- (\tikzinputsegmentlast);
			},
		},
	},
	mid arrow/.style={postaction={decorate,decoration={
				markings,
				mark=at position .5 with {\arrow[#1]{stealth}}
	}}},
}
\tikzset{%
	link/.style    = { white, double = black, line width = 1.8pt,
		double distance = 0.41pt },
	channel/.style = { white, double = black, line width = 0.8pt,
		double distance = 0.61pt },
}

\tikzcdset{%
	triple line/.code={\tikzset{%
			double distance = 3pt, 
			double=\pgfkeysvalueof{/tikz/commutative diagrams/background color}}},
	quadruple line/.code={\tikzset{%
			double distance = 5.3pt, 
			double=\pgfkeysvalueof{/tikz/commutative diagrams/background color}}},
	Rrightarrow/.code={\tikzcdset{triple line}\pgfsetarrows{tikzcd implies cap-tikzcd implies}},
	RRightarrow/.code={\tikzcdset{quadruple line}\pgfsetarrows{tikzcd implies cap-tikzcd implies}}
}

\newcommand*{\tarrow}[2][]{\arrow[Rrightarrow, #1]{#2}\arrow[dash, shorten >= 0.5pt, #1]{#2}}
\newcommand*{\qarrow}[2][]{\arrow[RRightarrow, #1]{#2}\arrow[equal, double distance = 0.25pt, shorten >= 1.28pt, #1]{#2}}

\begin{document}
	
	\vspace*{-1.5cm}
	
	\vspace{5mm}
	
	\begin{center}
		\textbf{\LARGE{Reflection Structures and Spin Statistics in Low Dimensions}}\\
		\vspace{1cm}
		{\large Lukas Müller $^{a}$} \ \ and \ \ {\large Luuk Stehouwer $^{b}$ }
		\\ 	\vspace{5mm}{\slshape $^a$ Perimeter Institute \\ 31 Caroline Street North \\  N2L 2Y5 Waterloo }	\\[7pt]	{\slshape $^b$\em  Max-Planck-Institut f\"ur Mathematik\\
			Vivatsgasse 7, \\ 53111 Bonn }\\
		
		\vspace{5mm}

	\end{center}
	\begin{abstract}\noindent 
  	We give a complete classification of topological field theories with reflection
  	structure and spin-statistics in one and two spacetime dimensions. Our answers
  	can be naturally expressed in terms of an internal fermionic symmetry group $G$ which is different from the spacetime structure group. Fermionic groups
  	encode symmetries of systems with fermions and time reversing symmetries. 
  	We show that 1-dimensional topological field theories with reflection structure and spin-statistics are classified by finite dimensional hermitian representations of $G$. In spacetime dimension two we give a classification in terms strongly $G$-graded stellar Frobenius algebras. Our proofs are based on the cobordism hypothesis. Along the way, we develop some
  	useful tools for the computation of homotopy fixed points of 2-group actions on bicategories.  
	\end{abstract}
	
	\tableofcontents

\section{Introduction}

Functorial field theories are a rigorous approach to quantum field theories. They are
best known in the incarnation of (and evolved from) Atiyah's definition of topological field theories~\cite{atiyah1988topological}. The framework in its traditional setting does not include reflection positivity (the Euclidean version of unitarity) or a connection between the spin and statistics of a particle. 
Both reflection structures\footnote{We comment on positivity in a moment.} and spin-statistics relations can 
be understood as an equivariance condition on the field theory~\cite{theospinstatistics,freedhopkins} with respect to a $\Z_2$ and $B\Z_2$ action, respectively. 

Most mathematical classification results and constructions in the context of topological field theories
do not consider these additional equivariance conditions. One big exception is the seminal work by Freed 
and Hopkins~\cite{freedhopkins} which classify fully extended reflection positive invertible theories using tools 
from homotopy theory. However, we are not aware of any work studying non-invertible theories in 
detail. For example there does not exist a definition of positivity for extended field theories. This paper provides a 
detailed study and classification of topological field theories with reflection structure and spin statistics in one and two spacetime dimensions. Even though we will not solve the problem of defining reflection positivity in general we hope that our results give some insights into how to
define it for extended non-invertible field theories.  
One (mathematically) surprising part of the work by Freed and Hopkins is that imposing reflection 
positivity leads to large simplifications in the computations. Something similar also happens in our work: 
By considering topological field theories with reflection structure and or spin statistics connection we 
will be able to arrive at a classification in terms of concrete algebraic structures. 

Due to the length and at times technical nature of the paper, we will give a detailed and extended summary 
of our results in the introduction. We start with a few comments about the connection to condensed 
matter physics partially motivating our work. 

\subsection{Gapped lattice systems and topological field theories}
A \emph{symmetry enriched topological phase of matter}~\cite{kitaevperiodic, gu2009tensor, Kapustin:2014dxa, freedhopkins, gaiottoSPT} is roughly speaking an equivalence class of 
gapped lattice Hamiltonians which can be continuously deformed into each other without breaking a symmetry or closing the energy gap.
A good tool to study them is their low energy effective field theory describing the ground states which turns out to be topological in many cases. 
Conjecturally, the effective low energy field theory is a complete invariant of the topological phase~\cite{Kapustin:2014tfa,Kapustin:2014dxa,Kapustin:2015uma,freedhopkins}. 
Topological quantum field theories are reasonably well understood mathematically allowing in principle explicit classifications and hence predictions of possible topological phases of matter.        
An interesting class of symmetry protected topological phases are the so called-short ranged entangled topological phases. 
They are characterized by the fact that their low energy effective field theory is invertible. 
Unitary, i.e. reflection positive, invertible topological field theories are classified in terms of bordism invariants~\cite{Yonekura:2018ufj,freedhopkins}. 
The non-invertible case is not nearly as well-understood.

Before providing more details on the realization of symmetries of a lattice system in the low energy effective field theory we give a rough introduction to the mathematical definition of topological field theories we use in this paper. More details can be found in Section~\ref{Sec:TFT}. 
The definition goes back to Atiyah~\cite{atiyah1988topological}, but has evolved significantly since its original definition. The definition comes in various flavors depending on the type of spacetime manifolds one considers. For example in theories involving fermions all manifolds should be equipped with a spin-structure. All possible structures relevant for this article can be 
formulated in terms of tangential structures (see Section~\ref{Sec:TFT} for details). To specify a type of tangential structure for $d$-dimensional manifolds, one has to fix a Lie group $H$ together with a real representation $\rho \colon H\to O_d$. 
A tangential structure on a $d$-dimensional manifold $M$ consists of a principal $H$-bundle $P$ together with an isomorphism of vector bundles $P\times_\rho \R^d \cong TM$. For example, tangential structures for the map $\Spin_d\to \SO_d \to O_d$ are spin structures on $M$. 

A topological field theory (with values in super vector spaces) assigns to every closed $d-1$-dimensional manifold $\Sigma$ equipped with a tangential structure a super vector space $\mathcal{Z}(\Sigma)$; the state space of the theory graded by fermion parity. Furthermore, $\mathcal{Z}$
assigns to every bordism, i.e. compact manifold $M$ with boundary $\partial M \cong \Sigma_1 \sqcup \Sigma_2$ equipped with a tangential structure a linear map $\mathcal{Z}(M)\colon \mathcal{Z}(\Sigma_1)\to \mathcal{Z}(\Sigma_2)$. 
The linear map associated to the gluing of two manifolds along a common boundary is required to be the composition of the linear maps for the individual pieces. 
There are a few more conditions usually imposed onto these data which can be conveniently encoded using the language of symmetric monoidal categories: Let
$\Bord^{H,\rho}_d$ be the symmetric monoidal category with objects closed $d-1$-dimensional manifolds equipped with a tangential $H$-structure and bordisms as morphisms. 
The symmetric monoidal structure is given by disjoint union of manifolds.
A \emph{$d$-dimensional topological field theory} is now a symmetric monoidal functor
\begin{align}
\mathcal{Z} \colon \Bord^{H,\rho}_d \to \sVect \ \ . 
\end{align} 
To fully capture locality one usually also allows cutting the $d-1$-dimensional manifolds into smaller pieces leading to a symmetric monoidal $d$-category $\Bord^{H,\rho}_{d,0}$~\cite{lurietft, calaquescheimbauer}. A fully local topological field theory with values in a symmetric monoidal $d$-category $\cat{S}$ is a symmetric monoidal functor $\Bord^{H,\rho}_{d,0} \to \cat{S}$. For us only the case $d=2$ will be relevant and the target bicategory $\sAlg$ of super algebras, bimodules, and intertwiners. We explain all this in more detail in Section~\ref{Sec:TFT} and Appendix~\ref{App: Super alg}. For the rest of this introduction, no detailed understanding of these concepts is needed.     

After this brief detour, we come back to symmetries of lattice systems. 
These are given by automorphisms of the Hilbert space associated to every site which commute with the Hamiltonian of the system. 
In theories with fermions, $-1$ to the number of fermions is always a non-trivial symmetry of the system denoted by $(-1)^F$. 
Furthermore, one distinguishes between symmetries which are unitary (time preserving) and anti-unitary (time reversing). 
We will call this local symmetry group the \emph{internal symmetry group} $G$ of the lattice system. Abstracting the structure of local symmetries in lattice
systems leads to the definition of a \emph{fermionic group} as a $\Z_2$-graded Lie group $G=G_0\sqcup G_1 $ together with an even central element $(-1)^F\in G_0$ squaring to $1$. We discuss fermionic groups in detail in Section~\ref{Sec:FG}.      
               
Since the internal symmetry is local, the low-energy effective field theory 
should be definable on manifolds with background fields for the internal symmetry. Here the time reversing symmetries and the fermion parity need to be combined 
appropriately with the local spacetime symmetries. This means that the low-energy effective topological field theory is not defined on manifolds with a tangential 
$G$ structure, but rather with a tangential structure for a spacetime structure
group $H_d$ associated to $G$. We explain in detail how to construct $H_d$ from
a fermionic group $G$ in Section~\ref{Sec:spacetime group}. 
For the moment, 
let us just give a few illustrating examples: 
\begin{itemize}
	\item If the internal symmetry group is trivial, then the associated 
	spacetime structure group is $SO_d$. Hence the low energy effective theory
	can be defined on oriented manifolds. 
	\item If the theory contains fermions, but no other internal symmetries, then the internal symmetry group is $\Z_2^F=\{ 1,(-1)^F\}$ and the spacetime structure group is $\Spin_d$. Hence the low energy effective theory
    is defined on spin manifolds.  
    \item In the case there is one time reversing symmetry $T$ squaring to $(-1)^F$,
    the spacetime structure group is $\Pin_d^+$. 
\end{itemize}   
One important feature about the theories arising from unitary lattice systems is that they are expected to be reflection positive and satisfy a version of the spin statistics relations. Both concepts are not well-defined for topological field theories with arbitrary tangential structure, but can be formulated for structure groups constructed from internal symmetry groups $G$. 
        
\subsection{Reflection structures and spin statistics}
By construction the spacetime structure group $H_d$ associated to an internal 
symmetry group $G$ is part of a short exact sequence 
\begin{align}
	1\to H_d\to \widehat{H}_d\to \Z_2 \to 1 
\end{align}
of Lie groups. 
This sequence can be used to define an involution $\overline{( \cdot )}$ on 
$H_d$-structured manifolds generalising orientation reversal. 
The construction goes roughly as follows (for details we refer to Section~\ref{Sec:Z2}): From the principal $H_d$-bundle $P$ we can form an
associated $\widehat{H}_d$-bundle $\widehat{P}$ which contains $P$ as a natural 
submanifold. 
We can construct a new principal $H_d$-bundle as $\overline{P} \coloneqq \widehat{P}\setminus P $. In addition the identification $\overline{P}\times_{\rho } \R^d \cong TM $ is constructed from the old identification and a chosen reflection in $\R^d$ (we choose the reflection along the $(e_1=0)$-plane). The involution $\overline{( \cdot )}$ gives rise to a $\Z_2$-action on $\Bord_d^{H_d, \rho}$ generalising the orientation-reversal action for $H = SO_d$. 

A reflection structure on a field theory is the requirement that `orientation reversal' corresponds to complex conjugation. 
Mathematically, a topological field theory with reflection structure is defined as a $\Z_2$-equivariant functor $\mathcal{Z} \colon \Bord_d^{H_d,\rho}\to \sVect$ with respect to the orientation reversal action on $\Bord_d^{H_d,\rho}$ and complex conjugation on $\sVect$~\cite{freedhopkins}. 
Note that being equivariant is a structure and not a property. 
A consequence of the definition is that the state space at every object is equipped with a hermitian inner product~\cite{freedhopkins}. 
A reflection structure is called \emph{positive} if all these products are positive.         

The spin statistic relation asserts that the transformation of a particle
under the element $c=(-1)\in \Spin_d$ determines its statistics, i.e. it acts as the identity on bosonic particles and by multiplication with $-1$ on fermionic particles. The type of particle is encoded by the grading of the super vector 
spaces: For a super vector space $V=V_0\oplus V_1$ we think of the elements of
$V_0$ as bosonic and $V_1$ as fermionic. The spin statistics connection is now 
the condition that $c$ acts on the space $V$ by the grading operator $(-1)^F_V$.
The action by $c$ defines an automorphism $c_\Sigma$ of every object $\Sigma \in \Bord_d^{\Spin_d}$. A topological field theory $\mathcal{Z}$ \emph{satisfies the spin statistics relation}~\cite{theospinstatistics} if $\mathcal{Z}(c_\Sigma)=(-1)^F_{\mathcal{Z}(\Sigma)}$. To extend this definition
to general spacetime structure groups $H_d$ constructed from an internal symmetry group $G$, it is enough to replace $c\in \Spin_d$ with an analogous element $c\in H_d$. We refer to Section~\ref{Sec:spacetime group} for its definition. We can reformulate this definition in a way more similar to the definition of a reflection structure: For this note that $(-1)^F_{-}$ induces a natural transformation from the identity on $\sVect$ to itself and hence an action of the 2-group $B\Z_2$ on $\sVect$. Similarly, acting with $c$ induces an action
of $B\Z_2$ on $\Bord_d^{H_d,\rho}$. The spin statistics relation is equivalent
to $\mathcal{Z}\colon \Bord_d^{H_d,\rho}\to \sVect$ being $B\Z_2$-equivariant. 

The $B\Z_2$-action encoding spin statistics and the $\Z_2$-action encoding reflection structures combine into a $\Z_2\times B\Z_2$-action on $\Bord_d^{H_d,\rho}$. 
The action extends to the fully local bordism category
$\Bord_{d,0}^{H_d,\rho}$ (this follows from the cobordism hypothesis as we explain in Section~\ref{Sec:Z2}). 
Furthermore, the $\Z_2\times B\Z_2$-action on $\sVect$ induces a $\Z_2\times B\Z_2$-action on the symmetric monoidal bicategory $\sAlg$ (see Appendix~\ref{App: Super alg}). 
Hence we can define extended field theories 
with reflection structure and spin statistics connection as $\Z_2\times B\Z_2$-equivariant functors. 
This paper is concerned with the classification and study of such field theories 
in low dimensions.  
We summarize our results in the next section.

\subsection{Main results}
Our main result is a classification of $\Z_2\times B\Z_2$-equivariant functors for an arbitrary internal
symmetry group $G$ in spacetime dimension 1 and 2.
One surprising outcome of our work is that topological field theories with
structure group $H_d$, reflection structure and spin statistics, associated to
an internal symmetry group $G$ admit a good description in terms of $G$ rather then
$H_d$ for $d$ equal to 1 and 2. 

To appreciate our result more we first describe the classification of 1-dimensional field theories without additional structure.  
In dimension one $\O_1=\Z_2$ and hence the tangential structure is fixed by a
Lie group homomorphism $\rho\colon H_1\to \Z_2$. The classification will only depend on
the map of discrete groups $\pi_0(\rho)\colon \pi_0(H_1)\to \Z_2$ and hence we will assume from 
now on (without loss of generality) that $G$ and hence $H_1$ is discrete. 
A 1-dimensional topological field theory $\mathcal{Z}\colon \Bord_1^{H_1,\rho}\to \sVect$ is completely described by a finite dimensional super vector space $V$ the value of $\mathcal{Z}$ on a point together with a representation of the even part of $(H_1)_{\operatorname{ev}}=\rho^{-1}(0)$ and a non-degenerate bilinear form $_{h}\langle -,- \rangle \colon V\otimes V\to \C $ for all odd elements of $H_1$ satisfying a bunch of not really enlightening relations (see Proposition~\ref{Prop: 1D TFT}). 

On the condensed matter side this is supposed
to correspond to a zero dimensional lattice system with internal $G$-symmetry. 
These are usually described by one Hilbert space $\mathcal{H}$, together with a Hamiltonian $H$, and an action of $G$ such that its even elements act unitarily
and its odd elements act antiunitaritly. 
We can weaken the condition of being
a Hilbert space to a hermitian vector space to also describe non-unitary systems\footnote{Non-Hermitian systems which have Hamiltonians that are not self adjoint have also found applications in condensed matter, but will not be considered here.}.
The low energy effective topological field theory of a gapped system should now be described by the
ground state of $H$, which still has an action by $G$. 
We can abstract the structure of the action in the following definition: 
A \emph{unitary fermionic representation} on a hermitian super vector space 
$V$ is a representation of $G$ on $V$ where even elements act by unitary maps
$\rho(g)\colon V \to V$,
odd elements act via $\C$-antiunitary maps $\rho(g)\colon \overline{V} \to V$, and $\rho(c)=(-1)^F_V$.  
However, this seems at first glance quite different from the data describing a $H_1$-topological field theory.  

This mismatch completely disappears when theories with reflection structure and 
spin statistics connection are considered. We prove in Section~\ref{Sec:1D} the
following classification result:

\begin{proposition}[Proposition~\ref{prop:1d classification}]
	Let $G$ be a fermionic group. The groupoid of 1-dimensional reflection and spin-statistics field 
	theories with internal symmetry group $G$ is equivalent to the core of the category of unitary fermionic
	representations of the fermionic group $\pi_0(G)$.  
\end{proposition}
Let us stress again that these are $\Z_2\times B\Z_2$ equivariant functors $\Bord_1^{H_1,\rho}\to \sVect$.
We will comment on the proof in the next part of the introduction.
This classification result is satisfying both from a mathematical as well as physical perspective. Mathematically, our result is appealing because it describes 1-dimensional
field theories in terms of the natural notion of representations of fermionic groups, compared to the slightly ugly classification in terms of a collection of 
bilinear inner products. 
Physically it is appealing because it recovers exactly the answer which is 
expected by considering 0-dimensional lattice systems. 

Section~\ref{Sec:2D} is concerned with the classification of fully extended 2-dimensional topological field theories with reflection structure and spin statistics connection with values in super algebras. 
These are $\Z_2\times B\Z_2$-equivariant symmetric monoidal functors
\begin{align}
\mathcal{Z}\colon \Bord_{2,0}^{H_2}\to \sAlg 
\end{align}   
between symmetric monoidal bicategories. We state the classification result as 
Theorem~\ref{th:main}. Instead of stating the result in full we just indicate
the structure we find. Similar to what happens in one dimension the answer will
only depend on low dimensional homotopical data corresponding to $G$. 
More concretely, it only depends on the fundamental groupoid $\Pi_1(G)$. 
This is a fermionic 2-group. We define them in more detail in Section~\ref{Sec:fermskeletal}. 
Roughly, it is a $\Z_2$-graded monoidal groupoid $(G,\otimes, 1,\alpha)$ with a canonical central element $c\in G$ which squares to $1$. 
We denote the grading by $\theta$. 
We assume for the sake of this introduction that the associator $\alpha$ and some of the other coherence isomorphisms featuring in the formal definition are trivial.
For any fermionic 2-group $G$, a \emph{$G$-graded algebra} consists 
the following data
\begin{itemize}
	\item A complex super vector space $\mathcal{A}=\bigoplus_{g\in G} A_g$
	\item with structure of a $G$-graded superalgebra over the real numbers such that $a_g\cdot i = (-1)^{\theta(g)}ia_g$ for all $a_g\in A_g$
	\item For every morphism $\gamma\colon g\to g'$ in $G$ an isomorphism $a_\gamma \colon A_g\to A_{g'}$
\end{itemize}
which has to satisfy a list of natural conditions. In particular, the component
$A_c$ is required to be generated as an $A_e$-module by one object $(-1)^F$ squaring to $1$ and satisfying $(-1)^Fa_g =(-1)^{|a_g|}a_g (-1)^F$. 
We discuss the details in Section~\ref{Sec:fermgraded}. 
A \emph{strongly} $G$-graded algebra is a $G$-graded algebra such that the multiplication induces an isomorphism $A_g\otimes_{A_e} A_g'\to A_{gg'}$. 
Our main result classifies 2-dimensional reflection and spin statistics topological field theories with structure group $H_2$ in terms of strongly $G$-graded Frobenius algebras which are additionally equipped with a generalization of a hermitian pairing, which is called a stellar algebra structure.
These are roughly speaking Morita invariant versions of $*$-algebras.  
They have been introduced in a linear setting in~\cite{schommerpriesthesis}, whereas we use the $\C$-antilinear version of the concept. 
For more details we refer to Section~\ref{Sec:stellar}. Our main theorem is now
\begin{theorem}[Theorem~\ref{th:main}]
	Let $G$ be a fermionic group. The 2-groupoid of 2-dimensional reflection and spin-statistics field 
	theories with internal symmetry group $G$ is equivalent to the core of the bicategory of strongly $\Pi_1(G)$-graded stellar Frobenius algebras.  
\end{theorem}     
 
\subsection{The cobordism hypothesis and general structure of the proof of the main result}   
Both the proof in dimension 1 and 2 follow the same line of reasoning and are based
on the cobordism hypothesis~\cite{baezdolan,lurietft}. The cobordism hypothesis is
a classification result for topological field theories with arbitrary target and 
tangential structure. We recall the concrete statement in Section~\ref{Sec:TFT}, but roughly it identifies fully extended d-dimensional framed field theories 
with target $\mathcal{S}$ with the full subgroupoid of $d$-dualisable objects $\mathcal{S}^{\operatorname{f.d.}}$. 
This implies that there is an $O_d$-action on $\mathcal{S}^{\operatorname{f.d.}}$ and topological field theories with
tangential structure $H\to \O_d$ are classified by homotopy fixed points for the induced $H$-action. 
A super algebra is 2-dualisable if and only if it is finite-dimensional and semi-simple. 
We describe the $O_2$-action in detail in Appendix \ref{Sec:salgdual}.      

The proof of our main theorem proceeds in the following steps
\begin{itemize}
	\item We use the cobordism hypothesis to identify $H_2$-topological field
	theories with the category of homotopy $H_2$ fixed points $H_2\operatorname{-TFT}\cong (\sAlg^{\operatorname{f.d.}})^{H_2}$. 
	$\Z_2\times B\Z_2$-equivariant functors can be identified with homotopy fixed
	points for the conjugation action of $\Z_2\times B\Z_2$ on the functor 
	category $\Bord_{2,0}^{H_2}\to \sAlg$. This means we want to compute 
	the bicategory of homotopy fixed points $(H_2\operatorname{-TFT})^{\Z_2\times B\Z_2} \cong ((\sAlg^{\operatorname{f.d.}})^{H_2})^{\Z_2\times B\Z_2} $. 
	\item Since we work with bicategories and compact Lie groups we can replace the actions of topological groups by actions of 2-groups. In Appendix~\ref{App: 2-Group} we collect some details on 2-groups. 
	In particular, we prove that iterative homotopy fixed points $(H_2\operatorname{-TFT})^{\Z_2\times B\Z_2} \cong ((\sAlg^{f.d.})^{H_2})^{\Z_2\times B\Z_2} $ are equivalent to fixed points for an action of the 2-group $H_2\rtimes (\Z_2\times B\Z_2)$ constructed as a twisted semi-direct product of 2-groups which sits in a short exact sequence of 2-groups 
	\begin{align}
	1 \to H_2\to H_2\rtimes (\Z_2\times B\Z_2) \to \Z_2\times B\Z_2	\to 1 \ \ . 
	\end{align}      
	\item The 2-group $H_2\rtimes (\Z_2\times B\Z_2)$ is non-trivially isomorphic to $O_2\times G_b$. This implies that if we want to compute $H_2\rtimes (\Z_2\times B\Z_2)$ fixed points we can also
	compute $G_b$-fixed points in $O_2$-fixed points. Here it is important to
	note that the $O_2$-action differs from the action appearing in the cobordism
	hypothesis. 
	\item In Section~\ref{Sec:O2 FP} we identify the bicategory of $O_2$-fixed
	points in $\sAlg^{\operatorname{f.d.}}$ with the bigroupoid of (antilinear) stellar algebras. 
	\item In Theorem~\ref{th:actionstfrob} we show that the induced $G_b$-action on stellar algebras takes a simple form. 
	\item In Theorem~\ref{th:main} we compute fixed points for this $G_b$-action finishing the proof of the classification.    
\end{itemize}
The 1-dimensional proof follows essentially the same line of reasoning, but at 
the level of 1-categories. This reduces the technical difficulties significantly. We present the 1-categorical proof in Section~\ref{Sec:1D}. 

Our results rely on the cobordism hypothesis for which no full published proof exists (however see the preprint~\cite{Grady2021TheGC}). The cobordism hypothesis is proven in dimension 1~\cite{Harpaz2012TheCH}. Hence, all our one dimensional results are theorems. 
In the bicategorical 2-dimensional setting the cobordism hypothesis has been proven for many tangential structures including framings~\cite{piotr}, orientations~\cite{schommerpriesthesis}, $G\times \SO_2$-structures~\cite{Sozer} and $r$-spin structures~\cite{LorantNils}. Our approach to the problem is 
to define the bicategory $\Bord_{2,0}^{H_2}$ through the universal property it satisfies by the cobordism hypothesis. Hence for the bordism bicategories 
considered in this paper the cobordism hypothesis holds by construction. However,
this comes at the cost that except in the cases mentioned above we cannot be sure
without proving the cobordism hypothesis that the bordism category used in this
paper agrees with the one constructed in~\cite{schommerpriesthesis,calaquescheimbauer}. 

\subsection{Outlook}

We conclude the introduction with a few speculations about possible extensions and 
consequences of our work.  

\subsubsection*{Comments on extended reflection positivity}
\label{Sec:positivity}
The work of Freed and Hopkins~\cite{freedhopkins} defines reflection positivity only for non-extended
field theories and fully extended invertible field theories. It is an interesting
question how to extend the definition to extended non-invertible theories. Our work suggests the following answer for once extended field theories: A consequence
of our work should be that the value of any once extended topological field theory
$\mathcal{Z}\colon \Bord_{d,d-2}^{H_d}\to \sAlg$ with reflection structure on a $d-2$ dimensional manifold $S$ is canonically a stellar algebra (this should follow via dimensional reduction). 
One source of stellar algebras are finite dimensional $C^*$-algebras. 
We suggest to define a \emph{positivity structure} on a given stellar algebra $A$
to be an isomorphism to a stellar algebra arising from a fixed $C^*$-algebra. 
Furthermore, the bimodules assigned by $\mathcal{Z}$ to cobordisms have the 
structure of stellar bimodules. 
Under the identifications with $C^*$-algebras these stellar bimodules come equipped with the structure of a $C^*$-algebra-valued hermitian inner product.
We can ask these bimodules to be positive, i.e. to be Hilbert $C^*$-bimodules.
To us it seems natural to now define a \emph{positivity structure} on
a once extended topological field theory with reflection structure as a coherent choice of positivity structures on all the values of the field theory, such that the bimodules associated to all bordisms are positive. 
One important subtlety in this definition is that there exist $*$-algebras that are not $C^*$, but are still isomorphic to a $C^*$-algebra as a stellar algebra, see Example \ref{ex:C*notmoritainv}.
A careful development of these ideas will most like also involve a better understanding `hermitian pairings' on objects in a bicategory and dagger bicategories.     

\subsubsection*{Smooth setting} 

Our classification only depends on the homotopy type of the classifying space of the internal symmetry group $G$ (even its 2-truncation). 
If one wants to access the geometric structure of $G$, one should work in the
setting of smooth field theories~\cite{Stolz2011SupersymmetricFT,BerwickEvans2015SmoothOT, Ludewig2020AFF}.
There the bordism category is in addition a smooth category, i.e.\ a sheaf of 
symmetric monoidal (higher) categories on the site of manifolds. 
We expect the
analogous smooth 1-dimensional classification to be in terms of smooth unitary fermionic representations of $G$ instead of $\pi_0(G)$. 
If one in addition includes a Riemannian metric we expect that the representations
are also allowed to be infinite dimensional and that one has to specify additional
data in terms of a Hamiltonian which commutes with the action of $G$. 

The concept of a $G$-graded algebra also has a smooth analogue, which we expect
to be related to 2-dimensional smooth field theories. 
It is given by a super vector bundle $\mathcal{A}\to G$ over $\C$ together with vector bundle
maps 
\begin{align}
\mu_A \colon \pr_1^*\mathcal{A}\otimes \pr_2^*\mathcal{A} \to \mu_G^* \mathcal{A} 
\end{align} 
 where $\mu_G$ is the multiplication in $G$, which is associative and admits a 
 unit. 
 It is straightforward to adapt this to fermionicaly graded algebras by requiring
 the fibers over odd points to anti-commute with $i$ and $\mathcal{A}_c=(\mathcal{A}_e)_{(-1)^F}$. 
 In addition we expect that $\mathcal{A}$ should be equipped with a connection which is compatible with all the other structures.  
 It would be interesting to define and study ``smooth fermionically graded algebras" as (higher) geometric objects. 
   
\subsubsection*{Categorification to fusion categories}
It would be interesting to extend our classification result to higher dimensions.
Unfortunately, with the current knowledge and tools this seems to be out of reach.  
There is ongoing work by Douglas, Schommer-Pries, and  Synder started in~\cite{douglasSPsnyder} 
which is supposed to be a first step into this direction.  
However, it seems to us that one should be able to come up with an educated guess for a partial answer and give explicit constructions of the 3-dimensional field 
theories. The answer should be an appropriate categorification of the structure
we found for 2-dimensional topological field theories.
A good replacement for $C^*$-algebras seem to be given by unitary 
super fusion categories. 
These allow for the construction of reflection positive oriented 3-dimensional field theories $\Bord_3^{\SO_3}\to \sVect$. 
In the discrete oriented bosonic case a good algebraic input seems to be 
that of a spherical $G$-graded fusion category~\cite{turaev20123}. This is a $G$-graded category
\begin{align}
	\cat{C}_G= \bigoplus_{g\in G} \cat{C}_g
\end{align}  
together with a monoidal product $\cat{C}_g\boxtimes \cat{C}_{g'} \to \cat{C}_{gg'}$. We think of this as a categorification of a graded algebra. 
The symmetric Frobenius structure is usually replaced by a pivotal structure for $\cat{C}$ which is required to be spherical.\footnote{Note that this is not the most straightforward categorification of a symmetric Frobenius algebra, which would be a cyclic associative algebra in linear categories. For a detailed discussion of cyclic associative algebras we refer to~\cite{MW}.} 
Recently, there has also appeared the definition of spherical fusion category graded 
by a 2-group in the literature together with a construction of the corresponding 3-dimensional topological field theory~\cite{Sozer2022MonoidalCG}.   

A pivotal structure on a fusion category $\cat{C}$ is given by a 
trivialisation of the monoidal functor $(-)^{\vee\vee}\colon \cat{C}\to \cat{C}$.
A result of~\cite{douglasSPsnyder} shows that the bimodule corresponding to $(-)^{\vee\vee}$ is the Serre automorphism in the 3-category of fusion categories. 
A pivotal structure hence gives rise to a trivialisation of the Serre automorphism.
Our result suggest that in situations where fermions are present instead of trivializing $(-)^{\vee\vee}$ we should identify it with an automorphism 
$(-1)^F\colon \cat{C}\to \cat{C}$. 
The functor $(-1)^F$ is the identity on objects and even morphisms and multiplication by $-1$ on odd morphisms. 
These considerations lead exactly to the notion of a super pivotal fusion category from~\cite{Aasen:2017ubm}, namely a $\sVect$-enriched fusion category together
with a monoidal natural isomorphism $(-)^{\vee \vee}\Longrightarrow (-1)^F$.
Following~\cite{douglasSPsnyder} it seems natural to call such a pivotal structure \emph{spherical} if the induced trivialisation of $(-)^{\vee \vee \vee \vee}$ agrees with the Radford isomorphism.
Time reversing symmetries might be encoded by the condition that the functor corresponding to tensoring with an odd element is $\C$-anti linear.

Based on these observations, we expect that there is an interesting concept of a unitary spherical fusion category graded by a fermionic 2-group which reduces to the 
examples above, plays an important role in 2+1 dimensional symmetry protected phases, and allows for the construction of 3-dimensional topological field theories with structure group $H_3$. We leave a precise definition and construction of the field theory for further work.    

\subsubsection*{Arbitrary dimensions}

The main strategy used in this paper to compute topological field theories with fermionic symmetries, reflection structure and spin-statitstics is amenable to generalizations to arbitrary dimensions.
To illustrate this, suppose $\mathcal{C}$ is a symmetric monoidal $(\infty,d)$-category with symmetric monoidal $\Z_2 \times B\Z_2$-action.
Let $T = \mathcal{C}^{\fd}$ be the space of framed TFTs with target $\mathcal{C}$.
By the cobordism hypothesis, $T$ comes equipped with a homotopy $O_d \times \Z_2 \times B\Z_2$-action.
Note that we now have two actions of $O_d$ on $T$; the restriction of the above to $O_d$ and the restriction to the second factor after the map $O_d \to O_\infty \to \pi_{\leq 1} O_\infty \cong \Z_2 \times B\Z_2$.
There is also the corresponding `diagonal action' which we will call the stellar action:
\[
BO_d \xrightarrow{\Delta} BO_d \times BO_d \xrightarrow{\id \times (w_1,w_2)} BO_d \times B\Z_2 \times B^2\Z_2.
\]

Let $G_b$ be a Lie group and $(\theta, \omega): BG_b \to B\Z_2 \times B^2 \Z_2$ a map\footnote{$BG_b$ can be replaced by any space, but we want to agree with the notation in the main text.}.
Define a space $B\hat{H}_d$ by the homotopy fiber
\[
 B\hat{H}_d \to BG_b \times BO_d \xrightarrow{w_2 + \omega + \theta \cup w_1} B^2\Z_2
\]
and a space $BH_d$ by its further fiber
\[
BH_d \to BG_b \times BO_d \xrightarrow{(w_1 + \theta, w_2 + \omega + \theta \cup w_1)} B\Z_2 \times B^2\Z_2.
\]
This is consistent with the definition of the groups $\hat{H}_d$ and $H_d$ in Section \ref{Sec:spacetime group} as the fermionic tensor product and its even part respectively.
We therefore obtain a homotopy action of $B\Z_2 \times B^2 \Z_2$ on $BH_d$, which agrees with the action used in the main text for $d = 1,2$.
We now define an action of the middle factor $BG_b \times BO_d$ on $T$ by letting $G_b$ act through $\Z_2 \times B\Z_2$ and $O_d$ by the stellar action.
This implies the restricted action of $H_d$ is given by the projection to $BO_d$ followed by the action of the cobordism hypothesis. 
Hence $H_d$-fixed points of this restricted action compute $H_d$-TFTs.
The space of TFTs with reflection structure and spin-statistics is then the iterated fixed point 
\[
(T^{h H_d})^{h(\Z_2 \times B\Z_2)} \cong T^{h(G_b \times O_d)} \cong (T^{hO_d})^{hG_b}.
\]
This paper works this out explicitly for the cases $d = 1,2$ where we found that $T^{hO_d}$ were Hermitian super vector spaces and stellar Frobenius algebras, respectively.

\subsubsection*{Conventions}
Here we summarise our notation and conventions for the convenience of the reader. 
Linear categories such as $\sVect, \sAlg$ etc are by default over $\C$.
An $(A,B)$-bimodule $M$ is regarded as a $1$-morphism $B \to A$ in $\sAlg$ so that composition of bimodules is given by $N \circ M := N \otimes_B M$.
If $\mathcal{C}$ is a bicategory, $\mathcal{C}^{\fd}$ denotes the maximal sub-bigroupoid on the fully dualizable objects.
The group $\Z_2=\Z/2\Z$ appears in many different situations throughout the paper.
Therefore we decided to adopt the physics convention to label the groups and generators by specific letters:
\begin{itemize}
\item the group $\Z_2^c$ is generated by an element $c$ of which the $+1$-eigenspaces are particles with integer spin and $-1$-eigenvectors have half-integer spin;
\item the group $\Z_2^F$ is generated by an element $(-1)^F$ given by the fermion parity, which is identified with $c$ in case spin-statistics holds; 
\item the group $\Z_2^R$ is generated by an element $R$ implementing the reflection action $A \mapsto \overline{A}^{\op}$. 
So fixed points for this action are typically vector spaces with Hermitian form, stellar algebras or TFTs with reflection structure;
\item the group $O_1 \subseteq O_2$ is generated by a single reflection in $\R^2$ denoted $s$ or simply $(-)$, which acts by the dual through the cobordism hypothesis;
\item the group $\Z_2^B$ is generated by an element $B$ implementing complex conjugation (`bar') so that $s = RB$;
\item the group $\Z_2^T$ denotes a time-reversal $T$ of square one, which depending on context either acts as $\Z_2^B$ or as $O_1$ (and these actions are identified in case a reflection structure is present). We also use the notation $\Z_4^T$ for a time reversal with square $c$.
\end{itemize}

\subparagraph{Acknowledgements.} We are grateful to Stephan Stolz and Peter Teichner for introducing us to many concepts appearing in this article and many discussions on reflection structures which had a great influence on this manuscript.
We thank
Bertram Arnold, Nils Carqueville, Theo Johnson-Freyd,
David Reutter, Chris Schommer-Pries, and Lukas Woike
for insightful discussions
and helpful comments.	
Both authors gratefully acknowledge support by the Max Planck Institute for Mathematics in
Bonn.
LM gratefully acknowledges support by the Simons Collaboration on Global Categorical Symmetries.
Research at Perimeter Institute is supported in part by the Government of Canada through the Department of Innovation, Science and Economic Development and by the Province of Ontario through the Ministry of Colleges and Universities. The Perimeter Institute is in the Haldimand Tract, land promised to the Six Nations.

\section{Topological field theories}\label{Sec:TFT} 
Topological field theories as defined by Atiyah~\cite{atiyah1988topological} are a mathematical axiomatisation of the formal
properties of the path integral formulation of metric independent quantum field theories under cutting and gluing of spacetime manifolds.
In this section we will briefly review the concept of topological field theories and their classification
through the cobordism hypothesis. After sketching the general case we will specialise to dimension two, which will be the most relevant for our paper.

\subsection{Review of the definition} 
Informally speaking a (oriented) $d$-dimensional 
topological field theory $\mathcal{Z}$ is: 
\begin{itemize}
\item 
an assignment of vector space to all closed $d-1$-dimensional manifolds $\Sigma$; the state spaces
of the theory. 
\item Furthermore, $\mathcal{Z}$ assigns to every bordism $M\colon \Sigma_1 \to \Sigma_2$ (i.e. a compact
manifold $M$ with boundary and an identification of its boundary $\partial M$ with $ -\Sigma_1\sqcup
\Sigma_2 $) a linear map $\mathcal{Z}(M)\colon \mathcal{Z}(\Sigma_1) \to \mathcal{Z}(\Sigma_2)$; the `time evolution operator'. 
\end{itemize}
Here $- \Sigma$ denotes the manifold $\Sigma$ with reversed orientation.
Gluing bordisms along boundaries is required to correspond to the composition of linear maps. Physical systems
stack by taking the tensor product of their state spaces, which is implemented in the definition 
by requiring $\mathcal{Z}(\Sigma_1 \sqcup \Sigma_2) \cong \mathcal{Z}(\Sigma_1)\otimes \mathcal{Z}(\Sigma_2)$. 

Before giving the formal definition we need to be more careful about the type of manifolds we consider. For 
example in the present of a time reversal symmetry we should work with unoriented manifolds and if there
are fermions all manifolds should be equipped with a $\Spin$-structure. A convenient way of treating
all types of possible structures relevant for us are tangential structures. 
\begin{definition}\label{Def: Tangential structure}
Let $H$ be a Lie group, $\rho \colon H\to \O_d$ a Lie group homomorphism and $M$ a $d$-dimensional 
manifold. A \emph{tangential $(H,\rho)$-structure} on $M$ is a homotopy commutative diagram 
\begin{equation}
\begin{tikzcd}
 \ar[rd, Rightarrow, "h" near end, shorten <= 25] &  BH \ar[d, "\rho"] \\ 
 M \ar[r, "TM", swap] \ar[ru, "\psi", bend left=10] & B\O_d	
\end{tikzcd}
\end{equation} 
where $TM\colon M \to B\O_d \sim BGL_d(\R)$ is the classifying map for the tangent bundle of $M$. 
The \emph{space of tangential structures} is the (derived) mapping space from $TM\colon M\to B\O_d$ to
$\rho \colon BG\to B\O_d$ in the category $\Top/B\O_d$ of topological spaces over $B\O_d$.  
A path in this 
space is a \emph{morphism of tangential structures}.
\end{definition} 
\begin{remark}\label{Rem: Tangential structures in terms of bundles}
There is an equivalent definition in terms of differential geometry. A tangential $(H,\rho)$-structure on 
$M$ can equivalently be described by a principal $H$-bundle $P\to M$ together with an isomorphism 
$P\times_\rho \R^d \cong TM $ of vector bundles on $M$. When working with the geometric definition, a morphism of tangential structures would be a gauge transformation $P\to P'$ 
such that the composition
\[
P \times_\rho \R^d \to P' \times_\rho \R^d \to TM
\]
agrees with the other identification with $TM$. However, this does not agree with morphisms as defined in 
Definition~\ref{Def: Tangential structure}, which is more homotopical in nature.  
Instead, we require a vector bundle homotopy between the two compositions $P \times_\rho \R^d \to TM$.
For example, a $O_d$-structure on a manifold is equivalent to a metric on it.
Geometrically, two $O_d$-structures are isomorphic if and only if they are isometric. 
However, topologically all $O_d$-structures are isomorphic, because the map $BO_d \to BGL_d(\R)$ is a homotopy equivalence.

\end{remark}     
\begin{example}
In this paper we are mostly interested in the structure groups $H_d$ associated to an internal symmetry group defined in Section~\ref{Sec:spacetime group}. Some standard examples of 
tangential structures appearing as special examples for $H_d$ are
\begin{itemize}
	\item A manifold equipped with a tangential structure for the map $*\to \O_d$ is the
	same as a framed manifold, i.e. a manifold equipped with a trivialisation of its tangent bundle. 
	\item A manifold equipped with a tangential structure for the identity $\O_d\to \O_d$ is the
	same as an unoriented manifold. 
	\item A manifold equipped with a tangential structure for the inclusion $\SO_d\to \O_d$ is 
	the same as an oriented manifold.
	\item A manifold equipped with a tangential structure for $\Spin_d \to \SO_d\to \O_d$ is 
	the same as a spin manifold. 
	\item Let $G$ be a Lie group. A manifold equipped with a tangential structure for the map $G\times \SO_d \xrightarrow{\pr_{\SO_d}} \SO_d\to \O_d$ is the same as an oriented manifold equipped with a principal $G$-bundle. 
\end{itemize} 
\end{example}
To extend Atiyah's definition of a topological field theory to manifolds with 
$(H,\rho)$-structures one defines a tangential structure on a $d-1$-dimensional manifold $\Sigma$ by a $H$-structure on the once stabilized tangent bundle.
In other words, we postcompose the map classifying the tangent bundle with the map $B\O_{d-1}\to B\O_d$, i.e. as a homotopy commutative diagram 
\begin{equation}
\begin{tikzcd}
\ar[rrd, Rightarrow, "h" near end, shorten <= 55 , shorten >= 5] &  & BG \ar[d, "\rho"] \\ 
M \ar[r, "TM", swap] \ar[rru, "\psi", bend left=10] & B\O_{d-1} \ar[r] & B\O_d	
\end{tikzcd}
\end{equation}   
Topological field theories have a concise definition using the language of category theory. 
For this one introduces a symmetric monoidal category $\Bord^{(H,\rho)}_d$~\cite{lurietft, calaquescheimbauer} where the objects are $d-1$-dimensional
closed manifold $\Sigma$ equipped with a tangential $(H,\rho)$-structure and morphism are
diffeomorphism classes of bordisms with tangential structure. The composition is defined by
gluing manifolds along their boundaries and the monoidal structure is given by the disjoint union of manifolds. 

The definition of a topological field theory makes sense in any symmetric monoidal category 
$\cat{C}$. 
\begin{definition}
Let $(H,\rho)$ be a tangential $d$-type and $\cat{C}$ a symmetric monoidal category.  
A \emph{topological $(H,\rho)$-field theory with values in $\cat{C}$} is a symmetric monoidal
functor 
\begin{align}
Z \colon \Bord_d^{(H,\rho)} \to \cat{C} \ \ .
\end{align}
\end{definition} 
For physical applications the category of (super) vector spaces is the most important
example to which we will now restrict our attention.
The value on the object given by the empty manifold of a topological field theory $\mathcal{Z}\colon \Bord_d^{(H,\rho)} \to \sVect$ with target complex super vector spaces comes with a preferred isomorphism to the complex numbers $\mathcal{Z}(\emptyset)\cong \C $, because $\mathcal{Z}$ is symmetric monoidal. A closed 
$(H,\rho)$-manifold $M$ defines a bordism from the empty set to itself. The value of $\mathcal{Z}$ on 
$M$ is a linear map $\mathcal{Z}(M)\colon \mathcal{Z}(\emptyset)\cong \C \to \C \cong \mathcal{Z}(\emptyset)$ and hence can 
canonically be identified with a complex number; the partition function of the field theory
on $M$.    
 
In the decades after Atiyah's work this definition has been extend in at least two important 
ways~\cite{lurietft}. The first one is not going to play an important role in this paper and hence
we only mention it in passing: In the definition of $\Bord_d^{(H,\rho)}$ we have identified
bordisms which are diffeomorphic relative to boundary. It is more natural to keep track of 
these diffeomorphisms and their isotopies in terms of higher morphisms. This leads to an $(\infty,1)$-category of bordisms whose homotopy category is $\Bord_d^{(H,\rho)}$.

The second one is of crucial importance to this paper: The definition of a topological field theory allows us to compute the partition function
by cutting spacetime manifolds into simpler pieces. This can be understood as an 
implementation of `locality'. However, the $d-1$-dimensional manifolds appearing as the objects of $\Bord^{(H,\rho)}_d$ cannot be decomposed. Allowing the decomposition of lower dimensional manifolds leads to a symmetric monoidal $d$-category of bordisms $\Bord^{(H,\rho)}_{d,0}$ with informally speaking $0$-dimensional closed manifolds with $(H,\rho)$-structure as objects, 1-dimensional 
bordisms between those as 1-morphisms, bordisms of bordisms (certain 2-dimensional manifolds with corners) as 2-morphisms, and so on up to $d$-dimensional manifolds with corners as $d$-morphisms. Composition is defined via gluing and the symmetric monoidal structure is defined by disjoint union. Defining the higher category $\Bord^{(H,\rho)}_{d,0}$ rigorously is quite involved and we refer to~\cite{calaquescheimbauer} for a definition using complete $d$-fold Segal space and to~\cite{schommerpriesthesis} for a definition of the symmetric monoidal bicategory $\Bord^{(H,\rho)}_{2,0}$. 
We will not need a detailed definition of the fully extended bordism 
category, since we assume the cobordism hypothesis explained in more detail in the 
next section.
This provides us with a description of the bordism category in terms of a universal property. 
\begin{definition}
Let $\cat{C}$ be a symmetric monoidal $d$-category. A \emph{fully extended $d$-dimensional
topological $(H,\rho)$-field theory with values in $\cat{C}$} is a symmetric monoidal functor
\begin{align}
\mathcal{Z}\colon \Bord^{(H,\rho)}_{d,0} \to \cat{C} . 
\end{align}
\end{definition}        
For any symmetric monoidal $d$-category $\cat{C}$ one can define a symmetric monoidal 
$(d-1)$-category $\Omega \cat{C} \coloneqq \End_{\cat{C}} (1)$ as the category of 
endomorphism of the symmetric monoidal unit of $\cat{C}$. Iteratively we can also define 
$\Omega^k \cat{C}$ for all $0<k<d$. We introduce the notation $\Bord^{(H,\rho)}_{d,k}\coloneqq \Omega^k \Bord^{(H,\rho)}_{d,0}$. Note that $\Bord^{(H,\rho)}_{d,d-1}$ is equal to $\Bord^{(H,\rho)}_{d}$. For $0<n<d$ a \emph{$n+1$-layered topological field theory with values in a symmetric monoidal $n$-category $\cat{D}$} is a symmetric monoidal functor between $n$-categories $\Bord^{(H,\rho)}_{d,d-n}\to \cat{D}$. 
Any fully extended field theory $\mathcal{Z}\colon \Bord^{(H,\rho)}_{d,0} \to \cat{C}$ can be restricted to a partially extended field theories
$\mathcal{Z}\colon \Bord^{(H,\rho)}_{d,k} \to \Omega^k \cat{C}$ for $0 \leq k \leq d$. We will be mostly interested in fully extended 2-dimensional field theories and hence have to pick a symmetric monoidal 2-category as target.\footnote{In this document, we will use the terms $2$-category and bicategory interchangably for the weak notion.}
 We will exclusively work with $\cat{C}=\sAlg$ the symmetric monoidal $2$-category of super algebras, bimodules, and even bimodule intertwiners. In Appendix~\ref{App: Super alg} we provide the needed details on this bicategory. Note that $\Omega \sAlg = \sVect$ is the category of super vector spaces.      

\subsection{The cobordism hypothesis}
The cobordism hypothesis~\cite{lurietft, baezdolan} is a classification statement for fully extended topological field
theories. Even though there does not exist a complete published proof of it in the literature, it has 
been proven in various special cases~\cite{piotr, schommerpriesthesis, HSV, Sozer} and there are detailed outlines of proofs~\cite{lurietft,AFCH}. 
The formulation of it which is most useful for us requires a small generalization of tangential structures 
from group homomorphisms $H\to \O_d$ to arbitrary topological spaces over $B\O_d$. For a space $\rho 
\colon X\to B\O_d$ over $B\O_d$ the construction of the bordism category $\Bord_{d,0}^{(X,\rho)}$ works 
exactly as in the previous section. The cobordism hypothesis is equivalent to the following two properties 
of the functor $\Bord_{d,0}^{-}\colon \Top/ BO_d \to d\text{-}\catf{Cat}^\otimes $
\begin{itemize}
	\item $\Bord_{d,0}^{-}$ preserves colimits. 
	\item The framed bordism category $\Bord_{d,0}^{X = *}$ is the free symmetric monoidal $d$-category on
	one fully dualisable object.  
\end{itemize}  
When talking about colimits we always mean the appropriate higher categorical concept, for example for $d = 2$ we mean what is usually called a pseudolimit \cite[Chapter 5]{Johnson2021}.
The general definition of a fully dualisable will not matter much for us and hence we refer to any of the 
following references for it~\cite{lurietft, ClaudiaOwen, CSPdualizability}. 
The second point implies that the category of framed
$d$-dimensional fully extended topological field theories with values in $\cat{C}$ is equivalent to the 
core of the full subcategory on all fully dualisable objects $\cat{C}^{\text{f.d.}}$. This 
space has an action by the group of automorphisms of $*\to B\O_d$ in $\Top/B\O_d$ which is homotopy 
equivalent to $\O_d$. Concretely, an element in $\O_d$ acts on the framed bordism category by rotating the 
framing which induces an action on the space of framed field theories.  

\begin{example}\label{Ex: Action on sVect}
A super vector space $V\in \sVect$ is 1-dualisable if and only if it is finite dimensional. The 
$\O_1=\Z_2$ action sends a finite dimensional vector space to its dual $V^*=\Hom(V_0,\C)\oplus 
\Hom(V_1,\C)$. An isomorphism $f\colon V\to V'$ is send to the linear map ${f^*}^{-1}\colon V^*\to 
{V'}^*$. This basic example already shows that the action is only defined on invertible dualisable morphisms. 

A super algebra $A\in \sAlg$ over $\C$ is 2-dualisable if and only if it is finite-dimensional and semi-simple~\cite[Example 2.5]{gunningham2016spin}. 
We 
describe the $\O_2$-action explicitly in Appendix~\ref{Sec:salgdual}.    
\end{example}

The space of topological field theories with a general tangential structure $G\to \O_d$ can be computed 
using that  $\Bord_{d,0}^{-}$ preserves colimits. For this note that we have 
$BH\to B\O_d = \colim_{BH} (*\to B\O_d)$ as
objects of $\Top/B\O_d$. The functor featuring in the colimit is induced by the map $H\to \O_d$ and can be 
informally described by saying that it sends an object $h\in H$ to the automorphism 
\begin{equation}
\begin{tikzcd}
	* \ar[rr, "*"] \ar[rdd] & & * \ar[ldd] \ar[lldd, Rightarrow, "\rho(h)" near start, shorten <= 5, shorten >=35,swap] \\
	 & & \\  
\  & B\O_d & 	
\end{tikzcd}
\end{equation} of $*\to B\O_d$    
where we denote by $\rho(h)$ the path in $B\O_d$ corresponding to $\rho(h)\in \O_d$. Having written 
$BH\to B\O_d$ as a colimit we can use that $\Bord^{(-)}_{d,0}$ preserves colimits
to conclude that
\begin{align}
	\Bord^{(H,\rho)}_{d,0} \simeq \colim_{BH} \Bord_{d,0}^* \ \ . 
	\label{Eq: Bord as colimit}
\end{align} 
A formal consequence of Equation~\ref{Eq: Bord as colimit} is a classification result for topological 
field theories (compare \cite[Theorem 2.4.18]{lurietft})
\begin{align}
\Cat_d^\otimes (\Bord_{d,0}^{(H,\rho)}, \cat{C} ) \simeq \Cat_d^\otimes (\colim_{BH} \Bord_{d,0}^*, \cat{C} ) \simeq \lim_{BH} \Cat_d^\otimes ( \Bord_{d,0}^*, \cat{C} ) \simeq \lim_{BH} \cat{C}^{f.d.} \ \  
\end{align} 
where the $H$-action on $\cat{C}^{f.d.}$ is the pullback of the $\O_d$-action along the group
homomorphism $\rho(h)\colon H \to \O_d$. 
The limit $\lim_{BH} \cat{C}^{f.d.}$ is the category 
of homotopy fixed points for this action. 

As a simple example we look at 1-dimensional topological field theories. 
A tangential structure in 1-dimension is given by a group homomorphism $\theta \colon H\to \Z_2=O_1$, i.e. a $\Z_2$-graded group. 
	\begin{proposition}\label{Prop: 1D TFT} 
	A one-dimensional topological field theory with tangential structure $\theta: H \to \Z_2$ is classified by a representation $(R,V)$ of $H_{0}$ together with a collection of non-degenerate bilinear forms $_g\langle .,. \rangle: V \otimes V \to \C$ for every $g \in H \setminus H_0$ such that for all $h \in H_0$, all $g,g' \in H \setminus H_0$ and all homogeneous vectors $v,w \in V$
	\begin{align}
		\label{eq: condition1}
		{}_{gh} \langle v,w \rangle &= {}_g \langle R(h) v,w \rangle 
		\\
		\label{eq: condition2}
		{}_{hg} \langle v,w \rangle &= {}_g \langle v, R(h^{-1}) w \rangle 
		\\
		\label{eq: condition3}
		{}_{g^{-1}}\langle v,w \rangle &= (-1)^{|v||w|} {}_{g'^{-1}} \langle w, R(gg') v \rangle \ \ .
	\end{align}
\end{proposition}
\begin{proof}
We have to compute homotopy fixed for the action induced from Example~\ref{Ex: Action on sVect}. We refer the reader to Remark~\ref{Rem: FP on Cat} for the definition. These consist of a family of maps 
\begin{align}
	\begin{cases}
		F(g) \colon V \longrightarrow V \ , \ \text{for $\theta(g)=1$} \\
		F(g) \colon V^* \longrightarrow V \ , \ \text{for $\theta(g)=-1$} 
	\end{cases}
\end{align}   
for $g\in H$ which has to be compatible with the group multiplication in $H$. 
There are four versions of the relevant diagram of homotopy fixed points given by the choices $\theta(g) = \pm 1$ and $\theta(h) = \pm 1$.
The first diagram tells us that $R:=F|_{G_0}$ is a representation of $G_0$ on $V$ that is even in the supergrading.	
For $g \in G_1 = G \setminus G_0$ we translate the fixed point structure $F(g): V^* \to V$ to the bilinear form $_g \langle.,. \rangle : V \times V \to \C$ by
	\[
	_g \langle v,w \rangle := F(g^{-1})^{-1}(v)(w).
	\]
	Our convention to choose $g^{-1}$ in the definition is to get nicer formulas.
	The fact that the $F(g)$ are even translates to the fact that $V_0$ and $V_1$ are orthogonal.
	The second diagram for $g_0 \in H_0$ and $g_1 \in H_1$ looks like
	\[
	\begin{tikzcd}
		V & V \arrow[l, "F(g_0)",swap]
		\\
		& V^* \arrow[ul, "F(g_0 g_1)"] \arrow[u, "F(g_1)", swap]
	\end{tikzcd}
	\]
	Written out in bilinear form notation we get 
	\[
	_{g_1^{-1}} \langle v, w \rangle = {}_{g_1^{-1} g_0^{-1}}\langle R(g_0) v,w \rangle
	\]
	which is equivalent to the first point.
	The third diagram looks like 
	\[
	\begin{tikzcd}
		V & V^* \arrow[l, "F(g_1)", swap]
		\\
		& V^* \arrow[ul, "F(g_1 g_0)"] \arrow[u, "F(g_0)^{* -1}", swap]
	\end{tikzcd}
	\]
	This gives us using the result of the last diagram
	\[
	_{g_1^{-1}}\langle v, R(g_0) w \rangle = {}_{g_0^{-1} g_1^{-1}} \langle v,w \rangle = {}_{g_1^{-1}} \langle R(g_1 g_0^{-1} g_1^{-1}) v, w \rangle
	\]
	or in other words $R(g_1 g_0^{-1} g_1^{-1})^T = R(g_0)^{-1}$ where the transpose is with respect to $_{g_1} \langle .,. \rangle$.
	The transpose of an even linear map $T: V \to V$ is defined with respect to a nondegenerate bilinear form $\langle.,. \rangle$ as
	\[
	\langle T v, w \rangle = \langle v, T^T w \rangle.
	\]
	Note that the order matters if the bilinear form is not symmetric or antisymmetric. In particular, $T^{TT} \neq T$ in general. 
	The final diagram for $g,g' \in H_1$ is
	\[
	\begin{tikzcd}
		V & V^* \arrow[l, "F(g)", swap]
		\\
		V \arrow[u, "F(g g')"] \arrow[r, "\sim",swap] & V^{**}  \ar[u, "F(g')^{* -1}", swap]
	\end{tikzcd}
	\]
	where the isomorphism $\operatorname{ev}: V \cong V^{**}$ maps $v \in V$ to the evaluation map $\operatorname{ev}_v : V^* \to \C$ given by $\operatorname{ev}_v(f) = (-1)^{|f| |v|} f(v)$.
	Now if $T: V \to V^*$ is an even linear map with corresponding bilinear form $_T \langle v,w \rangle := T(v)(w)$, then under this isomorphism $V \cong V^{**}$, the dual map $T^*: V \to V^*$ satisfies $_{T^*} \langle v,w \rangle = (-1)^{|v||w|} {}_T \langle w,v \rangle$ because
	\begin{align}
		[T^*(\operatorname{ev}_v)](w) = \operatorname{ev}_v(Tw) = (-1)^{|Tw||v|} (Tw)(v).
	\end{align}
	Using this the diagram gives
	\[
	{}_{g^{-1}}\langle v,w \rangle = (-1)^{|v||w|} {}_{g'{-1}} \langle w, R(gg') v \rangle \ \ .
	\]
\end{proof}
Mathematically, it is not clear how to simplify these conditions without making unnatural choices.
Note however that all other bilinear forms are determined by a single one $\langle .,. \rangle := _{g^{-1}}\langle.,. \rangle$ using condition \ref{eq: condition1} after an unnatural reference element $g \in H \setminus H_0$, i.e. a set-theoretic section of the short exact sequence
\[
1 \to H_0 \to H \to \Z_2 \to 1.
\]
is chosen.

When moving up to 2-dimensional topological field theories the concepts of a group
action on a bicategory and homotopy fixed points get more involved~\cite{HSV}. We
review the necessary definitions in Appendix~\ref{App: 2-Group}. 
There already exist a few computations of the relevant bicategory of homotopy 
fixed points for field theories with target a Morita 2-category:
\begin{itemize}
	\item In~\cite{HSV,janserre,janthesis} $SO_2$-homotopy fixed points
	in the Morita bicategory $\Alg^{f.d. }$ have been identified with
	symmetric semi-simple Frobenius algebras. 
	\item Without using the cobordism hypothesis, unoriented (i.e. $O_2$) field theories with values in $\Alg$ have been classified in terms of ($\C$-linear) stellar symmetric Frobenius algebras in~\cite{schommerpriesthesis}. These are a Morita invariant version of $*$-algebras, see Section~\ref{Sec:stellar} for more
	details on the $\C$-antilinear variant.   
	\item For a finite group $G$, let $H:=G\times SO_2$ with the structure map being projection onto the second factor. Then $H$-homotopy fixed points in $\Alg^{f.d.}$ have been identified with strongly $G$-graded symmetric Frobenius algebras~\cite{oritthesis}. These results have been extended to 
	$G\times O_2$-fixed points in~\cite{Sozer}. His results only admit a good algebraic description for certain groups.  
	\item In~\cite{gunningham2016spin} $\Spin_2$-homotopy fixed points in $\sAlg$ have been computed. His answer is not given in terms of a nice algebraic structure as in the cases above. 	     
\end{itemize}

The conclusion we draw from these examples is that there seems to be no good description of homotopy fixed points in terms of algebraic objects as soon as Spin-structures (i.e. fermions) or too many orientation reversing elements are involved in the tangential structure. One of the main points of this article will be that if one includes spin-statistics and reflection structures it turns out that there are nice classifications in these cases. 
         
\subsection{Pictorial interpretation of the cobordism hypothesis in dimension 2} 
In this section we focus on 2-dimensional fully extended topological field theories with values in $\sAlg$.
The goal is to give a more hands on description of the cobordism hypothesis in this setting. As explained
in the previous section,
framed topological field theories $\Bord^*_{2,0} \longrightarrow \sAlg$ are the starting point for the 
classification of other tangential structures. The cobordism hypothesis states that the 2-groupoid of
framed topological field theories is equivalent to the subgroupoid $\sAlg^{\text{f.d.}}$ of fully dualisable algebras. We will not explain this geometrically here, but refer to the
following accounts in the literature~\cite{lurietft,ClaudiaOwen,tanaka2020lectures}. The fully dualisable super algebras over $\C$ 
are exactly the finite-dimensional semi-simple ones.

For every tangential structure $\rho \colon H\longrightarrow O_2$  there is a canonical functor $\iota \colon \Bord_{2,0}^* \longrightarrow \Bord_{2,0}^{H,\rho}$ sending a framed manifold
\begin{equation}
\begin{tikzcd}
&  * \ar[d] \\
M \ar[ru,""{name=U, below},bend left=30] \ar[r,"{TM}",""{name=V, above}, swap] & BO_2
\ar[from=U, to=V, "h", Rightarrow]
\end{tikzcd}
\end{equation}   
to
\begin{equation}
\begin{tikzcd}
&  * \ar[d] \ar[r] & BH \ar[ld] \\
M \ar[ru,""{name=U, below},bend left=30] \ar[r,"{TM}",""{name=V, above}, swap] & BO_2 
\ar[from=U, to=V, "h", Rightarrow]
\end{tikzcd}
\end{equation}
where the new triangle is strictly commutative. Hence, every topological $(H,\rho)$-field theory $\mathcal{Z}\colon
\Bord_{2,0}^{H,\rho} \longrightarrow \sAlg$ has an underlying framed topological field theory $\mathcal{Z}_{\text{fr}} 
\coloneqq \iota^* \mathcal{Z}$ which is completely determined by its value on the positively framed point, i.e.
a fully dualizable algebra $A_\mathcal{Z}\in \sAlg^{\text{fd}}$ . We describe the additional data the 
cobordism hypothesis tells us to consider to fully describe $\mathcal{Z}$ in terms of 
$\mathcal{Z}_{\text{fr}}$. We write $+$ for the positively framed point and $-$ for the negatively framed point
as well as for their images in $\Bord_{2,0}^{H,\rho}$. 
Every element $g\in H$ corresponds to a path $\gamma_g \colon [0,1] \longrightarrow BH $. The classifying 
map for the tangent bundle of $[0,1]$ is trivial and hence we can upgrade this to a $(H,\rho)$-bordism  
\begin{equation}
\begin{tikzcd}
&  BH \ar[d] \\
{	[0,1]} \ar[ru,"\gamma_g",""{name=U, below},bend left=30] \ar[r,"*",""{name=V, above}, swap] & BO_2
\ar[from=U, to=V, "h", Rightarrow]
\end{tikzcd} \ \ .
\end{equation}
where $h$ is the homotopy which continues the loop in $B\O_2$ to its end point. Explicitly $h$ is given
by the map 
\begin{align}
h\colon I\times I & \to B\O_2 \\ 
(t,x) & \longmapsto \rho(\gamma_g(x+(1-x)t)) \ \ .
\end{align}
This defines a 1-morphism from $+$ to $+$ if $\rho(g)\in \O_2$ is in the connected component of the 
identity and from $-$ to $+$ otherwise. We draw these 1-morphisms as shown in Figure~\ref{fig:1-mor}.   
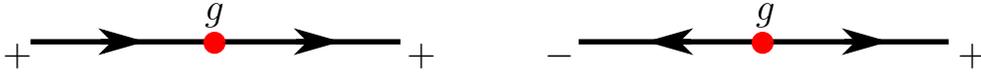
\begin{figure}[h!]
	\begin{center}
		\begin{overpic}[scale=1
			,tics=10]
			{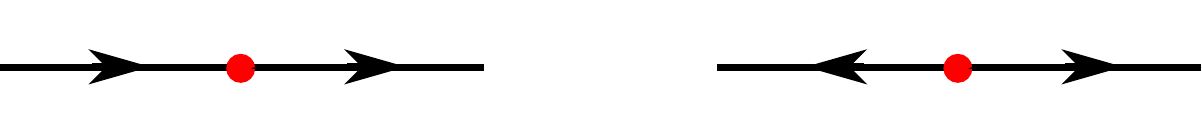}
			\put(19,8){{\Large$g$}}
			\put(79,8){\Large$g$}	
			\put(-3,3){\Large$+$}	
			\put(41,3){\Large$+$}	
			\put(56,3){\Large$-$}	
			\put(101,3){\Large$+$}	
		\end{overpic}  
		\vspace{0.5cm}
		\caption{Some 1-morphisms in $\Bord_{2,0}^{H,\rho}$.}
		\label{fig:1-mor}
	\end{center}
\end{figure}
The field theory $\mathcal{Z}$ assigns to them either an $\mathcal{A}_{\mathcal{Z}}$-$\mathcal{A}_{\mathcal{Z}}$-bimodules $\mathcal{A}_g$ in the case $\rho(g)\in SO_2$ or an $\mathcal{A}_{\mathcal{Z}}$-$\mathcal{A}_{\mathcal{Z}}^{\text{op}}$ bimodule otherwise. For 
both morphisms, we also have their orientation reversal, which is constructed by composing the homotopy for
the tangential structure with the non-trivial loop in $BO_2$. We denote these by reversed arrows. The 
bimodules associated to those are $\mathcal{A}_g^{\op,-1}$ where $^{-1}$ denotes a chosen adjoint inverse. Let
$g$ and $g'$ be elements of $H$. There is a canonical homotopy $\psi\colon \gamma_{g'} \circ \gamma_g \rightarrow \gamma_{g'g}$, which when combined with a scaling gives rise to 2-morphisms in $\Bord_{2,0}^{H,\rho}$ which we draw as indicated in Figure~\ref{fig:2-mor}.  
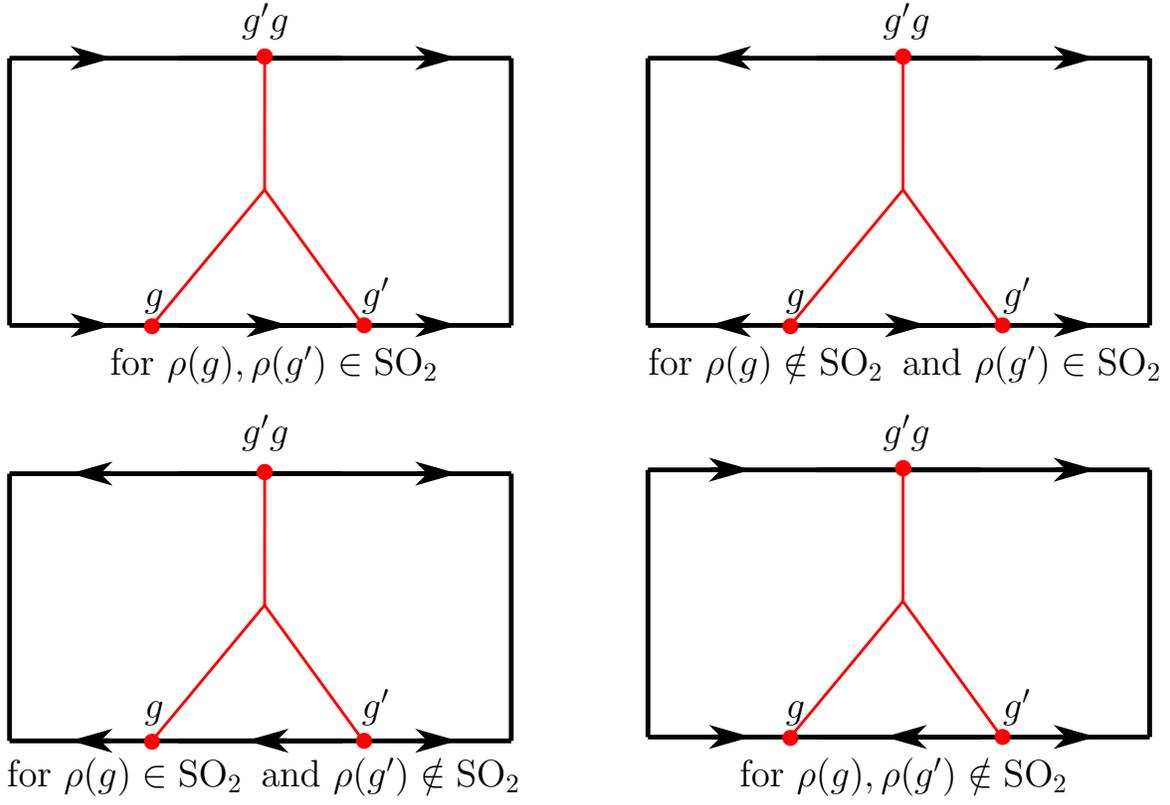
\begin{figure}[h!]
	\begin{center}
		\begin{overpic}[scale=0.9
			,tics=10]
			{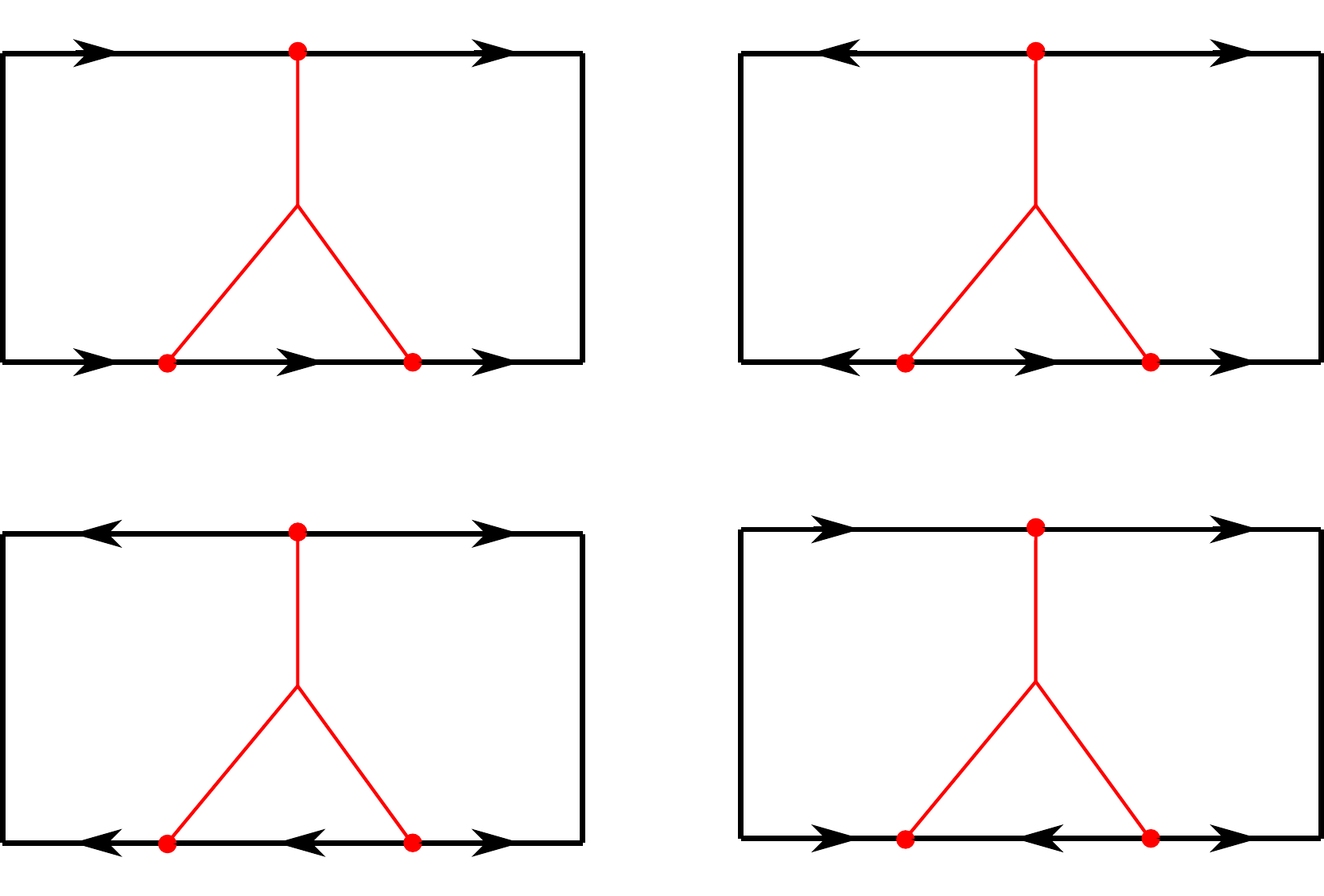}
			\put(12,42){{\Large$g$}}
			\put(31,42){{\Large$g'$}}
			\put(20.5,66){\Large$g'g$}	
			\put(12,6){{\Large$g$}}
			\put(31,6){{\Large$g'$}}
			\put(20.5,30){\Large$g'g$}
			\put(68,6){{\Large$g$}}
			\put(87,6){{\Large$g'$}}
			\put(76.5,30){\Large$g'g$}			
			\put(68,42){{\Large$g$}}
			\put(87,42){{\Large$g'$}}
			\put(76.5,66){\Large$g'g$}
			\put(9,36){\Large$\text{for } \rho(g),\rho(g')\in \SO_2$}
			\put(56,36){\Large$\text{for } \rho(g)\notin \SO_2\text{ and } \rho(g')\in \SO_2$}
			\put(0,0){\Large$\text{for } \rho(g)\in \SO_2\text{ and } \rho(g')\notin \SO_2$}	
	\put(64,0){\Large$\text{for } \rho(g),\rho(g')\notin \SO_2$}	
	\end{overpic}  
		\vspace{0.5cm}
		\caption{Some 2-morphisms in $\Bord_{2,0}^{H,\rho}$.}
		\label{fig:2-mor}
	\end{center}
\end{figure}
Note that these only describe the part of the 2-morphism given by a map $I^2\to BH$. The homotopy which is
part of the tangential structure can be constructed as before by following the homotopy $\psi$ to the end. Concretely, it is given by  
\begin{align}\label{Eq: Def h}
h\colon I\times I^2 & \to B\O_2 \\ 
(t,x,y) & \longmapsto \rho(\psi(x+(1-x)t,y+(1-y)t)) \ \ .
\end{align}
Evaluating $\mathcal{Z}$ on those morphism gives 
rise to bimodule intertwiners $\phi_{g,g'}\colon \mathcal{A}_{g'} \otimes_{\mathcal{A}_{\mathcal{Z}}} \mathcal{A}_g \to 
\mathcal{A}_{g'g}$ for $\rho(g),\rho(g')\in \SO_2$, $\phi_{g,g'}\colon \mathcal{A}_{g'} \otimes_{\mathcal{A}_{\mathcal{Z}} 
} \mathcal{A}_g \to \mathcal{A}_{g'g}$ for $\rho(g)\in \SO_2$ and $\rho(g')\notin \SO_2$, $\phi_{g,g'}\colon 
{\mathcal{A}_{g'}^{\op}}^{-1} \otimes_{\mathcal{A}_{\mathcal{Z}}} \mathcal{A}_g \to \mathcal{A}_{g'g}$ for 
$\rho(g)\notin \SO_2$ and $\rho(g')\in \SO_2$, and $\phi_{g,g'}\colon \mathcal{A}_{g'} 
\otimes_{\mathcal{A}^\op_{\mathcal{Z}} } {\mathcal{A}_g^{\op}}^{-1}  \to \mathcal{A}_{g'g}$ for $\rho(g),\rho(g')\notin \SO_2$, where we have used 
implicitly the canonical identification ${\mathcal{A}_{\mathcal{Z}}^{\op}}^{\op}\cong \mathcal{A}_{\mathcal{Z}}$.     

The morphisms defined above satisfy relations of the type sketched in Figure~\ref{fig:ass}
\begin{figure}[h!]
	\begin{center} 
		\begin{overpic}[scale=0.7
			,tics=10]
			{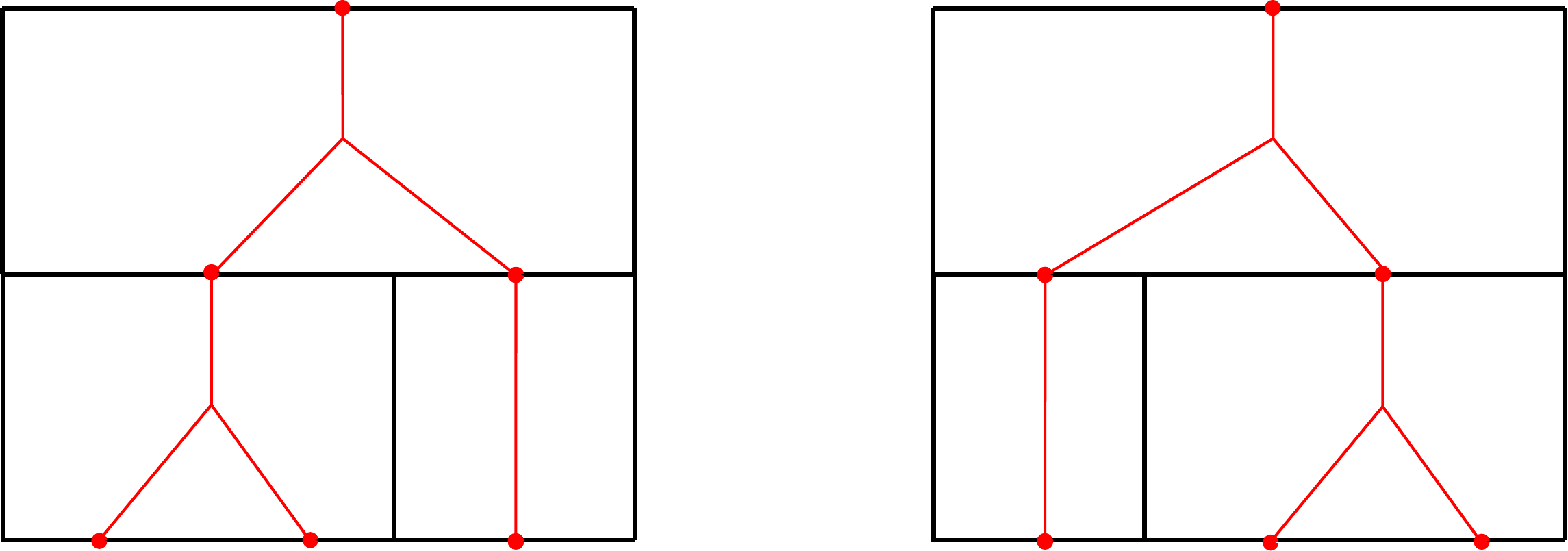}
			\put(48,16){{\huge$=$}}
		\end{overpic}  
		\vspace{0.5cm}
		\caption{A type of relation in $\Bord_{2,0}^{H,\rho}$.}
		\label{fig:ass}
	\end{center}
\end{figure}
inside $\Bord_{2,0}^{H,\rho}$, where we use a straight red line to denote the identity 2-morphism.  

For every path $\Gamma \colon g \longrightarrow g'$ in $H$ there is a homotopy 
$\Gamma \colon I^{2} \longrightarrow BG $. When building a 2-morphism from $\Gamma$ there is one important
subtlety: We defined 2-morphism to be constant along the vertical boundary, but if one naively uses the 
deformation to the right top corner to define the tangential structure as in Equation~\eqref{Eq: Def h} this is not the case. The change of
the tangential structure along the left vertical edge corresponds to the path $\rho(\Gamma) \in O_2$. We 
can solve this by including the path into the top line of the morphism. We draw 
\begin{figure}[h!]
	\begin{center}
		\begin{overpic}[scale=1
			,tics=10]
			{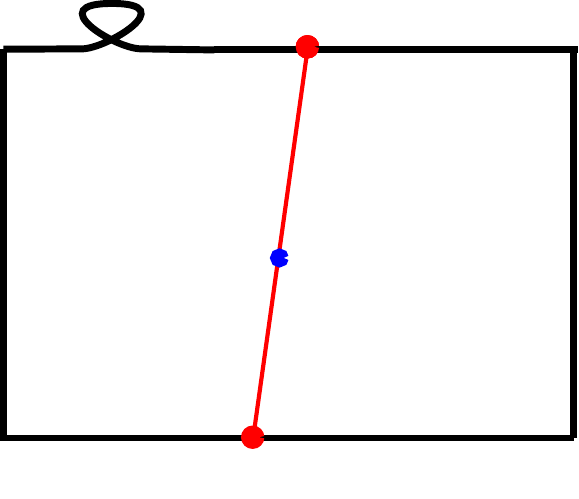}
			\put(52,39){{\Large$\Gamma$}}
			\put(16,88){{\Large$n$}}
		\end{overpic}  
		\vspace{0.5cm}
		\caption{More 2-morphisms in $\Bord_{2,0}^{H,\rho}$.}
		\label{fig:Gamma}
	\end{center}
\end{figure}
these as shown in Figure~\ref{fig:Gamma}, where $n$ is an integer indicating the loop $\rho(\Gamma)$ in $O_2$ used to define the morphism at the top. Note that the target of the 2-morphism just defined is the composition of the image of a morphism in $\Bord^*_{2,0}$ with the morphism corresponding to $g'$. The automorphism of the plus point corresponding to $n=1$ in $\Bord^*_{2,0}$ is called
the \emph{Serre automorphism}. All other values of $n$ can be constructed from it by composition and 
taking the inverse. The Serre automorphism can be defined in any fully dualisable bicategory and is part 
of a natural transformation from the identity to itself~\cite[Proposition 3.2.]{LorantNils}. It is part
of the $\O_2$-action mentioned in the previous section, see Appendix \ref{Sec:serre} and \ref{Sec:O2-action} for details. 

Evaluating the field theory on the 2-morphism from Figure~\ref{fig:Gamma} gives rise to an intertwiner 
$\lambda_{\Gamma}\colon \mathcal{A}_g \to \mathcal{A}_{g'}\circ S^n$ where $S$ is the Serre automorphism of $\mathcal{Z}$. 
These 2-morphisms satisfy relations of the form shown in Figure~\ref{fig:Rel1} and \ref{fig:Rel2}.
\begin{figure}[h!]
	\begin{center}
		\begin{overpic}[scale=1
			,tics=10]
			{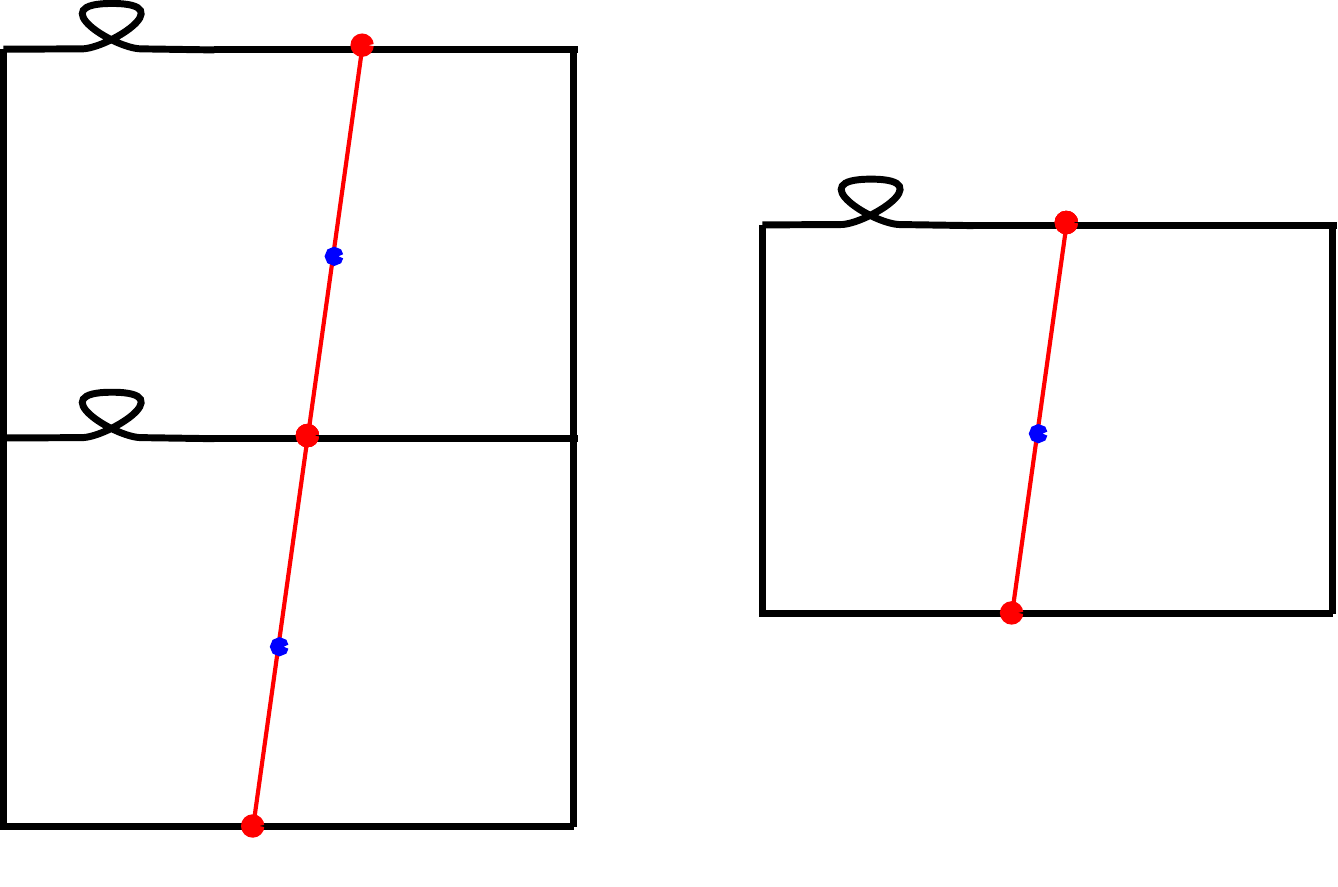}
			\put(22,16.5){{\Large$\Gamma$}}
			\put(26,46){{\Large$\Gamma'$}}
			\put(7.5,38){{\Large$n$}}
			\put(4.5,67){{\Large$n+m$}}
			\put(61,53.5){{\Large$n+m$}}
			\put(48,35){{\Large$=$}} 
			\put(79,32.5){{\Large$\Gamma' \circ \Gamma$}}
		\end{overpic}  
		\vspace{0.5cm}
		\caption{More relations in $\Bord_{2,0}^{H,\rho}$.}
		\label{fig:Rel1}
	\end{center}
\end{figure}
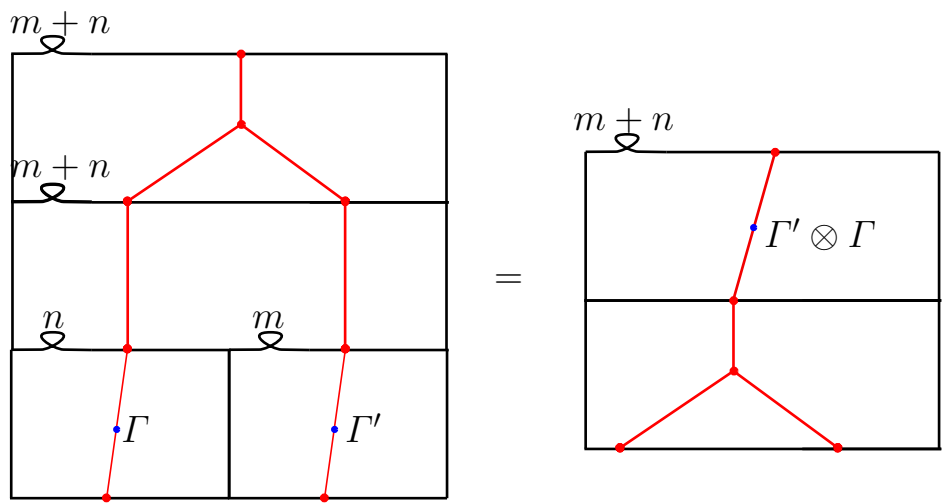
\begin{figure}[h!]
	\begin{center}
		\begin{overpic}[scale=0.5
			,tics=10]
			{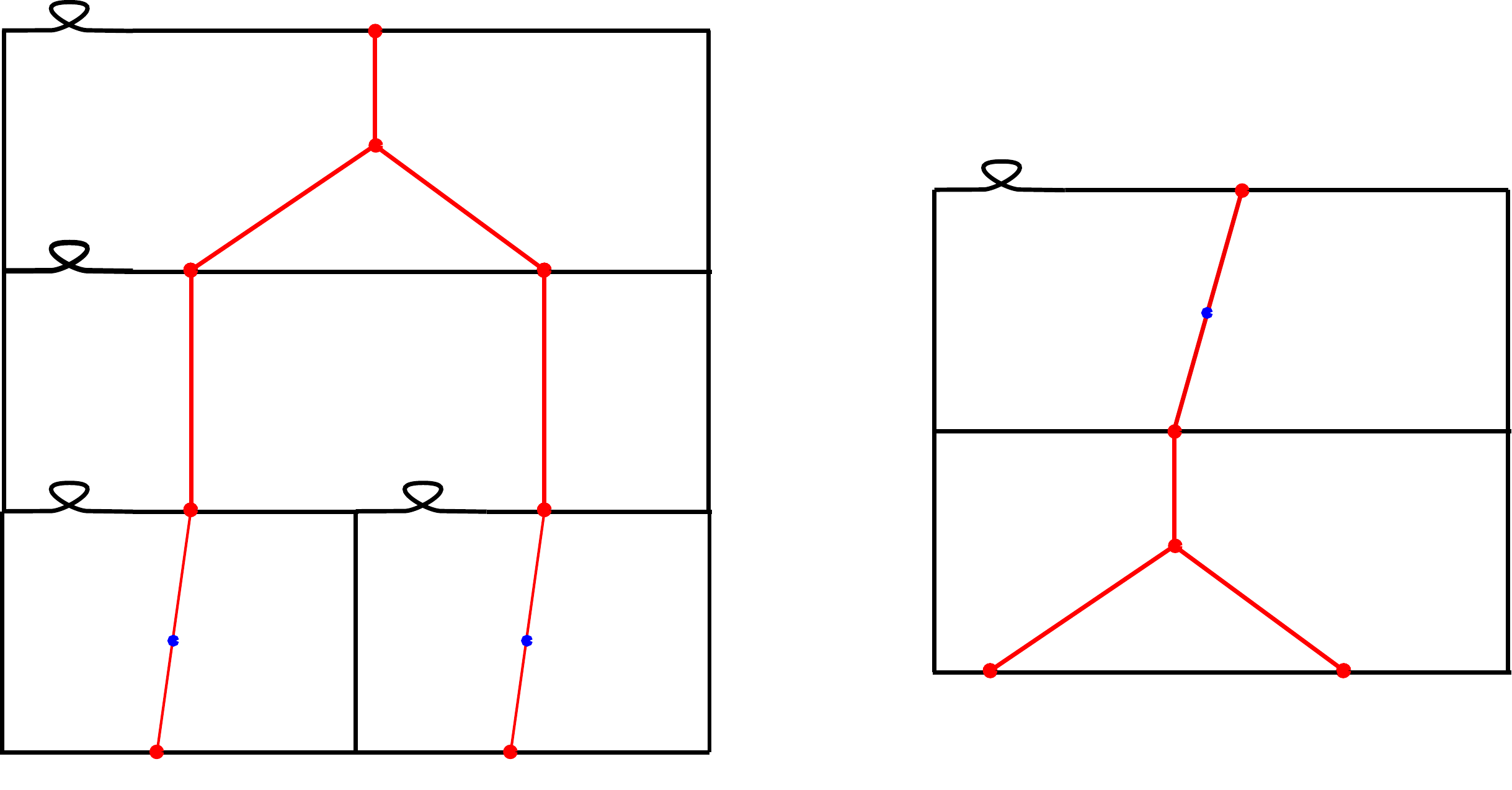}
			\put(12,8.5){\Large$\Gamma$}
			\put(36,8.5){\Large$\Gamma'$} 
			\put(3.5,20.7){\Large $n$} 
			\put(26,20.7){\Large $m$} 
			\put(0,37){\Large $m+n$} 
			\put(0,52.5){\Large $m+n$} 
			\put(60.5,42){\Large $m+n$} 
			\put(52,25){\Large $=$}
			\put(81,29.5){\Large$\Gamma'\otimes \Gamma$} 
		\end{overpic}  
		\vspace{0.5cm}
		\caption{Even more relations in $\Bord_{2,0}^{H,\rho}$.}
		\label{fig:Rel2}
	\end{center}
\end{figure}
For the first relation we suppress the isomorphism between $m \circ n$ and $n+m$ and the second involves
the naturality of the Serre isomorphism. 
The statement of the cobordism hypothesis in 2-dimensions can be understood as saying that these morphisms and relations generate 
$\Bord_{2,0}^{H,\rho}$ as a fully dualisable symmetric monoidal bicategory. Hence giving a fully 
extended topological $(H,\rho)$-field theory is equivalent to giving a fully dualisable algebra $A$ 
bimodules $A_g$, interwiners $\phi_{g,g'}$, and interwiners $\lambda_\Gamma$ satisfying relations of 
the type shown in Figure~\ref{fig:ass}, \ref{fig:Rel1}, and \ref{fig:Rel2}. For a proof of this statement in the case of $G\times \SO_2$-tangential structure
we refer to~\cite{Sozer}. 
As we explain in Appendix~\ref{App: 2-Group} this is exactly the same data as a homotopy fixed point for the action of $G$ on $\sAlg$.  

\begin{remark}
A similar, but significantly simpler pictorial description also works in dimension one. There one has morphisms as in Figure~\ref{fig:1-mor} for every element $g\in G$. They compose in accordance with the group law of $G$. The orientations can be 
taken into account exactly as in Figure~\ref{fig:2-mor}. More concretely instead
of having a 2-isomorphism from the bottom to the top, in 1-dimension, the bottom
and top are equal.

This description should generalise to arbitrary dimensions by using cubes $I^d$ instead of squares $I^2$. 
\end{remark}

\section{Fermionic groups}\label{Sec:FG}
In this section we introduce fermionic groups as a model for symmetry groups appearing in physics and study some of there basic properties. We also define the spacetime structure groups associated to an internal symmetries described by a fermionic group.
  
\subsection{Fermionic groups as internal symmetries}
We start by introducing an abstract definition for the internal symmetry group of a quantum
system which can contain both bosons and fermions. 
In this context it is important to distinguish between symmetries which act unitarily and 
anti-unitarily on the Hilbert space of the theory. For example, time 
reversal symmetry acts anti-unitarily.\footnote{From a relativistic perspective, time-reversal is not an internal symmetry. We will instead adopt the common convention to call time-reversing symmetries internal, even though they are only internal from a spatial perspective.}
We will encode this information by a $\Z_2^T$-grading.
Fermion parity usually denoted by $(-1)^F$ is 
a special symmetry of any quantum system, which we encode 
by an even central element $c$ of square $1$. We abstract these two pieces of information in the following
definition (which we learned from Peter Teichner) for the structure of an internal symmetry group:
\begin{definition}
A \emph{fermionic group} $G$ is a Lie group together with
\begin{itemize}
\item a continuous homomorphism $\theta: G \to \Z_2^T$, which we interpret as a $\Z_2^T$-grading $G = G_0 \sqcup G_1$;
\item an even central element $c \in G$ of square 1, i.e. $\theta(c) = 0$. 
\end{itemize}
A \emph{homomorphism of fermionic groups} $f\colon (G,c_G,\theta_G) \to (H,c_H,\theta_H)$ is a smooth group homomorphism $f\colon G \to H$ which preserves the grading and satisfies $f(c_G)=c_H$.
\end{definition}

Let $(G,c,\theta)$ be a fermoinic group. 
We denote by $G^{\op}$ the \emph{opposite fermionic group} which has the same underlying manifold, $c$, and $\theta$ as $G$, but multiplication

\begin{align}
g_1*^{\op}g_2= c^{\theta(g_1)\theta(g_2)}g_2 g_1 \ \ . 
\end{align}

If $c$ is equal to $1$ we call a fermionic group \emph{bosonic}.  
For a fermionic group $(G,c,\theta)$ with $c\neq 1$ we denote the central subgroup generated by $c$ by 
$\Z_2^c \subset G$ and define the \emph{underlying bosonic group} as the quotient $G_b \coloneqq G/\Z_2^c$.
For a fermionic group $G$ with non-trivial $c$ and non-trivial grading, both $G_0$ and $G_b$ are again fermionic groups (the former with trivial grading, the latter with $c=1$) fitting into a commutative diagram
\begin{equation}\label{Eq: subgroups}
	\begin{tikzcd} 
		\Z_2^c \ar[r, equals] \ar[d, hook] & \Z_2^c \ar[d, hook] & \\
		G_0  \arrow[d,two heads] \arrow[r, hook] 
		& G \arrow[d, two heads] \arrow[r, two heads] & \Z_2^T \ar[d, equals] \\
		G_{0,b} \arrow[r, hook] 
		& G_b \arrow[r, two heads] & \Z_2^T
	\end{tikzcd}
\end{equation} 
of short exact sequences of fermionic groups. 
\begin{example}\label{Ex: SA FG}
Let $A=A_0\oplus A_1$ be a real super algebra. We refer to Appendix~\ref{App: Super alg} for basics on super algebras. The group of homogenous units $A^\times$ is a fermionic group with $\Z_2^T$-grading 
induced by the grading of $A$ and $c=-1$. Furthermore, the sphere $S(A)= A^\times / \R_{> 0}$ is 
a compact fermionic group where the positive real numbers act by scalar multiplication. 
\end{example}

\begin{example}
For $p,q$ non-negative integers, the Pin-group $\Pin_{p,q}$ (i.e. the group generated by elements of $\R^{p,q}$ of 
norm $\pm 1$ in the Clifford algebra $\Cl_{p,q}$ of signature $(p,q)$) is a fermionic group.
 In this case the Diagram~\eqref{Eq: subgroups} becomes 
\begin{equation}
	 \begin{tikzcd}
	 \Spin_{p,q}  \arrow[d,two heads] \arrow[r, hook]
	 & \Pin_{p,q} \arrow[d, two heads] \\
	 \SO_{p,q} \arrow[r, hook]
	 & \O_{p,q},
	 \end{tikzcd}
\end{equation}
and hence defines canonical fermionic group structures on $\Spin_{p,q}$, $\SO_{p,q}$, and $\O_{p,q}$, where the latter two are bosonic. 
The special central element $c \in \Pin_{p,q}$ in the kernel of the map to $O_{p,q}$ is related to the fermion parity operator $(-1)^F$ through the spin-statistics theorem, see Section \ref{spin-st} for further motivation. 
The grading $\theta$ is given by the determinant $\theta: \Pin_{p,q} \to O_{p,q} \xrightarrow{\det} \{\pm 1\} \cong \Z_2^T$, which agrees with the $\Z_2$-grading on the Clifford algebra.
\end{example} 

\begin{definition}\label{Def: Fermionic tensor product}
Let $(G,c_G,\theta_G)$ and $(H,c_H,\theta_H)$ be fermionic groups. The \emph{fermionic tensor product} $G 
\otimes H$ is the set $(G \times H)/ \langle (c_G,c_H) \rangle$ with the operation
\[
(g_1 \otimes h_1) (g_2 \otimes h_2) = c_G^{\theta_G(g_2)\theta_H(h_1)} g_1 g_2 \otimes h_1 h_2,
\]
the central element $1 \otimes c_H = c_G \otimes 1 \in G \otimes H$ and the grading $\theta_{G\otimes H}(g \otimes h) = \theta_G(g) + \theta_H(h)$.
When it should be clear from the context we omit the subscripts $G$ and $H$ to improve readability.
\end{definition}

It is straightforward to show that the fermionic tensor product is indeed a group.
Note that $G \otimes H$ 
is naturally bigraded; there is a homomorphism $\theta_G \otimes \theta_H: G \otimes H \to \Z^T_2 \times 
\Z^T_2$.
The grading from Definition~\ref{Def: Fermionic tensor product} that we use to make $G \otimes H$ into a 
fermionic group is the bigrading composed with the sum operation $\Z^T_2 \times \Z^T_2 \to \Z_2^T$.
We will often consider the even part $(G \otimes H)_{0}$ under this grading.
Note that the other part of the bigrading still gives a grading $(G \otimes H)_{0}  = G_0 \otimes H_0 \sqcup G_1 \otimes H_1$ and so forms another fermionic group.

\begin{example}
	As a concrete special case of the Example~\ref{Ex: SA FG}, consider the Clifford algebra $D = \Cl_{-1}$ with one odd generator $f$ of square $-1$. 
	Then $S(D) = \{[1],[-1],[f],[-f]\} \cong \Z/4$.
	The fermionic group structure is expressed through the exact sequence
	\[
	1 \to \{\pm 1\} \to S(D) \overset{(-)^2}{\to} \{\pm 1\} \to 1.
	\]
	This fermionic group can be identified with $S(D)\cong \operatorname{Pin}^-_1 = \operatorname{Pin}_{0,1}$.
	
	Consider the fermionic tensor product $G:= S(D) \otimes S(D)$.
	It consists of eight elements $\pm a \otimes b$, where $a$ is either $1$ or $f$.
	It is generated by the two elements $f \otimes 1$ and $1 \otimes f$, which both square to $1$ and anticommute in the sense that
	\[
	(f \otimes 1) (1 \otimes f) = f \otimes f = (-1 \otimes f)(f \otimes 1).
	\]
	We conclude that $G$ is isomorphic to the quaternion group $\{\pm 1, \pm i, \pm j, \pm k\}$.
	The fermionic structure is given by the central element $-1$ and the unique  $\Z/2$-grading such that $i$ and $j$ are odd (so that $k$ is even).
	In particular note that the even part $G_0$ is isomorphic to $S(D)$.
	Also clearly $S(D) \otimes S(D) \neq S(D \otimes_\R D)$, even though the former is a subgroup of the latter.
\end{example}

\begin{remark}\label{Rem: fermionic groups as maps}
The definition of a fermionic group with non-trivial $c$ has the following topological interpretation. It 
is equivalent to giving a Lie group $G_b$ together with a continuous map $BG_b\to B\Z_2^T \times  B^2\Z_2^c $. 
Indeed, since $\Z_2^c$ is discrete, a map $\omega: BG_b \to B^2\Z_2^c$ is equivalent to an extension of topological groups
\[
1 \to \Z_2^c \to G \to G_b \to 1.
\]
Note that the map $\Z_2^c \to G$ being injective implies $c \neq 1$, but otherwise having such an exact sequence is equivalent to having a central square one element in $G$.
We emphasize that taking the map $BG_b \to B^2 \Z_2^c$ to be trivial does not lead to a fermionic group with $c =1$, but to the trivially split fermionic group $G = \Z_2^c \times G_b$ instead.
The map $BG_b \to B\Z_2^T$ is equivalent to a homomorphism $\pi_0(G_b) \to \Z_2^T$ and so equivalent to a continuous homomorphism $G_b \to \Z_2^T$.
This in turn is equivalent to a homomorphism $G \to \Z_2^T$ with the additional property that $c$ is even.

The opposite fermionic group also has an interpretation in this language.
Let $\Z_4^{Tc}$ denote the group in which $T^2 = c$.
Recall that the Steenrod square can be considered as the map $Sq^1: B\Z_2^T \to B^2\Z_2^c$ classifying the fibration
\[
B\Z_2^c \to B\Z_4^{T} \to B\Z_2^{T}
\]
induced by the corresponding short exact sequence of groups.
Let $\phi: B\Z_2^T \times B^2 \Z_2^c \to B\Z_2^T \times B^2 \Z_2^c$ be the map given by projection onto the first factor and the composition
\[
B \Z_2^T \times B^2 \Z_2^c \xrightarrow{Sq^1 \times \id } B^2 \Z_2^c \times B^2 \Z_2^c \xrightarrow{+} B^2 \Z_2^c
\]
on the second factor.
Then postcomposing $BG_b\to B\Z_2^T \times  B^2\Z_2^c$ with $\phi$ maps the cohomology classes $(\theta, \omega) \in H^1(BG_b; \Z_2^T) \times H^2(BG_b; \Z_2^c)$ to 
\[
(\theta, \omega + Sq^1(\theta)) =  (\theta, \omega + \theta^2) \in H^1(BG_b; \Z_2^T) \times H^2(BG_b; \Z_2^c)
\]
where $\theta^2 = \theta \cup \theta$.
The extension of $G_b$ by $\Z_2^c$ corresponding to $\theta^2$ is the set $\Z_2^c \times G_b$ with multiplication $g_1 * g_2 = c^{\theta(g_1) \theta(g_2)} g_1 g_2$ for $g_1,g_2 \in G_b$.
The sum in $H^2(G_b, \Z_2^c)$ corresponds to the Baer sum of extensions of groups.
A short computation now shows that $\phi(\theta, \omega)$ is the topological data classifying $G^{\op}$.
\end{remark}

\subsection{Representations of fermionic groups}
\label{Sec:fermrep}

Let $V=V_0\oplus V_1$ be a complex super vector space. We define a fermionic group $\Aut^f(V)= \Aut^f(V)_0 \sqcup \Aut^f(V)_1$ 
of linear and antilinear automorphisms of $V$. 
Its even component $\Aut^f(V)_0$ is the set of $\C$-linear 
automorphisms of $V$. The odd component $\Aut^f(V)_1$ is given by the set of $\C$-antilinear 
automorphisms of $V$. The preferred element $c\in \Aut^f(V)$ is the grading involution $(-1)^F_V$ 
which is the identity on $V_0$ and minus the identity on $V_1$. 
Physical intuition tells us that the following is the correct notion for representation of a fermionic internal symmetry group on the state space of its quantum theory. 
This intuition will be confirmed by Proposition \ref{prop:1d classification}.

\begin{definition}\label{Def: F Rep}
Let $(G,c, \theta_G)$ be a fermionic group and $V$ a complex super vector space. A \emph{representation} of 
$(G,c, \theta_G)$ on $V$ is a homomomorphism of fermionic groups $\rho \colon G\to \Aut^f(V)$. 

Let $\rho \colon G\to \Aut^f(V)$ and $\rho' \colon G\to \Aut^f(V')$ be representations of $G$.
A \emph{morphism of $G$-representations} is a $\C$-linear map $f\colon V\to V'$ which commutes with the 
action of $G$.    
\end{definition}      

Concretely, this means that the even elements $g\in G_0$ act via $\C$-linear maps $\rho(g)\colon V \to V$,
the odd elements act via $\C$-antilinear maps $\rho(g)\colon \overline{V} \to V$, and $\rho(c)=(-1)^F_V$. A 
\emph{morphism of representations} is a $\C$-linear map which satisfies $f(\rho(g)[v])=\rho'(g)[f(v)]$. The condition 
$\rho(c)=(-1)^F_V$ can be understood as an implementation of the spin-statistics relation, see Section \ref{spin-st}.
Note that any action of a fermionic group on a purely even vector space $V=V_0$ factors through the group 
$G_b$ and that any fermionic group with $c=1$ only admits representations on even vector spaces. 

\begin{remark}\label{Rem: action on sVect}
We give a more abstract perspective on the action of a discrete fermionic groups on super vector spaces which will be useful later on (but can be safely ignored for the moment). The 2-group $\Z_2^T \times B\Z_2^{F}$ has an action on the category $\sVect$ of 
complex super vector spaces (we refer to Appendix \ref{Sec:action on sVect} for details on this action) 
where $T$ acts by sending a vector space $V$ to its complex conjugated vector space $\overline{V}$ and 
$(-1)^F$ acts by the natural isomorphism $(-1)^F_- \colon \id_{\sVect} \Longrightarrow \id_{\sVect} $ with
component at a super vector space $V$ given by the grading automorphism $(-1)^F_V: V \to V$ defined by $v \mapsto (-1)^{|v|} v$. Recall from 
Remark~\ref{Rem: fermionic groups as maps} that a fermionic group can equivalently be described by 
a homomorphism $G_b \to \Z_2^T\times B\Z_2^{F} $ of $\infty$-groups. This implies that $G_b$ acts on 
$\sVect$ in the following way 
\begin{align}
\psi(g) \colon \sVect & \to \sVect  \\ 
V & \longmapsto 
\begin{cases}
V & \text{if } \theta(g)=0 \\
\overline{V} & \text{if } \theta(g)=1.
\end{cases} 
\end{align}   
In addition the action comes with coherence isomorphisms $\psi_{g,g'}\colon \psi(g') \circ \psi(g) 
\Longrightarrow \psi(g'g)$ given by $(-1)^F_-$ if the 2-cocycle $\omega(-,-)$ classifying the 
extension
\begin{align}
1\longrightarrow \Z_2^c \longrightarrow G \longrightarrow G_b \longrightarrow 1
\end{align}      
is non-trivial when evaluated on $g,g'$ and the identity otherwise. A homotopy fixed point (see Remark~\ref{Rem: FP on Cat} for the definition) 
for this action consists of a vector space $V$ and linear maps $\rho(g)\colon \psi(g)[V]\to V$ for all 
$g\in G_b$. This means that the even elements of $G_b$ act linearly and the odd ones act antilinearly.
These maps have to satisfy the twisted composition law $ \rho(g') \circ \rho(g) = {\psi_{g,g'}}_V \circ
\rho(g'g)$. Such a twisted action of $G_b$ is equivalent to an action of the central extension $G$ where
$c$ acts by $(-1)^F_V$. This shows that we can identify the category of homotopy fixed points for the
action of $G_b$ on $\sVect$ with the category of fermionic representations of $G$.  
\end{remark}   

In physics unitary representations play an important role. 
Before defining unitary representations of fermionic groups we need to explain our conventions related to super hermitian vector spaces. We restrict our attention
to finite dimensional super vector spaces, since these will be most prominent in
our paper. However, the generalisations to infinite dimensional situations is 
straightforward. 

\begin{definition}
\label{def:hermvect}
Given a super vector space $H\in \sVect$, define its \emph{complex conjugate} $\overline{H} \in \sVect$ to be the same as an abelian group, but with complex-conjugated scalar multiplication.
A \emph{hermitian super vector space} is a super vector space $H\in \sVect$ together with an isomorphism
 $h_H \colon  \overline{H}^* \to H$ such that 
\begin{align}
	H \cong ({\overline{({\overline{H}})^*}})^* \xrightarrow{\overline{h_H}^*} \overline{H}^* \xrightarrow{h_H} H
\end{align}    
is the identity on $H$.
A \emph{unitary morphism} $f\colon (H,h_H) \to (H',h_{H'})$ is given by a linear map $f\colon H\to H' $ such that the adjoint $f^\dagger$ defined by
\begin{equation}
	\begin{tikzcd}
		\overline{H}^* \ar[r, "h_H"]   & H  \\
		\overline{H'}^* \ar[u, "\overline{f}^*"]  & \ar[l,"h_{H'}^{-1}"] H' \ar[u, "f^\dagger",swap]
	\end{tikzcd}
\end{equation}
is inverse to $f$. 
\end{definition} 

Here the canonical isomorphism $({\overline{({\overline{H}})^*}})^* \cong H$ is given by the evaluation map with the appropriate Koszul sign rule
\[
v \mapsto \ev_v: \overline{\overline{H}^*} \to \C \quad \ev_v(f) = (-1)^{|v||f|} f(v) \quad v \in \overline{H}, f \in \overline{\overline{H}^*}.
\]
A hermitian structure on a super vector space can equivalently be defined in terms of an inner product 
$\langle -, - \rangle \colon \overline{H}\otimes H \longrightarrow H^*\otimes H \longrightarrow \C $ such that $H_0$ and $H_1$ are orthogonal and $\langle v, w \rangle = (-1)^{|v||w|} \overline{\langle w, v \rangle} $. This implies that the
pairing of two even elements is a real number as usual, but the pairing of two odd elements is purely
imaginary. 
A \emph{super Hilbert space} is a hermitian super vector space $H$ such that $\langle v_0,v_0 
\rangle \geq 0$ for all even elements $v_0\in H_0$ and $\tfrac{\langle v_1,v_1 \rangle}{i} \geq 0$ for 
all odd elements $v_1\in H_1$. 
	Our notions of super Hermitian vector space and super Hilbert space are equivalent to the more common convention which we will call $\Z_2$-graded Hermitian vector spaces and $\Z_2$-graded Hilbert space where the pairing $	\langle - | - \rangle$ instead satisfies
	\[
		\langle v | w\rangle 	= \overline{\langle w | v\rangle}.
	\]
The correspondence is given by defining
	\begin{equation}
	\label{eq:Z2Hilb}
	\langle v | w\rangle := 
	\begin{cases}
	\langle v,w \rangle & v,w \text{ even,}
	\\
	 \frac{\langle v,w \rangle}{i} & v,w \text{ odd}
	\end{cases}
	\end{equation}
	and zero for vectors of different degree.
	The sign choice of $i$ and not $-i$ in the above is consistent with the convention that when $v \in H$ is odd, then
	\[
	\langle v, v \rangle \in i \R_{\geq 0}.
	\]
	In other words, the correspondence maps super Hilbert spaces to the usual notion of $\Z_2$-graded Hilbert space in which $\langle v |v \rangle \in \R_{\geq 0}$ also for odd $v$.

\begin{remark}
For a super vector space $H$ the existence of an isomorphism $\overline{H}^* \cong H$ implies that $H$ is finite-dimensional.
Therefore to obtain infinite-dimensional Hilbert spaces we have to change Definition \ref{def:hermvect} appropriately.
As this article only considers topological field theories, state spaces will be finite-dimensional and so our definition suffices.
\end{remark}

Let $(H,h_H)$ be a hermitian super vector space. We call a $\C$-antilinear map $H \to H$ (equivalently a $\C$-linear map $\overline{H}\to H$ )
\emph{anti-unitary} if it is unitary with respect to the hermitian structure $(-1)^F_{\overline{H}} \circ 
\overline{h_H}$ on $\overline{H}$. To motivate the factor $(-1)^F$ appearing in the definition we remark
that in our conventions, if $(H,h_H)$ is a super Hilbert space so is $(\overline{H},(-1)^F_{\overline{H}} \circ 
\overline{h_H})$. However, a brief computation shows that $(\overline{H}, 
\overline{h_H})$ is only a Hilbert space if $H$ is bosonic. 
Explicitly this means that an even antilinear $T: H \to H$ is anti-unitary when
\[
\langle Tv, Tw \rangle = (-1)^{|v|} \overline{\langle v, w \rangle}
\]
for all homogeneous $v$ and $w$. 

We denote by $U^f(H)$ the fermionic subgroup of $\Aut^f(H)$ consisting of unitary and anti-unitary automorphisms.

\begin{definition}
Let $H$ be a hermitian vector space and $G$ a fermionic group. A \emph{unitary representation of $G$ on
$H$} is a homomorphism of fermionic groups $\varphi \colon G\to U^f(H)$.
\end{definition}    


\begin{example}
Let $G = \Pin^-_1 = \Z_4^{FT}$ be the fermionic group consisting of a single time-rerversal symmetry with square $(-1)^F$.
A fermionic representation of $G$ on a super vector space $H = H_0 \oplus H_1$ is given by a complex antilinear even map $\phi(T): H \to H$ which squares to the grading automorphism.
Therefore it consists of a real structure on $H_0$ and a quaternionic structure on $H_1$.
In particular, note that such a representation exists if and only if $H_1$ is even-dimensional.
If $H$ is additionally a super Hilbert space and the representation is unitary, we have to additionally require $\phi(T)$ to be anti-unitary.
\end{example}
  
\subsection{Spacetime structure groups}\label{Sec:spacetime group}
In this section we let $(G,\theta,c)$ be an arbitrary internal symmetry group. Associated to
$G$ there is for any spacetime dimension $d$ a Euclidean spacetime structure group $H_d$~\cite{freedhopkins}. The 
construction of $H_d$ involves combining the internal symmetries in $G$ with local symmetries of 
spacetime and Wick rotation from Lorenzian to Euclidean signature.   
The passage from internal symmetry groups to spacetime structure groups has a concise formulation in
the language of fermionic groups in terms of the fermionic tensor products
\[
H_d:= (\operatorname{Pin}_d^+ \otimes G)_{0}.
\]
Projecting down to the first component induces a map $\rho_d: H_d \to  \Pin^+_d / \Z_2^c \cong O_d$. This allows us to consider 
$d$-dimensional manifolds with a tangential $H_d\xrightarrow{\rho_d} O_d$-structure.
The group
\[
\hat{H}_d := \operatorname{Pin}_d^+ \otimes G.
\]
will play an important role in our formulation of reflection structures. 

\begin{example}
Note that $H_1 = (G \otimes \Pin^+_1)_{\operatorname{ev}}$ is in general not isomorphic to $G$.
Instead, a short computation shows it is isomorphic to $G^{\op}$ with the corresponding canonical map to $O_1$.
\end{example}

\begin{example}
Let $D$ be a superdivision algebra over $\R$. 
Recall from Example~\ref{Ex: SA FG} that the sphere $S(D) = \frac{D^\times}{\R_{>0}}$
is a fermionic group with central element $[-1] \in S(D)$ and grading induced by the grading on $D$.
As $D$ ranges over the ten nonisomorphic superdivision algebras \cite{wallgradedbrauer}, $S(D)$ will range over the ten `internal symmetry groups' $I$ in Freed-Hopkins' ten-fold way~\cite{freedhopkins}.
The fermionic groups 
\[
H_n(D) = (S(D) \otimes \operatorname{Pin}^+_n)_{0}
\]
will give the ten spacetime structure groups of Freed-Hopkins, see table \ref{tab:ten-fold way}.
\end{example}

\begin{table}[h!]
	\centering
	\resizebox{\columnwidth}{!}{%
		\renewcommand{\arraystretch}{1.5}
		\begin{tabular}{l|l|l||l|l|l|l|l|l|l|l}
			FH label $s$ & $0$ & $1$ & $0$ & $1$ & $2$ & $3$ & $4$ & $-3$ & $-2$ & $-1$ 
			\\
			\hline
			Superdivision Algebra $D$ & $\C$ & $\C \operatorname{l}_1$ & $\R$ & $\Cl_{+1}$ & $\Cl_{+2}$ & $\Cl_{+3}$ & $\mathbb{H}$ & $\Cl_{-3}$ & $\Cl_{-2}$ &$\Cl_{-1}$
			\\
			\hline
			Internal fermionic symmetry $S(D)$ & $\Spin_1^c$ & $\Pin_1^c$ & $\Spin_1$ & $\Pin_1^+$  & $\Pin_2^+$ & $\Pin_3^+$ & $\Spin_3$ & $\Pin_3^-$ & $\Pin_2^-$ & $\Pin_1^-$ 
			\\
			\hline
			Spacetime structure group $H_n(D)$ & $\Spin^c_d$ & $\Pin^c_d$ & $\Spin_d$ & $\Pin^-_d$ & $\Pin^{\tilde{c}-}_d$ & $G^+_d$ & $G^0_d$ & $G^-_d$ & $\Pin^{\tilde{c}+}_d$ & $\Pin^+_d$ 
			\\
			\hline
			Symmetry class & A & AIII & D & BDI & AI & CI &  C & CII & AII & DIII 
			\\
			\hline 
			Symmetries & $Q$ & $Q, K$ & - & $K$ & $Q, T'$ & $S, K$ & $S$ & $S, T$ & $Q, T$ & $T$
		\end{tabular}%
	}
	\caption{Abstract mathematical relations between the ten-fold way internal symmetry groups, spacetime structure groups and superdivision algebras.
	The notation for the spacetime structure groups is given in \cite{freedhopkins}.
	In the symmetries row $Q$ denotes $U_1$-charge, $K$ denotes an anti-unitary particle-hole symmetry such as a sublattice symmetry, $S$ denotes an internal $\Spin_3$-symmetry, $T$ denotes a time reversal with square $(-1)^F$ and $T'$ denotes a time-reversing symmetry with square $1$.}
	\label{tab:ten-fold way}
\end{table}

We conclude this section by deriving an exact sequence featuring $H_d$. 
This follows from the following general proposition of which the proof is straightforward. 
\begin{proposition}
		\label{prop: exact sequences}
		Let $G,H$ be fermionic groups. 
		If $H$ has a nontrivial central element $c \in H$, then there is an exact sequence of topological groups
		\[
		1 \to G \overset{i}{\to} G \otimes H \overset{\pi}{\to} H_b \to 1,
		\]
		where the maps are given by the obvious inclusion and projection.
		If $G$ additionally has a nontrivial grading homomorphism, then this restricts to an exact sequence
		\[
		1 \to G_0 \overset{i}{\to} (G \otimes H)_0 \overset{\pi}{\to} H_b \to 1.
		\]
	\end{proposition}
%
%
	
	\begin{corollary}
		\label{cor: exact sequence}
	Let $G$ be a fermionic group with nontrivial central element and induced spacetime structure group $H_d$. 
	The sequence
	\[
	1 \to \Spin_d \to H_d \to G_b \to 1
	\]
	is exact.
	\end{corollary}

\subsection{2-group models for fermionic groups}
\label{Sec:fermskeletal}

One goal of this article is to give explicit algebraic characterizations of topological field theories.
To do this nicely in terms of a finite amount of concrete information, it is convenient to use combinatorial models that are as small as possible for all homotopic information involved.
For example, the space $SO_2$ can be made into a simplicial set by dividing the circle into $m$ points and $m$ line segments.
Algebraically we could describe this as the monoidal category $\mathcal{C}_m$ which has $\Z_m$ as its group of objects and 
\[
\Hom(n_1, n_2) = \{ k \in \Z: n_2 - n_1 \equiv k \pmod m \}
\]
in which composition is given by addition.
Of course, for most applications it is easiest to use $m = 1$.
To this description corresponds the simplicial set coming from the fact that $SO_2 = B\Z$ is the nerve of $\mathcal{C}_1 = * \DS \Z$, the category with one object and a $\Z$ amount of morphisms.
There is a monoidal equivalence $\mathcal{C}_1 \cong \mathcal{C}_m$ which maps the unique object to $0 \in \Z_m = \operatorname{ob} \mathcal{C}_m$ and the morphism $k \in \Z$ to the automorphism $mk \in \Hom_{\mathcal{C}_m}(0,0)$.
Note that this is not a strict isomorphism, but it is fully faithful and essentially surjective.
Since in the category $* \DS \Z$ all isomorphic objects are equal, we call it a \emph{skeletal model} for $\mathcal{C}_m$. 
A more abstract way to think of the category $\mathcal{C}_m$ is as a homotopy quotient or semidirect product of $2$-groups 
\[
\mathcal{C}_m = \Z_m \rtimes (* \DS \Z)
\]
where the action of $* \DS \Z$ on $\Z_m$ is induced by the canonical surjective homomorphism $\Z \to \Z_m = Z(\Z_m)$.
We refer the reader to Appendix \ref{App: 2-Group} for the basic theory of $2$-groups and to Appendix \ref{App: Exact} for their semidirect products and short exact sequences.

From the introduction above we can conclude that skeletal models are often easiest to describe.
However, for applications to fermionic groups in this paper, slightly larger models are more convenient.
For example, from the model $* \DS \Z$, we can easily describe the exact sequence of $2$-groups
\[
1 \to * \DS \Z \xrightarrow{2 \cdot } * \DS \Z \to * \DS \Z^c_2 \to 1
\]
corresponding to the fibration of topological spaces
\[
B\Spin_2 \to BSO_2 \to B^2 \Z^c_2.
\]
However, if we instead want to describe the exact sequence
\[
1 \to \Z^c_2 \to \Spin_2 \to SO_2 \to 1
\]
in terms of $2$-groups, we prefer to use the model $\mathcal{C}_2 = \Z_2^c \rtimes (* \DS \Z)$ with two objects $1$ and $c$ for $\Spin_2$.
We call this the \emph{fermionically skeletal model} of $\Spin_2$.
Then the map to $SO_2 = * \DS \Z$ is given by mapping a morphism $k \in \Hom(n_1, n_2)$ for $n_1,n_2 \in \Z_2$ to $k \in \Z$ itself.
Under the equivalence $* \DS \Z \cong \Z_2^c \rtimes(* \DS \Z)$ described in last paragraph (specializing to the case that $m := 2$), the map $* \DS \Z \cong \Spin_2 \to SO_2 = * \DS \Z$ is indeed given by multiplication by $2$ as above.

We generalize the above discussion from $\Spin_2$ to a general fermionic group $G$ that is not bosonic, i.e. $c \neq 1$.
In this paper, it will turn out to suffice to look at the truncation $\mathcal{G} := \pi_{\leq 1} G$ in which we killed all higher homotopy groups, the fundamental $2$-group of $BG$.
Therefore we will adapt the definition of a fermionic group to fermionic $2$-groups:

\begin{definition}
A \emph{fermionic 2-group} is a $\Z^T_2$-graded 2-group $\mathcal{G}= \mathcal{G}_0 \sqcup \mathcal{G}_1$ together with a braided monoidal functor $\Z_2^c\to Z(\mathcal{G})$ where $Z(\mathcal{G})$ is the Drinfeld center of $\mathcal{G}$ such that the image lies in the degree zero component of $\mathcal{G}$
\end{definition}

\begin{remark}
When unpacking the definition of a braided functor from $\Z_2^c$ to $Z(\mathcal{G})$ we find that it is equivalent to specifying an element $c\in \mathcal{G}_0$ together with a half braiding $\sigma_{c,g}\colon c\otimes g \to g \otimes c $ satisfying $\sigma_{c,c}=\id$ and an isomorphism $c\otimes c \to 1$ compatible with the half braiding.	 
This makes precise the encoding of a central element $c$ of square $1$.
\end{remark}

Given a fermionic group $G$, the fundamental $2$-group is a fermionic $2$-group in which most coherence data is strict.
However, the canonical $2$-group model for the fundamental $2$-group is usually to big for practical purposes.  
Therefore we want to work with skeletal models, for which the coherence data typically becomes more interesting.
However, just like for the fermionic group $\Spin_2$ above, we will not use a skeletal model for $\pi_{\leq 1} G$ but start with a skeletal model $\mathcal{G}_b$ for $\pi_{\leq 1} G_b$ instead.

Recall from Remark \ref{Rem: fermionic groups as maps} that the extension 
\[
1 \to \Z_2^c \to G \to G_b \to 1
\]
can be equivalently described as a map $\omega: B G_b \to B^2 \Z_2^c$.
This induces a map $B\pi_{\leq 1} G_b \to BG_b \to B^2 \Z_2^c$, which we can equivalently give as a map between skeletal $2$-groups $F: \mathcal{G}_b \to * \DS \Z_2^c$.

\begin{definition}
Let $G$ be a fermionic group and let $\mathcal{G}_b$ denote the skeletal model for the $2$-group $ \pi_{\leq 1} G_b$.
The \emph{fermionically skeletal model} of $G$ is the fermionic $2$-group 
\[
\Z_2^c \rtimes \mathcal{G}_b
\]
where $\mathcal{G}_b$ acts on $\Z_2^c$ via the map $\mathcal{G}_b \to  * \DS \Z_2^c$.\footnote{If $A$ is an abelian group, then the $2$-group $*\DS A$ acts on $A$ by the identity map $*\DS A \to *\DS A$.
Explicitly this means letting a path $a: * \to *$ in $*\DS A$ corresponding to $a \in A$ be mapped to the natural automorphism of the identity functor on the bicategory $* \DS A$ mapping $*$ to $a$.
In monoidal category language, $\rho(a)$ is an unpointed natural transformation from $A$ to $A$ with $\rho(a)_* = a$ and otherwise trivial.}
\end{definition}

Unlike in the canonical model of $\pi_{\leq 1} G$, the associator of the fermionically skeletal model is in general nontrivial.
However, the element $c \in \mathcal{G}$ will still strictly square to 1 and commute with other objects.
By the classification of maps between skeletal $2$-groups given in Lemma \ref{Lem:skeletalmap} in Appendix \ref{App: 2-Group}, we get an explicit description of $2$-group homomorphisms $F: \mathcal{G}_b \to  * \DS \Z_2^c$ as follows.
The data consists of a map $\Xi: \pi_0(G_b) \times \pi_0(G_b) \to \Z_2^c$ (the monoidality data of the functor $F$) and a map $\Gamma: \pi_1(G_b) \to \Z_2^c$ (the functor $F$ on morphisms $1 \to 1$).
The functor condition is equivalent to the equalities
\[
\label{eq:functor cond}
\Gamma(\gamma \delta) = \Gamma(\gamma) \Gamma(\delta) \qquad \Gamma(g \gamma g^{-1}) = \Gamma(\gamma)
\]
for $\gamma, \delta \in \pi_1(G_b)$ and $g \in \pi_0(G_b)$.
The pentagon is equivalent to the equality 
\[
\Gamma \circ \alpha_{\mathcal{G}_b} = d R
\]
of maps $\pi_0 (G_b) \times \pi_0(G_b) \times \pi_0(G_b) \to \Z_2^c$, where $\alpha_{\mathcal{G}_b}: \pi_0 (G_b) \times \pi_0(G_b) \times \pi_0(G_b) \to \pi_1(G_b)$ is the associator of $\mathcal{G}_b$.
We derive the condition that $\Gamma \circ \alpha_{\mathcal{G}_b} \in H^3(\pi_0(G_b), \Z_2^c)$ is trivial in cohomology and moreover is uniquely determined by $\Xi$.

$2$-morphisms of $2$-groups are unpointed monoidal natural transformations between two monoidal functors.
From Lemma \ref{Lem:skeletalmap} it follows that in this case they are given by maps $\sigma: \pi_0(G_b) \to \Z_2^c$ such that $\Xi_1 = \Xi_2 \cdot d \sigma$.
In other words, $\Xi_1 \Xi_2^{-1}$ measures the failure of $\sigma$ to be a homomorphism.
In the special case where $\Gamma \circ \alpha_{\mathcal{G}_b} = 0$ we derive that $\Xi \in Z^2(\pi_0(G_b), \Z_2^c)$ is a cocyle and so up to $2$-isomorphism only the class of $\Xi$ in $H^2(\pi_0(G_b), \Z_2^c)$ matters.
Otherwise, when we fix one reference $\Xi_0$ trivializing $\Gamma \circ \alpha_{\mathcal{G}_b}$ in cohomology, all other $\Xi$ are up to $2$-isomorphism still classified by $H^2(\pi_0(G_b), \Z_2^c)$.
In other words, if we fix $\Gamma$ such that equations \ref{eq:functor cond} hold and $\Gamma \circ \alpha_{\mathcal{G}_b}$ is zero in cohomology, then ways to make it into a homomorphism $\mathcal{G}_b \to * \DS \Z_2^c$ are classified up to $2$-isomorphism by an $H^2(\pi_0(G_b), \Z_2^c)$-torsor.
Summarizing the above, we have

\begin{proposition}
Let $\mathcal{G}_b$ be a skeletal $2$-group. Then equivalence classes of homomorphisms 
\[
\mathcal{G}_b \to * \DS \Z_2^c
\]
fit in the short exact sequence
\begin{align}
1 &\to H^2(\pi_0(G_b), \Z_2^c) \to \pi_0 \Hom(\mathcal{G}_b , * \DS \, \Z_2^c) 
\\
&\to  \{\Gamma \in \Hom(\pi_1(\mathcal{G}_b), \Z_2^c): \Gamma(g \gamma g^{-1}) = \Gamma(\gamma), \Gamma \circ \alpha_{\mathcal{G}_b} = 0 \in H^3(\pi_0(\mathcal{G}_b), \Z_2^c)\} \to 1 \ \ . \label{ses maps to Z2}
\end{align}
\end{proposition}

We can derive the above data from the original fermionic group as follows.
The short exact sequence of Lie groups
\[
1 \to \Z_2^c \to G \to G_b \to 1
\]
induces a long exact sequence of homotopy groups of which the tail looks like
\[
1 \to \pi_1(G) \to \pi_1(G_b) \xrightarrow{\Gamma} \Z_2^c \to \pi_0(G) \to \pi_0(G_b) \to 1.
\]
In particular notice how the homomorphism $\Gamma$ is nonzero if and only if $1$ and $c$ are connected by a path in $G$.
Note that $\Gamma$ is given by $\pi_1$ of the map $G_b \to B^2 \Z_2$, which under the homotopy hypothesis corresponds to what the monoidal functor does to automorphisms of the identity object.
Therefore it corresponds to the morphism we called $\Gamma$ before.
In the case where $\Gamma$ is trivial we are in the kernel of the map of sequence \ref{ses maps to Z2} and so we get a canonical class $\Xi \in H^2(\pi_0(G_b), \Z_2^c)$.
This is the class classifying the short exact sequence
\[
1 \to \Z_2^c \to \pi_0(G) \to \pi_0(G_b) \to 1.
\]
More generally $\Xi$ has a homotopical interpretation as providing the comparison data of the $k$-invariants of the domain and range of the map of $2$-types $\pi_{\leq 2} G_b \to B^2 \Z_2$.
In this special case where the $k$-invariant of the codomain is trivial, this boils down to a trivialization of the composite
\[
B\pi_0(G) \xrightarrow{k} B^3 \pi_1(G) \to  B^3 \Z_2^c
\]
at least when the action of $\pi_0(G)$ on $\pi_1(G)$ is trivial.
As the $k$-invariant corresponds to the associator under the homotopy hypothesis, this is the description of $R$ from before.

\begin{example}
\label{Ex: pin ferm gp}
Consider the case $\pi_0 G_b = \Z_2$ and $\pi_1 G_b = \Z$.
The two possible actions of $\pi_0G_b$ on $\pi_1 G_b$ are the trivial one and $n \mapsto -n$, which we denote $\Z_-$.
The groups classifying associators are $H^3(B\Z_2,\Z) = 0$ for the trivial action and $H^3(B\Z_2,\Z_-) = \Z_2$ for the nontrivial one.
Therefore there are three isomorphism classes of $2$-groups with $\pi_0 G_b = \Z_2$ and $\pi_1 G_b = \Z$.
They are the fundamental 2-groups of 
\[
SO_2 \times \Z_2, \quad O_2 \cong SO_2 \rtimes \Z_2 \text{ and } \Pin^-_2 \cong \frac{SO_2 \rtimes \Z_4}{\Z_2}
\]
For example, the nontrivial element of $H^3(B\Z_2,\Z_-) \cong H^2(B\Z_2,U(1)_-)$ corresponds to the nontrivial extension of Lie groups of $\Z_2$ by $U(1)$ with nontrivial action of $\Z_2$ on $U(1)$, which is $\Pin^-_2$.

We study the possible maps $\mathcal{G}_b \to *\DS \Z_2$ for the $2$-group corresponding to $O_2$.
The cohomology group $H^2(BO_2, \Z_2^c) \cong \Z_2^2$ is generated by $w_2$ and $w_1^2$ to which correspond four extensions of Lie groups
\[
1 \to \Z_2^c \to G \to O_2 \to 1.
\]
These are the trivial extension, $\Pin^+_2, \Pin^-_2$ and the extension in which multiplying $A_1, A_2 \in O_2$ results in an extra $c \in \Z_2^c$ if and only if $\det A_1 = \det A_2 = -1$.
In terms of the corresponding skeletal $2$-group $\mathcal{G}_b = \Z_2\rtimes * \DS \Z$, we can describe the four maps $\mathcal{G}_b \to * \DS \Z_2^c$ as follows.
There are two homomorphisms $\Gamma: \Z = \pi_1(\mathcal{G}_b) \to \Z_2^c$; the trivial one and the surjective one.
The extra condition $\Gamma(n) = \Gamma(-n)$ that appears when the $\pi_0(G_b)$-action on $\pi_1(G_b)$ is nontrivial is automatically satisfied.
The pentagon condition is empty, because the associator of $\mathcal{G}_b$ is trivial.
If $\Gamma$ is nontrivial, then there is a loop in $G_b$ that lifts to a path $1 \to c$ and so the extension is either $\Pin^{+}_2$ or $\Pin^-_2$.
Next there are two possibilities for $\Xi$ up to $2$-isomorphism since $H^2(B\Z_2,\Z_2^c) = \Z_2$.
\end{example}


We proceed to write out the definition of the fermionically skeletal model $\Z_2^c \rtimes \mathcal{G}_b$ as a monoidal category.
We use the above explicit description of the $2$-group homomorphism $\mathcal{G}_b \to * \DS \Z_2$ in terms of $\Gamma$ and apply the definition of the semidirect product from Appendix \ref{App: 2-Group}.
The objects are the set $\pi_0(G_b) \times \Z_2^c$ in which the tensor product is given by 
\[
g_1 \otimes g_2 = \Xi_{g_1,g_2} g_1 g_2 \quad c \otimes c = 1 \quad c \otimes g = g \otimes c.
\]
In other words, in case that $1 \neq c \in \pi_0(G)$, this is exactly $\pi_0(G)$.
However, in the other case when $\pi_0(G) \cong \pi_0(G_b)$, this collection is twice as large.
Given $\epsilon_1, \epsilon_2 \in \{\pm 1\}$ and $g_1,g_2 \in \pi_0(G)$ the hom-set $g_1 c^{\epsilon_1} \to g_2 c^{\epsilon_2}$ is empty unless $g_1 = g_2$.
In that case we have
\[
\Hom_{\Z_2^c \rtimes \mathcal{G}_b}(g c^{\epsilon_1}, g c^{\epsilon_2}) = \{g \gamma \in  \Hom_{\mathcal{G}_b}(g,g) : \Gamma(\gamma) = c^{\epsilon_1 + \epsilon_2}\}
\]
where we used that by definition of the skeletal model, all morphisms $g \to g$ in $\mathcal{G}_b$ are of the form $g \gamma$ for $\gamma \in \pi_1(\mathcal{G}_b)$.
Composition and tensor product of morphisms in $\Z_2^c \rtimes \mathcal{G}_b$ is given by composition and tensor product in $\mathcal{G}_b$, after appropriately changing the domain and codomain.
The associator is similarly induced from the associator of $\mathcal{G}_b$.
The reader should be warned that even though a skeletal $2$-group $\mathcal{G}_b$ with trivial associator yields a fermionically skeletal $2$-group with trivial associator, the associator of the skeletal model of $\Z_2^c \rtimes \mathcal{G}_b$ might be nontrivial.

To make compositions in the fermionically skeletal model more intuitive, we introduce the following notation.
Given $\gamma \in \pi_1(G_b)$, we denote by the same symbol the morphism $\gamma: 1 \to \Gamma(\gamma)$ in $\Z_2^c \rtimes \mathcal{G}_b$.
The other morphism in $\Z_2^c \rtimes \mathcal{G}_b$ which projects to $\gamma \in \pi_1(G_b)$ under the map $G \to G_b$ is then given by $c \gamma = \id_c \otimes \gamma: c \to c \Gamma(\gamma)$.
We can then compose morphisms in $\mathcal{G}_b$ keeping in mind that $g_1 \otimes g_2 = \Xi_{g_1,g_2} g_1 g_2$ and of course $\gamma g$ possibly is not equal to $g \gamma$ already in $\mathcal{G}_b$.
So for example, if $\Gamma(\gamma) = c$, then $c\gamma \circ \gamma$ is a loop in $\mathcal{G}$.

We briefly turn to describing the opposite fermionic group in terms of $2$-group language.
Note that there is a single non-trivial $2$-group isomorphism
\[
\phi: \Z^T_2 \times *\DS \Z^c_2 \to \Z^T_2 \times * \DS \Z^c_2
\]
given by the identity functor together with nontrivial monoidality data over the object pair $(T, T)$.
Since this is the only nontrivial $2$-group isomorphism up to $2$-isomorphism, it corresponds under delooping to the (not nullhomotopic) map $\phi$ defined in Remark \ref{Rem: fermionic groups as maps}. 
Therefore, if $\mathcal{G}_b \to \Z_2^T \times *\DS \Z_2^c$ defines a fermionic $2$-group, to describe its opposite we have to compute the composition
\[
\mathcal{G}_b \to \Z_2^T \times *\DS \Z_2^c \xrightarrow{\phi} \Z_2^T \times *\DS \Z_2^c \to * \DS \Z_2^c.
\]
Since $\phi$ is the identity as a functor, the opposite has the same $\Gamma$.
A short computation shows that its $\Xi$ is changed to $\Xi^{\op}_{g_1,g_2} := c^{\theta(g_1) \theta(g_2)} \Xi_{g_1, g_2}$.

We now turn to spacetime structure groups.
Recall the exact sequence from Corollary \ref{cor: exact sequence}
\[
1 \to \Spin_d \to H_d \to G_b \to 1
\]
which we want to restrict to an exact sequence of $2$-groups
\[
1 \to \pi_{\leq 1} \Spin_d \to \pi_{\leq 1} H_d \to \pi_{\leq 1} G_b \to 1
\]
for the cases $d = 1,2$.
As before we will work with a skeletal model of $\pi_{\leq 1} G_b$ and a fermionically skeletal model of $\Spin_d$.
The idea is to take a set-theoretic section of the projection map and study the exact sequence through the behaviour under multiplication.
In other words we want to compute the action of the latter $2$-group on the first and realize the middle $2$-group as a semidirect product as explained in Appendix~\ref{App: Exact}.

Let $e_1 \in \Pin^+_d$ denote the generator corresponding to a vector in $\R^d$ of unit length and let $\Xi: \pi_0(G_b) \times \pi_0(G_b) \to \Z_2^c$ and $\Gamma: \pi_1(G_b) \to \Z_2^c$ be classifying $\mathcal{G}_b \to * \DS \Z_2^c$ as before.
Let $s: G_b \to G$ be a section which in general is neither a group homomorphism nor continuous. 
Then we can define a set-theoretic section of $H_d \to G_b$ by $s(g) \otimes e_1^{\theta(g)} \in H_d$ which induces a set theoretic section $\Obj \pi_{\leq 1} G_b \to \pi_{\leq 1} H_d$ which is the main ingredient in the computation of the action resulting from a short exact sequence at the beginning of Appendix~\ref{App: Exact}. The action of $g$ on $\pi_{\leq 1} \Spin_d$ is by conjugation with this element which turns out to be trivial for $d=1$. 
The short computation
\[
(s(g_1) \otimes e_1^{\theta(g_1)} ) (s(g_2) \otimes e_1^{\theta(g_2)}) = c^{\theta(g_1) \theta(g_2)} s(g_1) s(g_2) \otimes e_1^{\theta(g_1) + \theta(g_2)} = c^{\theta(g_1) \theta(g_2)} \omega_{g_1,g_2} s(g_1 g_2) \otimes  e_1^{\theta(g_1 g_2)}
\]
shows that the failure of the section to be multiplicative is measured by the $2$-cocycle $\omega^{\op}: G_b \times G_b \to \Z_2$ corresponding to the extension
\[
1 \to \Z_2^c \to G^{\op} \to G_b \to 1.
\]
Since this $2$-cocycle induces $\Xi^{op}$ on the level of $2$-groups, we can conclude that for $d =1$ the action of $\mathcal{G}_b$ on $\Spin_1 = \Z_2^c$ induced by the exact sequence of $2$-groups corresponds to the map $\mathcal{G}_b \to * \DS \Z_2^c$ classifying $G^{\op}$, so consisting of $\Xi^{\op}$ and $\Gamma$.

We now consider $d = 2$ and want to understand the induced action of $\mathcal{G}_b$ on $\Spin_2 = \Z_2^c \rtimes B\Z$.
Since we already computed the action of $\mathcal{G}_b$ on $\Z_2^c \subseteq \Spin_2$, we now consider the action on a generating morphism $\eta: 1 \to c$ in $\Spin_2$.
In explicit topological terms we express it as
\[
\eta(t) = \cos(\pi t) + e_1 e_2 \sin(\pi t),
\]
where $e_2 \in \Pin^+_2$ comes from a vector in $\R^2$ orthonormal to $e_1$.
We now compute the action of $g \in G_b$ on $\eta$ inside $H_2$ as
\begin{align*}
\eta(t) (s(g) \otimes e_1^{\theta(g)}) &= s(g) \otimes (e_1^{\theta(g)} \cos(\pi t) + (e_1 e_2) e_1^{\theta(g)} \sin(\pi t))
\\
&=  s(g) \otimes (e_1^{\theta(g)} \cos(\pi t) + c^{\theta(g)} e_1^{\theta(g)} (e_1 e_2) \sin(\pi t))
\\
&=  s(g) \otimes (e_1^{\theta(g)} \cos((-1)^{\theta(g)} \pi t) + e_1^{\theta(g)} (e_1 e_2) \sin((-1)^{\theta(g)}\pi t))
\\
&= (s(g) \otimes e_1^{\theta(g)}) \eta((-1)^{\theta(g)} t).
\end{align*}
Note that the inverse of $\eta$ is $\eta(1-t) = c\eta(-t)$, so that $t \mapsto \eta(-t)$ is $c \eta^{-1}$
We obtain that $\pi_0(G_b)$ maps $\eta$ to $\eta$ when $\theta(g) = 0$ and to the other generating morphism $c \eta^{-1}: 1 \to c$ when $\theta(g) = 1$.

There is more data involved in describing the action (in particular at the level of morphisms), but these are straightforward to compute following Appendix~\ref{App: Exact}.

\begin{remark}
For the purpose of exposition, we have been careful to distinguish for example the monoidal category $* \DS \Z$ seen as a $2$-group from the realization of its nerve $B\Z$.
In the rest of the paper we will often implicitly apply the homotopy hypothesis and use the notation $B\Z$ for both.
\end{remark}

\subsection{Fermionically graded algebras}
\label{Sec:fermgraded}

If $G$ is a discrete group, a strongly $G$-graded algebra is an algebra $\mathcal{A}$ together with a direct sum decomposition of vector spaces
\[
\mathcal{A} = \bigoplus_{g\in G} A_g
\]
with the property that the multiplication map restricts to vector space isomorphisms $A_g \otimes_A A_h \cong A_{gh}$.
In \cite{oritthesis}, it is shown that two-dimensional extended topological field theories with target $\operatorname{Alg}$ and structure group $SO_2 \times G$ are classified by strongly graded $G$-algebras with a symmetric Frobenius structure $\lambda: \mathcal{A} \to \C$ such that for all $g \in G$ not equal to $1$
\[
\lambda(a_g) = 0 \quad \forall a_g \in A_g.
\]
Since the next few sections in a certain sense generalize her results to the case where $G$ is a fermionic group, we briefly discuss analogues of strongly graded algebras to this setting.

So let $(G, c, \theta)$ be a discrete fermionic group.
The idea motivated by physics is that time-reversing elements of $G$ should anti-commute with $i$ in a similar way to how they are required to act antilinearly in Definition \ref{Def: F Rep}.
To include fermions properly, we will have to replace algebras over $\C$ with superalgebras over $\C$, which are defined to be $\Z_2$-graded algebras 
with the appropriate Koszul signs.
In Appendix~\ref{App: Super alg}, we summarize the basic definitions and technical results on superalgebras we need in the main text.
We start with a definition which does not take into account the fermion parity element $c \in G$:
\begin{definition}
A \emph{strongly $(G, \theta)$-graded algebra} is a complex super vector space $\mathcal{A}$ which is a superalgebra over $\R$ together with a direct sum decomposition of complex super vector spaces
\[
\mathcal{A} =  \bigoplus_{g\in G} A_g
\]
such that multiplication defines a real-linear isomorphism $A_g \otimes_{A_e} A_h \cong A_{gh}$ and $a_g i = (-1)^{\theta(g)} i a_g$ if $a_g \in A_g$.
\end{definition}

Note that the definition recovers the notion of a strongly $G$-graded algebra over $\C$ in the case $\theta$ is trivial.
We now additionally want to impose that $c \in G$ corresponds to the $\Z_2$-grading of the superalgebra, something that will be motivated by the spin-statistics connection explained later.
If $A$ is a superalgebra, we denote the $(A,A)$-bimodule associated to the grading isomorphism $(-1)^F_A: A \to A$ by $A_{(-1)^F}$, see Definition \ref{def:spin-statistics bimodule}.

\begin{definition}
A \emph{(strongly) fermionically graded algebra} is a $(G,\theta)$-graded algebra
\[
\mathcal{A} =  \bigoplus_{g\in G} A_g
\]
such that $A_c = A_{(-1)^F}$ compatibly with the multiplication in the sense that
\[
 (-1)^F \cdot (-1)^F = 1 \quad (-1)^F a_g = (-1)^{|a_g|} a_g (-1)^F \in A_{g c}
\]
for all $a_g \in A_g$.
In the case $G$ is bosonic, i.e. $c = 1$, we call $\mathcal{A}$ a \emph{bosonically graded algebra} instead and we still require $A_c = A = A_{(-1)^F}$.
\end{definition}

\begin{example}
Let $\theta: H \to \Z_2^T$ be a bosonic group and let $G = H \times \Z_2^c$ be the induced split fermionic group.
Let $\mathcal{B}$ be a strongly $(H,\theta)$-graded algebra.
Then 
\[
\mathcal{B} \oplus \mathcal{B}_{(-1)^F} = \frac{\mathcal{B}[x = (-1)^F]}{(x^2 -1,xb = (-1)^{|b|} bx)}
\]
 is a fermionically graded algebra.
 For example, if $H$ is trivial and $\mathcal{B} = Cl_{p,q}$ is a Clifford algebra, then 
 \[
 \mathcal{B} \oplus \mathcal{B}_{(-1)^F} \cong
 \begin{cases}
 Cl_{p,q} \oplus Cl_{p,q} & p-q \equiv 0 \pmod 4
   \\
  Cl_{p,q+1}  & p-q \equiv 1 \pmod 4
 \\
  \C l_{p,q} &  p-q \equiv 2 \pmod 4
  \\
  Cl_{p+1,q} & p-q \equiv 3 \pmod 4
 \end{cases}
 \]
To show this, let $\operatorname{vol} := e_1 \dots e_{p+q}$ be a volume element.
  Note that $a:= (-1)^F \operatorname{vol}$ is of degree $(-1)^{p+q} = (-1)^{p-q}$ and graded commutes with all $e_i$.
  Moreover, $a^2 = (-1)^{p+q} \operatorname{vol}^2$ and a short computation shows that $\operatorname{vol}^2 = 1$ if $p-q \equiv 0,1 \pmod 4$ and $-1$ otherwise.
  Therefore $a^2 = 1$ when $p-q \equiv 0,3 \pmod 4$ and $-1$ otherwise.
So we see that there are four cases as above.
\end{example}

\begin{example}
\label{ex:pin-}
Let $G = \Pin^+_1 = \Z_2^c \times \Z_2^T$ be the fermionic group with a single time-reversal of square $1$.
We look for fermionically $G$-graded algebras $\mathcal{A}$ with $A_1 = \C$. 
Let $x_T \in A_T$ denote a preferred basis element.
All of $\mathcal{A}$ is uniquely determined once we make the choice whether $A_T = \C$ or $\Pi \C$ and what the square of $x_T$ is.
Because of reasons that will become apparent when we will discuss this example in the setting of two-dimensional TFTs (Example \ref{ex:pin-tft}), we restrict to the cases $x_T^2 = \pm 1$.
Ignoring the $\Z_2^T$-grading, these choices result in the following real superalgebras
\[
A \oplus A_T \cong
\begin{cases}
M_2(\R) & x_T^2 = 1, |x_T| = 0
\\
\mathbb{H} & x_T^2 = -1, |x_T| = 0
\\
Cl_{+2}& x_T^2 = 1, |x_T| = 1
\\
Cl_{-2}& x_T^2 = -1, |x_T| = 1
\end{cases}
\]
From the last example we also obtain
\[
\mathcal{A} \cong
\begin{cases}
M_2(\R) \oplus M_2(\R)  & x_T^2 = 1, |x_T| = 0
\\
\mathbb{H} \oplus \mathbb{H} & x_T^2 = -1, |x_T| = 0
\\
\C l_{2} & x_T^2 = 1, |x_T| = 1
\\
\C l_{2}& x_T^2 = -1, |x_T| = 1
\end{cases}
\]
Note that even though the last two algebras are isomorphic as superalgebras, they are not isomorphic as $G$-graded algebras.

If we would have started with the other Morita-invertible complex superalgebra $A_1 =  \C l_1$ up to Morita equivalence, we would get a similar result as follows.
First of all, note that there are again two isomorphism classes of invertible $(A,A)$-bimodules given by $A$ and $\Pi A$.
If $|x_T| =0$, then $x_T$ satisfies $e x_T = x_T e$ and otherwise $e x_T = - x_T e$ without loss of generality because $\Pi \C l_1 \cong (\C l_1)_{(-1)^F}$.
In case $|x_T| = 0$, $x_T$ and $i$ generate an algebra $B$ isomorphic to $M_2(\R)$ if $x_T^2 = 1$ and $\mathbb{H}$ otherwise.
Because these generator commute with $e$ we obtain
\[
A \oplus A_T \cong B \otimes_\R \C l_1 \cong M_2(\C l_1)
\]
in both cases, even though they are not isomorphic as graded algebras.
Note that $\C l_1 \oplus (\C l_1)_{(-1)^F} \cong \C l_2$ because $ie(-1)^F$ and $e$ are anticommuting odd generators so that
\[
\mathcal{A} \cong M_2(\C l_1)
\]

In case $|x_T| = 1$, the three generators $T, e$ and $iT$ are odd and mutually anticommute, so that $A \oplus A_T$ is a Clifford algebra.
Computing the square $(iT)^2 = T^2$ results in
\[
A \oplus A_T \cong 
\begin{cases}
Cl_{+3} & T^2 = 1
\\
Cl_{1,2} & T^2 = -1
\end{cases}
\]
All in all we obtain
\[
\mathcal{A} \cong
\begin{cases}
M_2(\C l_1)  & x_T^2 = 1, |x_T| = 0
\\
M_2(\C l_1) & x_T^2 = -1, |x_T| = 0
\\
\C l_{4} & x_T^2 = 1, |x_T| = 1
\\
\C l_{2,2} & x_T^2 = -1, |x_T| = 1
\end{cases}
\]
We have now specified eight non-isomorphic fermionically $G$-graded algebras for which $A_1$ is Morita invertible. 
\end{example}

Later one we also need a generalisation of the definition of fermionically graded-algebras to 2-groups. 
We start with the most general definition

\begin{definition}\label{Def: 2-group graded alg}
Let $(\mathcal{G}=\mathcal{G}_0\sqcup \mathcal{G}_1 ,1, \alpha, c, \sigma_{c,-}, \omega\colon c\otimes c \to 1 )$ be a fermionic 2-group. A \emph{$\mathcal{G}$-graded algebra} is a complex super vector space
\begin{align}
\mathcal{A}= \bigoplus_{g\in \Obj(\mathcal{G})} A_g
\end{align} 
together with 
\begin{itemize}
	\item an element $1\in A_1$,
	\item $\R$-linear maps $\cdot \colon A_g\otimes A_{g'}\to A_{g\otimes g'}$, 
	\item and
$\C$-linear maps $F_\gamma\colon A_g\to A_{g'}$ for all morphisms $\gamma\colon g\to g'$ 
\end{itemize}
such that $1$ is a unit for the multiplication $\cdot$, and 
\begin{align}
F_{\alpha(g,g',g'')}[(a_g\cdot a_{g'})\cdot a_{g''}] & = a_g\cdot(a_{g'}\cdot a_{g''}) \\
 a_g i &= (-1)^{\theta(g)} i a_g \\ 
 F_{\gamma_1 \otimes \gamma_2} (a_{g_1}  a_{g_2}) &= F_{\gamma_1}(a_{g_1}) F_{\gamma_2}(a_{g_2}) \text{ for all morphisms } \gamma_1: g_1 \to h_1, \gamma_2: g_2 \to h_2
\end{align}
and $A_c=A_{(-1)^F}$ as a bimodule over $A_1$ such that the data of $\omega: c^2 \cong 1$ induces the multiplication in the sense that
\begin{align}
F_{\omega}[(-1)^F\cdot (-1)^F]=1 
\end{align}
holds and the commutation data $\sigma_{c,g}\colon g\otimes c \cong c \otimes g$ induces the naturality of $(-1)^F$ in the sense that
\begin{align}
F_{\sigma_{c,g}}[(-1)^F\cdot a_g] = (-1)^{|a_g|} a_g\cdot (-1)^F
\end{align}
holds.
\end{definition}

Note that in the case the fermionic 2-group is actually a 1-group the definition reduces to the definition of a fermionicaly graded algebra given above. 
We call a fermionically 2-group graded algebra \emph{strongly graded} if the multiplication induces isomorphisms $A_g\otimes_{A_e} A_{g'}\to A_{g\otimes g'}$.

\begin{remark}
It would be better to include the identification of the $(A_1,A_1)$-bimodules $A_{(-1)^F}$ and $A_c$ as data. However, both definitions lead to equivalent bicategories and hence we will not dwell on this minor adaptation further.
\end{remark}

\begin{remark}
Recall that the composition of morphisms in a $2$-group $\mathcal{G}$ is uniquely determined by the tensor product of morphisms by suitably translating the composable morphisms by tensoring with objects of the $2$-group.
Using such observations the condition an algebra graded by a fermionic $2$-group has to satisfy for morphisms is equivalent to the three conditions
\begin{align*}
F_\gamma\circ F_{\gamma'} &= F_{\gamma\circ \gamma'} \text{ for all composable morphisms $\gamma, \gamma'$} \\
a_{k'}\cdot F_\gamma(a_k) &=  F_{k' \otimes \gamma}(a_{k'}\cdot a_k) \text{ for all $a_k\in A_k$ and $a_{k'}\in A_{k'}$} \\ 
F_\gamma(a_k)\cdot a_{k'}  &=  F_{\gamma \otimes k'}(a_k \cdot a_{k'} ) \text{ for all $a_k\in A_k$ and $a_{k'}\in A_{k'}$} 
\end{align*}
\end{remark}

\begin{remark}
A lot of data in the definition is redundant. For example, it is enough
to specify $A_g$ for one representative in every isomorphism class of objects and $F_\gamma $ for morphisms of the form $\gamma\colon 1 \to g$.   
A very efficient presentation uses the fermionically skeletal model of the $2$-group.
\end{remark}

\section{Spin statistics and reflection structures for non-extended topological field theories} 
We have seen in Section~\ref{Sec:spacetime group} how to associate to an internal symmetry group $G$ 
a spacetime structure group $H_d$. In this section we define what it means for non-extended topological $H_d$-field theories with values in super vector spaces to satisfy spin statistics and admit a reflection structure.

\subsection{Physical motivation for the definition}\label{spin-st}

We start by motivating the definitions from a physical perspective. The mathematically inclined reader can
skip this section without problem.
	The spin-statistics connection tells us that for a particle we can determine its statistics (fermionic/bosonic) from its spin (integer/half-integer) and vice-versa.
The spin of a particle is determined by how it transforms under a Lorentz transformation.
Assuming Euclidean signature for simplicity, the Lorentz group becomes the basic spacetime structure group $SO_d$ under Wick-rotation, where $d$ is the dimension of spacetime. 
Bosonic particles live in some irreducible representation $(V,\rho)$ of $SO_d$.
These representations are classified by a single nonnegative integer $a \in \Z$ called spin.
If we want to define for example a particle with spin $a = 1/2$, we will need rotation with $2 \pi$ to give $-1$.
This should not be surprising, since particles with half-integer spin actually do not transform under a representation of $SO_d$, but instead under a representation of the spin group $\Spin_d$.
The spin group is a double cover $\Spin_d \to SO_d$; its kernel has two elements $1$ and $c$ with $c^2 = 1$.
Representations with integer spin satisfy $R(c) = 1$ and therefore give representations of $SO_d$ as well.
However, this is not true for representations with half-integral spin, which instead have $R(c) = -1$.
The spin-statistics theorem therefore can be reformulated as follows: particles on which $R(c)$ acts by $-1$ should satisfy Fermi statistics, while particles on which $R(c)$ acts by $1$ should satisfy Bose statistics.

Mathematically, the distinction between bosonic and fermionic statistics is formalized by assuming that the Hilbert space of states $\mathcal{H}$ is a superspace, i.e. has a $\Z_2$-grading $\mathcal{H} = \mathcal{H}_0 \oplus \mathcal{H}_1$.
If $\mathcal{H}$ is the Fock space of some single particle Hilbert space, then the grading operator is the parity $(-1)^F$ of the fermion number $F$.
In other words, $\mathcal{H}_0$ consists of multi-particle states of bosons and states consisting of an even amount of fermions.
Even if $\mathcal{H}$ is not a second quantized one particle space, we will abuse notation and still denote this operation by $(-1)^F$.
The actual Fermi-statistics is then implemented by requiring Koszul signs in the superspace $\mathcal{H}$.
In mathematical jargon one could say that we take the nontrivial (symmetric) braiding on the monoidal category $(\sVect_\C,\otimes)$.
We can now formulate the spin-statistics connection: the element $c \in \Spin_d$ of the (Euclidean) spin Lorentz group is required to act by $(-1)^F$ on $\mathcal{H}$. Recall that a general spacetime
structure group $H_d$ associated to an internal symmetry group $G$ has a preferred element $c\in H_d$ 
and hence the previous definition generalises to $H_d$. Note that if $G$ is bosonic, i.e. $c_G=1$ then
any representation satisfying the spin-statistics connection needs to be completely even. 

The classification result for one-dimensional topological field theories Proposition~\ref{Prop: 1D TFT} is not only mathematically but also physically unsatisfactory.
Indeed, one would expect a time-reversing symmetry $T$ to act antilinearly on $V$, i.e. to be a complex-linear map $\overline{V} \to V$.
Instead, the mathematical definition of a topological field theory gives a map $V^* \to V$.
Luckily, under the natural assumption that our state space is a Hilbert space, the Hilbert space inner product will provide a canonical choice of isomorphism $\overline{V} \cong V^*$.
Moreover, we expect time reversing symmetries to act (anti-)unitarily.
These requirements can be implemented mathematically by requiring that our field theory is unitary, which in the current Euclidean framework is called reflection positive. 
To allow for non-unitary field theories, we will however require the weaker notion of a reflection structure, which replaces the Hilbert space by a Hermitian structure that is not necessarily positive.

To illustrate the notion of a reflection structure, consider the case where our topological field theory is bosonic with structure group $H = SO_d$.
We have a notion of orientation-reversal on our manifolds and bordisms, which we suggestively write $Y \mapsto \overline{Y}$.
One might then be tempted for a TFT to postulate a relationship between $\mathcal{Z}(\overline{Y})$ and $\overline{\mathcal{Z}(Y)}$, for example one might require that the partition function satisfies $\mathcal{Z}(\overline{X}) = \overline{\mathcal{Z}(X)}$ for closed $d$-dimensional manifolds $X$.
On objects, such structure becomes data, which one can express mathematically using the observation that orientation-reversal equips $\Bord_{d}^{SO_d}$ with a symmetric monoidal $\Z_2$-action:

\begin{definition}
A \emph{reflection structure} on a $d$-dimensional (non-extended oriented) TFT $\mathcal{Z}: \Bord_{d}^{SO_d} \to \Vect$ is symmetric monoidal $\Z_2$-equivariance data for the actions $Y \mapsto \overline{Y}$ and $V \mapsto \overline{V}$.
\end{definition}

We now elaborate on the relationship between this definition and the structure of Hilbert spaces on the state spaces $\mathcal{Z}(Y^{d-1})$.
For this we remark that on objects, $\overline{Y}$ can be exhibited as the dual of $Y$ in the categorical sense.
In other words, there is a nondegenerate pairing $h_Y: \overline{Y} \sqcup Y \to \emptyset$.
In general, there can be multiple such pairings, but in this case there is a canonical choice induced by a choice of reflection along the spatial slice $Y$ in the surrounding spacetime $Y \times \{0\} \subset Y \times (- \epsilon, \epsilon)$.
This is also the historical reason this is called a reflection structure~\cite{jaffe2018reflection}.
Moreover, this choice of reflecting along the $d$th coordinate makes the pairing $h_Y: \overline{Y} \sqcup Y \to \emptyset$ symmetric in a certain sense.
If we happen to have a TFT $\mathcal{Z}$ with reflection structure, we obtain a nondegenerate sesquilinear pairing
\[
\overline{\mathcal{Z}(Y)} \otimes \mathcal{Z}(Y) \cong \mathcal{Z}(\overline{Y}) \otimes \mathcal{Z}(Y) \cong \mathcal{Z}(\overline{Y} \sqcup Y) \xrightarrow{\mathcal{Z}(h_Y)} \mathcal{Z}(\emptyset) = \C.
\]
Moreover, the symmetry property of $h_Y$ translates in the condition that this is a Hermitian pairing.
These considerations are crucial to understand reflection-positive TFTs, which are the TFTs with reflection structure for which the above pairing is positive definite.
However, in this article we content ourselves with the weaker notion and so we stick to TFTs equipped with $\Z_2$-equivariance data.

\begin{remark}
The notion of reflection structure discussed here is sometimes called Hermitian TFT.
In particular the definitions above agree with the notion of Hermitian TFT in the sense of Turaev \cite[III.5]{turaev2016quantum}.
\end{remark}

In the next section, we abstract this definition further in order to generalize to the fermionic and the once extended situation in which some subtleties arise.

\subsection{The $\Z^R_2\times B\Z^F_2$-action on bordism categories}\label{Sec:Z2}
Topological field theories with the properties motivated in the previous section have a convenient 
formulation in terms of $\Z^R_2\times B\Z^F_2$-equivariant functors. There is an action of $\Z_2^R\times 
B\Z_2^F$ on $BH_d\to BO_d$, we refer to Appendix~\ref{App: 2-Group} for a definition of an action of a higher 
group. The action on $BH_d$ is induced by an action on $H_d$ where the element 
$R$ acts by conjugation with $e=e_d \otimes 1 \in \widehat{H}_d=  \Pin_d^+ \otimes K $ in 
$\widehat{H}_d$ on $H_d$ and $(-1)^F$ acts through the central element $c\in H_d$. The action $\rho$ on
$BH_d$ does not directly extend to an action on $BH_d\to B\O_d$ because 
\begin{equation}
\begin{tikzcd}
	BH_d \ar[rdd] \ar[rr, "\rho(R)"] & & BH_d \ar[ldd] \\ 
	& & \\ 
	& BO_d &
\end{tikzcd}
\end{equation}   
does not commute. However, the square commutes up to a homotopy which induced by noting that
$ \rho_d \circ \rho(R) \colon H_d \to \O_d$ differs from $\rho_d$ by conjugation with the reflection
along $e_d$ in $\O_d$ which induces a homotopy filling the square. 

Using the functoriality of $\Bord_{d,0}^{-}$ there is an induced $\Z_2^R\times 
B\Z_2^F$-action on the $d$-category $\Bord_{d,0}^{(H_d,\rho_d)}$ and in particular on the 1-category 
$\Bord_{d}^{(H_d,\rho_d)}$. Note that the action involves all the structures of the fermionic group.  

\begin{remark}
We make this abstract perspective more geometric to connect with the discussion of reflection structures in the last section.
Recall from Remark~\ref{Rem: Tangential structures in terms of bundles} that we can describe a tangential $(H_d,\rho_d)$-structure also by a principal $H$ bundle $P\to M$ together with a vector bundle isomorphism $\psi \colon P\times_{\rho_d}\R^d \to TM$.
The new tangential structure given by postcomposing with $\rho(R)$ can be described by the pair $\overline{P},\overline{\psi}$ constructed as follows:
\begin{itemize}
	\item From $P$ we can form the principal $\widehat{H}_d$-bundle $\widehat{P}\coloneqq P\times_{H_d} \widehat{H}_d$. The principal bundle $\overline{P}$ can now be defined as $\widehat{P}\setminus P $. This bundle is isomorphic to the bundle constructed from $P$ by twisting the $H_d$ action with the automorphism $\rho(R)$. 
	\item The identification $\overline{\psi}$ is given by 
	\begin{align}
	\overline{\psi}\colon \overline{P}\times_{\rho_d} \R^d & \longrightarrow TM \\
	[\overline{p},x] \longmapsto \psi[\overline{p}\cdot (e_d\otimes 1) ,	R_{e_d} x] \ \ ,
	\end{align} 
where $R_{e_d}$ is the reflection along $e_d$.
\end{itemize}

\end{remark}
Recall from Remark~\ref{Rem: action on sVect} that there is also an $\Z^B_2\times B\Z^F_2$-action on $\sVect$. 
We can now define field theories with spin-statistics and reflection structure following~\cite{theospinstatistics,freedhopkins}
\begin{definition}\label{Def: non extended spin reflection TFT}
A topological \emph{reflection and spin statistics} $(H_d,\rho_{d})$-field theory is a $\Z^R_2\times 
B\Z^F_2$-equivariant functor
\begin{align}
\mathcal{Z}\colon \Bord_{d}^{(H_d,\rho_d)} \to \sVect \ \ .
\end{align} 
A functor which is only $\Z_2$ or $B\Z_2$-equivariant is called a \emph{reflection} or \emph{spin 
statistics} field theory, respectively. 
\end{definition}
Explicitly, a reflection structure on a topological field theory
\begin{align}
\mathcal{Z}\colon \Bord_{d}^{(H_d, \rho_d)}\to \sVect
\end{align} 
consists of a natural symmetric isomorphism $\omega\colon  \mathcal{Z} \circ \overline{(-)}  \Longrightarrow\overline{(-)} \circ \mathcal{Z}$ squaring to 1. 
The component of $\omega$ at an object $\Sigma$ gives 
a natural isomorphism $\omega_{\Sigma}\colon \mathcal{Z}(\overline{\Sigma})\longrightarrow \overline{\mathcal{Z}(\Sigma )}$. 
There are canonical isomorphisms $h_{\Sigma}$ from $\overline{\Sigma}$ to the categorical dual $\Sigma^\vee$ 
which define a hermitian structure on $\Sigma$~\cite{freedhopkins}. Similar to the discussion at the end of the previous section,
combining these isomorphisms with $\omega$ defines a hermitian pairing on the state space $\mathcal{Z}(\Sigma)$ 
\begin{align}
\overline{\mathcal{Z}(\Sigma)}\otimes \mathcal{Z}(\Sigma) 
\xrightarrow{\omega \otimes \id} \mathcal{Z}(\overline{\Sigma})\otimes \mathcal{Z}(\Sigma) \xrightarrow{\mathcal{Z}(h_{\Sigma})\otimes \id } \mathcal{Z}(\Sigma^\vee)\otimes \mathcal{Z}(\Sigma)\cong 
\mathcal{Z}(\Sigma)^*\otimes \mathcal{Z}(\Sigma) \xrightarrow{\ev } \C \ \ . 
\end{align}
The $B\Z^F_2$-equivariance which defines field theories with spin-statistics does not require the specification of additional data. 
It is given by the condition 
\begin{align}
\mathcal{Z}(c|_\Sigma) = (-1)^F_{\mathcal{Z}(\Sigma)} 
\end{align} 
for all $\Sigma\in \Bord_d^{H_d, \rho_d}$ where $c$ is the automorphism of the $H_d$-structured manifold $\Sigma$ induced by acting with $c$. 
For $H_d=\Spin_d$ this is exactly the spin-flip. 
\begin{example}
For $H_1=\Spin_1=\Z_2^c$ one-dimensional topological field theories $\mathcal{Z}$ are classified by a $\Z^c_2$-representation on the super vector space $\mathcal{Z}(+)$ assigned to the positively oriented point.  	
The only representations which give rise to spin-statistics field theories are those who satisfy the condition that $-1\in \Z_2$ act by the grading operator 
$(-1)^F$.
This shows that there are many topological field theories which do not satisfy spin-statistics. 
Interestingly every fully extended two-dimensional $\Spin_2$-field theory 
satisfies spin-statistics.  
\end{example}
\subsection{Classification in 1-dimension}\label{Sec:1D}
Let $G$ be an internal symmetry group and $H_1=(\Pin_1^+\otimes G)_0\cong G^{\operatorname{op}} $ the associated one dimensional 
spacetime symmetry group.
Our first result is to classify reflection and spin statistics $H_1$-field theories in dimension 1.
Recall that these are
$\Z_2^R\times B\Z_2^F$ equivariant symmetric monoidal functors $\mathcal{Z}\colon \Bord_1^{H_1,\rho_1} \to \sVect$, where 
the $\Z_2^R\times B\Z_2^F$ action on $\sVect$ is given by 
\begin{align}
	R\colon \sVect \to \sVect \ , \ \ V \longmapsto \overline{V} 
\end{align} 	
and
\begin{align}
	(-1)^F \colon \id_{\sVect} \Longrightarrow \id_{\sVect} \ , \ \ V \longmapsto {(-1)^F}_V 
\end{align} 
together with the canonical coherence isomorphisms spelled out in Appendix \ref{Sec:action on sVect}.

Instead of computing $\Z_2^R\times B\Z_2^F$-equivariant functors we can also compute fixed points for the 
combined action on the functor category. Using the cobordism hypothesis these can be computed as the 
homotopy fixed points
\begin{align}\label{Eq: RS-FP-1D}
	(([\Bord_1^{\fr}, \sVect])^{H_1})^{\Z_2^R\times B\Z_2^F} \simeq (( \sVect^\times)^{H_1})^{\Z_2^R\times B\Z_2^F}
	\simeq ( \sVect^\times)^{{H_1}\rtimes_\rho (\Z_2^R\times B\Z_2^F)}
\end{align} 
for the ${H_1} \rtimes_\rho (\Z_2^R\times B\Z_2^F)$ action $\psi$:
\begin{align}
	\psi(h) \colon \sVect^\times & \to \sVect^\times \ , \ \ V\longmapsto \begin{cases}
		V & \text{if } \rho_1(h)=1 \\
		V^* & \text{if } \rho_1(h)=-1
	\end{cases}  \ , \ \ f \longmapsto \begin{cases}
		f & \text{if } \rho_1(h)=1 \\
		(f^*)^{-1} & \text{if } \rho_1(h)=-1
	\end{cases} \\ 
	\psi(R) \colon \sVect^\times & \to \sVect^\times \ , \ \ V\longmapsto \overline{V}^* \ , \ \ f \longmapsto ({\overline{f}^*})^{-1} \\ 
	\psi((-1)^F) \colon \id_{\sVect^\times} & \Longrightarrow \id_{\sVect^\times} \ , \ \ V \longmapsto (-1)^F_V 
\end{align}    
with the canonical coherence isomorphisms. For the following computation it is enough to work with
a 2-group model for $H_1$, which we furthermore assume without loss of generality to be fermionically 
skeletal. 

The semi-direct product 2-group $H_1\rtimes (\Z_2^R\times B\Z_2^F)$ has the following concrete description (see Appendix~\ref{App: Exact}): 
Its objects are the set $H_1 \times \{1,R\}$ and their composition $\otimes $ is given by
\begin{align}
(h\rtimes R^{\epsilon})\otimes (h'\rtimes R^{\epsilon'}) = (h\otimes_{H_1}h'\otimes_{H_1} c^{\epsilon\cdot\theta(h')})\rtimes R^{\epsilon+\epsilon'}
\end{align} 
with $\epsilon =0,1$ and $\epsilon' =0,1$.
There are two types of morphisms. The first one is of the form $\gamma\rtimes \id \colon h\rtimes R^\epsilon \to h'\rtimes R^\epsilon$ where $\gamma$ is a morphism 
from $h$ to $h'$ in $H_1$. The other type is of the form $\gamma\rtimes (-1)^F \colon  h\rtimes R^\epsilon \to h'\rtimes R^\epsilon $ where $\gamma$ is now a morphism $h\to h'\otimes c$. 
All the unspecified data (such as the composition of morphisms) is given canonically in terms of the data of the $2$-group $H_1$.
We refer to the appendix for more details.

The first step in our computation is to show that $ H_1 \rtimes_\rho (\Z_2^R\times B\Z_2^F)$ is isomorphic to 
the direct product of 2-groups $G_b\times\Z_2^R $ in the case that $c\neq 1$     
which allows us to compute the fixed points appearing in~\eqref{Eq: RS-FP-1D} by first computing the $\Z_2^R $-fixed points and then the fixed points for the induced ${G}_b$-action. 
We comment on the simpler case $c=1$ in Remark~\ref{Rem: c=1 1D}.

To define the equivalence $\sigma\colon G_b\times \Z_2^R \to  H_1 \rtimes_\rho (\Z_2^R\times B\Z_2^F)
$ we recall from Corollary~\ref{cor: exact sequence} that there is a short exact sequence 
\begin{align}
1\to \Z_2^c \to H_1 \to G_b\to 1
\end{align} 
of Lie groups. We pick a section $s\colon G_b\to H_1$ satisfying $s(g)\otimes s(g')=\omega^{\operatorname{op}}(g,g') s(gg')$ with $\omega^{\operatorname{op}}(g,g') =\omega(g,g')+\theta(g)\theta(g') $ where
$\omega$ is the cocycle classifying the central extension
\begin{align}
	1\to \Z_2^c \to G \to G_b\to 1 \ \ .
\end{align} 
This is just another way of saying that $H_1=G^{\operatorname{op}}$. For every path $\gamma\colon g\to g' $ there is a canonical lift $s(\gamma)\colon s(g)\to c^{\Gamma(\gamma)}s(g)$ with $\Gamma(\gamma)\in \{0 , 1\}$.  
We can now define the 2-group isomorphism 
\begin{align}
\sigma \colon G_b\times \Z_2^R & \longrightarrow H_1\rtimes (\Z_2\times B\Z_2)  \\
g\times R^{\epsilon} & \longmapsto s(g)\rtimes R^{\theta(g)} R^\epsilon \\
\gamma\colon g\to g' & \longmapsto s(\gamma)\rtimes ((-1)^F)^{\Gamma(\gamma)}  \colon \sigma(g)\to \sigma(g')  \ \ .
\end{align} 
The functor $\sigma$ is not strictly compatible with the multiplication.
Instead there is a coherence isomorphism
\begin{align}
 \sigma(g\times R^\epsilon) \otimes \sigma(g'\times R^{\epsilon'}) & = s(g)\rtimes R^{\theta(g)} R^\epsilon \otimes s(g')\rtimes R^{\theta(g')} R^{\epsilon'} = s(g g')\otimes c^{\omega(g,g')+\epsilon\theta(g')}\rtimes R^{\theta(g)+\theta(g')+\epsilon+\epsilon'}  \\ 
 & \xrightarrow{1\rtimes {(-1)^F}^{\omega(g,g')+\epsilon\theta(g')}} s(gg')\rtimes R^{\theta(g)+\theta(g')+\epsilon+\epsilon'} =\sigma(gg'\times R^{\epsilon+\epsilon'}) 
\end{align} 
It is straightforward to verify that $\sigma$ is an equivalence of 2-groups.
The next step in our computation is to compute $\Z_2^R$-fixed points. 

\begin{proposition} \label{Prop: Action 1 D}
	For the actions introduced above the category of homotopy fixed points 
	\begin{align}
		(\sVect^\times)^{\Z_2^R} \simeq \hsVect^{\times u} 
	\end{align}
	is given by the groupoid of hermitian super vector spaces and compatible (i.e. even and unitary) isomorphisms.

\end{proposition}
\begin{proof} We use the definition of a homotopy fixed point spelled out in Remark~\ref{Rem: FP on Cat}. 
A fixed point consist of an object $V\in \sVect$ together with a morphism $h_V \colon \rho(R)(V)= \overline{V}^* \to 
	V$ such that 
	\begin{align}
		V \simeq ({\overline{({\overline{V}})^*}})^* \xrightarrow{\overline{h_V}^*} \overline{V}^* \xrightarrow{h_V} V = V \xrightarrow{\id_V} V \ \ . 
	\end{align}   
	It is straightforward to see that $h_V$ defines a hermitian structure on $V$. 
	A morphism of fixed points $(V,h_V) \to (V',h_{V'})$ is given by a map $f\colon V\to V' $ such that
	\begin{equation}
		\begin{tikzcd}
			\overline{V}^* \ar[r, "h_V"] \ar[d, "(\overline{f}^*)^{-1}",swap]  & V \ar[d, "f"] \\
			\overline{V'}^* \ar[r,"h_{V'}", swap] & V'
		\end{tikzcd}
	\end{equation}
	commutes. This implies that the adjoint of $f$ agrees with its inverse and hence $f$ must be unitary.
\end{proof}

\begin{remark}\label{Rem: Spin}
The 2-group $\Spin_1\rtimes (\Z_2^R \times B \Z_2^F)$ is equal to $(\Z_2^c \rtimes B\Z_2^F)\times \Z_2^R$, since $\Z_2^R$ acts trivially on $\Spin_1=\Z_2^c$. 
Furthermore, from the definition of the semi-direct product, we see directly that $\Z_2^c \rtimes B\Z_2^F$ is contractible (equivalent to the 2-group with one object and one 1-morphism). This shows that the previous proposition also classifies $\Spin_1$ reflection and spin-statistics field theories in terms of super hermitian vector spaces.      
\end{remark}
The next proposition describes the induced $G_b$-action on $\hsVect^{\times u}$.

\begin{proposition}
The induced ${G}_b$-action on $\hsVect^{\times u}$ is trivial for even elements $g\in G_b$ and given by 
\begin{align}
\gamma(g) \colon \hsVect^{\times u}  & \to \hsVect^{\times u}  \\ 
V & \mapsto (\overline{V}, (-1)^F_{\overline{V}} \circ \overline{h_V} ) \\ 
f\colon V\to V & \mapsto \overline{f}
\end{align}
for $g\in G_b$ odd with coherence data for $g$ and $g'$ 
\begin{align}
		\gamma_{g',g} = ((-1)^F)^{\omega(g',g)} 
\end{align}
combined with the canonical identification of the double bar with the original
vector space, where $\omega(g',g)$ is the $\Z_2^c$-valued 2-cocycle classifying the
central extension $1\to \Z_2^c \to G \to {G}_b \to 1$. Furthermore, a morphisms $P \colon g\to 
g$ (note that by assumption all morphisms are loops) act trivially if $\Gamma(\gamma)=0$ and by $(-1)^F$ otherwise.
\end{proposition}
\begin{proof}
An element $g\in G_b$ acts on a fixed point $(V,h_V)$ by sending it to 
\begin{align}
	\psi(\sigma(g))[V]= \psi\left(s(g)\rtimes R^{\theta(g)} \right) V =  \begin{cases}
		V & \text{if } \theta(g)=0 \\
		\overline{V^*}^*\cong \overline{V} & \text{if } \theta(g)=0
	\end{cases}
\end{align}      
equipped with the new hermitian structure 
\begin{align}
\psi(R)[\psi(\sigma(g))[V]]= \psi(\sigma(R)\otimes \sigma(g) )[V] \cong 
\psi(\sigma(R\otimes g) )[V] &= \psi(\sigma(g\otimes R) )[V] \cong \psi(\sigma(g))[\psi(R)[V]] \\ &\xrightarrow{\psi(\sigma(g))[h_V]} \psi(\sigma(g))[V]  
\end{align}
which when spelled out is just $h_V$ for even elements and $(-1)^F_{\overline{V}} \circ \overline{h_V}$ for odd elements. The coherence isomorphisms are induced by those of $\sigma$ leading directly to the formula stated in the proposition. 

The action of morphisms is given by the corresponding component of $\psi(\sigma(\gamma)) = \psi \left(s(\gamma)\rtimes ((-1)^F)^{\Gamma(\gamma)}\right) $ as claimed. 
\end{proof}

\begin{remark}
Note that the Hermitian structure on the complex conjugate of a Hermitian vector space we obtained here is the same as the one discussed in Section \ref{Sec:fermrep} obtained by a different argument.
\end{remark}

To finish the classification we have to compute fixed points of the induced ${G}_b$-action. We have already done this in Remark~\ref{Rem: action on sVect}
in terms of unitary representations of $G$ as defined in Definition~\ref{Def: F Rep} leading to the following proposition. 

\begin{proposition}
\label{prop:1d classification}
	Let $G$ be a fermionic group. The groupoid of 1-dimensional reflection and spin-statistics field 
	theories with internal symmetry group $G$ is equivalent to the core of the category of unitary 
	representations of the fermionic group $\pi_0(G)$.  
\end{proposition}
\begin{proof}
We will only discuss the case $c\neq 1$ here.  We comment on the simpler case $c=1$ in the next remark. 
To see that we can restrict to the fermionic group $\pi_0(G)$ instead of the 2-group model we need to distinguish two additional cases. The first one is the case that there are no morphisms from $1$ to $c$ in $G$. 
In this case  $\Gamma(\gamma)$ will always be $1$ and hence all morphisms in $G_b$ act trivial and we can restrict to $\pi_0(G_b)$. 
In the case that there is a morphism $\gamma\colon 1\to c$ in $G$ there is a morphism $\gamma'\colon 1\to 1 $ in $G_b$ with $\Gamma(\gamma')=c$. This implies
that for every homotopy fixed point $V$, the operator $(-1)^F_V $ must be equal to the identity
on $V$ and hence $V$ is completely even. This implies that we can restrict to the subcategory of even vector spaces where all morphisms act again trivial.    
Now the statement of the identification of the fixed points with the category of $\pi_0(G)$-representation is exactly the content of Remark~\ref{Rem: action on sVect}
\end{proof}

\begin{remark}\label{Rem: c=1 1D}
We conclude the section with the promised remark on the case $c=1$, i.e. $G=G_b$. Then the semi-direct product $H_1\rtimes (\Z_2\times B\Z_2)$ is actually directly equal to
$G\times \Z_2^R\times B\Z_2^F$. Hence when computing homotopy fixed points we can start by computing $B\Z_2^F$-homotopy fixed points. These are given by even vector space and the induced $G\times \Z_2^R$ is the restriction of the action to $G\times \Z_2^R$. Now the computation of fixed points is exactly the same as above, but 
all of the $(-1)^F$-factors drop out.   
\end{remark}	
\section{Once extended field theories}\label{Sec:2D}
We extend the results and definitions from the previous section to fully extended
2-dimensional topological field theories with values in the Morita bicategory of super algebras $\sAlg$.
\subsection{Spin-statistics and reflection structures for extended topological field theories}
We now extend the definition of spin-statistics and reflection structures to once extended
field theories 
\begin{align}
\mathcal{Z}\colon \Bord_{d,d-2}^{(H_d,\rho_d)} \to \sAlg
\end{align}
with values in super algebras. The $\Z_2^R\times B\Z_2^F$-action on 
$\Bord_{d,0}^{(H_d,\rho_d)}$ constructed in Section~\ref{Sec:Z2} restricts to an action 
on the bicategory $\Bord_{d,d-2}^{(H_d,\rho_d)}$. To generalise \ref{Def: non extended spin reflection TFT} we 
need to extend the $\Z_2^R\times B\Z_2^F$-action on $\sVect$ to an action on $\sAlg$. The extension can
abstractly be constructed from the fact that the functor sending a symmetric monoidal category to the Morita
category of $E_1$-algebras in it is functorial. Concretely the generator of $\Z_2^R$ acts by 
sending a super algebra $A$ to the complex conjugated algebra $\overline{A}$, see Definition \ref{def:C-conj algebra}. 
The non-trivial 1-morphism $(-1)^F$ of $B\Z_2^F$ acts by the natural transformation 
\begin{align}
(-1)^F \colon \id_{\sAlg} &\Longrightarrow \id_{\sAlg} 
\end{align} 
whose value at a super algebra $A\in \sAlg$ is given by the bimodule $A_{(-1)^F}$ from 
Definition~\ref{def:spin-statistics bimodule}
and at a 1-morphism $M\colon A \longrightarrow  B$ by the filling
\begin{equation}
\begin{tikzcd}
A \ar[dd,"M",swap] \ar[rr, "{A_{(-1)^F}}"] & & A  \ar[dd,"M"] \ar[lldd, Rightarrow, "(-1)^F_M",swap, shorten <=10, shorten >=10] \\ 
& & \\
B \ar[rr, "{B_{(-1)^F}}",swap] & & B
\end{tikzcd} \ \ 
\end{equation} 
given by the formula 
\[
m \otimes a (-1)^F \mapsto 	(-1)^F (-1)^{|m|+|a|} \otimes ma
\]
which is a bimodule isomorphism.
Part of the data for the $B \Z^F_2$-action is also the trivialization of the square of $(-1)^F$.
In this case it is given by the $(A,A)$-bimodule isomorphism $A_{(-1)^F} \otimes_A A_{(-1)^F} \cong A$ that is induced by the composition of algebra automorphisms $(-1)^F \circ (-1)^F = \id_A$.
Explicitly, this means that 
\[
a_1 (-1)^F \otimes a_2 (-1)^F \mapsto (-1)^{|a_2|} a_1 a_2.
\]
Intuitively, we can think of $(-1)^F$ as a formal even variable squaring to one such athat $(-1)^F a = (-1)^{|a|} a (-1)^F$.
With these preparations we can give the following definition

\begin{definition}\label{Def: extended spin reflection TFT}
	An extended topological \emph{reflection and spin statistics} $(H_d,\rho_{d})$-field theory is a $\Z^R_2\times 
	B\Z^F_2$-equivariant functor
	\begin{align}
	\mathcal{Z}\colon \Bord_{d,d-2}^{(H_d,\rho_d)} \to \sAlg \ \ .
	\end{align} 
	A functor which is only $\Z^R_2$ or $B\Z^F_2$-equivariant is called a \emph{reflection} or \emph{spin 
		statistics} field theory, respectively. 
\end{definition}

\begin{remark}
Note that an extended topological reflection and spin statistics field theory gives rise to an non-extended reflection and spin statistics field theory by restriction to the endomorphisms of the monoidal unit.
This follows because the endomorphisms of the unit in $\sAlg$ are $\sVect$ and the induced $\Z_2^R \times B\Z_2^F$-action is the same as considered in the last section.
\end{remark}

A reflection structure on a given topological field theory 
\begin{align}
	\mathcal{Z}\colon \Bord_{d,d-2}^{(H_d, \rho_d)}\longrightarrow \sAlg
\end{align}
consist of a lot of data, which we only partially spell out. 
For every object $S\in \Bord_{d,d-2}^{(H_d, \rho_d)}$ we get an invertible bimodule $(\mathcal{Z}(\overline{S}), \overline{\mathcal{Z}(S)})$ as part of the natural isomorphism 
\begin{align}
\mathcal{Z}\circ \overline{(-)} \Longrightarrow  \overline{(-)} \circ \mathcal{Z} \ \ . 
\end{align} 
Furthermore, the data involves a modification from the square of this natural isomorphism (combined with the coherence morphisms for $\overline{(-)}$) to the identity. 
When combining these bimodules with an isomorphism $\overline{S}\longrightarrow S^\vee$ in $\Bord_{d,d-2}^{(H_d,\rho_d)}$, we get a categorification of the concept of a hermitian structure on a vector space. 
These are stellar algebras which we discuss in detail in the next section. 

Spin-statistics also becomes data in an
extended field theory. Namely for every, object $S\in \Bord_{d,d-2}^{(H_d,\rho_d)}$ a bimodule isomorphism $\mathcal{Z}(c|_S)\longrightarrow \mathcal{Z}(S)_{(-1)^F}$ satisfying various coherence conditions.

\begin{remark}
Note that given a once extended theory $\mathcal{Z}$ with both $\Z_2^R$- as well as $B\Z_2^F$-equivariance data, there is an extra condition it has to satisfy for it to be $\Z_2^R \times B\Z_2^F$-equivariant.
Explicitly, it is saying that for every $d-2$-dimensional closed manifold $Y$ the diagram of $2$-morphisms
\[
\begin{tikzcd}
\mathcal{Z}(c|_{\overline{Y}}) \arrow[r] \arrow[d] & \mathcal{Z}(\overline{c|_{Y}}) \arrow[r] & \overline{\mathcal{Z}(c|_Y)}  \arrow[d]
\\
\mathcal{Z}(\overline{Y})_{(-1)^F} \arrow[r] & \overline{\mathcal{Z}(Y)_{(-1)^F}} \arrow[r] & \overline{\mathcal{Z}(Y)_{(-1)^F}}
\end{tikzcd}
\]
commutes.
Therefore having a reflection and spin statistics theory is not equivalent to separately specifying a spin-statistics relation and a reflection structure on $\mathcal{Z}$.
\end{remark}

The general approach to the classification in spacetime dimension 2 follows the 1-dimensional case we discussed in 
Section~\ref{Sec:1D}. 
Equivariant functors are equivalent to fixed points for the induced $\Z_2^R\times B\Z_2^F$-action on the functor category $[\Bord_{2,0}^{(H_2,\rho_2)}, \sAlg]$. Combined with the cobordism hypotheses this allows us to compute the reflection and spin statistics $(H_2,\rho_{2})$-field theories as homotopy fixed points 
\begin{align}\label{Eq: RS-FP-2D}
	(([\Bord_2^{\fr}, \sAlg])^{H_2})^{\Z_2^R\times B\Z_2^F} \simeq (( \sAlg^{\text{fd}})^{H_2})^{\Z_2^R\times B\Z_2^F}
	\simeq ( \sAlg^{\text{fd}})^{{H_2} \rtimes (\Z_2^R\times B\Z_2^F)} \ \ .
\end{align} 
Similar to the 1-dimensional case the group ${H_2} \rtimes (\Z_2^R\times B\Z_2^F)$ turns out to be equivalent to the direct product $O_2\times G_b$ (see Section~\ref{Sec:Computation 2D}). Fixed points for the $O_2$ part can be described by a generalization of hermitian structures to super algebras which we introduce next.    

\subsection{Stellar algebras}
\label{Sec:stellar}

In one-dimensional fermionic theories, we have seen that reflection structures are intimately connected to super Hermitian vector spaces.
Before diving into full depth for two-dimensional theories, we will start with a purely algebraic treatment of a two-dimensional analogue of this notion.
Following \cite{schommerpriesthesis}, we call these stellar algebras, which are a certain modification of the notion of a $*$-algebra.\footnote{In \cite{schommerpriesthesis} stellar algebras are introduced as an analogue of linear $*$-algebras, which are linear maps $*:A \to A$ such that $(ab)^* = b^* a^*$. In this paper we will need the complex-antilinear (super) analogue instead, where $*$ is a complex-antilinear map.}
In Appendix \ref{Sec:*-algs} we treat super $*$-algebras in detail.

To motivate stellar algebras as being the categorification of Hermitian vector spaces, first consider the symmetric monoidal functor $\End: \sVect^\times  \to \sAlg_1$ from the core of complex supervector spaces into the $1$-category of algebras and isomorphisms.
On morphisms it maps a linear operator to conjugation with that operator.
It sends the dual vector space $V^*$ to the (graded) opposite algebra $(\End V)^{\op}$ and the complex conjugate vector space $\overline{V}$ to the complex conjugate algebra $\overline{\End V}$. 
We refer the reader to Appendix \ref{App: Super alg} for the relevant definitions.
If we restrict to finite-dimensional vector spaces, we see that the $\Z^R_2$-action $V \mapsto \overline{V}^*$ is mapped to the $\Z_2^R$-action $A \mapsto \overline{A}^{\op}$ under $\End$.
As a Hermitian vector space is a $\Z_2^R$-fixed point in finite-dimensional vector spaces, it is mapped to a $\Z_2^R$-fixed point in $\sAlg^{\fd}_1$.
Explicitly, a fixed point for $A \mapsto \overline{A}^{\op}$ consists of a complex-antilinear even map $*: A \to A$ such that $a^{**} = a$ and $(ab)^* = (-1)^{|a||b|} b^* a^*$. \footnote{Our choices for natural isomorphisms $\overline{\overline{A}} \cong A, A^{\op \op} \cong A$ and $\overline{A}^{\op} \cong \overline{A^{\op}}$ are the obvious equalities, as specified in Appendix \ref{App: Super alg}. They are compatible with the functor $\End$.}
This is known as a (complex-antilinear) super $*$-algebra.
We deduce that for every Hermitian vector space $(V,h)$, the algebra $\End V$ is canonically a $*$-algebra.
The star is given by the graded adjoint with respect to the Hermitian form
\[
\langle Tv, w \rangle = (-1)^{|T| |v|} \langle v, T^* w \rangle \quad T \in \End V, v,w \in V.
\]
The fact that $*$-algebras are analogous to Hermitian vector spaces is also not surprising from the perspective of algebraic quantum mechanics, where $*$-algebras replace Hilbert spaces.

In extended two-dimensional topological field theory however, we instead work with the Morita $2$-category $\sAlg^{\text{fd}}$ of finite-dimensional semisimple complex superalgebras with invertible $1$- and $2$-morphisms.
We generalize the discussion to that situation by analogy. 
There is still a symmetric monoidal $\Z_2^R$-action $A \mapsto \overline{A}^{\op}$, see Appendix \ref{App: Super alg} for details.
On $1$-morphisms $R$ acts by $M \mapsto \overline{M}^{\op -1}$ to make it covariant, while on $2$-morphisms $\phi$ it acts by $\overline{\phi}^{\op}$.
We will now consider fixed points for this action as a two-dimensional version of Hermitian forms.
This results in the following explicit definition.

\begin{definition}
\label{Def:stAlg}
The $2$-category of \emph{($\C$-antilinear) stellar algebras} has 
\begin{itemize}
\item \textbf{Objects} called \emph{stellar algebras} are triples $(A,M,\sigma)$ consisting of a super algebra $A$ together with an invertible $(A, \overline{A}{}^\op)$ bimodule $M$ and a bimodule isomorphism $\sigma \colon M \to \overline{M}^\op $ such that 
\[
M \xrightarrow{\sigma} \overline{M}^{\op} \xrightarrow{\overline{\sigma}^{\op}} \overline{\overline{M}^{\op}}^{\op} = M
\]
is the identity.
\item \textbf{1-morphisms} from $(A_1,M_1,\sigma_1)$ to $(A_2,M_2,\sigma_2)$ called $(A_2,A_1)$-\emph{stellar bimodules} are pairs $(N,\phi)$ consisting of an invertible $(A_2,A_1)$-bimodule $N$ together with a \emph{unitarity datum}, which is a bimodule isomorphism
	\[
	\phi: N \otimes_{A_1}  M_1 \otimes_{\overline{A_1}^{\op}} \overline{N}^{\op}\to M_2
	\]
	satisfying the \emph{Hermiticity condition} which says that the compositions
\[
N \otimes_{A_1} \otimes M_1 \otimes_{\overline{A}_1^{\op}} \overline{N}^{\op} \xrightarrow{\id \otimes \sigma_1 \otimes \id}  N \otimes_{A_1} \otimes \overline{M}_1^{\op} \otimes_{\overline{A}_1^{\op}} \overline{N}^{\op} \cong \overline{N \otimes_{A_1} \otimes M_1 \otimes_{\overline{A}_1^{\op}} \overline{N}^{\op}}^{\op} \xrightarrow{\overline{\phi}^{\op}} \overline{M}_1^{\op}
\]
and $\sigma_2 \circ \phi$ are equal.
\item \textbf{2-morphisms} from $(N,\phi)$ to $(N',\phi')$ called \emph{unitary bimodule maps} are even invertible bimodule maps $\psi: N \to N'$ satisfying the unitarity condition that 
\[
N \otimes_A M_1 \otimes_{\overline{A_1}^{\op}} \overline{N}^{\op} \xrightarrow{\psi \otimes \id_{M_1} \otimes \overline{\psi}^{\op}} N' \otimes_{A_1} M_1 \otimes_{\overline{A_1}^{\op}} \overline{N'}^{\op} \xrightarrow{\phi'} M_2
\]
is equal to $\phi$.
\end{itemize}
The  composition of a $(B,A)$-stellar module $(N_1,\phi_1)$ and a $(C,B)$-stellar module $(N_2,\phi_2)$ is $(N_2 \otimes_B N_1, \phi_2 \circ \phi_1)$, where $\phi_2 \circ \phi_1$ is defined using the isomorphism of Lemma \ref{eq:optensor} as
\begin{align*}
(N_2 \otimes_B N_1) \otimes_A M_1 \otimes_{\overline{A}^{\op}} \overline{(N_2 \otimes_B N_1)}^{\op} &\cong N_2 \otimes_B ( N_1  \otimes_A M_1  \otimes_{\overline{A}^{\op}} \overline{N}_1^{\op}) \otimes_{\overline{B}^{\op}} \overline{N}_2^{\op}  
\\
&\xrightarrow{\phi_1} N_2 \otimes_{B} M_2 \otimes_{\overline{B}^{\op}} \overline{N}_2^{\op} \xrightarrow{\phi_2} M_3.
\end{align*}
\end{definition}

Before comparing this definition with $\Z_2^R$-fixed points in the $2$-category of algebras, we provide some intuition by comparing stellar algebras with $*$-algebras.
Every ($\C$-antilinear) super $*$-algebra $*: \overline{A}^{\op} \to A$ gives rise to a $\C$-linear stellar structure as follows.
Define $M := A_*$ to be the $(A,\overline{A}^{\op})$-bimodule induced by the homomorphism $*$ and $\sigma: \overline{M}^{\op} \to M$ the bimodule map $\sigma(\overline{a_*}^{\op}) = (a^*)_*$.
It is easy to check that $\overline{\sigma}^{\op} \sigma = \id$ using that $*$ is an involution.
A unitarity datum on a bimodule between $*$-algebras is equivalent to a kind of algebra-valued inner product.
Similar inner products are often considered in the $C^*$-algebra literature, where they are called Hilbert $C^*$-modules.
	
	\begin{proposition}
	\label{prop:Hilbert module}
Let $A,B$ be $\C$-antilinear super $*$-algebras considered as stellar algebras and $N$ an invertible $(B,A)$-bimodule.
Then a unitarity datum on $N$ is equivalent to a sequilinear pairing $\langle.,.\rangle: N \times N \to B$ complex-linear in the left argument and complex antilinear in the right argument
such that 
\begin{align*}
    \langle n_2, b n_1  \rangle &= (-1)^{|b||n_1|} \langle n_2, n_1 \rangle b^*
    \\
    \langle n_2 a, n_1 \rangle &= (-1)^{ |a| |n_1|}  \langle n_2, n_1 a^* \rangle 
    \\
    \langle b n_2, n_1 \rangle &= b \langle n_2,  n_1 \rangle
    \\
    \langle n_2, n_1 \rangle^* &= (-1)^{|n_1 | |n_2|} \langle n_1, n_2 \rangle.
\end{align*}
The unitarity condition on an invertible bimodule map $\psi: N \to N'$ between stellar $(B,A)$-modules is equivalent to
\[
\langle \psi(n_1), \psi(n_2) \rangle_{N'} = \langle n_1 , n_2 \rangle_{N}.
\]
for all $n_1,n_2 \in N$.
For a third $*$-algebra $C$, stellar $(B,A)$-bimodule $N_1$ and stellar $(C,B)$-bimodule $N_2$, the pairing on the tensor product $N_2 \otimes_B N_1$ is
\[
\langle n_2 \otimes n_1, n_2' \otimes n_1' \rangle_{N_2 \otimes_B N_1} = (-1)^{|n_2'| |n_1'|} \langle n_2 \langle n_1, n_1' \rangle_{N_1}, n_2 \rangle_{N_2}.
\]
\end{proposition}	

\begin{remark}
The first condition
\[
    \langle n_2, b n_1  \rangle = (-1)^{|b||n_1|} \langle n_2, n_1 \rangle b^*
\]
is redundant as it can be derived from the third and the fourth. 
But we include it for completeness as the first three conditions are equivalent to a bimodule map
\[
	\phi: N \otimes_A  M_1 \otimes_{\overline{A}^{\op}} \overline{N}^{\op}\to M_2.
	\]
\end{remark}

\begin{remark}
Because of the focus of this article on the cobordism hypothesis, we only defined the $2$-groupoid of stellar algebras.
However, the definition above generalizes to give a $2$-category of stellar algebras in which bimodules and bimodule maps need not be invertible.
This is straightforward up to the subtlety what to require from the unitarity datum $\phi$ on a noninvertible $(A_2,A_1)$-bimodule $N$ between stellar algebras $(A_1,M_1,\sigma_1)$ and $(A_2,M_2,\sigma_2)$.
For example, in the case that $A_1 = A_2 =\C$ with the trivial stellar structure, $N$ is a vector space and we want $\phi$ to define a nondegenerate Hermitian form on $N$.
In particular, requiring that $\phi$ is an isomorphism is certainly too strict, because it will imply $N$ is invertible and hence one-dimensional.
Therefore we instead need to require a nondegeneracy condition on $\phi$.
More precisely, it realizes $M_1 \otimes_{\overline{A_1}^{\op}} \overline{N}^{\op}$ as the right adjoint of $(M_2)^{-1} \otimes_{A_2} N$ in the $2$-category $\sAlg$, see Appendix \ref{Sec:duals}.
\end{remark}

\begin{remark}
The definition of the bicategory of stellar algebras above makes the pairings antilinear in the right argument and valued in the target $*$-algebra. 
The reason for our choice of convention comes from how we defined the directions of $\Z_2$-fixed points morphisms such as $*:\overline{A}^{\op} \to A$ in Definition \ref{Def: Hfixed point}.
Unfortunately, our convention not agree with the common convention of the operator algebra literature.
We briefly make the appropriate translation for readers familiar with the $C^*$-algebra literature.

Our conventions stellar modules are analogous to `left Hilbert modules with adjointable operators acting on the right'.
	We can get the more common convention (up to signs) by taking $*$ to map from $\overline{A}^{\op}$ to $A$.
	This gives a sequilinear pairing $\langle\cdot ,\cdot \rangle: N \times N \to A$ complex-linear in the right argument and complex antilinear in the left argument such that 
\begin{align*}
    \langle n_1, n_2 a \rangle &= \langle n_1, n_2 \rangle a
    \\
    \langle n_1 a, n_2 \rangle &= (-1)^{|n_1| |a|} a^* \langle n_1, n_2 \rangle
    \\
    \langle b n_1, n_2 \rangle &= (-1)^{|b| |n_1|} \langle n_1, b^* n_2 \rangle.
\end{align*}
The first two of the above three conditions say that we have a Hilbert $C^*$-right $A$-module, while the last condition says that $N$ is a $B$-representation acting by adjointable operators.

Note separately that the $C^*$-algebra literature uses the conventions where a $\Z_2$-graded $C^*$-algebra satisfies $(ab)^* = b^* a^*$ and a Hilbert module satisfies $\langle n_2, n_1 \rangle = \langle n_1, n_2 \rangle^*$, while we prefer to work with the appropriate Koszul signs.
See Appendix \ref{Sec:*-algs} for the relationship between our conventions on graded $*$-algebras which we call super $C^*$-algebras and the more common $\Z_2$-graded $C^*$-algebras. 
Under this correspondence there is also a bijection between $\Z_2$-graded Hilbert bimodules in the usual sense and Hilbert bimodules with appropriate Koszul signs, also see Lemma \ref{Lem:construct}.

A third difference is that we do not require any positivity conditions; $A$ and $B$ need not be $C^*$-algebras and the pairing on $N$ need not be positive.
Such considerations are important for defining unitary topological field theory and generalizations to non-topological field theories, see the discussion at \ref{Sec:positivity}.
In particular, we do not obtain a norm on the Hilbert bimodule and so we do not consider completeness properties.
\end{remark}

	\begin{example}
	We show that there are exactly two $\C$-antilinear stellar structures on $\C$ that are not Morita equivalent.
	We identify $\overline{\C}$ with $\C$ through the canonical algebra isomorphism $\overline{\C} \cong \C$ given by complex conjugation.
	Note that two stellar algebras $(\C,M_1,\sigma_1), (\C,M_2,\sigma_2)$ with underlying modules $M_1 = \C$ and $M_2 = \Pi \C$ will never be Morita equivalent. 
	Indeed, if $N$ is such a Morita equivalence, then $N = \C$ or $N= \Pi \C$ and so $\overline{N} \otimes M_1 \otimes N$ can never be isomorphic to $M_2$, independently of the definitions of $\sigma_1$ and $\sigma_2$.
	
	Now we show that if $(\C,M_1,\sigma_1), (\C,M_2,\sigma_2)$ are stellar structures such that $M_1 = M_2$, then the two stellar structures are Morita-equivalent.
	Given $M_1$ is one of $\C$ or $\Pi \C$, an intertwiner $\sigma_1: \overline{M_1} \to M_1^{\op}$ such that $\sigma_1^{\op} \circ \overline{\sigma_1} = \id$ is a real structure on the one-dimensional vector space $\C$ or $\Pi \C$.
	These are given by $\sigma_1: z \mapsto a_1 \overline{z}$ with $a_1 \in U(1)$.
	Similarly we obtain $\sigma_2(z) = a_2 z$ for $z \in M_2$.
	Now take $N := \C$ to be the trivial bimodule.
	To make $N$ into a stellar bimodule, we have to equip it with a unitarity datum $\phi: M_1 = N \otimes_\C M_1 \otimes_{\overline{\C}^{\op}} \overline{N}^{\op} \to M_2 = M_1$.
	Such an invertible bimodule map is necessarily given by multiplication by some $b \in \C^\times$.
	We will proceed by finding the suitable $b$ that makes the map $\phi$ fulfil the Hermiticity condition required to make it a $1$-morphism of stellar algebras.
	Working out this condition gives the equation
	\[
	\overline{b} a_1 = a_2 b \iff a_1 = a_2 \frac{b^2}{|b|^2}
	\]
	Pick $b \in U(1)$ to be a square root of $a_1/a_2$.
	Then the displayed equation holds since $a_1,a_2$ have unit norm and so $\phi$ defines a unitarity datum on the Morita equivalence $N$.
	This makes $(N,\phi)$ into a $1$-isomorphism in the $2$-category of stellar algebras.
	\end{example}
	
	\begin{example}
	\label{ex:C*notmoritainv}
	Since stellar algebras are a Morita-invariant notion, the last example also shows that there are two stellar structures on $M_n(\C).$
	Also note that stellar algebra structures on $M_n(\C)$ induced by a $*$-algebra always have $M \cong M_n(\C)$ as a bimodule and so are all Morita equivalent, therefore giving equivalent stellar structures.
	Since matrix algebras admit $*$-structures that are not $C^*$, this in particular shows that being a $C^*$-algebra is not a Morita-invariant notion.
	\end{example}

\begin{example}
If $(A,M,\sigma)$ is a stellar algebra, then $(A,\Pi M,\Pi \sigma)$ is a stellar algebra.
If we start with $*$-algebras $A,B$, then a stellar bimodule between $A_*$ and $\Pi B_*$ is a pairing similar to Proposition \ref{prop:Hilbert module}, the only difference being that the pairing now has odd degree.
A similar statement holds for stellar bimodule between $\Pi A_*$ and $B_*$.
Stellar bimodules between $\Pi A_*$ and $\Pi B_*$ are in bijection with stellar bimodules between $A_*$ and $B_*$.
\end{example}

\begin{example}
\label{ex:stellar A_c}
Let $(A,M,\sigma)$ be a $\C$-antilinear stellar algebra.
Then the $(A,A)$-bimodule $A_{(-1)^F}$ becomes a stellar module as follows.
Note first that $(\overline{A}^{\op})_{(-1)^F}$ is an inverse of $\overline{A_{(-1)^F}}^{\op}$ by 
\[
(\overline{A}^{\op})_{(-1)^F} \otimes_{\overline{A}^{\op}} \overline{A_{(-1)^F}}^{\op} \to \overline{A}^{\op} \quad \overline{a}_1^{\op} (-1)^F \otimes \overline{a_2 (-1)^F}^{\op} \mapsto \overline{a}_1^{\op} \overline{a}_2^{\op} = (-1)^{|a_1||a_2|} \overline{a_2 a_1}^{\op}
\]
and a similar map in the other direction.
Note that in contrast to what one might expect, there is no sign $(-1)^{|a_2|}$ coming from exchanging the $a_2$ and the $(-1)^F$ in the above equation.
This makes us able to fill the desired square
\[
\begin{tikzcd}
A \arrow[rr,"A_{(-1)^F}"]& \ & A
\\
\overline{A}^{\op} \arrow[rr, "(\overline{A}^{\op})_{(-1)^F}"] \arrow[u,"M"] & \ & \overline{A}^{\op} \arrow[u,"M"] \arrow[ll,"\overline{A_{(-1)^F}}^{\op}", bend left]
\end{tikzcd}.
\]
For the middle square we used the fact that $A \mapsto A_{(-1)^F}$ is a natural transformation.
This results in the unitarity datum
\[
a (-1)^F \otimes_A m \otimes_{\overline{A}^{\op}} \overline{ b(-1)^F}^{\op} \mapsto (-1)^{|m|} am \overline{b}^{\op}.
\]
Note that confusingly there is again no sign with respect to $b$.
For example, if $(A,M,\sigma)$ comes from a $*$-algebra, then the corresponding $A$-valued inner product on $A_{(-1)^F}$ is
\[
\langle a (-1)^F , b (-1)^F \rangle = a b^*.
\]
This is a sesquilinear, nondegenerate Hermitian pairing satisfying the desired equations.
For example
\begin{align*}
\langle a (-1)^F a_0, b (-1)^F \rangle &= (-1)^{|a_0|} a a_0 b^* = (-1)^{|a_0| + |b| |a_0|} a (b a_0^*)^* = (-1)^{|a_0| + |b| |a_0|} \langle a (-1)^F, b a_0^* (-1)^F \rangle 
\\
&= (-1)^{|b| |a_0|} \langle a (-1)^F, b (-1)^F a_0^* \rangle.
\end{align*}
\end{example}

	\begin{example}
	A brief computation shows that there are exactly two complex antilinear super $*$-structures $*_{\pm}$ on $\C l_1$ defined by $e^{*_\pm} = \pm i e$.
	The two complex-antilinear stellar structures induced by these are not Morita equivalent.
	Suppose they were. Let $M$ be an invertible stellar $(\C l_1,\C l_1)$-bimodule between the stellar algebras coming from the $*$-structures $*_+$ and  $*_-$ respectively.
We then compute the $(\C l_1,*_- )$-valued inner product
\begin{align*}
e \langle 1,1 \rangle = \langle e,1 \rangle = \langle 1, e^{*_+} \rangle = \langle 1, ie \rangle = \langle 1, i \rangle e^{*_-} = - \langle 1,i \rangle i e = - \langle 1,1 \rangle e
\end{align*}
Since $\langle 1,1 \rangle$ is even, it is a multiple of $1$ and so the computation implies it is zero.
This contradicts the fact that the pairing should be nondegenerate.
	\end{example}

\begin{example}
Given a stellar algebra $A$. the multiplication map 
\[
A \otimes_A A \to A
\]
is unitary.
For example, for $A$ a $*$-algebra, this follows from the computation
\begin{align*}
\langle a_1 \otimes a_2, b_1 \otimes b_2 \rangle_{A \otimes_A A} &= (-1)^{|b_1||b_2|} \langle a_1 \langle a_2, b_2 \rangle_A, b_1 \rangle_A = (-1)^{|b_1||b_2|} a_1 a_2 b_2^* b_1^{*} =  \langle a_1 a_2, b_1 b_2 \rangle_A
\end{align*}
\end{example}

The following theorem is relevant for classifying two-dimensional TFTs with reflection structures that do not necessarily satisfy spin-statistics.

\begin{theorem}
\label{th: stellar algs}
There is an equivalence of bicategories between the core of antilinear stellar algebras $\stAlg$ and $2$-category of $\Z_2^R$-fixed points $(\sAlg^{\text{fd}})^{\Z_2^R}$ under the action $A \mapsto \overline{A}^{\op}$. 
\end{theorem} 
\begin{proof}
A $\Z_2^R$-fixed point consists of a super algebra $A$ together with an invertible $(A,\overline{A}{}^\op)$ bimodule $M$ and a $2$-isomorphism \[
\sigma := \phi_{R,R}: M \to \overline{M}^{\op},
\]
see Definition~\ref{Def: Hfixed point}.
The only nontrivial condition is Diagram \ref{Diagram A} in the case where $g = g' = g'' = R$.
This gives the desired condition $\sigma^{\op} \sigma = \id$.

Looking at Definition \ref{def:1-mor of hfps}, a 1-morphism between fixed points $(A,M_1, \sigma_1)$ and $(B,M_2, \sigma_2)$ consists of an $B$-$A$ 
bimodule $N$ together with a 2-isomorphism $\phi := f_R$ filling the square 
\[
\begin{tikzcd}
A \arrow[r,"N"]  & B 
\\
\overline{A}^{\op} \arrow[u,"M_1"] & \overline{B}^{\op} \arrow[l,"\overline{N}^{\op}"] \arrow[u,"M_2", swap]
\end{tikzcd}
\]
Writing out what this entails gives the unitarity datum with the desired domain and target.
The only condition of being a $1$-morphism of fixed points is the equality
\[
\begin{tikzcd}
\arrow[dddd,equals, bend right=50] A \arrow[rr,"N"] & \ & B \arrow[dddd,equals, bend left=50]
\\
\arrow[rr,Rightarrow,"{\phi}", shorten <=15, shorten >=15] \arrow[dd, bend right=50, Rightarrow, "{\sigma_1}", swap, shorten >= 10, shorten <= 10]& \ & \ \arrow[dd, bend left=50, Rightarrow, "{\sigma_2}", shorten >= 10, shorten <= 10]
\\
\overline{A}^{\op} \arrow[uu,"M_1"] \arrow[dd,"\overline{M}_1^{\op}"] & \ & \overline{B}^{\op} \arrow[ll,"\overline{N}^{\op}"] \arrow[uu,"M_2"]  \arrow[dd,"\overline{M}_2^{\op}"]
\\
\arrow[rr,Rightarrow,"\overline{\phi}^{\op}", shorten >=15,  shorten <=15] & \ & \
\\
\overline{\overline{B}^{\op}}^{\op} \arrow[rr, "\overline{\overline{N}^{\op}}^{\op}"]& \ & \overline{\overline{B}^{\op}}^{\op}
\end{tikzcd}
\]
This is equivalent to the Hermiticity condition.
Composition of $1$-morphisms is given by the composition of stellar bimodules as in Definition \ref{Def:stAlg}.

A 2-morphism of fixed points $(N, \phi) \Longrightarrow (N',\phi')$ is a bimodule map $\psi: N \to N'$ such that 
\[
\begin{tikzcd}
A \arrow[rr,"N'"] & \ & B
\\
\arrow[rr,Rightarrow,"\phi' ", shorten >=10,  shorten <=10] & \ & \
\\
\overline{A}^{\op} \arrow[uu,"M_1"] & \ & \overline{B}^{\op} \arrow[ll, "\overline{N'}^{\op}"] \arrow[uu, "M_2"] \arrow[ll, bend left, "\overline{N}^{\op}"]
\end{tikzcd}
= 
\begin{tikzcd}
A \arrow[rr,"N"] \arrow[rr, bend left, "N '"]& \ & B
\\
\arrow[rr,Rightarrow,"\phi", shorten >=10] & \ & \
\\
\overline{A}^{\op} \arrow[uu,"M_1"] & \ & \overline{B}^{\op} \arrow[ll, "\overline{N}^{\op}"] \arrow[uu, "M_2"]
\end{tikzcd}
\]
where we fill $N \Longrightarrow N '$ with $\psi$ and $\overline{N}^{\op} \Longrightarrow \overline{N'}^{\op}$ with $\overline{\psi}^{\op}$.
This is equivalent to the unitarity condition for $\psi$.
\end{proof}


We finish the section by defining the complex conjugate of a stellar algebra in a similar spirit to how we defined the complex conjugate for Hermitian super vector spaces in Section \ref{Sec:fermrep}.
Indeed if $(A,M,\sigma)$ is a stellar algebra, then using the commutation data of the $\Z_2^R$ and $\Z_2^B$ actions on $\sAlg^{\text{fd}}$ and the above theorem, we obtain a canonical structure of a stellar algebra on $\overline{A}$.
However, similarly to one spacetime dimension, we will need to change this stellar structure by $(-1)^F$ to get the correct classification of two-dimensional reflection theories in the coming sections.

\begin{definition}
\label{def:barstellar}
Define the stellar structure on $\overline{A}$ to be 
\[
\overline{M} \otimes_{A^{\op}} A^{\op}_{(-1)^F} \cong \overline{A}_{(-1)^F} \otimes_{\overline{A}} \overline{M}
\]
 with
\[
\overline{M} \otimes_{\overline{\overline{A}}^{\op}} \overline{\overline{A}}^{\op}_{(-1)^F} \xrightarrow{\overline{\sigma} \otimes \id}  \overline{\overline{M}}^{\op} \otimes_{\overline{\overline{A}}^{\op}} \overline{\overline{A}}^{\op}_{(-1)^F} 
\cong \overline{(\overline{A}_{(-1)^F} \otimes_{\overline{A}} \overline{M})}^{\op} 
\cong \overline{\overline{M} \otimes_{\overline{\overline{A}}^{\op}} \overline{\overline{A}}^{\op}_{(-1)^F}  }^{\op}
\]
\end{definition}

In a concrete schematic formula the bimodule map does
\[
\overline{m} (-1)^F \mapsto (-1)^{|m|} \overline{\sigma(m)} (-1)^F.
\]
In case the stellar structure comes from a $*$-algebra structure, this formula agrees with the $*$-algebra structure on $\overline{A}$ explained in Appendix \ref{Sec:*-algs}.

If $(A_1,M_1,\sigma_1), (A_2,M_2,\sigma_2)$ are two stellar algebras and $N$ is an $(A_2,A_1)$-bimodule, we make the $(\overline{A_2},\overline{A_1})$ bimodule $\overline{N}$ into a stellar bimodule with respect to the stellar structures on $\overline{A_i}$ as above by filling the diagram
\[
\begin{tikzcd}
\overline{\overline{A_1}}^{\op} \arrow[d,"{\overline{M_1}}"]  &  \overline{\overline{A_2}}^{\op} \arrow[l,"{\overline{\overline{N}}^{\op}}"] \arrow[d,"{\overline{M_2}}"]
\\
\overline{A_1} \arrow[d,"{\overline{A_1}_{(-1)^F}}"] \arrow[r,"\overline{N}"] & \overline{A_2} \arrow[d,"{\overline{A_2}_{(-1)^F}}"]
\\
\overline{A_1} \arrow[r,"\overline{N}"] & \overline{A_2}
\end{tikzcd}
\]
The lower square is filled by naturality of $(.)_{(-1)^F}$ and the upper by functorialty data of the contravariant functor $\overline{(.)}^{\op}$.
We can describe this explicitly in the case where $A_1$ and $A_2$ are $*$-algebras.
Recall from Appendix \ref{Sec:*-algs} that the $*$-algebra structure on $\overline{A_i}$ is defined as $\overline{a}^* = (-1)^{|a|} \overline{a^*}$ in agreement with Definition \ref{def:barstellar}.
Let $\langle .,. \rangle$ denote the $A_2$-valued inner product on $N$.
The resulting inner product on $\overline{N}$ is 
\[
\langle \overline{n}_1, \overline{n}_2 \rangle = (-1)^{|n_2|}\overline{\langle n_1, n_2 \rangle}
\]
This strange sign indeed makes it into a Hilbert bimodule for the $*$-structures on the $\overline{A_i}$.
For example
\begin{align*}
\langle \overline{n}_1, \overline{b n_2} \rangle &= (-1)^{|b| + |n_2|} \overline{\langle n_1, b n_2 \rangle}
= (-1)^{|b| + |n_2| + |b| |n_2|} \overline{\langle n_1, n_2 \rangle b^*}
=(-1)^{|n_2| + |b| |n_2|} \overline{\langle n_1, n_2 \rangle} \cdot \overline{b}^*
\\
&= \langle \overline{n}_1, \overline{n}_2 \rangle \overline{b}^*
\end{align*}
for $b \in A_2$.

\subsection{$\Spin_2 \rtimes (\Z_2^R\times B\Z_2^F)$ fixed points and stellar Frobenius algebras}
\label{Sec:O2 FP}

In Section~\ref{Sec:1D} we computed the category of $\Z_2^R \cong \Spin_1 \rtimes (\Z_2^R\times B\Z_2^F)$-fixed points in finite-dimensional complex super vector spaces (see Remark~\ref{Rem: Spin}).
The result was the category of super Hermitian vector spaces.
We used this to obtain the classification for general fermionic symmetry group by computing further $G_b$-fixed points.
In this section, we will study the two-dimensional analogue.
Therefore we will compute two-dimensional $\Spin$-TFTs as the bicategory of $\Spin_2 \rtimes (\Z_2^R\times B\Z_2^F)$-fixed points in finite-dimensional semisimple superalgebras.
The resulting bicategory has objects finite-dimensional semisimple Frobenius algebras that are symmetric in the ungraded sense together with a compatible stellar structure.
In the next section, we compute the $G_b$-action on this bicategory and its fixed points.

Recall the existence of a symmetric monoidal $\Z_2^B \times B\Z_2^F$-action on $\sAlg$ shown in Appendix \ref{Sec:Z2xBZ2 action on sAlg}.
The symmetric monoidality together with the cobordism hypothesis leads to a $\Z_2^B \times B\Z_2^F \times O_2$-action on $\sAlg^{\text{fd}}$ explicitly spelled out in Appendix \ref{Sec:total action}.
We will slightly readjust this action to make it into the $\Spin_2 \rtimes (\Z_2^R\times B\Z_2^F)$-action that we will need.
First of all, the reflection action $\Z_2^R$ is the composition of the bar $\Z_2^B$ and the action of $s \in O_1 \subseteq O_2$ by the dual functor (defined to be the dual inverse on $1$-morphisms).
Therefore we will focus on the diagonal element $R := B s$ as a generator.
Because $B$ commutes with the $O_2$-factor, $R$ will satisfy the same nontrivial commutation-relations with the $SO_2$-factor as $s$, i.e. $R \gamma = \gamma^{-1} R$ for $\gamma$ a $1$-morphism in the $2$-group $SO_2$.
We thus obtain a $\Z_2^B \times (SO_2 \rtimes (\Z_2^R \times B\Z_2^F))$-action on $\sAlg^{\text{fd}}$, where the action of $B\Z_2^F$ on $SO_2$ is trivial while the action of $\Z_2^R$ is given by $\gamma \mapsto \gamma^{-1}$ on morphisms.
We now want to lift this to a $\Z_2^B \times (\Spin_2 \rtimes (\Z_2^R \times B\Z_2^F))$-action compatibly with the double cover map $\Spin_2 \to SO_2$.
For this we have to make a choice of how $B \Z_2^F$ acts on $\Spin_2$.
To implement the spin-statistics connection correctly, we take this action to be given by the inclusion of $c  \in \Spin_2$, which is indeed central and squares to one.
Therefore we have a $\Z_2^B \times (\Spin_2 \rtimes (\Z_2^R\times B\Z_2^F))$-action on $\sAlg^{\text{fd}}$.
To compute 2d $\Spin$-TFTs with reflection structure and spin-statistics connection, we will restrict this action and only compute $\Spin_2 \rtimes (\Z_2^R\times B\Z_2^F)$-fixed points.
For symmetry groups with time-reversal symmetries however, it turns out that we will need the $\Z_2^B$-action, so it is useful to keep it around.

To warm up for the general case, we first restrict $\Spin_2 \rtimes \Z_2^F$-fixed points, so spin theories satisfying spin-statistics but without reflection structures.
This computation is also of interest for the classification of 2d spin-statistics theories without reflection structure with a general fermionic symmetry group.
Recall that we use the fermionically skeletal model for $\Spin_2$ in which there are two objects $1, c$ and a generating morphism $\eta: 1 \to c$.
Under the double cover to $SO_2$, $\eta$ is mapped to a generator of $\pi_1(SO_2)$.
This generator maps $A$ to the Serre automorphism of $A$, which is the $(A,A)$-bimodule given by the linear dual $A^* = \operatorname{Lin}(A,\C)$ in $\sAlg$.
We emphasize that is has no relationship with the dual of the object $A \in \sAlg$, which is given by $A^{op}$.
See Section \ref{Sec:serre} for the definition of the Serre automorphism in a general bicategory and Section \ref{Sec:salgdual} for information specifically on $\sAlg$.
The naturality of this automorphism of $A$ as a natural transformation from the identity functor to itself contains more data, as discussed in Appendix \ref{Sec:salgdual}.
The $B \Z_2^F$ acts by $A \mapsto A_{(-1)^F}$ and we have a $\Spin_2\rtimes B\Z_2^F$-action because the two fillers of the square
\[
\begin{tikzcd}
A \arrow[r, "A_{(-1)^F}"] \arrow[d,"A^*",swap] & A \arrow[d,"A^*"]
\\
A \arrow[r,"A_{(-1)^F}",swap]& A
\end{tikzcd}
\]
corresponding to naturality of $A \mapsto A_{(-1)^F}$ and naturality of $A \mapsto A^*$ are equal.
This is already a consequence of $(-1)^F$ being a symmetric monoidal natural transformation so that it preserves the $SO_2$-action, but we also check it directly in Example \ref{ex: (-1)^F Serre nat}.

\begin{lemma}
The bicategory of two-dimensional $\Spin_2$-TFTs i.e. $\Spin_2$-fixed points in the core of finite-dimensonal semisimple superalgebras is equivalent to the bicategory with 
\begin{itemize}
\item \textbf{Objects:} strongly $\Z_2^c$-graded finite-dimensional semisimple superalgebras $A \oplus A_c$ together with an invertible $(A,A)$-bimodule isomorphism $A_\nu: A \to A^* \otimes_A A_c$ such that
\[
\begin{tikzcd}
A^* \otimes_A A_c \otimes_A A_c \arrow[r]
&
 A_c \otimes_A A^* \otimes_A A_c
\\
A \otimes_A A_c \arrow[u, "A_\nu"] \arrow[r]
&
A_c \otimes_A A \arrow[u, "A_\nu"]
\end{tikzcd}
\]
commutes, where the upper horizontal arrow is given by Serre naturality $S_{A_c}$.
\item \textbf{1-morphisms:} from $(A \oplus A_c, A_\nu)$ to $(B \oplus B_c, B_\nu)$ consist of a $(B,A)$-bimodule $N$ and a bimodule isomorphism $f_c: N \otimes_A A_c \to B_c \otimes_B N$ such that the diagrams
\[
\begin{tikzcd}
N \otimes_A A_c \otimes_A A_c \arrow[d,"{f_c \otimes \id}",swap] \arrow[r] & N
\\
B_c \otimes_B N \otimes_A A_c \arrow[r, "{\id \otimes f_c}",swap] & B_c \otimes_B B_c \otimes_B N \arrow[u]
\end{tikzcd}
\]
and 
\begin{equation}
\label{eq: compatible Frobenius}
\begin{tikzcd}
A^* \otimes_A A_c \otimes_A N \arrow[r]  & N \otimes_B B^* \otimes_B B_c 
\\
N\arrow[u,"A_\nu"] \arrow[ur,"B_\nu",swap]& \
\end{tikzcd}
\end{equation}
commute.
\item \textbf{2-morphisms:} from $(N,f_c)$ to $(N',f_c')$ are bimodule isomorphisms $\chi: N \to N'$ such that
\[
\begin{tikzcd}
N \otimes_A A_c \otimes_A A_c \arrow[d,"{f_c \otimes \id}",swap] \arrow[r] & N
\\
B_c \otimes_B N \otimes_A A_c \arrow[r, "{\id \otimes f_c}",swap] & B_c \otimes_B B_c \otimes_B N \arrow[u]
\end{tikzcd}
\]
commutes.
\end{itemize}
\end{lemma}
\begin{proof}
We use the decomposition theorem (see Proposition~\ref{Prop: Fixed points}) on the defining exact sequence
\[
1 \to \Z_2^c \to \Spin_2 \to SO_2 = B\Z \to 1.
\]
So the first step is computing $\Z_2^c$-fixed points for the trivial $\Z_2^c$-action.
In this bicategory objects are triples $(A,A_c, \phi_{c,c}: A_c \otimes_A A_c \to A)$ consisting of a finite-dimensional semisimple superalgebra $A$, an invertible $(A,A)$-bimodule $A_c$ and an even bimodule isomorphism $\phi_{c,c}: A_c \otimes_A A_c \to A$.
These satisfy the condition that the two maps $A_c \otimes_A A_c \otimes_A A_c \to A$ defined by $\phi_{c,c}$ are equal.
This is equivalent to $A \oplus A_c$ being a strongly $\Z_2^c$-graded superalgebra.
$1$-morphisms from $(A, A_c, \phi_{c,c})$ to $(B, B_c, \psi_{c,c})$ in this bicategory consist of an invertible $(B,A)$-bimodule $N$ and an even bimodule isomorphism $f_c: N \otimes_A A_c \to B_c \otimes_B N$ such that the diagram
\[
\begin{tikzcd}
N \otimes_A A_c \otimes_A A_c \arrow[d,"{f_c \otimes \id}",swap] \arrow[r,"{\phi_{c,c}}"] & N
\\
B_c \otimes_B N \otimes_A A_c \arrow[r, "{\id \otimes f_c}",swap] & B_c \otimes_B B_c \otimes_B N \arrow[u,"{\psi_{c,c}}",swap]
\end{tikzcd}
\]
commutes.
$2$-morphisms from $(N,f_c)$ to $(N',f_c')$ are even invertible bimodule maps $\chi: N \to N'$ such that the diagram
\[
\label{eq: 2-morph condition}
\begin{tikzcd}
N \otimes_A A_c \arrow[d,"f_c", swap] \arrow[r,"\chi"] & N' \otimes_A A_c \arrow[d,"f_c'"]
\\
B_c \otimes_B N \arrow[r,"\chi",swap] & B_c \otimes_B N'
\end{tikzcd}
\]
commutes.

Next we compute the $B\Z$-action on $\Z_2^c$-fixed points.
The action is uniquely specified once we give the action of a generator of $B\Z$.
Pick the generator which lifts to the morphism $\nu: 1 \to c$ in $\Spin_2$ and note it is not a loop.
In the original action of $\Spin_2$ on the core of fully dualizable superalgebras, $\rho(\nu)$ was a natural transformation from the identity to $\rho(c)$.
Even though $\rho(c)$ is the identity too, the induced action on $\Z_2^c$-fixed points coming from the decomposition theorem requires us to use the fixed point data $A_c$ to get $\rho(c) \Longrightarrow \id_{\sAlg^{\text{fd}}}$.
Therefore the action of $SO_2$ on $\Z_2^c$-fixed points maps $(A,A_c, \phi_{c,c})$ to the $1$-morphism of fixed points $(A,A_c, \phi_{c,c}) \to (A,A_c, \phi_{c,c})$ given by $A^* \otimes_A A_c$ as a $1$-morphism in $\sAlg$. (Equivalently we could have chosen to work with $A_c \otimes_A A^*$, which gives an equivalent action using Serre naturality applied to $A_c$)
To make it into a $1$-morphism in $(\sAlg^{\text{fd}})^{\Z_2^c}$ we also have to give an isomorphism
\[
(A^* \otimes_A A_c) \otimes_A A_c \to A_c \otimes_A (A^* \otimes_A A_c)
\]
which is given by Serre naturality.

The final step to compute the data of an arbitrary $\Spin_2$-TFTs is to compute $B\Z$-fixed points on the $\Z_2^c$-fixed point category.
The fixed point is uniquely specified by the fixed point on the generating morphism $\nu$, which is a single $2$-morphism $A_\nu: A \to A^* \otimes_A A_c$ in the bicategory of $\Z_2^c$ fixed points.
In other words, it is an even invertible bimodule isomorphism such that the diagram
\[
\begin{tikzcd}
A^* \otimes_A A_c \otimes_A A_c \arrow[r]
&
 A_c \otimes_A A^* \otimes_A A_c
\\
A \otimes_A A_c \arrow[u, "A_\nu"] \arrow[r]
&
A_c \otimes_A A \arrow[u, "A_\nu"]
\end{tikzcd}
\]
commutes.
So an object in the $\Spin_2$-fixed point bicategory is given by a quadruple $(A, A_c, \phi_{c,c},A_\nu)$.

A $1$-morphism $(A, A_c, \phi_{c,c},A_\nu) \to (B, B_c, \psi_{c,c},B_\nu)$ of fixed points is then given by an invertible $(B,A)$-bimodule $N$ and an invertible even bimodule map $f_c: N \otimes_A A_c \to B_c \otimes_A N$ satisfying the condition from before plus the extra condition that the diagram
\begin{equation}
\begin{tikzcd}
A^* \otimes_A A_c \otimes_A N \arrow[r]  & N \otimes_B B^* \otimes_B B_c 
\\
N\arrow[u,"A_\nu"] \arrow[ur,"B_\nu",swap]& \
\end{tikzcd}
\end{equation}
commutes.
The horizontal arrow is given by the composition of $f_c$ and the naturality of the Serre.
A $2$-morphism $(N,f_c) \Longrightarrow (N',f_c')$ of $\Spin_2$-fixed points is given by an even bimodule isomorphism $\chi: N \to N'$ such that the diagram \ref{eq: 2-morph condition} commutes as before.
\end{proof}

\begin{remark}
Our description of $\Spin_2$-TFTs is equivalent to the a priori simpler description in which we take fixed points for the $B\Z$-action generated by the Serre automorphism squared $A \mapsto A^* \otimes A^*$.
In that approach $\Spin_2$-fixed points are described by the single datum of a bimodule isomorphisms $A \mapsto A^* \otimes_A A^*$ instead of the more complicated quadruple $(A,A_c, \phi_{c,c},A_\nu)$.
The relationship between the two is that a quadruple can be mapped to
\[
A \xrightarrow{A_\nu \otimes A_\nu} A^* \otimes_A A_c \otimes_A A^* \otimes_A A_c \to A^* \otimes_A A^* \otimes_A A_c \otimes_A A_c \xrightarrow{\phi_{c,c}} A^* \otimes_A A^*
\]
where we used Serre naturality.

Our approach is preferred in this document for two reasons: when working with fermionic internal symmetry groups $G$ with nontrivial special central element $c \in G$, fixed points for the subgroup $\Z_2^c \subseteq G$ can be identified with the $\Z_2$-graded superalgebra $A \oplus A_c$.
This data fits very naturally in the rest of the fixed point structure.
Secondly, we are mainly interested in topological field theories with a spin-statistics connection, in which $A \oplus A_c$ can be canonically identified with the $\Z_2$-graded superalgebra $A \oplus A_{(-1)^F}$ as we will show in Proposition \ref{prop:spinstatisticsTFTs}.
\end{remark}

\begin{remark}
\label{Rem:O2iso}
There is an isomorphism of $2$-groups $\Spin_2 \rtimes B\Z_2^F \cong B\Z$ sending $c$ to $1$, $(-1)^F$ to the identity and $\gamma: 1 \to c$ to a generator.
Under this isomorphism, the $B\Z$-action on $\sAlg^{\text{fd}}$ is $A \mapsto A_{(-1)^F} \otimes_A A^*$. 
\end{remark}

\begin{proposition}
\label{prop:spinstatisticsTFTs}
The bicategory of $\Spin_2$-TFTs satisfying spin-statistics, i.e. $\Spin_2 \rtimes \Z_2^F$-fixed points, is equivalent to the bicategory $\catf{ungrFrob}$ of ungraded-symmetric Frobenius algebras which is defined as
\begin{itemize}
\item \textbf{Objects:} Finite-dimensional semisimple superalgebras $A$ together with a functional $\lambda: A \to \C$ vanishing on odd elements such that the pairing $(a,b) \mapsto \lambda(ab)$ is nondegenerate and $\lambda(ab) = \lambda(ba)$
\item \textbf{1-morphisms:} from $(A,\lambda)$ to $(A', \lambda')$ are invertible $(A',A)$-bimodules $M$ such that the Serre naturality isomorphism
\[
S_M: A'^* \otimes_{A'} M \to M \otimes_A A^*
\]
satisfies
\[
S_M(\lambda' \otimes m) = (-1)^{|m|} m \otimes \lambda.
\]
\item \textbf{2-morphisms:} from $M$ to $N$ are arbitrary even bimodule isomorphisms.
\end{itemize}
\end{proposition}
\begin{proof}
The most straightforward proof directly applies Remark \ref{Rem:O2iso}, also see the remark directly after the proof.
In general, the machinery of the decomposition theorem is useful and so we illustrate it in the case at hand.
Thus we consider the exact sequence of $2$-groups
\[
1 \to \Spin_2 \to \Spin_2 \rtimes B\Z_2^F \to B\Z_2^F \to 1.
\]
Recall that the corresponding action of $B \Z_2^F$ is given by picking out the nontrivial object $c$ in $\Spin_2$ which is central.
In other words, in the $2$-group $\Spin_2 \rtimes B\Z_2^F$ it defines a morphism $1 \to c$.
Just as before this implies that to obtain the $B\Z_2^F$-action on $\Spin_2$-fixed points, we have to compose the natural transformation $(-1)^F: \id \to \rho(c)$ with the fixed point data for $\Z_2^c$.
Using the notation from the proof of the last lemma, the result is that a $\Spin_2$-fixed point $(A, A_c, \phi_{c,c},A_\nu)$ is mapped to the $1$-morphism $A_{(-1)^F} \otimes_A A_c$ in $\sAlg^{\text{fd}}$.
The data of making it into a $1$-morphism of $\Spin_2$-fixed points still requires what we called $f_c$.
In this case it is given by the naturality of $(-1)^F$
\[
A_{(-1)^F} \otimes_A A_c \otimes_A A_c \to A_c \otimes_A A_{(-1)^F} \otimes_A A_c
\]
The remaining data of this defining a $B\Z_2^F$-action is the modification $(-1)^F \circ (-1)^F \Rrightarrow \id$ where the natural transformations are between functors on the bicategory of $\Spin_2$-fixed points.
The left natural tranformation is given by mapping a $\Spin_2$-fixed point $(A, A_c, \phi_{c,c}, A_\nu)$ to the bimodule $A_{(-1)^F} \otimes_A A_c \otimes_A A_{(-1)^F} \otimes_A A_c$ together with the bimodule isomorphism
\[
(A_{(-1)^F} \otimes_A A_c \otimes_A A_{(-1)^F} \otimes_A A_c) \otimes_A A_c \cong A_c \otimes_A (A_{(-1)^F} \otimes_A A_c \otimes_A A_{(-1)^F} \otimes_A A_c)
\]
given by applying the map from before twice.
The data of $(-1)^F$ squaring to one is given by composing the naturality of $A \mapsto A_{(-1)^F}$ with the multiplication maps $ A_{(-1)^F} \otimes_A A_{(-1)^F} \to A$ and $\phi_{c,c}$:
\begin{equation}
\label{eq: squaring condition}
A_{(-1)^F} \otimes_A A_c \otimes_A A_{(-1)^F} \otimes_A A_c \to A_{(-1)^F} \otimes_A A_{(-1)^F} \otimes_A A_c \otimes A_c \to A.
\end{equation}
We have thus computed the $B\Z_2^F$-action on $\Spin_2$-fixed points.

We now compute fixed points for this action to find the bicategory of $\Spin_2$-theories with spin-statistics connection.
Objects consist of a $\Spin_2$-fixed point $(A,A_c,\phi_{c,c},A_\nu)$ and a bimodule isomorphism $A_\gamma: A \to A_{(-1)^F} \otimes_A A_c$ satisfying several conditions, where we denoted the morphism $(-1)^F \in B\Z_2^F$ acting on $\Spin_2$-fixed points by $\gamma$.
First of all $A_\gamma$ has to be a $2$-morphism of $\Spin_2$-fixed points, leading to the commutative diagram
\[
\begin{tikzcd}
A_c \arrow[r, "A_\gamma"] \arrow[d,equals] & (A_{(-1)^F} \otimes_A A_c) \otimes_A A_c \arrow[d]
\\
A_c \arrow[r, "A_\gamma"] & A_c \otimes_A (A_{(-1)^F} \otimes_A A_c)
\end{tikzcd}
\]
Let $c \in A_c$ be the unique element such that $A_\gamma(1) = (-1)^F \otimes c$.
Then the above commutative diagram is equivalent to the commutation relation $ac = (-1)^{|a|}ca$ in the $(A,A)$-bimodule $A_c$.
Next, the fact that this is not just a $B\Z$ but a $B\Z_2$-fixed point gives the condition that the composition
\[
 A \xrightarrow{A_\gamma \otimes A_\gamma} A_{(-1)^F} \otimes_A A_c \otimes_A A_{(-1)^F} \otimes_A A_c \to A
\]
is the identity, where the last morphism is given in Equation \ref{eq: squaring condition}.
This condition on $A_\gamma$ is equivalent to $\phi_{c,c}(c,c) = 1$, i.e. $c^2 = 1$ in the $\Z_2^c$-graded superalgebra $A \oplus A_c$.

We turn to $1$-morphisms
\[
\underline{A} := (A,A_c,\phi_{c,c}, A_\nu, A_\gamma) \to  \underline{B} := (B,B_c,\psi_{c,c}, B_\nu, B_\gamma)
\]
between $\Spin_2$-TFTs with spin-statistics connection.
Being a $1$-morphism in $\Spin_2$-fixed points, they consist of an invertible $(B,A)$-bimodule $N$ and an invertible even bimodule map $f_c: N \otimes_A A_c \to B_c \otimes N$ which intertwine $\phi_{c,c}$ with $\psi_{c,c}$ and $A_\nu$ with $B_\nu$ as in Diagram \ref{eq: compatible Frobenius}.
The extra condition on $(N,f_c)$ being a morphism between $\Spin_2 \rtimes B\Z_2^F$-fixed points is
\begin{equation}
\begin{tikzcd}
A_{(-1)^F} \otimes_A A_c \otimes_A N \arrow[r]  & N \otimes_B B_{(-1)^F} \otimes_B B_c 
\\
N\arrow[u,"A_\gamma"] \arrow[ur,"B_\gamma"]& \
\end{tikzcd}
\end{equation}
where we used $f_c$ and the naturality of $A \mapsto A_{(-1)^F}$ in the horizontal arrow.
A $2$-morphism $\chi: (N,f_c) \to (N',f_c')$ of $\Spin_2 \rtimes B\Z_2^F$-fixed points is simply a $2$-morphism of $\Spin_2$-TFTs.

We claim that this bicategory is equivalent to the bicategory $\catf{ungFrob}$ of ungraded-symmetric Frobenius algebras.
The key idea is that $A_\nu$ provides a canonical isomorphism between $A_c$ and $A_{(-1)^F}$, allowing us to effectively identify $A_c$ with $A_{(-1)^F}$.
Note that there is no reason to be careful about distinguishing $A_c$ and its inverse, because we have been given a canonical way to identify $A_c$ as its own inverse through $\phi_{c,c}$.
Since the isomorphism preserves all relevant data, such as mapping the isomorphism $\phi_{c,c}$ to the data of the $B\Z_2^F$-action squaring to one $A_{(-1)^F} \otimes_A A_{(-1)^F} \to A$, this allows us to forget the information $A_c, \phi_{c,c}$ and $A_\gamma$.
More precisely, let $A_{\gamma \nu}: A \to A^* \otimes_A A_{(-1)^F}$ denote the composition of $A_\nu$ with the identification $A_c \cong A_{(-1)^F}$ given by $A_\gamma$.
Such a bimodule isomorphism is equivalent to an ungraded-symmetric Frobenius structure $\lambda$ as follows.
It is uniquely specified by $A_{\gamma \nu}(1) \in A^* \otimes_A A_{(-1)^F}$, which without loss of generality is an element of the form $ \lambda \otimes (-1)^F \in A_{(-1)^F} \otimes A^*$ for some element $\lambda \in A^*$.
The condition that $A_{\gamma \nu}$ is a bimodule map is equivalent to $\lambda(ab) = \lambda(ba)$.
Its invertibility is equivalent to the pairing $(a,b) \mapsto \lambda(ab)$ being nondegenerate.
The condition that $A_{\gamma \nu}$ is even is equivalent yo $\lambda$ being even, which means it vanishes on odd elements of $A$.
We have thus provided a map $F$ from objects of $(\sAlg^{\text{fd}})^{\Spin_2 \rtimes B\Z_2^F}$ to objects of $\catf{ungFrob}$.
We will now extend $F$ to a functor.

On $1$-morphisms we map 
\[
\underline{A} \xrightarrow{(N,f_c)} \underline{B}
\]
to the $1$-morphism $N$.
Consider the following diagram
\[
\begin{tikzcd}
N \arrow[rr, bend left, "{A_{\gamma \nu}}"] \arrow[dd, bend right=80, "{B_{\gamma \nu}}", swap] \arrow[r, "{A_\nu}"] \arrow[d, "{B_\nu}"]  & N \otimes_A A^* \otimes_A A_c \arrow[r,"{A_\gamma}"] & N \otimes_A A^* \otimes_A A_{(-1)^F} \arrow[dd]
\\
B^* \otimes_{B} B_c \otimes_{B} N \arrow[d,"{B_\gamma}"] \arrow[ur] & \ & \
\\
B^* \otimes_{B} B_{(-1)^F} \otimes_{B} N \arrow[rr] \arrow[uurr] & \ & B^* \otimes_{B} N \otimes_A A_{(-1)^F}
\end{tikzcd}
\]
The left-upper triangle commutes because $N$ is a $1$-morphism of $\Spin_2$-fixed points, see Diagram \ref{eq: compatible Frobenius}.
The middle part commutes because $(N,f_c)$ is a $1$-morphism of $B\Z_2^F$-fixed points.
The right-lower triangle commutes on the nose.
The diagram evaluated on an arbitrary $n \in N$ yields the desired equation
\[
S_N(\lambda_B \otimes n) = (-1)^{|n|} n \otimes \lambda_A.
\]
On $2$-morphisms, the functor is the identity.
This clearly defines a functor between morphism categories
\[
\Hom_{(\sAlg^{\text{fd}})^{\Spin_2 \rtimes B\Z_2^F}}(\underline{A}, \underline{B}) \to \Hom_{\catf{ungrFrob}}((A,\lambda_A), (B,\lambda_B)).
\]
We show it is fully faithful and essentially surjective by showing that the condition on an even bimodule isomorphism $\chi: N \to N'$ that it is is a $2$-morphism of $\Spin_2$-fixed points is automatic.
For this, consider the diagram
\[
\begin{tikzcd}
N \otimes_A A_c \arrow[rrr, "\chi"] \arrow[dr,"A_\gamma"] \arrow[ddd,"f_c"]
& \ 
& \
& N' \otimes_A A_c \arrow[dl,"A_\gamma"] \arrow[ddd,"f_c'"]
\\
\ 
& N \otimes_A A_{(-1)^F} \arrow[r,"\chi"] \arrow[d] 
& N' \otimes_A A_{(-1)^F} \arrow[d]
& \
\\
\
& B_{(-1)^F} \otimes_B N \arrow[r,"\chi"] 
& B_{(-1)^F} \otimes_B N'
& \
\\
 B_c \otimes_B N \arrow[rrr, "\chi"] \arrow[ur,"B_\gamma"]
& \ 
& \
& B_c \otimes_B N' \arrow[lu,"B_\gamma"]
\end{tikzcd}
\]
The middle square commutes by naturality of $A_{(-1)^F}$.
The left quadrilateral commute because $(N,f_c): \underline{A} \to \underline{B}$ is a $1$-morphism of $\Spin_2 \rtimes B \Z_2^F$-fixed points and the right quadrilateral because $(N',f_c')$ is.
The upper and lower quadrilateral commute by definition.
The resulting outer diagram shows that $\chi$ is a $2$-morphism of $\Spin_2 \rtimes B \Z_2^F$-fixed points.
We conclude that 
\[
\Hom_{(\sAlg^{\text{fd}})^{\Spin_2 \rtimes B\Z_2^F}}(\underline{A}, \underline{B}) \to \Hom_{\catf{ungrFrob}}((A,\lambda_A), (B,\lambda_B)).
\]
is an equivalence of categories. 

Composition of $1$-morphisms in $(\sAlg^{\text{fd}})^{\Spin_2 \rtimes B\Z_2^F}$ is preserved on the nose by the functor $F$ to ungraded-symmetric Frobenius algebras.
It also preserves identities.
We show that $F$ is essentially surjective.
Given an ungraded-symmetric Frobenius algebra $(A,\lambda)$, the algebra $A$ assembles into a $\Spin_2 \rtimes B\Z_2^F$-fixed point using  $A_c := A_{(-1)^F}$, trivial $A_{\gamma}$, $\phi_{c,c}$ the multiplication map $A_{(-1)^F} \otimes_A A_{(-1)^F} \to A$ and $A_\nu: A \to A^* \otimes_A A_{(-1)^F}$ is the bimodule map satisfying $A_\nu(1) = \lambda \otimes (-1)^F$.
Indeed, because $A_{(-1)^F}$ is a $B\Z_2^F$-action, this defines a $\Z_2^c$-fixed point.

We now show this is a $\Spin_2$-fixed point.
First note that because the naturality data of $(-1)^F$ applied to the bimodule $A_{(-1)^F}$ is the identity, the two bimodule isomorphisms
\[
A^* \otimes_A A_{(-1)^F} \otimes_A A_{(-1)^F} \to A_{(-1)^F} \otimes_A A^* \otimes_A A_{(-1)^F}
\]
given by applying the naturality to the bimodule $A^* \oplus_A A_{(-1)^F}$ and applying the naturality to $A^*$ and then tensoring with $\id_{A_{(-1)^F}}$ are equal.
Moreover, for any even bimodule map $\psi: N \to N'$ between $(B,A)$-bimodules, the diagram
\[
\begin{tikzcd}
N' \otimes_A A_{(-1)^F} \arrow[r] & B_{(-1)^F} \otimes_B N'
\\
N \otimes_A A_{(-1)^F} \arrow[r] \arrow[u,"\psi"] & B_{(-1)^F} \otimes_B N \arrow[u,"\psi"]
\end{tikzcd}
\]
commutes by naturality of the $B\Z_2^F$-action (or more concretely, because $\psi$ is even).
Now, for showing $A_\nu: A \to A^* \otimes_A A_{(-1)^F}$ is a $2$-morphism of $\Z_2^c$-fixed points, we have to show that the diagram
\[
\begin{tikzcd}
(A^* \otimes_A A_{(-1)^F}) \otimes_A A_{(-1)^F} \arrow[r]
&
 A_{(-1)^F} \otimes_A (A^* \otimes_A A_{(-1)^F})
\\
A \otimes_A A_{(-1)^F} \arrow[u, "A_\nu"] \arrow[r]
&
A_{(-1)^F} \otimes_A A \arrow[u, "A_\nu"]
\end{tikzcd}
\]
commutes, where we used the naturality of the Serre automorphism on the top. 
Since the naturality data of the Serre automorphism applied to the $(A,A)$-bimodule $A_{(-1)^F}$ agrees with the naturality data of the $B\Z_2^F$-action on $\sAlg^{\text{fd}}$ applied to $A^*$ by Example \ref{ex: (-1)^F Serre nat}, we can replace Serre naturality by $(-1)^F$-naturality.
Now the above remarks together with the above diagram for the case $B = A, N = A, N' = A^* \otimes_A A_{(-1)^F}$ and $\psi = A_\nu$ shows that the last diagram commutes.
So we have defined a $\Spin_2$-fixed point.
Finally, the $(A,A)$-bimodule $A_{(-1)^F}$ satisfies $a (-1)^F = (-1)^{|a|} (-1)^F a$ and $\phi_{c,c}$ satisfies $((-1)^F)^2 = 1$ so that $A_\gamma$ is a $2$-morphism of $\Spin_2$-fixed points.
We have shown $(A, A_{(-1)^F}, \phi_{c,c}, A_\nu, A_\gamma)$ is a $\Spin_2 \rtimes B\Z_2^F$-fixed point.
As it clearly maps to $(A,\lambda)$ under $F$, we conclude $F$ is essentially surjective and so it is an equivalence of $2$-categories.
\end{proof}

\begin{remark}
The condition for an $(A',A)$-bimodule $N$ between ungraded-symmetric Frobenius algebras to be a $1$-morphism in $\catf{ungFrob}$ can be alternatatively expressed by saying that the diagram
\[
\begin{tikzcd}
N \arrow[r,"{\id_N \otimes \lambda}"] \arrow[d,"{\lambda' \otimes \id_{N}}",swap]& N \otimes_A A_{(-1)^F} \otimes_A A^* \arrow[d] 
\\
A'_{(-1)^F} \otimes_{A'} A'^* \otimes_{A'} N \arrow[r] & A'_{(-1)^F} \otimes_A N \otimes_A A^*
\end{tikzcd}
\]
commutes, where the two unlabeled arrows are the naturality of the Serre automorphism and $A_{(-1)^F}$.
\end{remark}

We now turn to the full $\Spin_2\rtimes (\Z_2^R\times B\Z_2^F)$-action on $\sAlg^{\text{fd}}$ and so we will provide some of the remaining explicit coherence data, details are explained in Appendix \ref{App: Super alg}.
The fact that the $\Z_2^R$-action $A \mapsto \overline{A}^{\op}$ and the $B\Z_2^F$-action combine to a $\Z_2^R \times B\Z_2^F$-action is given by the $2$-isomorphism
\[
(\overline{A}^{\op})_{(-1)^F} \cong \overline{{}_{(-1)^F} A}^{\op} = \overline{A_{(-1)^F}}^{\op -1}.
\]
The remaining data of the $\Spin_2$- and $\Z_2^R$-actions combining to a $\Spin_2 \rtimes \Z_2^R$-action is the $2$-isomorphism 
\[
(\overline{A}^{\op})^* \cong \overline{A^*}^{\op}.
\]
It corresponds to the equality $\gamma R = R c \gamma^{-1}$.
Indeed, note that there is no inverse on the right hand side of the above equation and $R$ acts by dual inverse on $1$-morphisms.

We now introduce the concept of a stellar Frobenius algebra with the goal of proving Theorem \ref{Th:stfrob}.

\begin{definition}
\label{Def:stfrob}
The bicategory of \emph{stellar Frobenius algebras} $\catf{stFrob}$ is the bicategory in which 
\begin{itemize}
\item \textbf{Objects:} are quadruples $(A,M,\sigma, \lambda)$ consisting of a finite-dimensional semisimple ungraded-symmetric Frobenius superalgebra $(A, \lambda) \in \catf{ungrFrob}$ and a complex-antilinear stellar structure $(M \in \mathbf{Bim}(A,\overline{A}^{\op}), \sigma: \overline{M}^{\op} \to M)$.
The Frobenius structure and the stellar structure are compatible in the sense that the following diagram in the category of $(A, \overline{A}^{\op})$-bimodules commutes
\[
\begin{tikzcd}
M \arrow[r,"{\overline{\lambda}^{\op}}"] \arrow[d,"\lambda"]& 
M \otimes_{\overline{A}^{\op}} \overline{A^* \otimes_A A_{(-1)^F}}^{\op} \arrow[r] &
M \otimes_{\overline{A}^{\op}} (\overline{A}^{\op})^* \otimes_{\overline{A}^{\op}} (\overline{A}^{\op})_{(-1)^F}
\\
A^* \otimes_A A_{(-1)^F} \otimes_A M \arrow[rr] & 
\ &
 A^* \otimes_A M \otimes_{\overline{A}^{\op}} (\overline{A}^{\op})_{(-1)^F} \arrow[u]
\end{tikzcd}
\]


\item \textbf{1-morphisms:} $(A,M,\sigma, \lambda) \to (A', M', \sigma', \lambda')$ are stellar $(A',A)$-bimodules $(N,h)$ which intertwine the Frobenius structures in the sense that $N$ is a $1$-morphism in $\catf{ungFrob}$.

\item \textbf{2-morphisms:} $(N_1, h_1) \Longrightarrow (N_2, h_2)$ are unitary even bimodule maps.
\end{itemize}
\end{definition}

\begin{theorem}
\label{Th:stfrob} The bicategory of $\Spin_2$ spin statistics and reflection field theories is equivalent
to the bicategory of stellar Frobenius algebras $\stFrob$. 
\end{theorem}
\begin{proof}
We apply the decomposition theorem to the exact sequence
\[
1 \to B \Z_2^F \rtimes \Spin_2 \to (\Z_2^R \times B \Z_2^F) \rtimes \Spin_2 \to \Z_2^R \to 1
\]
We know that $(\sAlg^{\text{fd}})^{B \Z_2^F \rtimes \Spin_2}$ is equivalent to the slightly smaller bicategory $\catf{ungFrob}$ by Proposition \ref{prop:spinstatisticsTFTs}.
Therefore it suffices to compute fixed points for the induced $\Z_2^R$-action on $\catf{ungFrob}$.

We start by computing this $\Z_2^R$-action on $\catf{ungFrob}$ using the decomposition theorem.
Recall from the proof of Proposition \ref{prop:spinstatisticsTFTs} that the Frobenius structure $\lambda: A \to A^* \otimes_{A} A_{(-1)^F}$ corresponds under the equivalence $(\sAlg^{\text{fd}})^{B \Z_2^F \rtimes \Spin_2} \cong \catf{ungFrob}$ to the fixed point with respect to the path $c \gamma \nu: 1 \to 1$ in $B \Z_2^F \rtimes \Spin_2$.
Therefore the functor $\rho(R): \catf{ungFrob} \to \catf{ungFrob}$ sends an ungraded Frobenius algebra $(A,\lambda)$ to $\overline{A}^{\op}$ with Frobenius structure
\[
\overline{A}^{\op} \xrightarrow{\overline{\lambda}^{\op}} \overline{A^* \otimes_{A} A_{(-1)^F}}^{\op} \cong (\overline{A}^{\op})^* \otimes_{\overline{A}^{\op}} (\overline{A}^{\op})_{(-1)^F}.
\]
In the last line we used the natural isomorphism $\rho(R) \rho(c \gamma \nu) \cong \rho(\gamma^{-1}\nu) \rho(R)$ between functors $\sAlg^{\text{fd}} \to \sAlg^{\text{fd}}$ provided by the corresponding commutation relation in $\Spin_2\rtimes (\Z_2^R\times B\Z_2^F)$.
Since $1$- and $2$-morphisms in $\catf{ungFrob}$ are $1$- and $2$-morphisms in $\sAlg^{\text{fd}}$ with extra conditions, the functor $\rho(R)$ is equal to the action of $R$ on $\sAlg^{\text{fd}}$ on these.
In particular the data of preservation of composition of $1$-morphisms $M: (A, \lambda_A) \to (B,\lambda_B), N: (B, \lambda_B) \to (C,\lambda_C)$ is given by the corresponding natural isomorphism
\[
\overline{N \otimes_B M}^{\op -1} \cong \overline{N}^{\op -1} \otimes_{\overline{B}^{\op-1}} \overline{M}^{\op -1}
\]
in $\sAlg^{\text{fd}}$.
Similarly given $(A,\lambda_A) \in \catf{ungFrob}$, the natural isomorphism $\overline{\overline{A}^{\op-1}}^{\op -1} \cong A$ is simply the corresponding natural isomorphism in $\sAlg^{\text{fd}}$.

We turn to fixed points.
Objects in $\catf{ungFrob}^{\Z_2^R}$ consists of an ungraded-symmetric Frobenius algebra $(A,\lambda)$ together with a $1$-morphism $\rho(R)(A, \lambda) \to (A, \lambda)$ in $\catf{ungFrob}$ and a certain $2$-morphism in $\catf{ungFrob}$.
Forgetting the fact that these are $1$- and $2$-morphisms in $\catf{ungFrob}$ and viewing them in $\sAlg^{\text{fd}}$ shows that these two data assemble into a stellar structure on $A$ in the same way as in the proof of Theorem \ref{th: stellar algs}.
Now, the above $1$-morphism in $\catf{ungFrob}$ is an invertible $(A,\overline{A}^{\op})$-bimodule $M$ satisfying the condition that the Frobenius structures of $A$ and $\overline{A}^{\op}$ are compatible.
Looking at the definition of the Frobenius structure on $\overline{A}^{\op}$, this is exactly the condition on objects given in the statement of the desired theorem.

So we turn to $1$-morphisms of $\Z_2^R$-fixed points.
They consist of a $1$-morphism $N: (A,\lambda) \to (A',\lambda')$ in $\catf{ungFrob}$ together with a $2$-morphism $h$ in $\catf{ungFrob}$ satisfying a certain condition.
$2$-morphisms in $\catf{ungFrob}$ are simply $2$-morphisms in $\catf{ungFrob}$ and the condition on $h$ is the same as the condition of it being a morphism of $\Z_2^R$-fixed points in $\sAlg^{\text{fd}}$.
We can now repeat the relevant part of the proof of Theorem \ref{th: stellar algs} to conclude that $h$ is a unitarity datum on $N$.
Similarly $2$-morphisms in $\catf{ungFrob}^{\Z_2^R}$ are simply $2$-morphisms in $\sAlg^{\text{fd}}$ satisfying a compatibility condition between the two stellar modules and so we can reduce back to Theorem \ref{th: stellar algs} again to conclude.

\end{proof}

\begin{remark}
Let $(A, M, \sigma)$ be a finite-dimensional semisimple stellar algebra.
Note that $A_{(-1)^F}$ becomes a stellar bimodule through the naturality of $A \mapsto A_{(-1)^F}$.
This coincides with the stellar bimodule structure discussed in Example \ref{ex:stellar A_c}.
However, because of the twisted commutation relation between $R$ and $\nu$, $A^*$ does not become a stellar bimodule through the naturality of the Serre automorphism.
Instead, it gives a filling of the slightly different diagram
\[
\begin{tikzcd}
A \arrow[r,"A^*"] & A
\\
\overline{A}^{\op} \arrow[u,"M"] \arrow[r,"(\overline{A}^{\op})^*",swap] & \overline{A}^{\op} \arrow[u,"M",swap]
\end{tikzcd}
\]
The square is filled by Serre naturality and the other two compositions by the isomorphism $(\overline{A}^{\op})^* \cong \overline{A^*}^{\op}$ expressing the nontrivial commutation relation between $R$ and $SO_2$.
In particular, we can not express the condition of compatibility of the Stellar algebra with an ungraded-symmetric Frobenius structure $\lambda$ as unitarity of $\lambda: A_{(-1)^F} \to A^*$ .
\end{remark}

\begin{example}
We provide some intuition for the condition saying the stellar structure and the Frobenius structure are compatible by looking at the case where the stellar structure comes from a $*$-algebra.
We will now show that in that case the condition is equivalent to
\[
\lambda(a^*) = (-1)^{|a|} \overline{\lambda(a)} = \overline{\lambda(a)}
\]
where the last equation follows from the fact that $\lambda$ vanishes for odd $a$.
For the stellar algebra coming from a $*$-algebra we have $M = A_*$.
Stellar modules in that case correspond to $A$-valued inner products called Hilbert modules as we saw in Section \ref{Sec:stellar}.
The stellar module $A_{(-1)^F}$ is described in this fashion in Example \ref{ex:stellar A_c}.
Describing $A^*$ explicitly as a Hilbert module is mildly inconvenient as it involves picking an inverse for $A^*$, so we will refrain from doing that here.
However, the naturality of the Serre isomorphism in this case is explained in the Appendix in example \ref{ex:*-alg Serre nat}
\end{example}

\subsection{Computing 2d TFTs with spin-statistics and reflection structure}

In this section we will compute the bicategory of two-dimensional extended TFTs with fermionic symmetry $G$, reflection structure and spin-statistics. 
Without loss of generality we will from now on assume $G$ is a fermionic $2$-group.
Similar to the computation in Section~\ref{Sec:1D} we will start with the case $c\neq 1$ and 
comment on the simpler case later on in Remark~\ref{Rem: c=1 in 2 D}. 
Recall that by the cobordism hypothesis, this means we have to compute fixed points for the action of
\[
H_2 \rtimes (\Z_2^R \times B \Z_2^F)
\]
on the bicategory of finite-dimensional semisimple superalgebras over $\C$.
Here the $\Z_2^R \times B \Z_2^F$-action is as in the last section and $H_2$ acts by the cobordism hypothesis through the map to $O_2$.
With the decomposition theorem for iterated fixed points Proposition~\ref{Prop: Fixed points} in mind and our understanding of stellar Frobenius algebras, we turn to a study of the exact sequence
\[
1 \to \Spin_2 \rtimes (\Z_2^R \times B\Z_2^F) \to H_2 \rtimes (\Z_2^R \times B\Z_2^F) \to G_b \to 1.
\]
It turns out this exact sequence is split, as we will soon show.

However, we first recall from the last section that there is an isomorphism of $2$-groups
\[
\Spin_2 \rtimes (\Z_2^R \times B\Z_2^F) \cong (\Z_2^c \rtimes B\Z) \rtimes (\Z_2^R \times B\Z_2^F) \to B\Z \rtimes \Z_2^R \cong O_2
\]
given by killing $c$ and $(-1)^F$ and sending $\eta$ to a generating loop.
We reiterate the fact that the induced $O_2$-action on $\sAlg^{\text{fd}}$ is \emph{not} given by the cobordism hypothesis.
Instead, the nontrivial object acts by $\overline{(.)}^{\op}$ while the loop acts by $A \mapsto A^* \otimes_A A_{(-1)^F}$, see Remark \ref{Rem:O2iso} and the proof above it.
This also agrees with Theorem \ref{Th:stfrob}, which tells us that fixed points for this $O_2$-action are stellar Frobenius algebras.

Next we specify an inverse 
\[
\sigma: O_2 \to \Spin_2 \rtimes (\Z_2^R \times B\Z_2^F)
\]
by sending $R$ to itself and the preferred generating loop in $B\Z$ to 
\[
1 \xrightarrow{\eta} c \xrightarrow{(\id_c \rtimes (-1)^F)} 1.
\]
Here we recall that $(-1)^F: 1 \to 1$ has its codomain changed in the semidirect product to $1 \rtimes (-1)^F: 1 \to c$ due to the fact that $\rho((-1)^F)_* = c$ in the action $\rho$ of $B\Z_2^F$ on $\Z_2^c$, see Definition \ref{Def:semidirectproduct}.

To show the splitting of the exact sequence, we want to construct a commutative diagram
\[
\begin{tikzcd}
1 \arrow[r] & \Spin_2 \rtimes (\Z_2^R \times B\Z_2^F) \arrow[r] & H_2 \rtimes (\Z_2^R \times B\Z_2^F) \arrow[r] & G_b \arrow[r] & 1 
\\
1 \arrow[r] & O_2 \arrow[u, "\sigma"] \arrow[r] & O_2 \times G_b \arrow[u] \arrow[r] & G_b \arrow[r] \arrow[u, equals] & 1
\end{tikzcd}
\]
in which all arrows going up are isomorphisms.
Hence we have to extend the map $\sigma$ we specified above to the domain $O_2 \times G_b$.
It will then follow that the resulting $\sigma$ is still an isomorphism of $2$-groups by the $5$-lemma.  Indeed, the two horizontal exact sequences above induce fibrations on classifying spaces and $\sigma$ defines a map between these fibrations. The long exact sequence of homotopy groups of a fibration is functorial. By the $5$-lemma when then obtain that $\sigma$ is a weak homotopy equivalence.

Let $s$ denote the canonical section of the exact sequence
\[
1 \to \Spin_2 \to H_2 \to G_b \to 1
\]
given by the fact that we use the fermionically skeletal model
\[
H_2 = \Spin_2 \rtimes G_b.
\]
Recall from the end of Section \ref{Sec:fermskeletal} how this semidirect product and section behave:
firstly we have $s(g) = 1 \rtimes g$ and the failure of section to preserve the multiplication is measured by
\[
\Xi^{\op}(g_1,g_2) := \Xi(g_1,g_2) + \theta(g_1) \theta(g_2)
\]
in the sense that
\[
s(g_1) \otimes s(g_2) = \Xi^{\op}(g_1, g_2) s(g_1 g_2).
\]
We warn the reader that in general $\Xi$ is not a $2$-cocycle on $\pi_0(G_b)$.
Since $G_b$ is skeletal, we are only interested in lifting automorphisms $\gamma: g \to g$ and they lift to morphisms in which the codomain is possibly changed by $c$:
\begin{align*}
s(\gamma) = 1 \rtimes \gamma: g \to g \otimes \Gamma(\gamma).
\end{align*}
Note that these morphisms compose well in the sense that
\[
s(\gamma_2 \circ \gamma_1) = (s(\gamma_2: g \to g) \otimes \id_{\Gamma(\gamma_1)}) \circ s(\gamma_1: g \to g)
\]
Moreover, in Section \ref{Sec:fermskeletal} it was shown that 
\[
(\eta \rtimes 1) \otimes \id_{s(g)} = 
\begin{cases}
\id_{s(g)} \otimes (\eta \rtimes 1) & \theta(g) = 0
\\
\id_{s(g)} \otimes (c \eta^{-1} \rtimes 1) & \theta(g) = 1
\end{cases}
\]

We proceed to extend $\sigma$ to $G_b$.
The main trick is to perturb the section $s$ into a $2$-group morphism by moving from $1$ to $c$ using the morphism $1 \rtimes (-1)^F$.
Define 
\begin{align*}
\sigma(g) = s(g) \rtimes R^{\theta(g)}
\end{align*}
on objects and for morphisms $\gamma: g \to g$ we define the morphism $\sigma(g) \to \sigma(g)$ by
\begin{align*}
\sigma(\gamma) = 
\begin{cases}
s(\gamma) \rtimes \id_{R^{\theta(g)}} & \Gamma(\gamma) = 1
\\
(\id_{c} \rtimes (\id_{R^{\theta(g)}}\otimes (-1)^F )) (s(\gamma) \rtimes \id_{R^{\theta(g)}}) & \Gamma(\gamma) = c.
\end{cases}
\end{align*}
Note that even in the latter case this is an automorphism as $s(\gamma) \rtimes \id_{R^{\theta(g)}}$ is a morphism from $R^{\theta(g)} $ to $R^{\theta(g)} c$.
Recall that the $\gamma$'s compose well under $s$ and $\Gamma$ is a group homomorphism $\Aut(g) \to \Z_2^c$.
Also using that $B\Z_2^F$ acts trivially on the $G_b$-part, we have by definition of horizontal composition in the semidirect product
\[
(1 \rtimes (-1)^F) (\id_k \rtimes 1) =  (\id_{k c^{\theta(k)}} \rtimes 1) (1 \rtimes (-1)^F).
\]
Therefore $s$ defines a functor:
Using that $R$ acts on objects of $H$ by multiplication with $c$ through its grading homomorphism $\theta$ on the $G_b$-part, we compute 
\begin{align*}
\sigma(g_1) \sigma(g_2) &= (s(g_1) \rtimes R^{\theta(g_1)}) (s(g_2) \rtimes R^{\theta(g_2)})
 = s(g_1) s(g_2) c^{\theta(g_2) \theta(g_1)} \rtimes R^{\theta(g_1) + \theta(g_2)} 
 \\
 &= s(g_1 g_2) \Xi(g_1, g_2) \rtimes R^{\theta(g_1 g_2)}
 = \Xi(g_1,g_2) \sigma(g_1 g_2).
\end{align*}
So for the monoidal data 
\[
s(g_1 g_2) \cong s(g_1) s(g_2)
\]
we take the isomorphism $1 \rtimes ((-1)^F)^{\Xi(g_1,g_2)}$.
Similarly, there can be an off by $c$ error in the commutation relation between $R$ and $G_b$.
Indeed, define $\sigma(g R) := \sigma(g) \sigma(R)$, then multiplying in the other direction gives
\begin{align*}
\sigma(Rg) &= \sigma(gR) = \sigma(g) \sigma(R) = s(g) \rtimes R^{\theta(g) + 1}
\\
\sigma(R) \sigma(g) &= c^{\theta(g)} s(g) \rtimes R^{\theta(g) + 1}
\end{align*}
So we define the monoidal data $\sigma(Rg) \cong \sigma(R) \sigma(g)$ to be $1 \rtimes (-1)^F$ if $\theta(g) = 1$.
The monoidal data in which at least one of the elements is of the form $Rg$ is similarly defined.
We now have to show the monoidal data satisfies the desired associativity condition.
For three elements of $G_b$ this follows by the relationship between $\Xi$, $\Gamma$ and the associator of $G_b$, see Section \ref{Sec:fermskeletal}.
For cases with $R$ we additionally have to commute $1 \rtimes (-1)^F$ through some elements of $G_b$.
Note that $1 \rtimes R$ and $\sigma(g) \rtimes 1$ for $\theta(g) = 1$ have the same commutation relation with $\sigma(\gamma) = s(\gamma) \rtimes 1$ so that
\[
\sigma(\gamma) \otimes \sigma(\id_g) = \sigma(\id_g) \otimes \sigma(\gamma)
\]
The naturality of the monoidality isomorphism follows by by the commutation properties between $(-1)^F$ and $c$ and $\eta$.

Let $\psi$ denote the full $H_2 \rtimes (\Z_2^R \times B\Z_2^F)$-action on $\sAlg^{\text{fd}}$.
We make use of the splitting $\sigma$: we first take $O_2$-fixed points and the consequent $G_b$-fixed points are now relatively simple to compute as a special case of the decomposition theorem (Proposition~\ref{Prop: Fixed points}) in which the semidirect product is a direct product.
The only thing to be careful about is to record the amount of $(-1)^F$s we have introduced at several points above to make $\sigma$ into a $2$-group homomorphism.
The equivalent $O_2 \times G_b$-action is given as follows
\begin{enumerate}
\item $\psi(\sigma(R)) = \overline{(\cdot)}^{\op}$;
\item $\psi(\sigma(g)) = \psi(s(g) \rtimes R^{\theta(g)}) = \psi(s(g) \rtimes 1) \psi(1 \rtimes R^{\theta(g)})$ is the identity if $\theta(g) = 0$ and $\overline{(\cdot)}^{\op \op} = \overline{(\cdot)}$ if $\theta(g) = 1$;
\item we have $\psi(\sigma(g)) \psi( \sigma(R)) = \psi(\sigma(g R))$ but $\psi(\sigma(R)) \psi( \sigma(g)) \cong \psi(\sigma(R g))$ has an extra $(-1)^F$ if $\theta(g) = 1$.
\item $\psi(\sigma(\eta: 1 \to 1))$ is given by the horizontal composition of the Serre and $(\cdot)_{(-1)^F}$, e.g. $\psi(\sigma(\eta))[A] = A_{(-1)^F} \otimes_A A^*$;
\item $\psi(\sigma(\gamma: g \to g))$ is given by the horizontal composition of $\id_{\psi(s(g))}$ and possibly $(\cdot)_{(-1)^F}$ if $\Gamma(\gamma) = c$, e.g.
\[
\psi(\sigma(\gamma))[A] = 
\begin{cases}
A & \theta(g) = 0, \Gamma(\gamma) = 1
\\
\overline{A} & \theta(g) = 1, \Gamma(\gamma) = 1
\\
A_{(-1)^F} & \theta(g) = 0, \Gamma(\gamma) = c
\\
\overline{A}_{(-1)^F} & \theta(g) = 1, \Gamma(\gamma) = c
\end{cases}
\]
where it is understood that we consider for example $\overline{A}$ as an $(\overline{A}, \overline{A})$-bimodule;
\item looking at the monoidality data of $\sigma$, the natural transformation $\psi(\sigma(g_1)) \psi( \sigma(g_2)) \cong \psi(\sigma(g_1 g_2))$ is the identity when $\Xi(g_1,g_2) = 1$. 
If $\Xi(g_1,g_2) = c$ it is given by $(\cdot)_{(-1)^F}: \id_{\sAlg^{\text{fd}}} \Longrightarrow \id_{\sAlg^{\text{fd}}}$ if $\theta(g_1 g_2) = 0$ and whiskered with $\overline{(\cdot)}$ otherwise;
\item corresponding to the relation $R \eta = \eta^{-1} R$ we have the modification 
\[
\overline{A_{(-1)^F} \otimes_A A^*}^{\op -1} \cong \overline{A_{(-1)^F}}^{\op -1} \otimes_{\overline{A}^{\op -1}} \overline{A^*}^{\op}  \cong ((\overline{A}^{\op})_{(-1)^F})^{-1} \otimes_{\overline{A}^{\op}} ((\overline{A}^{\op})^*)^{-1}
\]
which already played an important role in the appearance of stellar algebras in the last section.
\end{enumerate}
All other data of the action are modifications obtained from combining three objects or two $1$-morphisms, which are either trivial or canonically defined.

We turn to the action $\xi$ of $G_b$ on $O_2$-fixed points, which are stellar Frobenius algebras by Theorem \ref{Th:stfrob}.
We start to compute $\xi(g)$ on objects for $g \in G_b$ an object.
Let $(A,M,\sigma, \lambda)$ be a stellar Frobenius algebra.
We describe the stellar Frobenius algebra $\xi(g)[A,M,\sigma, \lambda]$.
If $\theta(g) = 0$, this is the original object, so assume $\theta(g) = 1$.
The underlying algebra is $\overline{A}$.
The stellar module is computed as the fixed point data for $R$, where we have to use the interesting commutation relation between $R$ and $g$:
\[
\overline{\overline{A}}^{\op} =  \psi(\sigma(R)) \psi(\sigma(g)) A \cong \psi(\sigma(Rg)) A = \psi(\sigma(gR)) A \cong \psi(\sigma(g)) \psi(\sigma(R)) A \xrightarrow{\psi(\sigma(g))[M]} \psi(\sigma(g)) A = \overline{A}.
\]
So we see that the stellar structure on $\overline{A}$ is $\overline{M} \otimes_{\overline{\overline{A}}^{\op}} \overline{\overline{A}}^{\op}_{(-1)^F}$.
This is the stellar module on $\overline{A}$ we introduced at the end of Section \ref{Sec:stellar}.
We will now show that also $\sigma$ is as advertised there.

Let $c_{n,h}: \rho(n) \rho(h) \Longrightarrow \rho(n) \rho(h)$ denote the composition of natural isomorphisms $R_{n,h}^{-1}$ and $R_{h,n}$ using that $hn = nh$.
For defining $\xi(h)[\underline{F}]_{n_1,n_2}$ we fill the following
\[
\begin{tikzcd}[column sep = 2.1cm]
\rho(n_1 n_2) \rho(h)[ F] \arrow[dd,"{R_{n_1,n_2}[\rho(h)[F]]}",swap]  \arrow[r,"{c_{n_1 n_2, h}[F]}"] & \rho(h) \rho(n_1 n_2) [F] \arrow[r,"{\rho(h)[F_{n_1 n_2}]}"] \arrow[d,"{\rho(h)[R_{n_1,n_2}[F]]}",swap] & \rho(h)[F]
\\
\ & \rho(h) \rho(n_1) \rho(n_2) [F] \arrow[r,"{\rho(h) \rho(n_1)[F_{n_2}]}"] & \rho(h) \rho(n_1) [F] \arrow[u,"{\rho(g)[F_{n_1}]}",swap]
\\
\rho(n_1) \rho(n_2) \rho(h) [F] \arrow[r,"{\rho(n_1)[c_{n_2,h}[F]]}",swap] & \rho(n_1) \rho(h) \rho(n_2)[F] \arrow[r,"{\rho(n_1) \rho(h)[F_{n_2}]}",swap] \arrow[u,"{c_{n_1,h}[\rho(n_2)[F]]}"]& \rho(n_1) \rho(h)[F] \arrow[u,"{c_{n_1,h}[F]}",swap]
\end{tikzcd}
\]
The upper right square is filled by $\rho(h)[F_{n_1,n_2}]$.
The lower right square is naturality data of $c_{n_1,h}$.
The left rectangle can be expanded as follows
\[
\label{diagramexchange}
\begin{tikzcd}
\rho(n_1 n_2) \rho(h)\arrow[Rightarrow,ddd] & \arrow[Rightarrow, l]\rho(h n_1 n_2) \arrow[equals,d] \arrow[Rightarrow,r] & \rho(h) \rho(n_1 n_2) \arrow[Rightarrow,d]
\\
\ & \arrow[Rightarrow,dd] \rho(n_1 h n_2) \arrow[Rightarrow,dr] & \rho(h) \rho(n_1) \rho(n_2)
\\
\ & \ & \rho(h n_1) \rho(n_2) \arrow[Rightarrow,u] \arrow[Rightarrow,d]
\\
\rho(n_1) \rho(n_2) \rho(h) & \rho(n_1) \rho(h n_2) \arrow[Rightarrow,r] \arrow[Rightarrow,l] & \rho(n_1) \rho(h) \rho(n_2)
\end{tikzcd}
\]
The left rectangle can be filled by $\omega_{n_1,n_2,h}$, the right upper part by $\omega_{h,n_1,n_2}$ and the lower right part by $\omega_{n_1,h,n_2}$.

To see how $\sigma$ changes under the action of $g$, we have to take $n_1 = n_2 = R$ and $h = g$ in the above diagram.
 We get using naturality of $A_{(-1)^F}$
\[
\overline{M} \otimes_{\overline{\overline{A}}^{\op}} \overline{\overline{A}}^{\op}_{(-1)^F} \xrightarrow{\overline{\sigma} \otimes \id}  \overline{\overline{M}}^{\op} \otimes_{\overline{\overline{A}}^{\op}} \overline{\overline{A}}^{\op}_{(-1)^F} 
\cong \overline{(\overline{A}_{(-1)^F} \otimes_{\overline{A}} \overline{M})}^{\op} 
\cong \overline{\overline{M} \otimes_{\overline{\overline{A}}^{\op}} \overline{\overline{A}}^{\op}_{(-1)^F}  }^{\op}
\]
which is indeed the stellar structure we defined in Section \ref{Sec:stellar}.

To compute the Frobenius structure on $\overline{A}$, we have to commute the time-reversing $g$ with $\eta$.
This corresponds to the isomorphism $\overline{A^* \otimes_A A_{(-1)^F}} \cong \overline{A}^* \otimes_{\overline{A}} \overline{A}_{(-1)^F}$ so that the Frobenius structure is simply the composition
\[
\overline{A} \xrightarrow{\overline{\lambda}} \overline{A^* \otimes_A A_{(-1)^F}} \cong \overline{A}^* \otimes_{\overline{A}} \overline{A}_{(-1)^F}
\]
which after plugging in the relevant isomorphisms is simply
\[
\overline{a} \mapsto \overline{\lambda(a)}.
\]
We have now determined $\xi(g)$ on objects. 

Now let $(N,h): (A_1, M_1, \sigma_1, \lambda_1) \to (A_2, M_2, \sigma_2, \lambda_2)$ be a $1$-morphism of $O_2$-fixed points, i.e. a stellar bimodule intertwining the ungraded Frobenius structures as in Theorem \ref{Th:stfrob}.
We want to equip the $(\overline{A_2}, \overline{A_1})$-bimodule $\overline{N}$ with a unitarity datum for the stellar structures looking like $(\overline{A_1})_{(-1)^F} \otimes_{\overline{A_1}} \overline{M_1}$ as explained above.
The general diagram of $1$-morphisms between fixed points of a direct product of $2$-groups is 
\[
\begin{tikzcd}[column sep=1.8cm]
\rho(n) \rho(h) F \arrow[r,"\rho(n) \rho(h) f"] \arrow[d,"{c_{n,h}[F]}",swap] &  \rho(n) \rho(h) F' \arrow[d,"{c_{n,h}[F']}"]
\\
\rho(h) \rho(n) F \arrow[d,"{\rho(h)[F_n]}",swap] \arrow[r,"{\rho(h) \rho(n) f}"] & \rho(h) \rho(n) F' \arrow[d,"{\rho(h)[F_n']}"]
\\
\rho(h) F \arrow[r,"\rho(h) f"] & \rho(h) F'
\end{tikzcd}
\]
The upper square is filled by naturality of $c_{n,h}$ and the lower is filled by $\rho(h)[f_n]$.
Specializing to the case at hand and using $\overline{A}^{\op} = \overline{A^{\op}}$ recovers the diagram defining $\overline{N}$ as discussed at the end of Section \ref{Sec:stellar}:
\[
\begin{tikzcd}
\overline{\overline{A_1}}^{\op} \arrow[d,"{\overline{M_1}}",swap]  &  \overline{\overline{A_2}}^{\op} \arrow[l,"{\overline{\overline{N}}^{\op}}",swap] \arrow[d,"{\overline{M_2}}"]
\\
\overline{A_1} \arrow[d,"{\overline{A_1}_{(-1)^F}}",swap] \arrow[r,"\overline{N}"] & \overline{A_2} \arrow[d,"{\overline{A_2}_{(-1)^F}}"]
\\
\overline{A_1} \arrow[r,"\overline{N}",swap] & \overline{A_2}
\end{tikzcd}
\]
Since the functor $\xi(g)$ simply becomes $\overline{(.)}: \sAlg^{\text{fd}} \to \sAlg^{\text{fd}}$ when forgetting the stellar Frobenius structure and stellar Frobenius algebras only have extra conditions (not data) on $2$-morphisms, the rest of the data of $\xi(g)$ is now determined.

We proceed to compute the natural isomorphisms $\xi(g_1 g_2) \cong \xi(g_1) \xi(g_2)$.
Recall that depending on $\Xi(g_1,g_2)$ and $\theta(g_1 g_2)$, they are given by an extra $(-1)^F$ and/or $\overline{(.)}$.
The only thing that remains is computing the stellar bimodule structures on these.
The general diagram looks like
\[
\begin{tikzcd}[column sep=2.5cm]
\rho(n) \rho(h_1) \rho(h_2) F \arrow[d,"{c_{n,h_1}[\rho(h_2)[F]]}", swap] & \rho(n) \rho(h_1 h_2) F \arrow[l,"{\rho(n)[R_{h_1,h_2}[F]]}",swap] \arrow[dd,"{c_{n,h_1 h_2}[F]}"]
\\
\rho(h_1) \rho(n) \rho(h_2) F \arrow[d,"{\rho(h_1)[c_{n,h_2}[F]]}", swap] & \
\\
\rho(h_1) \rho(h_2) \rho(n) F \arrow[d,"{\rho(h_1) \rho( h_2)[F_n]}",swap] & \rho(h_1 h_2) \rho(n) F \arrow[d,"{\rho(h_1 h_2)[F_n]}"] \arrow[l,"{R_{h_1,h_2}[\rho(n)[F]]}"]
\\
\rho(h_1) \rho(h_2) F  & \rho(h_1 h_2) F \arrow[l,"{R_{h_1,h_2}[F]}"]
\end{tikzcd}
\]
The lower part is filled by naturality of $R_{h_1,h_2}$.
The upper part is filled by various $\omega$ in a similar way to Diagram \ref{diagramexchange}.
Plugging in $n = R$ gives the necessary stellar structure on the canonical $(\overline{\overline{A}^{\theta(g_1)}}^{\theta(g_2)}, \overline{A}^{\theta(g_1 g_2)})$-bimodule.
The vertical arrows can contain extra $(-1)^F$ depending on $\theta(g_1)$ and $\theta(g_2)$ and the horizontal arrows can contain extra $(-1)^F$ depending on $\Xi(g_1,g_2)$.
So one depending on this one might have to use $(-1)^F$-naturality data, but otherwise the unitarity data is obvious.
We summarise the computations so far in the following theorem.
\begin{theorem}
\label{th:actionstfrob}
The $G_b$-action on $\catf{stFrob}$ is given as follows.
\begin{enumerate}
\item for $g \in G_b$ the functor $\psi(g)$ is the identity if $\theta(g) = 0$.
Otherwise it maps a stellar Frobenius algebra $(A,M,\sigma, \lambda)$ to the stellar algebra $(\overline{A}, \overline{A}_{(-1)^F} \otimes_{\overline{A}} \overline{M}  , \overline{\sigma}, \overline{\lambda})$ and a stellar module $(N,h): (A_1, M_1, \sigma_1, \lambda_1) \to (A_2, M_2, \sigma_2, \lambda_2)$ to $(\overline{N},\overline{h})$.
The other data of $\psi(g)$ is identical to the functor $\overline{(\cdot)}$ on $\sAlg^{\text{fd}}$.
\item for $g_1,g_2 \in G_b$, the natural transformation $\rho(g_1 g_2) \Longrightarrow \rho(g_1) \circ \rho(g_2)$ is given by $(-1)^F$ if $\Xi(g_1,g_2)$ is nontrivial (after horizontally composing with $\overline{(\cdot)}$ appropriately if $\theta(g) = 1$). 
The stellar bimodule structure on this natural transformation is canonically given by the naturality data of $(-1)^F$ and $\overline{\overline{(\cdot)}} \cong (\cdot)$.
\item For a morphism $\gamma: g \to g$ in $G_b$ the natural transformation $\rho(\gamma): \rho(g) \Longrightarrow \rho(g)$ is given by $(-1)^F$ if $\Gamma(\gamma) =c$ and the identity otherwise (after horizontally composing with $\overline{(\cdot)}$ appropriately if $\theta(g) = 1$). 
\item All other data of the action ($R_{\gamma,\gamma'}, \alpha_{\gamma,\gamma'}, \omega_{g_1,g_2,g_3}$ etc.) is canonically induced by the naturality of the $\Z_2^B \times B\Z_2^F$-action on $\sAlg^{\operatorname{fd}}$.
\end{enumerate}
\end{theorem}

To state our main theorem we need to combine the definition of algebras strongly graded over a fermionic $2$-group (Definition \ref{Def: 2-group graded alg}) with the definition of ungraded symmetric stellar Frobenius algebras:

\begin{definition}
Let $\mathcal{G}$ be a fermionic $2$-group.
A \emph{strongly $\mathcal{G}$-graded stellar Frobenius algebra} is a strongly $\mathcal{G}$-graded algebra 
\[
\mathcal{A} = \bigoplus_{g \in\operatorname{Obj} \mathcal{G}}  A_g
\]
such that 
\begin{enumerate}
\item $A:=A_1$ is equipped with the structure of a stellar Frobenius algebra $(M,\sigma, \lambda)$, in particular $A$ is finite-dimensional and semisimple.
\item For every object $g$ of $\mathcal{A}$ the $(A, \overline{A}^{\theta(g)})$-bimodules $A_g$ are $1$-morphisms in stellar Frobenius algebras. 
In other words, they come equipped with the structure of stellar bimodules, where the stellar structure on $\overline{A}$ is induced by $(M,\sigma)$ as in Definition \ref{def:barstellar} and the bimodule is compatible with the Frobenius structures on $A$ and $\overline{A}$.
\item For every morphism $\gamma:g \to g'$ the maps $F_\gamma: A_g \to A_{g'}$ are unitary and for every pair of objects $g_1,g_2$ the multiplication maps $A_{g_1} \otimes_{\overline{A}^{\theta(g)}} \overline{A_{g_2}}^{\theta(g_1)} \to A_{g_1 g_2}$ are unitary, where we equip the complex conjugation of a bimodule with the stellar structure explained at the end of Section \ref{Sec:stellar}.
\end{enumerate}
\end{definition}

Before continuing to prove the main theorem, we first make some of these conditions more explicit. 

\begin{proposition}
\label{Prop: ungrsym}
The condition for strongly $\mathcal{G}$-graded stellar Frobenius algebras that the $(A, \overline{A}^{\theta(g)})$-bimodules $A_g$ preserve the Frobenius structures $\lambda$ on $A$ and the one it induces on $\overline{A}^{\theta(g)}$ is equivalent to
\[
\lambda(a_g a_{g^{-1}}) = \overline{\lambda(a_{g^{-1}} a_g)}^{\theta(g)}
\]
for all $a_g \in A_g$ and $a_{g^{-1}} \in A_{g^{-1}}$.
\end{proposition}
\begin{proof}
Recall from Proposition \ref{prop:spinstatisticsTFTs} that $A_g$ preserving the Frobenius structure is equivalent to the requirement that the Serre naturality isomorphism $S_g: A^* \otimes_A A_{g} \cong A_g \otimes_{\overline{A}^{\theta(g)}} (\overline{A}^{\theta(g)})^*$ satisfies
\[
S_g(\lambda \otimes_A a_g) = (-1)^{|a_g|} a_g \otimes_A \overline{\lambda}^{\theta(g)}
\]
where $\overline{\lambda}$ is the ungraded symmetric Frobenius structure on $\overline{A}$ given in Theorem \ref{th:actionstfrob}.
In Lemma \ref{Lem:ungrsym} we show this condition is equivalent to requiring the desired $\overline{\lambda(a_g a_{g^{-1}})}^{\theta(g)} = \lambda(a_{g^{-1}} a_g)$ for all $a_g \in A_g$ and $a_{g^{-1}} \in A_{g^{-1}}$.
\end{proof}

\begin{example}
To obtain more intuition, we consider the case where $A$ is a super $*$-algebra.
Then since $A_g$ is a stellar $(A, \overline{A}^{\theta(g)})$-bimodule, it comes equipped with an $A$-valued inner product satisfying Hilbert module type identities, where the only sign subtlety comes up in the $*$-structure $\overline{a}^* = (-1)^{|a|} \overline{a^*}$ on $\overline{A}$.
For example, this results in the case $\theta(g) = 1$ in an inner product $\langle \cdot ,\cdot \rangle$ on $A_g$ satisfying
\[
\langle a_g  \overline{a}, a_g' \rangle = (-1)^{|a_g'| |a| + |a|}  \langle a_g,  a_g'  a^* \rangle
\]
because in the expression on the right we had to compute $\overline{a}^* = (-1)^{|a|} \overline{a^*}$.
Using the definition of the stellar structure on the composition,
the multiplication maps $A_g \otimes_{\overline{A}^{\theta(g)}} \overline{A_h}^{\theta(g)} \to A_{gh}$ being unitary gives the requirement
\begin{align*}
\langle a_g a_{h}, b_g b_{h} \rangle_{A_{gh}} = (-1)^{|b_{g}| |b_h|} \langle a_g \langle  a_h, b_h \rangle_{A_h}, b_g \rangle_{A_g}
\end{align*}
for $\theta(g)$ trivial.
For $\theta(g)$ nontrivial the same computation results in
\[
\langle a_g a_{h}, b_g b_{h} \rangle_{A_{gh}} = (-1)^{|b_{g}| |b_h|} \langle a_g \langle \overline{a}_h, \overline{b}_h \rangle_{\overline{A_h}}, b_g \rangle_{A_g} = (-1)^{|b_{g}| |b_h| + |b_h|} \langle a_g \overline{\langle a_h, b_h \rangle_{A_h}}, b_g \rangle_{A_g}
\]
where the sign comes from the complex conjugation of the stellar module $A_h$ as discussed in the beginning of this section.
For example when $h = g^{-1}$ this gives
\begin{align*}
 a_g a_{g^{-1}} (b_g b_{g^{-1}})^* = (-1)^{|b_{g}| |b_{g^{-1}}| + \theta(g) |b_{g^{-1}}|} \langle a_g, \langle a_{g^{-1}}, b_{g^{-1}} \rangle_{A_{g^{-1}}} b_g \rangle_{A_g}.
\end{align*}
\end{example}

We now prove our main theorem.

\begin{theorem}
\label{th:main}
Let $(G,c,\theta)$ be a fermionic group and $\mathcal{G}$ the fundamental $2$-group of $G$.
Then two-dimensional TFTs with fermionic symmetry $G$, spin-statistics and reflection structure are classified by strongly $\mathcal{G}$-graded stellar Frobenius algebras.
\end{theorem}
\begin{proof}
By the cobordism hypothesis we have to compute
\begin{align*}
\Fun(\Bord^{H_2}, \sAlg^{\text{fd}})^{\Z_2^R \times B\Z_2^F} &= ((\sAlg^{\text{fd}})^{H_2})^{\Z_2^R \times B\Z_2^F} = (\sAlg^{\text{fd}})^{H_2 \rtimes (\Z_2^R \times B\Z_2^F)} \cong (\sAlg^{\text{fd}})^{O_2 \times G_b} 
\\
&\cong ((\sAlg^{\text{fd}})^{O_2})^{G_b} \cong \stFrob^{G_b}.
\end{align*}
Let $\mathcal{G}_b$ denote the fundamental $2$-group of $G_b$.
Recall that the map of $2$-groups $\mathcal{G}_b \to B\Z_2^c$ classifying the extension to $\mathcal{G}$ is given by a pair $(\Gamma, \Xi)$ as explained in Section \ref{Sec:fermskeletal}.
Following the isomorphisms around, it suffices to compute $\mathcal{G}_b$-fixed points for the action on stellar algebras given in Theorem \ref{th:actionstfrob}.
We proceed to explicitly give all data using Definition \ref{Def: Hfixed point}.
The data consists of
\begin{enumerate}
\item a Frobenius stellar algebra $(A,M,\sigma, \lambda) \in \ob \stFrob$;
\item for every $g^b \in \operatorname{ob}\mathcal{G}_b$ a stellar $(A,\overline{A}^{\theta(g^b)})$-bimodule $A_{g^b}$ preserving the Frobenius structure where we take the $(-1)^F$-twisted stellar structure on $\overline{A}$ as explained at the end of Section \ref{Sec:stellar};
\item a unitary bimodule isomorphism 
\[
\phi_{g^b_1, g^b_2}: A_{g^b_1} \otimes_{\overline{A}^{\theta(g^b_1)}} \overline{A_{g^b_2}}^{\theta(g^b_1)} \otimes_{\overline{A}^{\theta(g_1^b g_2^b)}} R_{g_1^b, g_2^b}[A] \to A_{g^b_1 g^b_2}
\]
for every $g^b_1, g^b_2 \in \operatorname{ob}\mathcal{G}_b$ where 
\[
R_{g_1^b, g_2^b}[A] :=
\begin{cases}
\overline{A}^{\theta(g_1^b g_2^b)} & \Xi(g^b_1, g^b_2) = 1
\\
(\overline{A}^{\theta(g_1^b g_2^b)})_{(-1)^F} & \Xi(g^b_1, g_2^b) = c.
\end{cases}
\]
\item for every path $\gamma^b: g^b_1 \to g^b_2$ in $\mathcal{G}_b$ a unitary bimodule isomorphism 
\[
F_{\gamma^b}: A_{g_1^b} \to A_{g_2^b}
\]
when $\Gamma(\gamma^b) = 1$ and
\[
F_{\gamma^b}: A_{g_1^b} \to A_{g_2^b} \otimes_{\overline{A}^{\theta(g_2^b)}}  (\overline{A}^{\theta(g_2^b)})_{(-1)^F}
\]
when $\Gamma(\gamma^b) = c$.
\end{enumerate}

Before moving to the conditions these data have to satisfy, we introduce some notation to connect with the statement of the theorem.
We will work with the model $\mathcal{G} = \Z_2^c \rtimes \mathcal{G}_b$ of the fundamental groupoid of $G$, where the semidirect product is formed using $(\Gamma, \Xi)$.
This gives us a preferred splitting of the exact sequence of $2$-groups
\[
1 \to B\Z_2^c \to \mathcal{G} \to \mathcal{G}_b \to 1.
\]
Concretely this means objects of $\mathcal{G}$ are of the form $g = c^\epsilon \rtimes g^b$, where $\epsilon = 0$ or $\epsilon = 1$ and $g^b$ is an object of $\mathcal{G}_b$.
Given such an object, we define the stellar bimodule
\[
A_g := 
\begin{cases}
A_{g_b}  & \epsilon = 1
\\
A_{g_b} \otimes_{\overline{A}^{\theta(g^b)}} (\overline{A}^{\theta(g^b)})_{(-1)^F} & \epsilon = -1.
\end{cases}
\]
We extend the definition of $\phi$ to $\phi_{g_1,g_2}$ for $g_1, g_2 \in \operatorname{ob} \mathcal{G}$ by inserting the unitary isomorphisms $A_{g_b} \otimes_{\overline{A}^{\theta(g_b)}} (\overline{A}^{\theta(g_b)})_{(-1)^F} \cong A_{(-1)^F} \otimes_A A_{g_b}$ and $A_{(-1)^F} \otimes_A A_{(-1)^F} \to A$ wherever relevant.
These isomorphisms are induced by naturality of $(-1)^F$ and the data of $(-1)^F$ squaring to the identity respectively.
They are unitary because of the relevant compatibility data with the $\Z_2^R$-action $A \mapsto \overline{A}^{\op}$.
Recall that the natural isomorphisms $\rho(g_1^b g_2^b) \cong \rho(g_1^b) \circ \rho(g_2^b)$ between complex conjugation (or identity) functors obtain an extra $(-1)^F$ when $\Xi(g_1^b, g_2^b)$ is nontrivial.
So by definition of the tensor product on $\operatorname{ob} \mathcal{G}$, $\phi_{g_1,g_2}$ becomes a bimodule map $A_{g_1} \otimes_{\overline{A}^{\theta(g_1)}} \overline{A_{g_2}}^{\theta(g_1)}  \to A_{g_1 g_2}$.
Define the $\operatorname{ob} \mathcal{G}$-graded complex super vector space
\[
\mathcal{A} := \bigoplus_{g \in \operatorname{ob} \mathcal{G}} A_g.
\]
We first work in the image of the forgetful functor $\sAlg^{\text{fd}}_\C \to \sAlg^{\text{fd}}_\R$ under which the complex conjugation functor becomes the identity.
There $\phi$ defines a $\R$-linear multiplication map $\mathcal{A} \otimes \mathcal{A} \to \mathcal{A}$ which is strongly $\operatorname{ob} \mathcal{G}_b$-graded.
We will denote this multiplication $\cdot$ in the rest of the proof, even though we will see it is only associative up to the associator of $\mathcal{G}$. 

Given $\gamma^b: g^b \to h^b$ a morphism in $\mathcal{G}_b$, $F_{\gamma^b}$ becomes a unitary bimodule map $A_{g_b} \to A_{\Gamma(\gamma^b) \rtimes h_b }$.
Note that all morphisms $\gamma: g \to h$ in $\mathcal{G}$ are of the form $\gamma^b \rtimes c^\epsilon: g^b \rtimes c^\epsilon \to h^b \rtimes \Gamma(\gamma^b) c^\epsilon$ for some morphism $\gamma^b: g^b \to h^b$ in $\mathcal{G}_b$ and $\epsilon = 0,1$.
Define $F_\gamma: A_g \to A_{h}$ by $F_{\gamma^b}$ in case $\epsilon = 0$ and as the composition
\begin{align*}
A_g = A_{g_b} \otimes_{\overline{A}^{\theta(g^b)}} (\overline{A}^{\theta(g^b)})_{c^\epsilon} \xrightarrow{F_{\gamma^b} \otimes \id}& A_{h_b} \otimes_{\overline{A}^{\theta(h^b)}} (\overline{A}^{\theta(h^b)})_{\Gamma(\gamma^b)} \otimes_{\overline{A}^{\theta(g^b)}} (\overline{A}^{\theta(g^b)})_{c^\epsilon} 
\\
\to& A_{h_b} \otimes_{\overline{A}^{\theta(h^b)}} (\overline{A}^{\theta(g^b)})_{c^{\epsilon} \Gamma(\gamma^b)} 
\end{align*}
otherwise.
We have now supplied all the data promised in the statement of the theorem.

We turn to the $\alpha$-twisted associativity condition the $\phi_{g^b_1,g^b_2}$ have to satisfy when combining three elements given by Diagram~\eqref{Diagram A}.
The associativity diagram for $\phi$ exactly says that the two $2$-morphisms (with slightly different codomains)
\[
A_{g_1} \otimes_{\overline{A}^{\theta(g_1)}} \overline{A_{g_2}}^{\theta(g_1)}  \otimes_{\overline{A}^{\theta(g_1 g_2)}} \overline{A_{g_3}}^{\theta(g_1 g_2)} \to A_{(g_1 g_2) g_3}, A_{g_1 (g_2 g_3)}
\]
are related by $F_{\alpha(g_1,g_2,g_3)}$ for elements of the form $g_i = 1 \rtimes g^b_i \in \operatorname{ob} \mathcal{G}$.
To prove $\mathcal{A}$ satisfies the relevant associativity axiom, we additionally have to show twisted associativity for arbitrary elements $g_i = c^\epsilon_i \rtimes g^b_i \in \operatorname{ob} \mathcal{G}$ with $\epsilon_i \in \{\pm\}$.
For $g_i = 1$ this follows by associativity of the $\Z_2$-graded superalgebra $A \oplus A_{(-1)^F}$.
In general we have to apply the compatibility of naturality data of $A_{(-1)^F}$ with the composition of $2$-morphisms and multiplication in $A \oplus A_{(-1)^F}$ in the sense that
\[
\begin{tikzcd}
M \otimes_A A_{(-1)^F} \arrow[d] \arrow[r] & B_{(-1)^F} \otimes_B M \arrow[d]
\\
N \otimes_A A_{(-1)^F} \arrow[r] & B_{(-1)^F} \otimes_B N
\end{tikzcd}
\]
and
\[
\begin{tikzcd}
M \otimes_A A_{(-1)^F} \otimes_A A_{(-1)^F} \arrow[d]  \arrow[r] & B_{(-1)^F} \otimes_B M \otimes_A A_{(-1)^F} \arrow[d]
\\
M & B_{(-1)^F} \otimes_B B_{(-1)^F} \otimes_B M \arrow[l]
\end{tikzcd}
\]
commute for $(B,A)$-bimodules $M$ and $N$ and a bimodule map $M \to N$.
We obtain that $\mathcal{A}$ satisfies the twisted associativity condition as a superalgebra over $\R$.
The strongly $\Z_2^c$-graded superalgebra $A \oplus A_{(-1)^F} \subseteq \mathcal{A}$ satisfies $((-1)^F)^2 = 1$ by definition.
Also given that the naturality of $(-1)^F$ for $(B,A)$-bimodules $M$ is defined by
\[
M \otimes_A A_{(-1)^F} \to  B_{(-1)^F} \otimes_B M \quad m \otimes (-1)^F \mapsto (-1)^{|m|} (-1)^F \otimes m
\]
the definition of multiplications $A_{g^b} \otimes_{\overline{A}^{\theta(g^b)}} \overline{A_c}^{\theta(g^b)} \to A_{c \rtimes g^b}$ and $A_c \otimes_A A_{g^b} \to A_{c \rtimes g^b}$ imply that $(-1)^F a_{g^b} = (-1)^{|a_{g^b}|} a_{g^b} (-1)^F$ for all $a_{g^b} \in A_{g^b}$.
The same formula works for multiplication of $a_g \in A_g$ with $(-1)^F$ in $A_{gc} = A_{g^b}$ if $g = c \rtimes g^b$.
Requiring that the $(A,A)$-bimodules $A_{g^b}$ over the $\R$-superalgebra $A$ are $(A,\overline{A}^{\theta(g^b)})$-bimodules over the $\C$-super algebra $A$ is equivalent to requiring the condition that $i a_{g^b} = (-1)^{\theta(g^b)} a_{g^b} i$ for all $a_{g^b} \in A_{g^b}$.
This extends from $g^b \in  \operatorname{ob} G_b$ to arbitrary $g \in \operatorname{ob} \mathcal{G}$ as $A_{(-1)^F}$ is an $(A,A)$-bimodule and so $(-1)^F$ is `$\C$-linear'.
The multiplication maps $\phi_{g_1^b,g_2^b}$ being complex-linear is automatic by associativity of multiplying $i$ with elements $a_{g^b_1} \in A_{g^b_1}$ and $A_{g_2^b} \in A_{g^b_2}$.

We have now applied all conditions following from the associativity diagram of $\phi$ and proceed to the Diagram \eqref{Diagram B}. 
This telling us what the relationship is between $F_{\gamma_1^b}, F_{\gamma_2^b}$ and $F_{\gamma_1^b \otimes \gamma_2^b}$ for  $\gamma^b_i: g_i^b \to h_i^{b}$ two morphisms in $\mathcal{G}_b$.
We obtain morphisms $\gamma_i: 1 \rtimes g^b_i \to \Gamma(\gamma_i) \rtimes h^b_i$ in $\mathcal{G}$.
Working out the diagram in this case yields the equation
\[
F_{\gamma_2 \otimes \gamma_1}(a_{g^b_1} a_{g^b_2}) = F_{\gamma_1}( a_{g^b_1}) F_{\gamma_2}( a_{g^b_2})
\]
for all $a_{g^b_1} \in A_{g^b_1}$ and $a_{g^b_2} \in A_{g_2}$.
Note how this equation makes sense as $\Gamma(\gamma_2 \otimes \gamma_1) = \Gamma(\gamma_2) \Gamma(\gamma_1)$ and $c$ is central.
We are left with showing the same formula for arbitrary morphism of $\mathcal{G}$, which are of the form $c^{\epsilon_i} \gamma_i$ for $\epsilon_i =0,1$.
Using how $F_{c \gamma_i}$ is defined and the multiplication on elements, we obtain a large diagram we have to show commutes.
This will follow by the relationship between $F_{\gamma_1^b}, F_{\gamma_2^b}$ and $F_{\gamma_1^b \otimes \gamma_2^b}$ together with several applications of the naturality conditions for the $B\Z_2^F$-action expressed by the commutative squares above.
\end{proof}
\begin{remark}\label{Rem: c=1 in 2 D}
	Now we comment on the case where the symmetry group $G$ is bosonic, i.e. $c = 1$.
	In this case, analogous but slightly simplified proofs lead to a similar classification of spin-statistics $G$-TFTs with reflection structure after small adaptations given as follows. 
	We still have a strongly $\mathcal{G}:=\pi_{\leq 1} G$-graded algebra 
	\[
\mathcal{A} = 	\bigoplus_{g \in \operatorname{Obj} \mathcal{G}} A_g
	\]	
	which in particular has $A:= A_1 = A_c = A_{(-1)^F}$.
We still have that $A_1$ is a stellar algebra and $A_g$ have stellar structures such that multiplication is unitary.
Paths $\gamma: g \to g'$ will still give bimodule maps $F_\gamma: A_g \to A_{g'}$ satisfying the appropriate coherence condition, keeping in mind that $A_c = A$.
	The trivialization of the Serre gives an $(A,A)$-bimodule isomorphism $A \cong A^*$, which will now give a graded-symmetric Frobenius structure instead.
Moreover, in the analogous proof of Proposition \ref{Prop: ungrsym} and Lemma \ref{Lem:ungrsym}, Koszul signs will now show up.
The result is then that $A_g$ preserves the Frobenius structure if and only if 
	\[
	\overline{\lambda(a_g a_{g^{-1}})}^{\theta(g)} = (-1)^{|a_g||a_{g^{-1}}|} \lambda(a_{g^{-1}} a_g)
	\]
	for all $a_g \in A_g$ and $a_{g^{-1}} \in A_{g^{-1}}$.
\end{remark}

In the next section, we spell out some concrete algebraic consequences of the above characterization of spin-statistics and reflection TFTs.

\subsection{Examples and computations}

In practice, most topological field theories naturally arise in the form of $*$-algebras, not stellar algebras.
With the goal of fitting many examples in this framework, we therefore look at the special case where our stellar algebra is a $*$-algebra.
We thus prove a construction lemma for two-dimensial TFTs with reflection structure and spin-statistics below.
For simplicity, we restrict to discrete symmetry groups in this lemma.

Recall our convention for $\Z_2$-graded vs super $*$-algebras; super $\Z_2$-graded algebras $(A,\dagger)$ satisfy $(ab)^\dagger = b^\dagger a^\dagger$ without the Koszul sign, but super $*$-algebras $(A,*)$ satisfy $(ab)^* = (-1)^{|a||b||} b^* a^*$.
Over the complex numbers, super and $\Z_2$-graded $*$-algebras are equivalent under the one-to-one correspondence
\[
a^* =
\begin{cases}
a^\dagger & |a| = 0 
\\
i a^\dagger & |a| = 1.
\end{cases}
\]
We have already seen that in this setting the $A_g$ will come equipped with Hilbert bimodule structures and spelled out explicitly what it means for the multiplication maps to be unitary.
The maps $F_\gamma$ being unitary is equivalent to the element $a_\gamma$ being unitary in $A$, i.e. $a a^* = a^* a =1$.
For the case $\Gamma(\gamma) = 1$ this follows immediately from the fact that unitary bimodule maps $A \to A$ are exactly invertible elements $a \in A$ such that $aa^* = 1$. 
For the case $\Gamma(\gamma) = c$ it is also true; the subtle signs in the definition of the stellar module structure on $A_{(-1)^F}$ explained in Example \ref{ex:stellar A_c} are irrelevant because $a_\gamma$ is even.
Indeed, note that the map $A \to A_{(-1)^F}$ being unitary is equivalent to
\[
\langle 1, 1 \rangle_A = \langle a_\gamma, a_\gamma \rangle = (-1)^{|a_\gamma|} a_\gamma a_\gamma^* = a_\gamma a_\gamma^*.
\]
The following theorem constructs two-dimensional TFTs with reflection structure and spin-statistics for discrete symmetry groups $G$ starting from a $G$-graded $*$-algebra.

\begin{proposition}
\label{Lem:construct}
Let $(G,\theta, c)$ be a discrete fermionic group that is not bosonic.
Consider $G$ as a fermionic $2$-group with only trivial morphisms and let
\[
\mathcal{A} = \bigoplus_{g \in G} A_g
\]
be a fermionically graded algebra.
Suppose $\mathcal{A}$ is additionally a $\Z_2$-graded real $*$-algebra such that $i^\dagger = -i$ and $A_g^\dagger = A_{g^{-1}}$.
Define the $A$-valued inner product on $A_g$ by 
\[
\langle a_g, b_g \rangle := 
\begin{cases}
 a_g b_g^\dagger & |b_g| = 0
\\
i a_g b_g^\dagger & |b_g| = 1
\end{cases}
\]
Make $A$ into a super $*$-algebra by restricting the $\Z_2$-graded $*$-algebra structure on $\mathcal{A}$.
Then $A_g$ is a stellar $(A, \overline{A}^{\theta(g)})$-bimodule where we use the interesting stellar structure of Definition \ref{def:barstellar} on $A$ and the multiplication of $\mathcal{A}$ is unitary.
Moreover, if $\lambda$ is an ungraded symmetric Frobenius structure on $A$ such that $\lambda(a_g a_{g^{-1}}) = \overline{\lambda(a_{g^{-1}} a_g)}^{\theta(g)}$, it is compatible with the induced stellar structure if and only if
\[
\lambda(a^\dagger) = \overline{\lambda(a)}
\]
In that case $\mathcal{A}$ is a strongly $G$-graded stellar Frobenius algebra.
\end{proposition}
\begin{proof}
We start by postulating a formula of the form
\[
\langle a_g, b_g \rangle := \alpha(a_g, b_g) a_g b_g^\dagger
\]
where $\alpha(a_g,b_g) = \alpha_{\theta(g)}(|a_g|,|b_g|) \in \C$ only depends on $|a_g|, |b_g|$ and $\theta(g)$.
We then show that $A_g$ being stellar modules is equivalent to $\alpha_{\theta(g)}(0,0)$ and $\alpha_{\theta(g)}(1,0)$ being real and
\begin{align}
\label{eq:alphaconditions}
\alpha_{\theta(g)} (0,1) &= i \alpha_{\theta(g)}(1,0)
\\
\alpha_{\theta(g)}(1,1) &= i \alpha_{\theta(g)}(0,0)
\\
\alpha_{\theta(g)}(1,0) &= \alpha_{\theta(g)}(0,0)
\end{align}
Then we show that multiplication is unitary if and only if $\alpha$ is as in the statement of the proposition (except that $i$ could be replaced by $-i$, but it turns out the choice $+i$ is fixed by our conventions for the stellar structure on a complex $*$-algebra.).

We proceed to go through the four conditions of Proposition \ref{prop:Hilbert module}, skipping the first as it is redundant.
Note that for $a \in A$ the comparison of the $\Z_2$-graded and super $*$-structures on $\overline{A}$ are
\[
\overline{a}^* = (-1)^{|a|} \overline{a^*} = 
\begin{cases}
\overline{a^\dagger}  & |a| = 0,
\\
i \overline{a^\dagger} & |a| = 1.
\end{cases}
\]
and so $\overline{a}^\dagger = \overline{a^\dagger}$.
Hence for $\theta(g) = 1$ we get
\begin{align*}
(-1)^{|a| |b_g|} \langle a_g, b_g \overline{a}^* \rangle &= (-1)^{|a| |b_g|} \langle a_g, b_g \textcolor{blue}{i} \overline{a^\dagger} \rangle 
= (-1)^{|a| |b_g|} \alpha( a_g, b_g  \overline{a^\dagger}) a_g (b_g \textcolor{blue}{i} \overline{a^\dagger})^\dagger
\\
&= \textcolor{blue}{-i} (-1)^{|a| |b_g|} \alpha( a_g, b_g  \overline{a^\dagger}) a_g \overline{a} b_g^\dagger
\end{align*}
where the blue symbols are only there if $|a| = 1$.
Comparing this with 
\[
\langle a_g \overline{a}, b_g \rangle = \alpha(a_g \overline{a}, b_g) a_g \overline{a} b_g^\dagger
\]
gives 
\[
\alpha(a_g \overline{a}, b_g) = \textcolor{blue}{-i} (-1)^{|a| |b_g|} \alpha( a_g, b_g  \overline{a^\dagger}).
\]
The case $\theta(g) = 0$ yields the same formula.
Plugging in all eight possibilities, we see this is equivalent to the first two equations of \ref{eq:alphaconditions}.

We turn to the Hermiticity condition which says that
\begin{align*}
\alpha(a_g, b_g) a_g b_g^\dagger =
(-1)^{|a_g| |b_g|} \langle b_g, a_g \rangle^*
&= 
(-1)^{|a_g| |b_g|} \overline{\alpha(b_g, a_g)}\, \textcolor{blue}{i}\, (b_g a_g^\dagger)^\dagger
\end{align*}
where the blue imaginary unit is there if and only if $|a_g| + |b_g| = 1$.
This results in 
\[
\alpha(a_g, b_g) = (-1)^{|a_g| |b_g|} \textcolor{blue}{i} \, \overline{\alpha(b_g, a_g)} 
\]
Writing out all four possibilities for the degrees, we conclude that $\alpha_{\theta(g)}(0,0)$ is real, $\alpha_{\theta(g)}(1,1)$ is imaginary and
\[
\alpha_{\theta(g)}(0,1) = i \, \overline{\alpha_{\theta(g)}(1,0)}.
\]
Together with Equations \ref{eq:alphaconditions}, this implies $\alpha_{\theta(g)}(1,0)$ is real and so we have found exactly the conditions on $\alpha$ as claimed at the beginning of the proof.
The third condition on being a Hilbert module of Proposition \ref{prop:Hilbert module} is obvious.

The only thing that remains to be checked is that multiplication is unitary.
For $\theta(g) = 0$ we compute
\begin{align*}
\alpha(a_g a_h, b_g b_h) a_g a_h (b_g b_h)^\dagger &= \langle a_g a_h, b_g b_h \rangle 
= (-1)^{|b_g| |b_h|} \langle a_g \langle a_h, b_h \rangle, b_g \rangle
= (-1)^{|b_g| |b_h|} \langle a_g \alpha( a_h, b_h) a_h b_h^\dagger, b_g \rangle
\\
&=  (-1)^{|b_g| |b_h|} \alpha( a_g \alpha( a_h, b_h) a_h b_h^\dagger, b_g )a_g \alpha( a_h, b_h) a_h b_h^\dagger b_g^\dagger
\\
&= (-1)^{|b_g| |b_h|} \alpha( a_g a_h b_h^\dagger, b_g )  \alpha( a_h, b_h) a_g a_h b_h^\dagger b_g^\dagger
\end{align*}
giving us the condition
\begin{equation}
\alpha(a_g a_h, b_g b_h) =  (-1)^{|b_g| |b_h|} \alpha( a_g a_h b_h^\dagger, b_g )  \alpha( a_h, b_h) 
\end{equation}
The same computation for $\theta(g) = 1$ gives
\begin{equation}
\alpha(a_g a_h, b_g b_h) =  (-1)^{|b_g| |b_h| + |b_h|} \alpha( a_g a_h b_h^\dagger, b_g )  \overline{\alpha( a_h, b_h)}
\end{equation}
Given all the conditions on $\alpha$ from before this yields $\alpha_{\theta(g)} (0,0) = \alpha_{\theta(g)} (1,0) = 1$ and $\alpha_{\theta(g)} (0,1) = \alpha_{\theta(g)} (1,1) = \pm i$ independent of $\theta(g)$.
Moreover, the sign in front of the imaginary unit is the same for both choices of $\theta(g)$.
The compatibility between the Frobenius structure and the stellar structure in the case where $A$ is a $*$-algebra is $\lambda(a^*) = \overline{\lambda(a)}$ which gives the last statement since $\lambda$ is zero on odd elements and on even elements we have $a^\dagger = a^*$.
This finishes the proof.
\end{proof}

\begin{remark}
We warn the reader that $\dagger: \mathcal{A} \to \mathcal{A}$ is not a $\C$-antilinear map if $\theta$ is nontrivial since 
\[
(i a_g)^\dagger = a_g^\dagger i^\dagger = - a_{g}^\dagger i = i a_{g}^\dagger
\]
for $\theta(g) = 1$.
\end{remark}

\begin{remark}
The above proposition immediately implies a reality condition on the Frobenius structure through
\[
\overline{\lambda(1)} = \lambda(1^\dagger) = \lambda(1) \in \R^\times.
\]
\end{remark}

We provide a few concrete examples of TFTs that can be constructed using the lemma above and indicate some generalizations to the case where $G$ is not discrete along the way.

\begin{example}
Consider for a finite fermionic group $G$ the trivial theory, which has $A = \C$ with stellar structure induced by its canonical $C^*$-algebra structure and the Frobenius structure is given by $\lambda(z) = z$.
The fermionically $G$-graded algebra is the fermionic group algebra
\[
\mathcal{A} = \frac{\R[i,x_g: g\in G ] }{(i^2 = -1, x_g i = (-1)^{\theta(g)} ix_g, x_g x_h = x_{gh})}
\] 
defined to be purely even.
Note that $\lambda$ satisfies the required condition.
We make this into a $\Z_2$-graded real $*$-algebra by $x_g^\dagger = x_{g^{-1}}$ and $i^\dagger = -i$.
Clearly $\lambda(z^\dagger) = \overline{\lambda(z)}$ for $z \in A_1 = \C$ and so by Proposition \ref{Lem:construct}, this defines an $H$-TFT with spin-statistics and reflection structure.
Note that we obtained a pretty interesting algebra by starting with the trivial theory. 
In particular since $\mathcal{A}$ is semisimple, we could assign the Frobenius algebra $\mathcal{A}$ to a point to define a two-dimensional TFT with structure group $O_2$.
We expect that this theory is the quantization of the trivial theory described above obtained by path integrating over the $G$-backgrounds.
Note also that if $G$ only contains time-preserving symmetries, the resulting theory is oriented.
\end{example}

For explicit computations, it is useful to minimize the amount of algebraic information necessary to describe a topological field theory.
For this we take the fermionically skeletal model of the fermionic $2$-group $\mathcal{G} = \pi_{\leq 1} G$.
We recall that it is defined as $\mathcal{G} = \Z_2^c \rtimes_{\Gamma, \Xi} \mathcal{G}_b$ where $\mathcal{G}_b$ is the skeletal model of $\pi_{\leq 1} G_b$.
Here the action of $\mathcal{G}_b$ on $\Z_2^c$ defined by the map $\mathcal{G}_b \to *\DS \Z_2^c$ classifying the extension of $G_b$ by $\Z_2^c$ is summarized by the data $\Gamma: \pi_1(G_b) \to \Z_2^c$ and $\Xi: \pi_0(G_b) \times \pi_0(G_b) \to \Z_2^c$ as explained in Section \ref{Sec:fermskeletal}.
The first advantage of this situation is that the objects form the relatively small set $\pi_0(G_b) \times \Z_2^c$ over which the algebras will be graded.
Additionally a lot of the data on $1$-morphisms is determined by loops.

So let 
\[
\mathcal{A} = \bigoplus_{g \in \Z_2^c \times \pi_0(G_b)} A_g
\]
denote a strongly $\mathcal{G}$-graded stellar Frobenius algebra.
Note that if $\gamma^b \in \pi_1(G_b)$ is a loop, then the data of a bimodule isomorphism $F_{\gamma^b}: A \cong A_{\Gamma(\gamma^b)}$ is equivalent to a choice of even invertible element $a_{\gamma^b} \in A_{\Gamma(\gamma^b)}$ satisfying the additional condition that
\[
a_{\gamma^b} a = a a_{\gamma^b} \quad \forall a \in A.
\]
We will sometimes identify $a_{\gamma^b}$ with an invertible element $a_{\gamma^b}'$ of $A$ using the canonical generator $(-1)^F \in A_{(-1)^F}$.
But given the twisted commutation condition $(-1)^F$ satisfies, this also changes the centrality condition on $a_{\gamma^b}'$.
Namely, it is central in case $\Gamma(\gamma^b) = 1$ and satisfies
\[
a_{\gamma^b} a = (-1)^{|a|} a a_{\gamma^b} \quad \forall a \in A
\]
in case $\Gamma(\gamma^b) = c$.
We emphasize that this latter condition is not saying that $a_{\gamma^b}$ is in the graded center since $a_{\gamma^b}$ is even and so the condition of it being in the ungraded center and in the supercenter are equal.

Note that the condition on the $F_\gamma$ for two loops $\gamma^b_1, \gamma^b_2$ is equivalent to $a_{\gamma^b_1 \gamma^b_2} = a_{\gamma^b_1} a_{\gamma^b_2}$ independent of $\Gamma(\gamma^b_i)$.
Moreover, by translating paths to start at the identity using the relationship between $F_{\id_g \otimes \gamma}$ and $F_\gamma$, the data of the $F_\gamma$ is determined by specifying it for all paths $1 \to 1$ and $1 \to c$ determined by the elements $a_{\gamma^b}$.
They then satisfy an extra condition coming from the action of $\pi_0(G)$ on $\pi_1(G)$ by conjugation.
Indeed, Diagram \ref{Diagram B} applied to the constant path at $g$ and a loop $\gamma^b$ gives a condition relating $F_{g \gamma^b g^{-1}}$ with $F_\gamma$ using $F_{\id_{g}} = \id_{A_{g^{-1}}}$, which can also be obtained as a special case of the equation the $F_\gamma$ have to satisfy.
On the level of elements this results in
\[
a_{g \gamma^b g^{-1}} = a_\gamma \quad g \in \pi_0(G),  \gamma \in \pi_1(G).
\]

\begin{remark}
	If additionally the associator of the skeletal model of $\pi_{\leq 1} G_b$ vanishes, then it also vanishes in the fermionically skeletal model of $\mathcal{G}$ and so the object set of $\mathcal{G}$ is a group.
	Alternatively, we can see $\operatorname{ob} \mathcal{G}$ as the extension of $\pi_0(G_b)$ by $\Z_2^c$ defined by $\Xi$, which is a $2$-cocycle on $\pi_0(G_b)$ because we assumed the associator to be trivial.
	In particular, $\operatorname{ob} \mathcal{G}$ comes with a canonical fermionic group structure and the algebra $\mathcal{A}$ is then also strongly fermionically graded under this $1$-group.
\end{remark}

In the case $G_b$ has nontrivial associator, the multiplication on $\operatorname{ob} \mathcal{G}$ need not be associative.
This is the case if and only if $\Gamma \circ \alpha = d\Xi \neq 0$, which tells us that $g_1 (g_2 g_3)$ and $(g_1 g_2) g_3$ differ by $\Gamma (\alpha(g_1, g_2,g_3)) \in \Z_2^c$.
The failure of $\mathcal{A}$ to be associative - even when $\Gamma \circ \alpha = 0$ - is now controlled by the element $a_{\alpha(g_1,g_2,g_3)} \in A$ obtained by taking the path $\gamma :=\alpha(g_1,g_2,g_3) \in \pi_1(G_b)$.
More precisely for all $g_1, g_2, g_3 \in G, a_{g_1} \in A_{g_1}, a_{g_2} \in A_{g_2}$ and $a_{g_3} \in A_{g_3}$ we have
\[
a_{g_1} (a_{g_2} a_{g_3}) = (a_{g_1} a_{g_2}) a_{g_3} a_{\alpha(g_1,g_2,g_3)} \in A_{g_1 (g_2 g_3)}.
\]
Note that since $a_\gamma \in A_{\Gamma(\gamma)}$, both sides of this equation are in the same set.

\begin{example}
\label{Ex:spinc}
Take $G = \Spin_2$ so that $G_b = SO_2$ and $H_d = \Spin^c_d$.
In the physics literature, this internal symmetry group is often written $G = U_1$ because it typically occurs in the case of charged particles.
We decide to refrain from this notation, because $\Spin_2$ has a nontrivial central element $c$ imposing a spin-charge relation.
A $G$-TFT consists of a stellar Frobenius algebra $(A, M, \sigma, \lambda) \in \stFrob$ together with a single even invertible element $a_\gamma \in A_{(-1)^F}$ such that $a_\gamma a =  a a_\gamma$ for all $a \in A$ and multiplication by it defines a unitary map $A \to A_{(-1)^F}$.
In case $A$ is a $*$-algebra, the requirement is that $a_\gamma \in A_{(-1)^F}$ is a unitary element of the $\Z_2$-graded super $*$-algebra $A \oplus A_{(-1)^F}$.
Equivalently, $a_\gamma' := a_\gamma (-1)^F$ is an even invertible unitary element of $A$ such that $a_\gamma ' a = (-1)^{|a|} a a_\gamma'$.
\end{example}

\begin{example} 
\label{ex:pin-tft}
Take $G = \Pin^+_1$ to be a time-reversal $g$ with square $1$, so $G_b = O_1$ and $H_d = \Pin^-_d$.
Then a $G$-TFT with spin-statistics connection and reflection structure contains the data of a finite-dimensional semisimple superalgebra $A$ and a strongly $\Z^c_2 \times \Z^g_2$-graded superalgebra over $\R$ of the form
\[
\mathcal{A} = A \oplus A_{(-1)^F}  \oplus A_g \oplus (A_g \otimes_A A_{(-1)^F})
\]
determined by the strongly $\Z_2$-graded superalgebra $A \oplus A_g$ such that $i a_g = - a_g i$ and $\lambda(a_g b_g) = \overline{\lambda(b_g a_g)}$.
In case $A$ is a $*$-algebra, $A_g$ has an $A$-valued nondegenerate inner product $\langle \cdot , \cdot \rangle$ (complex linear in the right argument in our convention) such that
\[
\langle a_g \overline{a}, a_g' \rangle = (-1)^{|a_g'| |a| + |a|}  \langle a_g,   a_g' a^* \rangle
\]
and the multiplication being unitary is equivalent to
\begin{align*}
 a_g a_{g}' (b_g b_{g}')^* = (-1)^{|b_{g}| |b_{g}'| + |b_{g}|} \langle a_g, \langle a_{g}', b_{g}' \rangle_{A_{g}} b_g \rangle_{A_g}
\end{align*}
for all $a_g,a_g',b_g,b_g' \in A_g$.

Assuming we are in the setting of our construction lemma, this $A$-valued inner product is induced by a $\Z_2$-graded $*$-algebra structure $\dagger$ on $\mathcal{A}$ satisfying additional conditions.
We illustrate the situation by specializing to the case $A = \C$ as we did in Example \ref{ex:pin-}.
Let $x_g$ denote a generator of $A_g$.
By multiplying $x_g$ with an appropriate element of $\C$ we can assume without loss of generality that $x_g^2 \in U_1$.
Suppose for now that $x_g$ is even, which uniquely specifies $\mathcal{A}$ as a superalgebra.
A choice of Frobenius structure on $A$ is uniquely specified by $\lambda(1) \in \R^\times$.
The condition for it be compatible with $A_g$ gives $\lambda(x_g^2) = \overline{\lambda(x_g^2)}$, which implies that $x_g^2 = \pm 1$.
We therefore necessarily have $x_g^\dagger = \pm x_g$ and either makes $\mathcal{A}$ into a $\Z_2$-graded $*$-algebra.
Note that $A \oplus A_g$ is a $C^*$-algebra (isomorphic to $M_2(\R)$ or $\mathbb{H}$) if and only if the two sign choices $x_g^2 = \pm 1$ and $x_g^\dagger = \pm x_g$ agree, which happens if and only if the $A$-valued inner product on $A_g$ is positive.
If this is the case, the reflection structure is reflection positive.

Now assume $x_g$ is odd instead.
We have $A \oplus A_g \cong Cl_{\pm 2}$ given by $i \mapsto e_1 e_2$ and $x_g \mapsto e_1$.
A $*$-algebra structure on $Cl_{\pm 2}$ is uniquely determined by the two sign choices $e_i^* = \pm e_i$ and again this gives a $C^*$-structure if and only if these signs agree with the signs $e_i^2$.
A condition in our construction lemma implies that 
\[
-e_1 e_2 = -i = i^\dagger = (e_1 e_2)^\dagger = e_2^\dagger e_1^\dagger = - e_1^\dagger e_2^\dagger
\]
so that the two sign choices $e_i^\dagger = \pm e_i$ have to be made equal.
Therefore, whether $A \oplus A_g$ will be a $C^*$-algebra will again depend on a single sign choice.
In particular, note that if we choose the sign $x_g^\dagger = \pm x_g$ equal to the sign $x_g^2$ as in the case where $x_g$ is even, we obtain that $A_g$ has a positive $A$-valued inner product.
Indeed, we have in that case that
\[
z x_g (z x_g)^\dagger = z x_g x_g^\dagger \bar{z} = |z|^2 \geq 0.
\]
Hence this choice will define a unitary theory.
It corresponds to twice a generator in the group of unitary invertible field theories $\Hom(\Omega^{\Pin^-}_2, \C^\times) \cong \Z_8$.
\end{example} 

\begin{example}
\label{Ex:Pin1-}
Take $G = \Pin^-_1$ to consist of a single time-reversal $T$ of square $(-1)^F$, so $G_b = O_1$ and $H_d = \Pin^+_d$. 
A $G$-TFT with spin-statistics connection and reflection structure is a stellar Frobenius algebra $A$, a strongly $\Z_4$-graded superalgebra
\[
\mathcal{A} = A \oplus A_{(-1)^F}  \oplus A_T \oplus (A_T \otimes_A A_{(-1)^F})
\]
such that $i a_T = - a_T i$.
The condition on the Frobenius structure is $\lambda(a_T b_{T^{-1}}) = \overline{\lambda(b_{T^{-1}} a_T)}$
Since $T^{-1} = c T$, this is equivalent to
\[
\lambda(a_T  b_T (-1)^F) = \overline{\lambda(b_T (-1)^F a_T)} = (-1)^{|a_T|} \overline{\lambda(b_T  a_T (-1)^F)}
\]
for all $a_T,b_T \in A_T$.
$A_T$ is a stellar bimodule such that the multiplication map $A_T \otimes_{\overline{A}} \overline{A_T} \to A_{(-1)^F}$ is unitary.
The multiplication being unitary is equivalent to
\begin{align*}
(-1)^{|b_{T}| |b_{T}'| + |b_{T}|} \langle a_T, \langle a_{T}', b_{T}' \rangle_{A_{T}} b_T \rangle_{A_T}  =  \langle a_T a_{T}' , b_T b_{T}' \rangle_{A_{(-1)^F}} = a_T a_{T}' (-1)^F (b_T b_{T}' (-1)^F)^* = a_T a_T' b_T b_T'
\end{align*}
for all $a_T,a_T',b_T,b_T' \in A_T$, see Example \ref{ex:stellar A_c}.

Note that alternatively, we could have implemented the twisting cocycle $\Xi$ by giving a nonassociative $\Z_2$-graded algebra $A \oplus A_T$ with twisted associativity conditions such as
\[
a_T (b_T a) =(-1)^{|a|} (a_T b_T) a \quad a_T, b_T \in A_T,  a \in A.
\]
In that convention the Frobenius structure satisfies
\[
\lambda(a_T  b_T ) = (-1)^{|a_T|} \overline{\lambda(b_T  a_T )}
\]
However, we prefer to work with the twice larger associative algebra $\mathcal{A}$.
\end{example}

\begin{example}
For a more complicated example, consider the internal symmetry group consisting of charge and a time-reversal with square $(-1)^F$, also called class AII in the literature on topological phases of matter.
This means that 
\[
G = \Pin_2^- \cong \frac{\Spin_2 \rtimes \Z_4^T}{\Z_2^c}
\]
with nontrivial $\theta$ and nontrivial $c$ so that $G_b = O_2$.
We have seen in Example \ref{Ex: pin ferm gp} that the skeletal model of $G$ has $\pi_0(G) = \Z_2$ acting nontrivially on $\pi_1(G) = \Z$ with nontrivial $k$-invariant.
However, we will work with the fermionically skeletal model.
So start with the skeletal model of $G_b = O_2$, which has $\pi_0(G_b) = \Z_2$ acting nontrivially on $\pi_1(G_b) = \Z$ with trivial associator.
We obtain the fermionically skeletal model $\mathcal{G}$, which has a $\Z_4$-worth of objects: $1,c,T$ and $cT$, where $T^2 = c$ is a lift of a reflection in the plane to $\Pin_2^-$.
Let $\gamma: 1 \to c$ be a lift of a generator of $\pi_1(G_b)$ to $\mathcal{G}$, which in the Lie group presentation we can explicitly write 
\[
\gamma(t) = e^{i \pi t} \in U_1 \cong \Spin_2 \subseteq\frac{ \Spin_2 \rtimes \Z_4^T}{\Z_2^c}
\]
for $t \in [0,1]$.
Because of the semidirect product, we obtain the relation $T \gamma(t) = \gamma(t)^{-1} c T$.

We now consider $2$d TFTs with spin-statistics and reflection structure and internal symmetry group $G$.
Because $\Pin^-_1$ naturally insides $\mathcal{G}$, we can start by copying the data from Example \ref{Ex:Pin1-}.
In particular we have a $\Z_4$-graded algebra
\[
A \oplus A_{(-1)^F} \oplus A_T \oplus (A_{T} \otimes_A A_{(-1)^F})
\]
with antilinear $A_T$.
The path $\gamma$ induces an even invertible element $a_\gamma' \in A$ such that $a a_\gamma' = (-1)^{|a|} a_\gamma' a$ for all $a \in A$ as in Example \ref{Ex:spinc}.
The commutation relation between $T$ and $\gamma$ will give the relation
\[
a_\gamma' a_T = (-1)^{|a_T|} a_T a_\gamma'^{-1}
\]
for all $a_T \in A_T$.
Indeed, this follows by the relation
\[
a_T a_\gamma ' (-1)^F = a_T a_\gamma  = F_{T \gamma}(a_T) = F_{\gamma^{-1} cT }(a_T) = (a_\gamma')^{-1} (-1)^F a_T.
\]
We will omit the discussion of the stellar Frobenius structure in this example.
\end{example}

\begin{example}
Take $G = SU_2$ where $c \in SU_2$ is the unique element of order two. 
Since $\pi_{\leq 1} SU_2 = 0$, the internal fermionic $2$-group is trivial and so we don't see anything of the interesting higher homotopy groups of $SU_2$ until spacetime dimension three.
However, note that $G_b = SO_3$ and we have $\pi_{\leq 1} G_b = * \DS \Z_2$, so the resulting TFTs will neither be the same as theories with symmetry group $\Z^c_2$, nor will they be the same as theories with trivial bosonic symmetry group.
Now $\Gamma$ is the identity on $\Z_2$ and its fermionically skeletal model is $\mathcal{G} = \Z_2 \rtimes */\Z_2$
The classification of two-dimensional topological field theories with symmetry group $G$ is thus equal to the one for $G = \Spin_2$ with the extra condition that $a_\gamma^2 = 1$.

Suppose more generally that the fermionic symmetry group is of the form $G = \frac{SU_2 \times K}{\Z_2}$ where $K$ is a finite fermionic group, the quotient is by the diagonal $\Z_2^c$.
The fermion parity is $(c,1) = (1,c)$ and $\theta: G \to \Z_2$ is the unique extension of $\theta: K \to \Z_2$.
Then $\operatorname{ob}\mathcal{G} = K$ and $\Gamma$ is again the identity.
Therefore in addition to the data from the case $G = SU_2$ we have a fermionically $K$-graded algebra and all $A_k$ are stellar modules with unitary multiplication such that $\lambda(a_k b_{k^{-1}}) = \overline{\lambda(b_{k^{-1}} a_{k})}$ for all $k \in K$.
\end{example}

\subsection{Classifications of some other types of 2d TFTs}\label{Sec:Computation 2D}

Our main theorem classifies two-dimensional fully local TFTs with fermionic symmetry, reflection structure, and spin-statistics connection.
In this section we will summarize how our results extend when we modify several adjectives in the last sentence.
Additionally, we will compare with known results in the literature.

In the bosonic setting discussed in Remark \ref{Rem: c=1 in 2 D}, it is easy to see what happens when we leave out the spin-statistics connection.
The only thing that changes is that we no longer identify $A \cong A_{(-1)^F}$ so that we just get a $(G,\theta)$-graded algebra.
However, it turns out that finite-dimensional semisimple superalgebras have a graded-symmetric Frobenius structure if and only if $A \cong A_{(-1)^F}$.
For example, suppose $\C l_1 = \frac{\C [e]}{(e^2 = 1)}$ would have a graded-symmetric Frobenius structure $\lambda: A \to \C$.
Then 
\[
\lambda(1) = \lambda(e^2) = - \lambda(e^2) = -\lambda(1)
\]
and so since $\lambda$ is an even map, the Frobenius form is degenerate.

Note that leaving out spin-statistics in the case where the symmetry group is not bosonic is more awkward.
In that case we will not have an identification of the arbitrary $\Z_2^c$-graded superalgebra $A \oplus A_c$ with the $\Z_2^c$-graded superalgebra $A \oplus A_{(-1)^F}$.
Consequently, the Frobenius structure is replaced by an even bimodule isomorphism $A \to A^* \otimes_A A_c$ satisfying a symmetry condition.
In case the graded algebra $A \oplus A_c$ is induced by an automorphism $\phi: A \to A$ so that $A_c = A_\phi$ compatible with the multiplication, this bimodule isomorphism can be replaced by an even Frobenius structure on $A$ with Nakayama automorphism $(-1)^F_A \circ \phi$.
This discussion also applies to $\Spin^r$-TFTs~\cite{LorantNils}; in our framework they are given by a $\Z_r$-graded superalgebras
\[
\mathcal{A} = A_1 \oplus A_g \oplus \dots \oplus A_{g^{r-1}}
\]
together with a bimodule isomorphism $A \to A^* \otimes_A A_g$ satisfying a condition.
In case $A_g$ is induced by an automorphism $\phi$ of order $r$, this bimodule map is equivalent to a Frobenius structure with Nakayama isomorphism $(-1)^F_A \circ \phi$.

In the bosonic case where $\theta$ is additionally trivial, leaving out the reflection structure gives a classification of arbitrary $H_2 = G \times SO_2$-TFTs in terms of a finite-dimensional strongly $\pi_0(G)$-graded algebra
\[
\mathcal{A} = \bigoplus_{g \in \pi_0(G)} A_g
\]
together with a graded-symmetric Frobenius stucture $\lambda$ on $A$ such that
\[
\lambda(a_g a_{g^{-1}}) = (-1)^{|a_g||a_{g^{-1}}|} \lambda(a_{g^{-1}} a_g)
\]
for all $a_g \in A_g$ and $a_{g^{-1}} \in A_{g^{-1}}$ and even invertible elements $a_\gamma \in Z(A)$ for every $\gamma \in \pi_0(G)$ such that
\[
a_{g \gamma g^{-1}} = a_\gamma
\]
for all $g \in \pi_0(G)$.

In particular, restricting further to the case where $G$ has no morphisms allows us to compare with the results of \cite{oritthesis} (up to the minor discrepancy that we work with superalgebras and she works with ungraded algebras).
To prove the equivalence of the above description with Davidovich' classification, the only nontrivial thing to remark is that our condition $\lambda(a_g a_{g^{-1}}) = (-1)^{|a_g||a_{g^{-1}}|} \lambda(a_{g^{-1}} a_g)$ on the graded-symmetric Frobenius structure $\lambda:A_1 \to \C$ is exactly equivalent to saying that $\lambda$ extends to a graded-symmetric Frobenius structure on $\mathcal{A}$ with the property that $\lambda(a_g) = 0$ unless $g = 1$.

Note that similarly to leaving out spin-statistics in the non-bosonic case, leaving out the reflection structure in the time-reversing case makes the results harder to interpret, independently of whether we allow fermions.
For example, bosonic TFTs with only a time-reversing symmetry $T$ of square one and reflection structure are classified by bosonically $\Z_2^T$-graded algebras
\[
A \oplus A_T
\]
in which elements of $A_T$ anti-commute with $i$.
Here $A$ is stellar, $A_T$ is a stellar bimodule such that multiplication is unitary and $A$ has a graded-symmetric Frobenius structure such that $\lambda(a_T b_T) = (-1)^{|a_T| |b_T|} \overline{\lambda(b_T a_T)}$ for all $a_T,b_T \in A_T$.
Without assuming a reflection structure, these TFTs have been classified in \cite{schommerpriesthesis} (he works with target ungraded algebras, but the generalization is straightforward).
Except for the graded Frobenius structure on $A$, the extra data is a complex-linear stellar structure on $A$.
In other words, it consists of an $(A,A^{\op})$-bimodule $N$ and a bimodule isomorphism $\tau: N^{\op} \to N$ such that $\tau^2 = \id_N$.
There is an additional compatibility condition between the complex-linear stellar structure and the Frobenius structure \cite[Definition 3.82]{schommerpriesthesis}.
Note that our data classifying such TFTs with an additional reflection structure in particular give such a $\C$-linear stellar algebra; by using the identification of $A^{\op}$ with $\overline{A}$ using the $\C$-antilinear stellar structure, we can identify $A_T$ with $N$ and the multiplication map $A_T \otimes_{\overline{A}} \overline{A_T} \to A$ with $\tau$.

\appendix 

\section{2-groups and their actions on bicategories}
\label{App: 2-Group}

This appendix contains some details and our conventions related to 2-groups and their
actions on bicategories. We assume some familiarity with the basic definitions related to bicategories. For details we refer to~\cite[Appendix A]{schommerpriesthesis}. We denote by $\BiCat$ the tricategory of bicategories. The appendix also contains some new result. In particular, 
related to fixed point categories of semi-direct products, see Proposition~\ref{Prop: Fixed points}. When it comes to coherence isomorphisms we
will not always be consistent with their direction. This is never a problem because we can always replace them by their (adjoint) inverse. The main reason for this is 
that depending on the situations different choices of direction seem more natural.  

\subsection{2-groups}\label{App:2-groups}

2-groups are categorifications of ordinary groups.
\begin{definition}
	A \emph{(weak) 2-group} is a monoidal groupoid $(\mathcal{G},\otimes , \alpha_\mathcal{G}, 1)$ such that every object $g\in \mathcal{G}$ is invertible with respect to the tensor product, i.e. there exists $g^{-1}\in G $ such that $g\otimes g^{-1}\cong 1 \cong g^{-1}\otimes g$. 
\end{definition} 
Equivalently, 2-groups can be described by a 2-groupoid $B\mathcal{G}$ with one object $*$ and
$\mathcal{G}$ as category of endomorphisms. 

\begin{example}
	Let $G$ be a topological group. Its foundamental groupoid $\Pi_1(G)$, i.e. the category with
	object points $g\in G$ and morphisms homotopy classes of paths. The monoidal structure induced
	by the multiplication in $G$. The corresponding 2-groupoid $B\Pi_1(G)$ is equivalent to the
	fundamental 2-groupoid of the classifying space $\Pi_2(BG)$. 
\end{example}

Since a (weak) 2-group $\mathcal{G}$ is in particular a monoidal category, $B \mathcal{G}$ is a bicategory with one object.
Through the definition of functors, natural transformations and modifications of bicategories we get notions of 1- 2- and 3-morphisms between 2-groups, which we now spell out using~\cite[Appendix A.1]{SP}.
The main issue that has to be treated with care is that we don't want to assume our 2- and 3-morphisms to be pointed.
Therefore for example 2-morphisms are slightly more general than just monoidal natural transformations.
The result is the following definition (which would work just as well for monoidal categories but we are mainly interested in 2-groups).

\begin{definition}
	A \emph{1-morphism} $F: \mathcal{G}_1 \to \mathcal{G}_2$ of $2$-groups is a monoidal functor.
	A \emph{2-morphism} $\sigma:F \to G$ consists of an object $\sigma_* \in \operatorname{ob} \mathcal{G}_1$ and a collection of natural isomorphisms $\sigma_g: \sigma_* \otimes F(g) \to G(g) \otimes \sigma_*$ for every $g \in \operatorname{ob}\mathcal{G}_1$ such that the diagram 
	\begin{equation}
		\begin{tikzcd}[column sep=1.7cm]
			& G(g_1 \otimes g_2) \otimes \sigma_* 
			\ar[ld,"{\chi_{g_1,g_2} \otimes 1_{\sigma_*}}",swap] &  
			\\ 
			(G(g_1) \otimes G(g_2)) \otimes \sigma_* \arrow[d] & & \sigma_* \otimes F(g_1 \otimes g_2) \ar[d, "{1_{\sigma_*} \otimes \phi_{g_1,g_2}}"]  \ar[lu, "{ \sigma_{g_1\otimes g_2} }",swap]  
			\\ 
			G(g_1) \otimes (G(g_2) \otimes \sigma_*) & &  \sigma_* \otimes (F(g_1) \otimes F(g_2)) 
			\\ 
			G(g_1) \otimes (\sigma_* \otimes F(g_2))  \ar[u, "{1_{G(g_1)} \otimes \sigma_{g_2}}"] & (G(g_1) \otimes \sigma_*) \otimes F(g_2) \arrow[l]  &  (\sigma_* \otimes F(g_1)) \otimes F(g_2) \ar[l,"{\sigma_{g_1} \otimes 1_{F(g_2)}}"] \arrow[u]
		\end{tikzcd}
	\end{equation}
	commutes for all $g_1,g_2 \in \operatorname{ob} \mathcal{G}_1$
	(Here $\phi$ and $\chi$ are the monoidality data for $F$ and $G$, respectively), and 
	\begin{equation}
		\begin{tikzcd}
			\sigma_* \otimes F(e)\ar[r, "\sigma_{1_1}"] \ar[d] & G(e)\otimes \sigma_* \ar[d] \\ 
			 	\sigma_* \otimes 1_2 \ar[r] & 1_2\otimes \sigma_*
		\end{tikzcd}
	\end{equation} 
commutes, where the vertical unlabeled isomorphisms are part of the monoidality data for $F$ and $G$ and the horizontal isomorphism is build from the unit constrains in $\mathscr{G}_2$.
	We say $\sigma$ is \emph{pointed} if $\sigma_* = 1$, i.e. if $\sigma$ is a monoidal natural transformation.
	A 3-morphism $m: \sigma \to \tau$ is a morphism $m_*: \sigma_* \to \tau_*$ such that
	\[
	\begin{tikzcd}
		\sigma_* \otimes 		F(g_1)  \arrow[d,"{m_* \otimes F(f)}", swap]  \arrow[r,"\sigma_{g_1}"] & G (g_1) \otimes \sigma_* \arrow[d,"{G(f) \otimes m_*}"]
		\\
		\tau_* \otimes F(g_2) \arrow[r,"{\tau_{g_2}}"] & G(g_2)  \otimes \tau_*
	\end{tikzcd}
	\]
	commutes for all morphisms $f: g_1 \to g_2$ in $\mathcal{G}_1$. Note that all
 	pointed 3-morphisms are trivial. 
	We denote by $\TGrp$ the tricategory of 2-groups which is a full subcategory of the tricategory of 2-groupoids $\TGrpd$. 
\end{definition}

A 2-group is called \emph{skeletal} if every isomorphism class of objects only contains one
object. Every 2-group is equivalent to a skeletal one, called its \emph{skeletal model}.
The skeletal model can be reconstructed from the following combinatorial data
\begin{equation}\label{Eq: Skeletal model}
	(\pi_0(\mathcal{G}),\pi_1(\mathcal{G}_1), \alpha: \pi_0(\mathcal{G})  \curvearrowright \pi_1(\mathcal{G}), k \in H^3(\pi_0(\mathcal{G}),\pi_1(\mathcal{G})_\alpha)).
\end{equation}
Here 
\begin{enumerate}
	\item $\pi_0(\mathcal{G})$ is the set of isomorphism classes of objects with group structure given by $\otimes$;
	\item $\pi_1(\mathcal{G})$ is the set of 1-morphisms $1 \to 1$ with abelian group structure given by either $\otimes$ or composition (which are equal and commutative by Eckman-Hilton);
	\item $\alpha$ is the action of $\pi_0(\mathcal{G})$ on $\pi_1(\mathcal{G})$ given by $\gamma \mapsto 1_g \otimes \gamma \otimes 1_{g^{-1}}$;
	\item $k(g_1,g_2,g_3)$ is the associator $(g_1 \otimes g_2) \otimes g_3 \to g_1 \otimes (g_2 \otimes g_3)$.
\end{enumerate}
Every such 4-tuple gives a skeletal 2-group with objects $\pi_0(\mathcal{G})$, morphisms from $g_1$ to $g_2$ equal to $\pi_1(\mathcal{G})$ if $g_1 = g_2$ and empty otherwise.
Given $g \in \pi_0 \mathcal{G}$, denote the element of $\Aut g$ corresponding to $\gamma \in \pi_1(\mathcal{G})$ by $\gamma_{g}$.
Then the tensor product $\Aut g_1 \times \Aut g_2 \to \Aut g_1 g_2$ is defined (in our conventions) as
\[
\gamma_{g_1} \otimes \delta_{g_2} = ( \gamma(\alpha(g_1)\delta))_{g_1 g_2}
\]
In other words, we have $\gamma_g = \gamma \otimes 1_g$ but $1_g \otimes \gamma = (\alpha(g) \gamma)_g$.
With this convention the pentagon identity for the monoidal structure is equivalent to $k$ being a $3$-cocycle with values in the left $\pi_0(\mathcal{G})$-module $\pi_1(\mathcal{G})$:
\[
k(g_1,g_2,g_3 g_4) k(g_1 g_2, g_3, g_4) = \alpha(g_1)(k(g_2,g_3,g_4)) k(g_1, g_2 g_3, g_4) k(g_1, g_2, g_3).
\]

Morphism between skeletal models can also be described explicitly:
\begin{lemma}
\label{Lem:skeletalmap}
	A homomorphism between skeletal $2$-groups $F: (\mathcal{G},\alpha,k) \to (\mathcal{H},\beta,l)$ consists of a homomorphism $F_1: \pi_1(\mathcal{G}) \to \pi_1(\mathcal{H})$, a map $F_0: \pi_0(\mathcal{G}) \to \pi_0(\mathcal{H})$  such that
	\[
	F_1(\gamma \cdot \alpha(g_1)(\delta)) = F_1(\gamma) \cdot \beta(g_1)(F(\delta))
	\]
	and a collection of elements $\phi(g_1,g_2) \in \pi_1(\mathcal{H})$ for every $g_1,g_2\in \pi_0(\mathcal{G})$ satisfying 
	\begin{align*}
		F_1(k(g_1,g_2,g_3))& \phi(g_1 g_2, g_3) \phi(g_1,g_2) = \phi(g_1, g_2 g_3) \cdot \beta(F_0(g_1))(\phi(g_2,g_3)) \cdot  l(F_0(g_1), F_0(g_2), F_0(g_3))
	\end{align*}
	for all $g_1,g_2,g_3 \in \pi_0(\mathcal{G})$.
	
	A 2-morphism $\sigma: (F_0,F_1,\phi) \Longrightarrow (G_0,G_1,\chi)$ of skeletal 2-groups consists of an object $\sigma_* \in \pi_0(\mathcal{H})$ and a collection $\sigma(g) \in \pi_1(\mathcal{H})$ for every $g \in \pi_0(\mathcal{G})$ such that
	\begin{align*}
		l(G_0(g_1),& G_0(g_2), \sigma_*)\chi(g_1,g_2) \sigma(g_1 g_2) \\
		&= \beta(G_0(g_1))(\sigma(g_2)) \cdot l(G_0(g_1), \sigma_*, F_0(g_2)) \cdot \sigma(g_1) \cdot  l(\sigma_*, F_0(g_1), F_0(g_2)) \cdot \beta(\sigma_*)(\phi(g_1,g_2)) \ \ .
	\end{align*}
	The (vertical) composition of two 2-morphisms $\sigma: (F_0,F_1,\phi) \Longrightarrow (G_0,G_1,\chi)$ and $\tau: (G_0,G_1,\chi) \Longrightarrow (H_0,H_1,\psi)$ is given by $(\tau \circ \sigma)_* = \tau_* \cdot \sigma_*$ and $(\tau \circ \sigma)(g) = \tau(g) \cdot \beta(\tau_*)(\sigma(g))$.
\end{lemma}
\begin{proof}
	We first note that we can assume without loss of generality that the monoidal units are strictly preserved. 
	A monoidal functor between skeletal $2$-groups $F: (\mathcal{G},\alpha,k) \to (\mathcal{H},\beta,l)$ consists of a homomorphism $F_1: \pi_1(\mathcal{G}) \to \pi_1(\mathcal{H})$, a map $F_0: \pi_0(\mathcal{G}) \to \pi_0(\mathcal{H})$ and a collection of isomorphisms $\phi_{g_1,g_2}: F(g_1 g_2) \cong F(g_1) F(g_2)$ corresponding to $\phi(g_1,g_2) \in \pi_1(\mathcal{H})$ for every $g_1,g_2\in \pi_0(\mathcal{G})$. We can assume without loss of generality that
 $\phi_{1,g} = 1$ for all $g \in \pi_0(\mathcal{G})$
	so that $F_1$ is determined by $F_1(\gamma_g) = F_1(\gamma \otimes 1_g) = F_1(\gamma)_g$.
	The commuting diagram
	\[
	\begin{tikzcd}[column sep=1.7cm]
		F_0(g_1 \otimes g_2) \arrow[r,"F_1(\gamma_{g_1} \otimes \delta_{g_2})"] \arrow[d,"\phi_{g_1,g_2}"] & F_0(g_1 \otimes g_2) \arrow[d,"\phi_{g_1,g_2}"] 
		\\
		F_0(g_1) \otimes F_0(g_2) \arrow[r,"{F_1(\gamma_{g_1}) \otimes F_1(\delta_{g_2})}"] & F_0(g_1) \otimes F_0(g_2)
	\end{tikzcd}
	\]
	translates into formulas using the skeletal data as follows.
	First note that
	\[
	F_1(\gamma_{g_1} \otimes \delta_{g_2}) = F_1((\gamma \cdot \alpha(g_1)(\delta))_{g_1 g_2})
	\]
	and similarly
	\[
	F_1(\gamma_{g_1}) \otimes F_1(\delta_{g_2}) =F_1(\gamma)_{g_1} \otimes F_1(\delta)_{g_2} = (F_1(\gamma) \cdot \beta(g_1)(F_1(\delta)))_{g_1 g_2}
	\]
	We arrive at the following equation in $\pi_1(\mathcal{H})$
	\[
	\phi(g_1,g_2) F_1(\gamma \cdot \alpha(g_1)(\delta)) = F_1(\gamma) \beta(g_1)(F_1(\delta)) \phi(g_1,g_2).
	\]
	Since $\pi_1(\mathcal{H})$ is abelian we can get rid of the $\phi(g_1,g_2)$.
	
	The other diagram saying that $F$ is a monoidal functor looks as follows
	\[
	\begin{tikzcd}[column sep=2cm]
		(F(g_1) \otimes F(g_2)) \otimes F(g_3) \arrow[r,"l_{F(g_1),F(g_2),F(g_3)}"] \arrow[d,"{\phi_{g_1,g_2} \otimes 1_{F(g_3)}}"] & F(g_1) \otimes (F(g_2) \otimes F(g_3)) \arrow[d,"{1_{F(g_1)} \otimes \phi_{g_2,g_3}}"] 
		\\
		F(g_1 \otimes g_2) \otimes F(g_3) \arrow[d,"{\phi_{g_1 g_2, g_3}}"] & F(g_1) \otimes F(g_2 \otimes g_3) \arrow[d,"{\phi_{g_1,g_2 \otimes g_3}}"]
		\\
		F((g_1 \otimes g_2) \otimes g_3) \arrow[r,"{F(k_{g_1,g_2,g_3})}"] & F(g_1 \otimes (g_2 \otimes g_3))
	\end{tikzcd}
	\]
	Noticing that $1_{F(g_1)} \otimes \phi_{g_2,g_3} = \beta(F(g_1))(\phi_{g_2,g_3})$ this finishes the proof of what morphisms between skeletal 2-groups look like.
	The expressions for $2$- and $3$-morphisms follow directly from the relevant diagrams in the definition above.
\end{proof}

\subsection{Actions and homotopy fixed points} 
\label{Sec:actions and hfps}

This part of the appendix contains some details on the definition of homotopy fixed points for 
actions of 2-groups on bicategories essentially following~\cite{HSV}. The motivation for this is that in the main part of this paper we want to compute 
homotopy fixed points of actions of topological groups on bicategories. As we 
explain now in all cases we are interested in this is equivalent to computing
fixed points for the action a 2-group action.  
To define actions of a topological group $G$ on a bicategory we proceed as follows. We 
denote by $BG$ the classifying space of $G$ which we consider as an $\infty$-groupoid. A $G$-action on
a bicategory can now be described by a 3-functor $\rho \colon BG \longrightarrow \BiCat$ which will factor
through the fundamental 3-groupoid $\Pi_3(BG)$. Since we are only interested in compact Lie groups (and hence groups with $\pi_2(G)=0$) in this paper we can replace $\Pi_3(BG)$ with $\Pi_2(BG)$. $\Pi_2(BG)$ is the delooping of a 2-group $\Pi_{\leq 1} (G)$.

We now explicitly spell out the data corresponding to the action of a 2-group $\mathcal{G}$ on a bicategory.
Since $B\mathcal{G}$ has one object, 	we can describe an action $\rho \colon B\mathcal{G} \longrightarrow \BiCat$ with $\rho(*) = \mathcal{B}$ equivalently as a monoidal functor $\mathcal{G} \to \Aut \mathcal{B}$.
Here the domain has monoidal structure given by composition of functors and for the composition of $1$-morphisms in $\Aut \mathcal{B}$ we have fixed a choice of horizontal composition of natural transformations.

\begin{definition}[see also Remark 3.8~\cite{HSV}]
	\label{def:action}
	Let $\mathcal{B}$ be a bicategory and $(\mathcal{G}, \otimes, \alpha_\mathcal{G}, 1 )$ a 2-group. An \emph{action} of $\mathcal{G}$
	on $\mathcal{B}$ consists of a functor $\rho: \mathcal{G} \to \Aut \mathcal{B}$ between bicategories together with monoidality data, which is a collection of natural isomorphisms $R_{g',g}:\rho(g' \otimes g)  \overset{\sim}{\Longrightarrow} \rho(g') \circ \rho(g)$ natural in $g$ and $g'$.
	Explicitly, this is the following data:
	\begin{itemize}
		\item for every object $g \in \mathcal{G}$ a functor $\rho(g)\colon \mathcal{B}\longrightarrow \mathcal{B}$ 
		\item For every morphism $\gamma \colon g \longrightarrow g' \in \mathcal{G}$ a natural isomorphism 
		$\rho(\gamma)\colon \rho(g) \longrightarrow \rho(g')$. Concretely, those consists of 1-isomorphisms $\rho(\gamma)(b) \colon \rho(g)[b] \longrightarrow \rho(g')[b]$ for all $b\in \mathcal{B}$ and 2-isomorphisms
		\begin{equation}
			\begin{tikzcd}
				\rho(g)[b]  \ar[rr, " {{\rho(\gamma)}(b)}"]  \ar[dd,swap, "{ \rho(g)[f]}"] & & \rho(g')[b] \ar[dd, " {\rho(g')[f]}"] \ar[lldd, Rightarrow, "{\rho(\gamma)(f)}"] \\ 
				& &
				\\ 
				\rho(g') [b'] \ar[rr,swap, " {\rho(\gamma)(b')}"]  &  & \rho(g' )[b']
			\end{tikzcd}
		\end{equation}
		for all $f\colon b \longrightarrow b' \in \mathcal{B}$ satisfying a natural coherence condition.
		
		\item for every pair of composable morphisms $\gamma \colon g \longrightarrow g'$ and 
		$\gamma' \colon g' \longrightarrow g'' $ an invertible modification 
		\begin{equation}
			\begin{tikzcd}
				& \rho(g) \ar[dd, Rightarrow, "\rho (\gamma'\circ \gamma )"] \ar[ld,swap, Rightarrow, "\rho(\gamma)"] \\ 
				\rho(g') \ar[rd, swap, Rightarrow, "\rho(\gamma')"] \tarrow[ "{\alpha_{\gamma',\gamma}}"]{r} & \ \\ 
				& \rho(g'') 
			\end{tikzcd}
		\end{equation}

		\item for every pair of objects $g,g'\in \mathcal{G}$ natural isomorphisms $R_{g',g}\colon \rho(g'\otimes g) \longrightarrow \rho(g') \circ \rho(g)$. Concretely, those consists of 1-isomorphisms $R_{g',g}(b) \colon \rho(g'\otimes g)[b] \longrightarrow \rho(g') \circ \rho(g)[b]$ for all $b\in \mathcal{B}$ and 2-isomorphisms
		\begin{equation}
			\begin{tikzcd}
				\rho(g'\otimes g)[b] \ar[dd, " {\rho(g' \otimes g)[f]}"]  \ar[rr, " {R_{g',g}(b)}"] & &  \rho(g') \circ \rho(g)[b]    \ar[dd, "{\rho(g') \circ \rho(g)[f]}"] \ar[lldd, Rightarrow, "{R_{g',g}(f)}"]\\ 
				& &
				\\ 
				\rho(g' \otimes g)[b'] \ar[rr,swap, " {R_{g',g}(b')}"] &  & \rho(g') \circ \rho(g)[b'] 
			\end{tikzcd}
		\end{equation}
		for all $f\colon b \longrightarrow b' \in \mathcal{B}$ satisfying a natural coherence condition.
		
		\item the structure from the previous point is functorial in $g,g'$, i.e. 
		for every pair of morphisms $\gamma \colon g_1 \longrightarrow g_2$ and $\gamma'\colon
		g_1' \longrightarrow g_2'$ we have to specify a modification 
		\begin{equation}
			\begin{tikzcd}
				\rho(g_1'\otimes g_1)\ar[d, Rightarrow, "\rho(\gamma'\otimes \gamma)", swap]  \ar[r, Rightarrow, "R_{g_1',g_1}"] 
				& 
				\rho(g_1') \circ \rho(g_1) \ar[d, Rightarrow, "\rho(\gamma') \bullet \rho(\gamma)"] \tarrow[ "{R_{\gamma',\gamma}}"]{ld} 
				\\ 
				\rho(g_2'\otimes g_2) \ar[r, swap, Rightarrow, "R_{g_2',g_2}"] 
				& 
				\rho(g_2') \circ \rho(g_2)
			\end{tikzcd}
		\end{equation}
		where we used a bullet to denote horizontal composition.
		
		\item 
		for every triple of objects $g'',g',g \in \mathcal{G} $ invertible modifications 
		\begin{equation}
			\begin{tikzcd}
				& \rho (g'') \circ \rho (g') \circ \rho(g) 
				& \tarrow[ "{\omega_{g'',g',g}}"near end, shorten >=0.5cm,shorten <=2.5cm]{lldd} \\ 
				\rho(g'')\circ \rho(g'\otimes g) \ar[ur, Rightarrow, "{\id_{\rho(g'')}\bullet R_{g',g}}"]  & & \rho(g''\otimes g')\circ \rho(g)   \ar[ul, Rightarrow ,"{ R_{g'',g'} \bullet \id_{\rho(g)} }", swap] \\ 
				\rho(g''\otimes (g'\otimes g)) \ar[u, Rightarrow, "{R_{g'',g'\otimes g}}"] \ar[rr,swap, Rightarrow, "{\rho(\alpha_\mathcal{G}(g'',g',g))}"] &  &  \rho((g''\otimes g')\otimes g) \ar[u, Rightarrow,"{R_{g''\otimes g',g}}", swap]
			\end{tikzcd}
		\end{equation}
		satisfying a pentagon-type identity HTA1 \cite[Figure 2.1]{janthesis} ~\cite{streettricat}
	\end{itemize}
	such that identities are preserved strictly\footnote{This can always be assumed up to isomorphism.} 
\end{definition}

\begin{remark}
We will often use the fact that in a $2$-group, composition of morphisms is determined by the tensor product of morphisms.
This implies for example that the data of $\alpha$ in the above definition is redundant. 
\end{remark}

\begin{remark}\label{Rem: Action on Cat}
	The definition above can be specialized to the case of ordinary categories by considering them as bicategories with only identity 2-morphisms. Concretely an 
	action of a 2-group $(\mathcal{G},\otimes,1,\alpha)$ on a category $\mathcal{C}$ consists of 
	\begin{itemize}
		\item An equivalence $\rho(g)\colon \mathcal{C} \longrightarrow \mathcal{C}$ for all $g\in \Obj(\mathcal{G})$, 
		\item natural isomorphisms $R_{g',g}\colon \rho(g'\otimes g)\Longrightarrow \rho(g')\circ \rho(g)$ for all pairs of objects $g',g\in \Obj(\mathcal{G})$,
		\item and natural isomorphisms $\rho(\gamma)\colon \rho(g)\Longrightarrow \rho(g')$ for all morphisms $\gamma\colon g\to g'$ in $\mathcal{G}$
	\end{itemize} 
	such that the diagrams 
\begin{equation}
		\begin{tikzcd}
		& \rho(g) \ar[dd, Rightarrow, "\rho (\gamma'\circ \gamma )"] \ar[ld,swap, Rightarrow, "\rho(\gamma)"] \\ 
		\rho(g') \ar[rd, swap, Rightarrow, "\rho(\gamma')"] & \ \\ 
		& \rho(g'') 
	\end{tikzcd}
\ , \ \
		\begin{tikzcd}
			\rho(g_1'\otimes g_1)\ar[d, Rightarrow, "\rho(\gamma'\otimes \gamma)", swap]  \ar[r, Rightarrow, "R_{g_1',g_1}"] 
			& 
			\rho(g_1') \circ \rho(g_1) \ar[d, Rightarrow, "\rho(\gamma') \bullet \rho(\gamma)"] 
			\\ 
			\rho(g_2'\otimes g_2) \ar[r, swap, Rightarrow, "R_{g_2',g_2}"] 
			& 
			\rho(g_2') \circ \rho(g_2)
		\end{tikzcd} \ , 
	\end{equation}
	\begin{equation}
		\ \ \text{and} \ \ 
	\begin{tikzcd}
		& \rho (g'') \circ \rho (g') \circ \rho(g) 
		&  \\ 
		\rho(g'')\circ \rho(g'\otimes g) \ar[ur, Rightarrow, "{\id_{\rho(g'')}\bullet R_{g',g}}"]  & & \rho(g''\otimes g')\circ \rho(g)   \ar[ul, Rightarrow ,"{ R_{g'',g'} \bullet \id_{\rho(g)} }", swap] \\ 
		\rho(g''\otimes (g'\otimes g)) \ar[u, Rightarrow, "{R_{g'',g'\otimes g}}"] \ar[rr,swap, Rightarrow, "{\rho(\alpha_\mathcal{G}(g'',g',g))}"] &  &  \rho((g''\otimes g')\otimes g) \ar[u, Rightarrow,"{R_{g''\otimes g',g}}", swap]
	\end{tikzcd}
\end{equation}
commute and $\rho(1)=\id_{\mathcal{C}}$. Note that every action on a bicategory induces an action in this sense on the homotopy 1-category.
\end{remark}

Having defined group actions on bicategories we can define the corresponding notion
of homotopy fixed points. There exists a concise definition in terms of 
tricategories. Let $\rho \colon B\mathcal{G} \longrightarrow \catf{BiCat}$ describe
an action of $\mathcal{G}$ on a bicategory $\mathcal{B}$. A \emph{homotopy fixed 
	point} is a natural transformation 
\begin{equation}
	\begin{tikzcd}
		B\mathcal{G} \ar[rr, bend left=40, ""{name=U, below}, "{\rho}"] \ar[rr, bend right=40, swap, ""{name=T, above},  "{*}"] & & \catf{BiCat}
		\ar[Rightarrow, from=T, to=U, "F",swap]
	\end{tikzcd}
\end{equation}
from the constant functor at the bicategory $*$ with one object, and only identity 
morphisms. Following~\cite{HSV} we concretely spell out the definition in the 
case of a 2-group action. 

\begin{definition}\label{Def: Hfixed point}
	Let $(\mathcal{G}, \otimes, \alpha_{\mathcal{G}},1)$ be a 2-group with action $(\mathcal{B},\rho(-), R_{-,-}, \alpha, \omega ) $ on a bicategory $\mathcal{B}$. 
	A \emph{homotopy fixed point} consists of the following data
	\begin{itemize}
		\item an object $F\in \mathcal{B}$ 
		
		\item for all elements $g\in \mathcal{G}$ a 1-isomorphism $F_g \colon \rho(g)[F]\to F$
		
		\item for all pairs of elements $g,g'\in \mathcal{G} $ 2-isomorphisms 
		\begin{equation}
			\begin{tikzcd}
				\rho(g')\circ \rho(g)[F] \ar[rr, "{\rho(g')[F_g]}"]  & \ar[dd, Rightarrow, "\varphi_{g',g}", shorten >=8] & \rho(g')[F] \ar[dd, "{F_{g'}}"]  \\  
				& & \\ 
				\rho(g'\otimes g)[F] \ar[uu, "{R_{g',g}[F]}"] \ar[rr, "{F_{g'\otimes g}}", swap]  & \ &  F
			\end{tikzcd}
		\end{equation}
		\item 
		for all 1-morphisms $\gamma \colon g \longrightarrow g'$ in $\mathcal{G}$ 2-isomorphisms 
		\begin{equation}
			\begin{tikzcd}
				\rho(g)[F] \ar[dd, "{\rho(\gamma)[F]} ", swap] \ar[rd,"{F_g}"] &  \ar[ldd,Rightarrow, "F_\gamma "near end,shorten <=1.1cm, shorten >= .3cm,swap] \\ 
				\  &  F \\ 
				\rho(g')[F]  \ar[ru, "{F_{g'}}",swap ] & 
			\end{tikzcd}
		\end{equation}
	\end{itemize}
	such that all units are preserved strictly, i.e. $F_e \colon \rho(e)[F]=F \xrightarrow{\id_F} F$, $F_{\id_g}=\id_{F_g}$ and $\varphi_{g',g}$ is trivial as
	soon as one of the components is.
	Furthermore, the data has to satisfy the following two relations 
	\newpage
	\begin{equation}
		\begin{tikzcd}[row sep=1.2cm,column sep=2cm]
			\rho(g'')\circ \rho(g')\circ \rho(g)[F]  \ar[d, swap, "{\rho(g'')\circ \rho(g')[F_g]}"]  \ar[dr, Rightarrow, "R_{g'',g'}(F_g)",shorten <=10, shorten >=10]    & \rho(g''\otimes g') \circ \rho(g)[F]  \ar[l, "{R_{g'',g'}(\rho(g)[F])}", swap] \ar[d, "{\rho(g''\otimes g')F_g}"] & \rho((g''\otimes g') \otimes g)[F]\ar[dddll, bend left= 50, "F_{(g''\otimes g') \otimes g}"] \ar[l, "{R_{g''\otimes g',g}(F)}", swap] \\ 
			\rho(g'')\circ \rho(g')[F]  \ar[d, swap, "{\rho(g'')[F_{g'}]}"] & \rho(g''\otimes g')[F] \ar[ddl, bend left=20, "{F_{g''\otimes g' }}"] \ar[l, "R_{g'',g'}(F)"] \ar[rd, Rightarrow, "\varphi_{g''\otimes g',g}" near start, shorten >=50] & \\
			\rho(g'')[F] \ar[d, swap, "{F_{g''}}"] \ar[r, Rightarrow, "\varphi_{g'',g'}" near start, shorten >=60]  & \ & \ \\ 
			F & & 
		\end{tikzcd} 
	\end{equation}
	\[
	=
	\]
	\begin{equation}
		\begin{tikzcd}[row sep=1.2cm,column sep=2cm]
			\rho(g'')\circ \rho(g')\circ \rho(g)[F]    \ar[d, swap, "{\rho(g'')\circ \rho(g')[F_g]}"]  \ar[ddr,Rightarrow, "{\rho(g)[\varphi_{g',g}]}",shorten <=25, shorten >=70, swap]  & \rho(g''\otimes g') \circ \rho(g)[F]  \ar[l, "{R_{g'',g'}(\rho(g)[F])}", swap] & \rho((g''\otimes g') \otimes g)[F] \ar[l, "{R_{g''\otimes g',g}(F)}", swap] \ar[dddll, bend left= 145, "F_{(g''\otimes g') \otimes g}"] \\ 
			\rho(g'')\circ \rho(g')[F]  \ar[d, swap, "{\rho(g'')[F_{g'}]}"]&  \ar[lu, "{\rho(g'')[R_{g',g}(F)]}", swap] \ar[u, Rightarrow, "{\omega_{g'',g',g}(F)}", swap]  \rho(g'')\circ \rho(g'\otimes g)[F]  \ar[ld, "{\rho(g'')[F_{g'\otimes g}]}"]  & \rho(g''\otimes (g'\otimes g))[F] \ar[u, "{\rho(\alpha_\mathcal{G}(g'',g',g))(F)}"] \ar[l, "{R_{g'',g'\otimes g}(F)}"] \ar[lldd, "{F_{g''\otimes(g'\otimes g)}}" near end] \\
			\rho(g'')[F] \ar[d, swap, "{F_{g''}}"] \ar[r,Rightarrow, "{\varphi_{g'',g'\otimes g}}",shorten <=10, shorten >=10]  & \ \ar[rd, Rightarrow, "{F_{\alpha_\mathcal{G}(g'',g',g)}}" near start,shorten <=10, shorten >=90] & \ \\ 
			F & \ & \
		\end{tikzcd} 
	\label{Diagram A}
	\end{equation}
	for all $g'',g',g \in  \mathcal{G}$
	and 		
	\begin{equation}
		\begin{tikzcd}[row sep=1.2cm,column sep=1.2cm]
			\rho(g_2) \circ  \rho(g_1) [F] \ar[d,"{\rho(g_2)[F_{g_1}]}", swap] \ar[rrd,Rightarrow, " \varphi_{g_2,g_1}",shorten <=10, shorten >=60, near start, swap]     && \rho(g_2\otimes g_1 ) [F] \ar[ll, "R_{g_2,g_1}(F)"] \ar[rr, "{\rho(\gamma_2\otimes \gamma_1)[F]}"] \ar[lldd,"F_{g_2\otimes g_1}"]   && \rho(g_2'\otimes g_1')[F] \ar[lllldd, "F_{g_2'\otimes g_1'}", bend left=25]&& \\ 
			\rho(g_2)[F] \ar[d, "F_{g_2}",swap] & \ar[rrru,Rightarrow, " F_{\gamma_2\otimes \gamma_1}",shorten <=25, shorten >=10, swap]& \ && \ \\ 
			F  && \ && \ \
		\end{tikzcd}
	\end{equation}
	\[
	= \]
	\begin{equation}
		\begin{tikzcd}[row sep=1.2cm,column sep=2cm]
			\rho(g_2) \circ  \rho(g_1) [F] \ar[rd, "{\rho(g_2)[\rho(\gamma_1)(F)] }"] \ar[d,"{\rho(g_2)[F_{g_1}]}", swap]    & \rho(g_2\otimes g_1 ) [F] \ar[r, "{\rho(\gamma_2\otimes \gamma_1)[F]}"] \ar[ld,Leftarrow, " {\rho(g_2)[F_{\gamma_1}]}" near end,shorten <=60, shorten >=10,swap]  \ar[l, swap, "R_{g_2,g_1}(F)"] & \rho(g_2'\otimes g_1')[F] \ar[rdd, "F_{g_2'\otimes g_1'}", bend left=45] \ar[d, "{R_{g_2',g_1'}(F)}"] & \\ 
			\rho(g_2)[F] \ar[d, "F_{g_2}",swap] \ar[rd, " {\rho(\gamma_2)[F]} ", swap] \ar[rr,Rightarrow, "{\rho(\gamma_2)(F_{g_1'})}", shorten <=50, shorten >=30, bend right=14]  & \ar[ru,Rightarrow, " {R_{\gamma_2,\gamma_1}(F)}",shorten <=20, shorten >=20] \ar[l, "{\rho(g_2)[F_{g_1'}]}"] \rho(g_2)\circ \rho(g_1')[F]  \ar[r, "{\rho(\gamma_2)(\rho(g_1')[F])}"]  \ar[ld,Leftarrow, " F_{\gamma_2}" near end,shorten <=80, shorten >=10,swap]  & \rho(g_2')\circ \rho(g_1')[F]  \ar[ld, "{\rho(g_2')[F_{g_1'}]}" ] \\ 
			F  \ar[rrr, "\id", bend right=20,swap] &  \ar[l, "F_{g_2'}"] \rho(g_2')[F] \ar[rr,Rightarrow, " \varphi_{g_2',g_1'}",shorten <=40, shorten >=40]  & &  F \\ 
			& & & 
		\end{tikzcd}
	\label{Diagram B}
	\end{equation}	
	for all morphisms $\gamma_2\colon g_2 \to g_2'$ and $\gamma_1\colon g_1 \to g_1'$ in $\mathcal{G}$. 
\end{definition}
As a direct consequence of the definition we find 
\begin{proposition}
	Let $(\mathcal{G}, \otimes, \alpha_\mathcal{G},1)$ be a 2-group. Homotopy fixed points for the trivial
	action on a bicategory $\mathcal{B}$ are objects $b\in \mathcal{B}$ equipped with $\mathcal{G}$-action.
\end{proposition}  
\begin{remark}
	General homotopy fixed points can be understood as actions twisted by the $\mathcal{G}$-action
	on the bicategory. Indeed, let $G$ be an ordinary discrete group which acts
	on the 1-category $\catf{Vect}$ through a 2-cocycle $\alpha \in C^2(G;\C^\times)$. Homotopy 
	fixed points for this action are exactly projective representations of $G$ with twist $\alpha$. 
\end{remark}

Homotopy fixed points naturally form a bicategory $\catf{HFP}(\rho)$: 
\begin{definition}
\label{def:1-mor of hfps}
	Let $(\mathcal{G}, \otimes, \alpha,1)$ be a 2-group with action $(\mathcal{B},\rho(-), R_{-,-}, \alpha, \omega ) $ on a bicategory $\mathcal{B}$.
	A \emph{1-morphism between homotopy fixed points} $(F,F_g,F_\gamma, \varphi_{g',g})$ and $(F',F'_g,F'_\gamma, \varphi'_{g',g})$ consists of 
	\begin{itemize}
		\item a 1-morphism $f\colon F\to F'$ in $\mathcal{B}$, 
		\item for every $g \in \mathcal{G}$ a 2-isomorphism 
		\begin{equation}
			\begin{tikzcd}[row sep=1.0cm,column sep=1.2cm]
				\rho(g)[F'] \ar[r, "{\rho(g)[f]}"] \ar[d,"F_g",swap] & \rho(g)[F'] \ar[d,"F'_g"] \ar[ld, Leftarrow, "f_g", swap, shorten <=10, shorten >=10] \\ 
				F \ar[r,"f",swap] & F'
			\end{tikzcd}
		\end{equation}
	\end{itemize}
	compatible with identities and such that 
	\begin{equation}
		\begin{tikzcd}[row sep=1.0cm,column sep=1.2cm]
			\rho(g')\circ \rho(g)[F] \ar[d,"{\rho(g')[F_g]}",swap] \ar[rd, Rightarrow, "{\varphi_{g',g}}" near start,shorten <=10, shorten >=30,swap] \ar[rdd, Rightarrow, "{f_{g'\otimes g}}" near end,shorten <=70, shorten >=10,swap] & \rho(g'\otimes g)[F]  \ar[l, "{R_{g',g}[F]}", swap] \ar[ldd, "{F_{g'\otimes g}}",swap] \ar[d, "{\rho(g'\otimes g)[f]}"]  \\ 
			\rho(g')[F] \ar[d,"{F_{g'}}",swap] & \rho(g'\otimes g)[F'] \ar[d,"{F'_{g'\otimes g}}"] \\ 
			F \ar[r,swap, "{f}"] & F'
		\end{tikzcd} 
		= 
		\begin{tikzcd}[row sep=1.0cm,column sep=1.4cm]
			\rho(g')\circ \rho(g)[F] \ar[d,"{\rho(g')[F_g]}",swap]  \ar[r, "{\rho(g')\circ \rho(g)[f]}"]  & \rho(g')\circ \rho(g)[F'] \ar[d,"{\rho(g')[F'_g]}"]  \ar[r, Rightarrow,shorten <=10, shorten >=10, "{R_{g',g}[f]}"] & \rho(g'\otimes g)[F]  \ar[d, "{\rho(g'\otimes g)[f]}"]  \ar[ll, "{R_{g',g}[F]}", bend right=30] \\ 
			\rho(g')[F] \ar[d,"F_{g'}",swap ] \ar[r, "{\rho(g')[f]}", swap] \ar[ru, Rightarrow,shorten <=10, shorten >=10, "{\rho(g')[f_g]}"] & \rho(g')[F'] \ar[rd, "{F'_{g'}}", swap] \ar[d, Leftarrow, "{f_{g'}}",shorten <=5, shorten >=5] \ar[r, Rightarrow,shorten <=10, shorten >=10, "{\varphi'_{g',g}}"] & \rho(g'\otimes g)[F'] \ar[d,"F'_{g'\otimes g}"] \ar[lu, "{R_{g',g}[F']}", swap] \\ 
			F \ar[rr,swap, "{f}"] & \ & F'
		\end{tikzcd} 
	\end{equation}
	for all $g',g \in \mathcal{G}$ and 
	\begin{equation}
		\begin{tikzcd}[column sep=1.1cm]
			\rho(g)[F]   \ar[r, "{\rho(\gamma)[F]}"] \ar[dd, swap, "{F_g}"] & \rho(g')[F] \ar[r, "{\rho(g')[f]}"] \ar[ldd, "{F_{g'}}"]& \rho(g')[F'] \ar[dd, "{F'_{g'}}"] \\
			\ \ar[ru, Rightarrow, shorten <=5, shorten >=5, "{F_{\gamma}}" ]& & \\ 
			F \ar[rr, swap, "f"] &\ \ar[ruu, Rightarrow, shorten <=5, shorten >=5, "{f_{g'}}" ] & F' 
		\end{tikzcd}
		=
		\begin{tikzcd}[column sep=1.4cm]
			\rho(g)[F] \ar[rd, "{\rho(g)[f]}",swap]  \ar[r, "{\rho(\gamma)[F]}"] \ar[dd, swap, "{F_g}"] & \rho(g')[F] \ar[r, "{\rho(g')[f]}"] & \rho(g')[F'] \ar[dd, "F'_{g'}"] \\
			\ & \rho(g)[F']\ar[dr,"F'_g"] \ar[ru, "{\rho(\gamma)[F']}"] \ar[u, Rightarrow,  "{\rho(\gamma)[f]}" ] \ar[r,Rightarrow, "{F'_\gamma}",shorten <=5, shorten >=5] & \ \\ 
			F \ar[rr, swap, "f"] \ar[ru, Rightarrow, shorten <=5, shorten >=5, "f_g" ]  &  & F' 
		\end{tikzcd}
	\end{equation} 
	for all $\gamma \colon g\to g'$.
		Composition of $1$-morphisms	$(F,F_g,F_\gamma, \varphi_{g',g}) \xrightarrow{(f,f_g)} (F',F'_g,F'_\gamma, \varphi'_{g',g})$ and $(F',F'_g,F'_\gamma, \varphi_{g',g}') \xrightarrow{(f',f'_g)} (F'',F''_g,F''_\gamma, \varphi''_{g',g})$ is given by $f' \circ f$ together with the $2$-isomorphism
		\begin{equation}
			\begin{tikzcd}[row sep=1.0cm,column sep=1.2cm]
				\rho(g)[F'] \ar[r, "{\rho(g)[f]}"] \ar[d,"F_g",swap] & \rho(g)[F'] \ar[d,"F'_g"] \ar[ld, Leftarrow, "f_g", swap, shorten <=10, shorten >=10] \ar[r, "{\rho(g)[f']}"] & \rho(g)[F''] \ar[d,"F''_g"] \ar[ld, Leftarrow, "f'_g", swap, shorten <=10, shorten >=10] \\ 
				F \ar[r,"f",swap] & F' \arrow[r,"f'", swap] & F''
			\end{tikzcd}
		\end{equation}

\end{definition}
\begin{definition}
	Let $(\mathcal{G}, \otimes, \alpha_\mathcal{G},1)$ be a 2-group with action $(\mathcal{B},\rho(-), R_{-,-}, \alpha, \omega ) $ on a bicategory $\mathcal{B}$. 
	Furthermore let $(f,f_g)$ and $(f',f'_g)$ be 1-morphisms $(F,F_g,F_\gamma, \varphi_{g',g})\to (F',F'_g,F'_\gamma, \varphi'_{g',g}) $ of homotopy fixed points. A \emph{2-morphism of homotopy fixed points} consists of a 2-morphism 
	\begin{equation}
		\begin{tikzcd}
			&  \arrow[dd,Rightarrow,"{\Gamma}", shorten <=10, shorten >=10,swap]  & \\
			F \arrow[rr, bend left=45, ""{name=U, below},"{f}"] \arrow[rr, bend right=45, ""{name=D},"{f'}",swap]& & F' \\ 
			& \ & 
		\end{tikzcd} 
	\end{equation}
	such that 
	\begin{equation}
		\begin{tikzcd}[row sep=1.4cm]
			& \ar[d, Rightarrow, "{\rho(g)[\Gamma]}" near end, shorten <=25, shorten >=5,swap] & \\
			\rho(g)[F] \ar[rrd, Rightarrow, "f'_g" , shorten <=5, shorten >=8,swap] \ar[d, "F_g", swap]  \ar[rr,"{\rho(g)[f]}",bend left =50] \ar[rr,"{\rho(g)[f']}",swap] & \ & \rho(g)[F'] \ar[d, "{F'_g}"] \\ 
			F \ar[rr, "f'"]  &  & F' \ 
		\end{tikzcd} 
		=
		\begin{tikzcd}[row sep=1.4cm]
			\rho(g)[F] \ar[rrd, Rightarrow, "f_g", shorten <=5, shorten >=8,swap]   \ar[rr,"{\rho(g)[f]}"] \ar[d,"F_g",swap] & \ & \rho(g)[F'] \ar[d, "{F'_g}"] \\ 
			F\ar[rr, "f"]  \ar[rr,"{f'}",bend right =50,swap] & \ar[d, Rightarrow, "\Gamma" near start, shorten <=1, shorten >=25,swap] & F' \ \\
			& \ & 
		\end{tikzcd} 
	\end{equation}
	for all $g\in \mathcal{G}$. 

\end{definition}

\begin{remark}\label{Rem: FP on Cat}
As in Remark~\ref{Rem: Action on Cat} we can restrict the definition of a homotopy fixed point to ordinary categories. Let $(\mathcal{G},\otimes,1,\alpha)$ be a 2-group and $(\rho(-), R_{-,-})$ an action of $\mathcal{G}$ on a category $\mathcal{C}$. A homotopy fixed point for $\rho$ consists of an object $F\in \mathcal{C}$ together with isomorphisms $F_g\colon \rho(g)[F]\longrightarrow F$ for all objects $g\in \mathcal{G}$ such that the diagrams
		\begin{equation}
	\begin{tikzcd}
		\rho(g')\circ \rho(g)[F] \ar[rr, "{\rho(g')[F_g]}"]  &  & \rho(g')[F] \ar[dd, "{F_{g'}}"]  \\  
		& & \\ 
		\rho(g'\otimes g)[F] \ar[uu, "{R_{g',g}[F]}"] \ar[rr, "{F_{g'\otimes g}}", swap]  & \ &  F
	\end{tikzcd}
\text{  and  } \ \ 
	\begin{tikzcd}
		\rho(g)[F] \ar[dd, "{\rho(\gamma)[F]} ", swap] \ar[rd,"{F_g}"] &  \\ 
		\  &  F \\ 
		\rho(g')[F]  \ar[ru, "{F_{g'}}",swap ] & 
	\end{tikzcd}
\end{equation} 
commute in $\mathcal{C}$ and $F_1=\id_F$. 

A 1-morphism between homotopy fixed points is given by morphism $f\colon F\to F'$ in 
$\mathcal{C}$ such that $f \circ F_g=  F'_g \rho(g)[f]$. There are no non-trivial 2-morphisms between homotopy fixed points. Hence homotopy fixed points form a 1-category $\mathcal{C}^\mathcal{G}$.  
\end{remark}

\subsection{Exact sequences and semi-direct products of 2-groups}
\label{App: Exact} 

We now turn our attention to exact sequences of 2-groups and semi-direct products. An exact
sequence of discrete groups 
\begin{align}
	1\to N\xrightarrow{\ \iota \ } G \xrightarrow{\ \lambda \ } H \to 1
\end{align}  
consists of a surjective group homomorphism $\lambda$ with kernel $K\xrightarrow{\ \iota \ } G$.
To generalise the definition of an exact sequence to 2-groups we start with the following 
definition
\begin{definition}
	Let $\lambda: \mathscr{G} \to \mathscr{H}$ be a homomorphism of $2$-groups.
	Its \emph{kernel} $\lambda^{-1}[1]$ is the monoidal category with objects consisting of $g \in \operatorname{ob} \mathcal{G}$ together with an isomorphism $n_g: \lambda(g) \to 1$ (`nullifier')
	and morphisms $(\gamma: g_1 \to g_2) \in \operatorname{Mor} \mathcal{G}$ such that 
	\[
	\begin{tikzcd}
		\lambda(g_1) \arrow[rr,"\lambda(\gamma)"] \arrow[rd,"n_{g_1}", swap] && \lambda(g_2)\arrow[ld,"n_{g_2}"]
		\\
		& 1 &
	\end{tikzcd}
	\]
	commutes.
\end{definition}
The monoidal structure is given by $n_{g_1} \otimes n_{g_2}$ on objects, which is clearly a functor.
The associator between objects in $\ker \lambda$ is just the associator in $\mathscr{G}$.

Let 
\begin{align}
	1\to \mathcal{N} \xrightarrow{\ \iota \ } \mathscr{G} \xrightarrow{\ \lambda \ } \mathscr{H} \to 1
\end{align}  
be a sequence of 2-groups together with a pointed 2-isomorphism $\epsilon\colon \lambda \circ \iota \to 1$. The 2-isomorphism $\epsilon$ defines a morphism of 2-groups 
\begin{align}
	\iota^\epsilon \colon  \mathscr{N} & \to \lambda^{-1}[1] \\ 
	k & \longmapsto(\iota(n) , \lambda \circ \iota(n)\xrightarrow{\epsilon_n} 1)
\end{align} 

\begin{definition}
	An \emph{exact sequence of 2-groups} is a sequence of 2-groups 
	\begin{align}
		1\to \mathscr{N} \xrightarrow{\ \iota \ } \mathscr{G} \xrightarrow{\ \lambda \ } \mathscr{H} \to 1
	\end{align}  
	together with a pointed 2-isomorphism $\epsilon\colon \lambda \circ \iota \to 1$ such that
	the map $\iota^\epsilon$ is an equivalence of 2-groups.
\end{definition}

\begin{remark}
	We can equivalently describe an exact sequence of 2-groups by a sequence of 2-groupoids
	\begin{align}
		1\to B\mathscr{N} \xrightarrow{\ \iota \ } B\mathscr{G} \xrightarrow{\ \lambda \ } B\mathscr{H} \to 1
	\end{align} and a pointed 1-morphism $\epsilon\colon \lambda \circ \iota \to 1$
	such that the induced map $B\mathscr{N}\to \lambda/*$ is an equivalence.
	In this case we can conclude that the slice category $\lambda/*$ is connected and that there is a an
	equivalence of 2-groupoids $B \lambda^{-1}[1] \to \lambda/* $. Hence the canonical
	map of 2-groups $\mathcal{N} \to \lambda^{-1}[1]$ is also an equivalence.
\end{remark}
As for ordinary groups an exact sequence induces an action of $\mathscr{H}$ on $B \mathscr{N}$. The 
action can be abstractly\footnote{Even more abstractly it is the straightening of the
	fibration $\lambda\colon B\mathscr{G}\to B\mathscr{H}$.} constructed by noting that $\mathscr{H}$ acts
canonically on the slice $ \lambda / *$ by post composition. Since $\iota^\epsilon$
is an equivalence we can transfer the action to $B \mathscr{N}$. We describe the induced action
in the case that $\epsilon$ is the identity and $\lambda$ is a fibration. This assumption only streamlines the presentation and 
is enough for the extensions considered in this paper. However, it is straightforward to generalise 
to the case of non-trivial $\epsilon$'s. For this we describe
the bicategory $\lambda/*$ more concretely: It consists of 

\begin{itemize}
	\item \textbf{Objects:} Morphisms $h\colon \lambda(*) \to *$ in $B \mathscr{H}$.
	\item \textbf{1-Morphisms:} Consist of a pair of a
	1-morphism $g\colon *\to *$ in $\mathscr{G}$ and a 
	2-morphism
	\begin{equation}
		\begin{tikzcd}
		&	* \ar[dd,"\lambda(g)",swap] \ar[rd,"h"] &  \\
		& \	 & \ar[l, Rightarrow, "\gamma",swap ,shorten <=0, shorten >=0 ] * \\ 
		\	& * \ar[ru,"h'", swap] &
		\end{tikzcd}
	\end{equation}
	\item \textbf{2-Morphisms:} A 2-morphism $\Omega \colon (g,\gamma)\Longrightarrow (g',\gamma')$ is a 2-morphism 
	$\Omega \colon g \to g'$ in $B\mathscr{G}$ such that 
	\begin{equation}
		\begin{tikzcd}
			&	* \ar[dd,"\lambda(g)",near start] \ar[dd, "\lambda(g')", bend right=70, swap ] \ar[rd,"h"] & \\
			\ar[r, Leftarrow, "\lambda(\Omega)", near end, shorten <=15, shorten >=0 ] & \	\ar[r, Leftarrow, "\gamma" ] & * \\ 
			&	* \ar[ru,"h'", swap] &
		\end{tikzcd}
		=	\
		\begin{tikzcd}
			* \ar[dd,"\lambda(g)",swap] \ar[rd,"h"] & \\
			\ar[r, Leftarrow, "\gamma'" ] & * \\ 
			* \ar[ru,"h'", swap] &
		\end{tikzcd}
	\end{equation}
\end{itemize}  
with the obvious composition.
Because $\iota \colon B\mathscr{N} \to \lambda/*$ is essentially surjective there exists for every
element $\lambda(*)\xrightarrow{\ h \ } * $ a diagram
\begin{equation}
	\begin{tikzcd}
		* \ar[dd,"\lambda(s(h))",swap] \ar[rd,"h"] & \\
		& * \\ 
		* \ar[ru,"*", swap] &
	\end{tikzcd} \ \ .
\end{equation}   
with an element $s(h)\in \mathscr{G}$ which we can assume to commute strictly. We pick such an element 
for all $h\in \mathscr{H}$, i.e. we pick a set theoretical section at the level of objects. Being 
essentially surjective on the level of 1-morphisms implies that for all 1-morphisms $(n,\gamma)\colon 
(*,1)\to (*,1)$ there exists an element $\Gamma(\gamma)\in \mathscr{N}$ and a path $s(\gamma)\colon \Gamma(\gamma) \to n $ in $\mathscr{G}$ such 
that
\begin{equation}
	\begin{tikzcd}
		* \ar[dd,"\lambda(n)",swap] \ar[rd,"1"] & \\
		\ar[r, Leftarrow, "\gamma" ] & * \\ 
		* \ar[ru,"1", swap] &
	\end{tikzcd}
	=	\
	\begin{tikzcd}
		&	* \ar[dd,"\lambda(\Gamma(\gamma))", bend left=5] \ar[dd, "\lambda(n)", bend right=95, swap ] \ar[rd,"1"] & \\
		\ar[r, Leftarrow, "\lambda(s(\gamma))", near end, shorten <=10, shorten >=0 ] & \	& * \\ 
		&	* \ar[ru,"1", swap] &
	\end{tikzcd} \ . 
\end{equation}
We again pick these `lifts'. 
In particular, for every path $\gamma\colon h \to h'$ in $\mathscr{H}$ we can consider 
\begin{equation}
	\begin{tikzcd}\label{Eq: diagram gamma}
		* \ar[d, "\lambda(s(h')^{-1})", swap] \ar[rrdd, "1"]  & & \\ 
		* \ar[dd, "\id", swap] \ar[rrd, "h", near start] & & \\ 
		\ar[rr, Leftarrow, "\gamma"] & & * \\ 
		* \ar[d,"\lambda(s(h))",swap] \ar[rru, "h'",near start, swap] & & \ \\ 
		* \ar[rruu, "1",swap] & & 	  
	\end{tikzcd}
\end{equation} 
and get a path $s(\gamma) \colon \Gamma(\gamma)  \to s(h') \otimes s(h)^{-1} $ or equivalently a path
$s(\gamma) \colon s(h) \to \Gamma(\gamma)^{-1} s(h')$, i.e. a lift of $\gamma$ to a path in $\mathscr{G}$ 
which starts at $s(h)$, but does not end at $s(h')$. Finding a lift which also ends at $s(h')$ is 
impossible in general. In cases where the identity is a possible lift we always pick it. Based on these choices we can pick a specific inverse to $\iota \colon B \mathscr{N} \to \lambda /*$:
\begin{align}
	\iota^{-1} \colon \lambda /* & \to B \mathscr{N}, \ \ (h\colon \lambda(*)\to *) \longmapsto * \\ 
	(g,\gamma)\colon h \to h' & \longmapsto \Gamma (s(h') \otimes \gamma \otimes s(h)^{-1}) 
\end{align}  
where we denote $s(h') \otimes \gamma \otimes s(h)^{-1}$ the 1-morphism of a form similar to
Equation~\eqref{Eq: diagram gamma} constructed from $(g,\gamma)$.
The value on a 2-morphism $\Omega\colon (g_1,\gamma_2) \to (g_2,\gamma_2)$ can be constructed as the 2-morphism
\begin{align}
	\Gamma(\lambda(s(h')) \otimes \gamma_1 \otimes \lambda(s(h)^{-1}))  & \xrightarrow{s(\lambda(s(h')) \otimes \gamma_1 \otimes \lambda(s(h)^{-1})} s(h') \otimes g_1 \otimes s(h)^{-1} \xrightarrow{\Omega} s(h') \otimes g_2 \otimes s(h)^{-1} \\ & \xrightarrow{s(\lambda(s(h')) \otimes \gamma_1 \otimes \lambda(s(h)^{-1})^{-1}} \Gamma(\lambda(s(h')) \otimes \gamma_2 \otimes \lambda(s(h)^{-1}))
\end{align}
which is an 2-morphism in $B\mathscr{N}$ because $\iota$ is fully faithful on 2-morphisms. By construction we have $\iota^{-1}\circ \iota = \id$ and a canonical natural isomorphism $\iota\circ \iota^{-1} \to \id_{\lambda/*}$.
 
The action of an element $H\in \mathscr{H}$ on $B\mathscr{N}$ is now $\psi(H)\colon B\mathscr{N}\xrightarrow{\iota }
\lambda/* \xrightarrow{H \circ -} \lambda/* \xrightarrow{\iota^{-1} } B\mathscr{N} $ which concretely
maps a 1-morphism $n$ to $s(H)\otimes n \otimes s(H)^{-1}$ and a 2-morphism $\gamma_{\mathscr{N}}$ to 
the induced 2-morphism $s(H)\otimes \gamma_{\mathscr{N}} \otimes s(H)^{-1}$.

A path $\gamma_\mathscr{H}$ acts by the natural transformation $\psi(\gamma_H)\colon \psi(H) \to 
\psi(H')$ whose value at $*\in B\mathscr{N}$ is $\Gamma(\gamma)$. The data making this into a natural
transformation is induced by the 2-isomorphisms $s(\gamma)\colon \Gamma(\gamma)\Longrightarrow s(h') \otimes s(h)^{-1}$. 

From this description we can also work out the coherence isomorphism involving two elements $H,H'\in \mathscr{H}$. It is given abstractly by 
\begin{equation}
\begin{tikzcd}
B\mathscr{N}\ar[r,"\iota"] & \lambda/*  \ar[r, "H \circ -"] & \lambda/* \ar[r, "\iota^{-1}"] \ar[rr, "\id", bend right =40, swap] & B\mathscr{N} \ar[d, Rightarrow, shorten >= 5] \ar[r,"\iota"] & \lambda/*  \ar[r, "H' \circ -"] & \lambda/* \ar[r, "\iota^{-1}"] & B\mathscr{N} \\ 
& & & \ & & &
\end{tikzcd}
\end{equation}
The component of this natural isomorphism at the base point is
\begin{equation}
	\iota^{-1}\left( 
\begin{tikzcd}
	* \ar[rd, "1"] \ar[rrd, "H'", bend left =10] \ar[d,"\lambda(s(H)^{-1})", swap] & \\ 
	* \ar[r,"H", swap] & * \ar[r, "H'", swap] & *
\end{tikzcd}
\right) =  s(H'\circ H) \circ s(H)^{-1} \circ s(H')^{-1}
\end{equation}
measuring the failure of $s$ to be compatible with the composition in $B\mathscr{H}$.

All exact sequences of groups can be reconstructed from the action of $\mathscr{H}$ on $\mathscr{N}$ 
via the twisted semi-direct products. We now explain the analogous construction for 2-groups. 
Let $\mathscr{G}$ be a $2$-group equipped with an action on another $2$-group $\mathscr{N}$.
We are going to define the semi-direct product $2$-group explicitly.
Conceptually it is easiest to start with writing down what should be the (derived) coinvariants (i.e. homotopy quotients) for such an action.
We do this by thinking of $\mathscr{G}$ and $\mathscr{N}$ as bicategories and describe the semidirect product as the colimit of the functor  $\rho: B \mathscr{G} \to \BiCat$ between tricategories describing the action.
After this we derive expressions in terms of monoidal categories from those which are less conceptual but useful for computations.
By the coherence theorem for monoidal categories, we from now on omit the associators of $\mathcal{G}$ and $\mathcal{N}$ without loss of generality.
This simplifies the already tedious formulas and it is straightforward to generalize to nontrivial associators in practice.

Concretely, the bicategory $B(\mathcal{N} \rtimes \mathcal{G})$ will have 1-morphisms consisting of a 1-morphism $g: * \to *$ of $B\mathcal{G}$ together with an object $n: * = \rho(g)(*) \to *$ of $\mathcal{N}$.
The composition of $(n', g')$ and $(n,g)$ can be defined by taking the composition 
\[
\rho(g' \otimes g)[*] \xrightarrow{R_{g',g}(*)} \rho(g') [\rho(g)[*]] \xrightarrow{\rho(g')[n]} \rho(g')[*] \xrightarrow{n'} *
\] 
to be the first entry and $g' \otimes g$ be the second.
In analogy with the twisted semi-direct product of groups, we think of the multiplication in the semidirect $2$-group as being specified by the `condition' $g n g^{-1} = \rho(g)(n)$ together with the `twisting 2-cocycle' $R_{g',g}$.
2-morphisms from $(n, g)$ to $(n', g')$ consist of a 2-morphism $\gamma: g \to g'$ in $B\mathcal{G}$ and a 2-morphism in $B\mathcal{N}$ filling the diagram
\begin{equation}
	\begin{tikzcd}
		\rho(g)[*] \ar[dd, "{n} ", swap] \ar[rd,"{\rho(\gamma)_*}"] & \\ 
		\ar[r,Rightarrow, "\nu "near end,shorten <=0.4cm] &  \ar[dl, "{n'}"] 					\rho(g')[*] \\ 
		*  & 
	\end{tikzcd}
\end{equation}
The vertical composition of 2-morphisms $\gamma: g \to g', \nu:  n \to n' \otimes \rho(\gamma)_*$ and $\gamma': g' \to g'', \nu: n' \to n'' \otimes \rho(\gamma')_*$ is then given by 
\[
\begin{tikzcd}[column sep=1.7cm,row sep=1.2cm]
	\rho(g)[*] \ar[rr,"\rho(\gamma' \circ \gamma)_*"] \ar[dr,"\rho(\gamma)_*", swap] \ar[ddr,"n", bend right, swap] & \ & \rho(g'')[*] \ar[ddl, "n''", bend left] 
	\\
	\ar[r, "\nu", Rightarrow, shorten >=5, shorten <=18, swap] & \rho(g')[*] \ar[ru,"{\rho(\gamma')_*}"] \ar[d,"n'"] \ar[u,Rightarrow,"{\alpha_{\gamma',\gamma}}", shorten <=4 , shorten >=4]  & \
	\\
	\ & * \ar[uur, "\nu'",Rightarrow, shorten <=15, shorten >=15] & \
\end{tikzcd}
\]
For the horizontal composition between $\gamma_i: g_i \to g_i', \nu_i: n_i \to n_i' \otimes \rho(\gamma_i)_*$ for $i = 1,2$ we have to give a morphism 
\[
n_2 \otimes \rho(g_2)[n_1] \otimes R_{g_2,g_1}(*) \to n_2' \otimes \rho(g_2')[n_1'] \otimes R_{g_2',g_1'}(*) \otimes \rho(\gamma_2 \otimes \gamma_1)_*.
\]
Hence we seek to fill the diagram
\[
\begin{tikzcd}
	\rho(g_2 g_1)[*] \arrow[rr,"\rho(\gamma_2 \otimes \gamma_1)_*"] \arrow[d,"R_{g_2,g_2}"] & \ & \rho(g_2' g_1')[*] \arrow[d,"R_{g_2',g_1'}(*)"]
	\\
	\rho(g_2) \circ \rho(g_1)[*] \arrow[rr,"{(\rho(\gamma_2) \bullet \rho(\gamma_1) )_*}"]  \arrow[d,"{\rho(g_2)[n_1]}"]& \ &  
	\rho(g_2')[ \rho(g_1')[*] ]
	\arrow[d,"{\rho(g_2')[n_1']}"]
	\\
	\rho(g_2)[*] \arrow[rr,"\rho(\gamma_2)(*)"] \arrow[dr,"n_2"] & \ & \rho(g_2')[*] \arrow[dl,"n_2'"]
	\\
	\ & * & \
\end{tikzcd}
\]
The upper rectangle can be filled with $R_{\gamma_2,\gamma_1}(*)$ and the lower triangle by $\nu_2$.
After plugging in the definition of horizontal composition of natural transformations, the middle rectangle can be filled as follows:
\[
\begin{tikzcd}[column sep=1.7cm,row sep=1.2cm]
	\rho(g_2) \rho(g_1)(*) \arrow[r,"{\rho(g_2)(\rho(\gamma_1)_*)}"] \arrow[dd,"\rho(g_2)(n_1)"]
	& \rho(g_2) \rho(g_1')(*) \arrow[r,"\rho(\gamma_2)_{\rho(g_1')(*)}"] 
	\arrow[ddl, "\rho(g_2)(n_1')"] 
	&  \rho(g_2') \rho(g_1')(*) 
	\arrow[dd,"{\rho(g_2')(n_1')}"]  \arrow[ddll,"\rho(\gamma_2)_{n_1'}", Rightarrow, shorten <=10,  shorten >=10, bend left=5]
	\\
	{\ } \arrow[ur,Rightarrow,"{\rho(g_2)(\nu_1)}", shorten <=10,  shorten >=10] & \ & \
	\\
	\rho(g_2)(*) \arrow[rr,"\rho(\gamma_2)(*)"] & \ & \rho(g_2')(*) 
\end{tikzcd} \ \ .
\]
For the associator, let us consider the triple tensor product of $(n_3,g_3), (n_2,g_2)$ and $(n_1,g_1)$.
The two possible tensor product orders are the left- and rightmost composition of 1-morphisms in the following diagram:
\[
\begin{tikzcd}
	\ & \rho(g_3 g_2 g_1)[*] \arrow[dl, "{R_{g_3 g_2, g_1}(*)}"] \arrow[dr,"{R_{g_3,g_2 g_1}(*)}"]& \
	\\
	\rho(g_3 g_2) \rho(g_1)[*] \arrow[d, "{\rho(g_3 g_2)[n_1]}"] \arrow[drr, "{R_{g_3, g_2}(\rho(g_1)[*])}"] & \ & \rho(g_3) \rho(g_2 g_1)[*] 
	\arrow[d,"{\rho(g_3)(R_{g_2,g_1}(*))}"] 
	\arrow[ddd, "{\rho(g_3)[n_2 \rho(g_2)[n_1]R_{g_2,g_1}(*)]}", bend left=100]
	\\
	\rho(g_3 g_2)[*] \arrow[dr,"{R_{g_3,g_2}(*)}"] & \ & \rho(g_3) \rho(g_2) \rho(g_1)[*] \arrow[dl,"{\rho(g_3) [\rho(g_2)[n_1]]}"]
	\\
	\   & \rho(g_3) \rho(g_2)[*]  \arrow[dr,"{\rho(g_3)[n_2]}"] & \
	\\
	\  & & \rho(g_3)[*] \arrow[d,"{n_3}"] & \
	\\
	\ & & * & \
\end{tikzcd}
\]
We can fill the upper quadrilateral by $\omega_{g_3,g_1,g_1}(*)$ and the lower one by the naturality isomorphism $R_{g_3,g_2}(n_1: \rho(g)[*] \to *)$.
We can get to the rightmost bent arrow by using the monoidality data of the monoidal functor $\rho(g_3)$.

Translating the formulas above, we obtain an equivalent description of the semi-direct product as a monoidal category.

\begin{definition}
\label{Def:semidirectproduct}
	The \emph{semidirect product} $\mathcal{N} \rtimes \mathcal{G}$ is the monoidal category with objects pairs $(n, g)$ with $g \in \mathcal{G}$ and $n \in \mathcal{N}$ together with monoidal product
	\[
	(n', g') \otimes (n, g) = (n' \otimes \rho(g')(n) \otimes R_{g',g}(*), g \otimes g').
	\]
	A morphism $(n, g) \to (n', g')$ consist of a morphism $\gamma: g \to g'$ in $\mathcal{G}$ and a morphism $\nu:n \to n' \otimes \rho(\gamma)_*$ in $\mathcal{N}.$
	The composition of $(\nu, \gamma): (g, n) \to ( g', n')$ and $(\nu', \gamma'): (  n', g') \to ( n'', g'')$ is defined by the pair consisting of $\gamma' \circ \gamma$ and the composition
	\[
	n \xrightarrow{\nu} n' \otimes \rho(\gamma)_* \xrightarrow{\nu' \otimes \id_{\rho(\gamma)_*}} n'' \otimes \rho(\gamma')_* \otimes \rho(\gamma)_* \xrightarrow{\id_{n''} \otimes \alpha_{\gamma', \gamma}} n'' \otimes \rho(\gamma' \circ \gamma)_* .
	\]
	The horizontal composition of $(\nu_1, \gamma_1): (n_1,g_1) \to (n_1',g_1')$ and $(\nu_2, \gamma_2): (n_2,g_2) \to (n_2',g_2')$ is defined by the pair consisting of $\gamma_2 \otimes \gamma_1$ and the composition
	\begin{align*}
		& n_2 \otimes \rho(g_2)[n_1] \otimes R_{g_2,g_1}(*)
		\xrightarrow{\nu_2 \otimes \id_{\rho(g_2)[n_1] \otimes R_{g_2, g_1}(*)}} 
		\\
		&n_2' \otimes \rho(\gamma_2)_* \otimes \rho(g_2)[n_1] \otimes R_{g_2, g_1}(*)
		\xrightarrow{\id_{n_2' \otimes \rho(\gamma_2)_* } \otimes \rho(g_2)[\nu_1] \otimes \id_{R_{g_2,g_1}(*)}}
		\\
		& n_2' \otimes \rho(\gamma_2)_* \otimes \rho(g_2)[n_1'] \otimes \rho(g_2)[\rho(\gamma_1)_*]  \otimes R_{g_2, g_1}(*)
		\xrightarrow{\id_{n_2'} \otimes \rho(\gamma_2)_{n_1'} \otimes \id_{\rho(g_2)[\rho(\gamma_1)_*] \otimes R_{g_2,g_1}(*)}}
		\\
		& n_2' \otimes \rho(g_2')[n_1'] \otimes \rho(\gamma_2)_* \otimes \rho(g_2)[\rho(\gamma_1)_*]  \otimes R_{g_2, g_1}(*)
		\xrightarrow{\id_{n_2' \otimes \rho(g_2')[n_1']} \otimes R_{\gamma_2, \gamma_1}(*)}   n_2' \otimes \rho(g_2')[n_1'] \otimes R_{g_2',g_1'} \otimes \rho(\gamma_2 \otimes \gamma_1)_* 
	\end{align*}
	The associator of $(n_3,g_3), (n_2,g_2)$ and $(n_1,g_1)$ is the isomorphism
	\begin{align*}
		((n_3, g_3) \otimes (n_2, g_2)) \otimes (n_1,g_1) &= (n_3 \rho(g_3)[n_2] R_{g_3,g_2}(*) \rho(g_3 g_2)[n_1] R_{g_3 g_2, g_1}(*), g_3 g_2 g_1) \to
		\\
		(n_3, g_3) \otimes ((n_2, g_2) \otimes (n_1,g_1)) &= (n_3 \rho(g_3)[n_2 \rho(g_2)[n_1]R_{g_2,g_1}(*)] R_{g_3,g_2 g_1}(*), g_3 g_2 g_1)
	\end{align*}
	which is the identity on $g_3 g_2 g_1$ and given by the composition
	\begin{align*}
		& n_3 \rho(g_3)[n_2] R_{g_3,g_2}(*) \rho(g_3 g_2)[n_1] R_{g_3 g_2, g_1}(*) \xrightarrow{\id_{n_3 \rho(g_3)[n_2]} \otimes R_{g_3,g_2}(n_1) \otimes \id_{R_{g_3 g_2, g_1}(*)}}
		\\
		& n_3 \rho(g_3)[n_2] \rho(g_3)[\rho(g_2)[n_1]] R_{g_3,g_2}(*) R_{g_3 g_2, g_1}(*)
		\xrightarrow{\id_{n_3 \rho(g_3)[n_2] \rho(g_3)[\rho(g_2)[n_1]]} \otimes \omega_{g_3,g_2,g_1}(*)}
		\\
		& n_3 \rho(g_3)[n_2] \rho(g_3)[\rho(g_2)[n_1]] \rho(g_3)[R_{g_2,g_1}(*)] R_{g_3,g_2 g_1}(*) \cong 
		n_3 \rho(g_3)[n_2 \rho(g_2)[n_1]R_{g_2,g_1}(*)] R_{g_3,g_2 g_1}(*)
	\end{align*}
	on the $\mathcal{N}$ factor.
\end{definition}

\paragraph{Decomposing actions}
Every exact sequence of 2-groups $B\mathscr{N}\to B\mathscr{G} 
\to B\mathscr{H}$ corresponds to an action $\psi$ of 
$\mathscr{H}$ on $B\mathscr{N}$. In this case the bicategory $B\mathscr{G}$
can be reconstructed from the action as 
$\colim_{B\mathscr{H}} \psi$, which can be explicitly modelled by the semi direct product described in the previous section. The description as a colimit 
can be used to give a description of a $\mathscr{G}$-action on
a bicategory $\cat{B}$ in terms of $\mathscr{H}$ and $\mathscr{N}$. Using the universal property of the colimit 
\begin{align}
	\catf{TriCat}(\colim_{B\mathscr{H}}\psi, \BiCat ) \cong [B\mathscr{H}, \catf{TriCat}](\psi , \Delta \BiCat)
\end{align} 
where $\Delta \BiCat$ is the constant diagram at the 
tricategory $ \BiCat$. This implies that we can describe a
$\mathscr{G}$-action by a natural transformation
$\rho \colon \psi \Longrightarrow \Delta \BiCat$ of 4-functors. Concretely,
this consists of
\begin{itemize}
	\item A 3-functor $\rho_{\mathscr{N}} \colon B\mathscr{N}\longrightarrow \BiCat $, i.e. an action of
	$\mathscr{N}$ on a bicategory $\cat{B}$.  
	\item For all 1-morphisms $h\in B \mathscr{H}$ a invertible natural
	transformation $\rho_h \colon  \rho \circ \psi_h 
	\Longrightarrow \rho_h $. 
	\item For all pairs of 1-morphisms $h,h' \in B \mathscr{H}$ invertible modifications 
\end{itemize}
\begin{equation}
	\begin{tikzcd}
		& \ & \ & \ & \ & \ & \ & \ & \\
		&\	&\ &   \tarrow["\rho_{h,h'}",  shorten >=15,  shorten <=15]{rr} & \ & \ & \ & \ & & \\
		B\mathscr{N} \ar[r, "\psi_h",swap] \ar[rrr, bend left =90, "\rho_{\mathscr{N}}"]&	B \mathscr{N} \ar[r, "\psi_h'",swap] \ar[rr, bend left=60, "\rho_{\mathscr{N}}"] \ar[uu, Rightarrow, "\rho_{h}",  shorten >=5] & \ar[u, Rightarrow, "\rho_{h'}",  shorten >=0,swap] B \mathscr{N} \ar[r, "\rho_{\mathscr{N}}",swap] & \BiCat	& & 	B\mathscr{N} \ar[r, "\psi_{h}",swap] \ar[rrr, bend left =90, "\rho_{\mathscr{N}}"] \ar[rr, bend left=60, "\psi_{(h\otimes h')}"] &	B \mathscr{N} \ar[r, "\psi_{h'}",swap]  \ar[u, Rightarrow,   shorten >=0,swap] &  B \mathscr{N} \ar[r, "\rho"] \ar[uu, Rightarrow, "\psi_{h\otimes h'}",  shorten >=5,swap] & \BiCat
	\end{tikzcd}
\end{equation}
\begin{itemize}
	\item for all triples of 1-morphisms $h,h',h'' \in B\mathscr{H}$ invertible perturbations 
\end{itemize}
\begin{equation}
	\begin{tikzcd}
		& \ & \ & \ & \ & \ & \ & \ & \ \\
		&\	&\ &  \  & \tarrow["\rho_{h',h''}"]{r} \ & \ & \ & \ & \ & \ \\
		B\mathscr{N} \ar[rrrr, bend left=70,"\rho_{\mathscr{N}}"] \ar[r, "\psi_h",swap] & B\mathscr{N} \ar[r, "\psi_h'",swap] \ar[rrr, bend left =55, "\rho_{\mathscr{N}}"] \ar[uu, Rightarrow, "\rho_{h}",  shorten >=0]& \tarrow["\rho_{h, h'}"]{dd}	B \mathscr{N} \ar[r, "\psi_h''",swap] \ar[rr, bend left=35, "\rho_{\mathscr{N}}"] \ar[uu, Rightarrow, "\rho_{h'}",  shorten >=15] & \ar[u, Rightarrow, "\rho_{h''}" near start,  shorten >=5,swap] B \mathscr{N} \ar[r, "\rho_{\mathscr{N}}",swap] & \BiCat	& B\mathscr{N} \ar[rrrr, bend left=70,"\rho_{\mathscr{N}}"] \ar[r, "\psi_h", swap] & 	B\mathscr{N} \ar[uu, Rightarrow, "\rho_{h}",  shorten >=0] \ar[r, "\psi_{h'}",swap] \ar[rrr, bend left =60, "\rho_{\mathscr{N}}"] \ar[rr, bend left=40, "\psi_{(h'\otimes h'')}"] & \tarrow["\rho_{h, h' \otimes h''}"]{dd}	B \mathscr{N} \ar[r, "\psi_{h''}",swap]  \ar[u, Rightarrow,   shorten >=5,swap] &  B \mathscr{N} \ar[r, "\rho"] \ar[uu, Rightarrow, "\psi_{h'\otimes h''}" near start,  shorten >=15,swap] & \BiCat
		& \ & \ & \ & \ & \ & \ & \ & \ \\
		& \ & \ & \ & \  & \ & \ & \ & \ \\
		& \ & \ & \ & \ \qarrow["\rho_{h,h',h''}"]{rr} & \ & \ & \ & \ \\ 
		& \ & \ & \ & \ & \ & \ & \ & \ \\
		&\	&\ &  \  &  \ & \ & \ & \ & \ & \ \\
		B\mathscr{N} \ar[rrrr, bend left=70,"\rho_{\mathscr{N}}"] \ar[rr, "\psi_{h\otimes h'}", bend left=35] \ar[r, "\psi_h",swap] &  B\mathscr{N} \ar[r, "\psi_h'",swap]  \ar[u, Rightarrow,   shorten >=5]& \tarrow["\rho_{h\otimes h',h''}"]{dd}	B \mathscr{N} \ar[r, "\psi_h''",swap] \ar[rr, bend left=40, "\rho_{\mathscr{N}}"] \ar[uuu, Rightarrow, "\rho_{h\otimes h'}",  shorten >=15] & \ar[u, Rightarrow, "\rho_{h''}" near start,  shorten >=5,swap] B \mathscr{N} \ar[r, "\rho_{\mathscr{N}}",swap] & \BiCat	& B\mathscr{N} \ar[rrrr, bend left=70,"\rho_{\mathscr{N}}"] \ar[rrr, bend left=50,"\psi_{h\otimes (h' \otimes h'')}"] \ar[r, "\psi_h", swap] & 	B\mathscr{N} \ar[uu, Rightarrow,  shorten >=20] \ar[r, "\psi_{h'}",swap]  \ar[rr, bend left=40, "\psi_{(h\otimes h')}"] &	B \mathscr{N} \ar[r, "\psi_{h''}",swap]  \ar[u, Rightarrow,   shorten >=5,swap] &  B \mathscr{N} \ar[r, "\rho"] \ar[uuu, Rightarrow, "\rho_{h\otimes(h' \otimes h'')}" near start,  shorten >=25,swap] & \BiCat
		& \ & \ & \ & \ & \ & \ & \ & \ \\
		& \ & \ & \ & \ & \ & \ & \ & \ \\
		& \ & \ & \ & \ & \ & \ & \ & \ \\ 
		& \ & \ & \ &  \ \tarrow[""]{rruu} & \  &  \ & \ & \ \\
		&\	&\ &  \  &  \ & \ & \ & \ & \ & \ \\
		B\mathscr{N} \ar[rrrr, bend left=70,"\rho_{\mathscr{N}}"] \ar[rr, "\psi_{h\otimes h'}", bend left=35]  \ar[rrr, "\psi_{(h\otimes h')\otimes h''}", bend left=50] \ar[r, "\psi_h",swap] & B\mathscr{N} \ar[r, "\psi_h'",swap]  \ar[u, Rightarrow,   shorten >=5]&	B \mathscr{N} \ar[r, "\psi_h''",swap]  \ar[uu, Rightarrow,   shorten >=15] & \ar[uu, Rightarrow, "\rho_{(h\otimes h')\otimes h''}" near start,  shorten >=0,swap] B \mathscr{N} \ar[r, "\rho_{\mathscr{N}}",swap] & \BiCat	
	\end{tikzcd} \ \ ,
\end{equation}
where the unlabelled morphisms are part of the action of $\mathscr{H}$ on $B\mathscr{N}$
\begin{itemize}
	\item For all 2-morphisms $\gamma \colon h \to h'$ 
	modifications
\end{itemize} 
\begin{equation}
	\begin{tikzcd}
		& \ &  \\
		B\mathscr{N} \ar[r,"\psi_h",swap] \ar[rr,"{\rho_{\mathscr{N}}}", bend left=60] & \ar[u, Rightarrow, "\rho_h"] B \mathscr{N} \ar[r,"\rho_{\mathscr{N}}",swap] & \BiCat
	\end{tikzcd}
	\begin{tikzcd}
		\tarrow["\rho_\gamma"]{r}  & \
	\end{tikzcd}
	\begin{tikzcd}
		& \ & \ & \ \\
		B\mathscr{N} \ar[rr,"\psi_h",swap] \ar[rr,"\psi_{h'}",bend left=40] \ar[rrr,"{\rho_{\mathscr{N}}}", bend left=60] & \ar[u,Rightarrow, "\psi_{\gamma}" near start, shorten >=6] & \ar[ru, Rightarrow, "\rho_{h'}" near start, shorten >= 14] B \mathscr{N} \ar[r,"\rho_{\mathscr{N}}",swap] & \BiCat
	\end{tikzcd} 
\end{equation}
\begin{itemize}
	
	\item for all pairs of compossible 2-morphisms $\gamma_1\colon h \to h'$ and $\gamma_2 \colon h'\to h''$ perturbation 
	\begin{equation}
		\begin{tikzcd}
			& \ & \ & \ & \ & \ & \  \\
			B\mathscr{N} \ar[r,"\psi_h",swap] \ar[rr,"{\rho_{\mathscr{N}}}", bend left=60] & \ar[u, Rightarrow, "\rho_h"] B \mathscr{N} \tarrow["\rho_{\gamma' \circ \gamma}",swap]{dd} \ar[r,"\rho_{\mathscr{N}}",swap] & \BiCat  \tarrow["\rho_\gamma"]{r} &  B\mathscr{N} \ar[rr,"\psi_h",swap] \ar[rr,"\psi_{h'}",bend left=40] \ar[rrr,"{\rho_{\mathscr{N}}}", bend left=60] & \ar[u,Rightarrow, "\psi_{\gamma}" near start, shorten >=6] & \ar[ru, Rightarrow, "\rho_{h'}" near start, shorten >= 14] B \mathscr{N} \ar[r,"\rho_{\mathscr{N}}",swap] & \BiCat \tarrow["\rho_{\gamma'}",swap]{dd} \qarrow["\rho_{\gamma,\gamma'}", shorten <= 40]{ddllll} \\ 
			& \ & \ & \ & \ & \ & \  \\
			& \ & \ & \ & \ & \ & \  \\
			& \ & \ & \ & \ & \ & \ & \  \\
			B\mathscr{N} \ar[rr,"\psi_h",swap] \ar[rr,"\psi_{h''}",bend left=40] \ar[rrr,"{\rho_{\mathscr{N}}}", bend left=60] & \ar[u,Rightarrow, "{\psi_{ \gamma' \circ \gamma}}" near start, shorten >=6] & \ar[ru, Rightarrow, "\rho_{h'}" near start, shorten >= 14] B \mathscr{N} \ar[r,"\rho_{\mathscr{N}}",swap] & \BiCat   & \tarrow["\psi_{\gamma,\gamma'}",swap]{l}	B\mathscr{N} \ar[rr,"\psi_h",swap] \ar[rr,bend left=20] \ar[rr,"\psi_{h''}" near end,bend left=45] \ar[rrr,"{\rho_{\mathscr{N}}}", bend left=60] & \ar[u,Rightarrow, "\psi_{\gamma}" near start, shorten >=15] \ar[u,Rightarrow, "\psi_{\gamma'}" near end, shorten <=16] & \ar[ru, Rightarrow, "\rho_{h'}" near start, shorten >= 14] B \mathscr{N} \ar[r,"\rho_{\mathscr{N}}",swap] & \BiCat
		\end{tikzcd}
	\end{equation}
	\item For all pairs of 2-morphisms $\gamma_1\colon h_1 \to 
	h'_1$ and $\gamma_2 \colon h_2 \to h_2'$ perturbations
\end{itemize} 
\begin{equation}
	\begin{tikzcd}
		& \ & \ & \ & \ & \ & \ & \ & \\
		&\	&\ &   \tarrow["\rho_{h_1,h_2}",  shorten >=15,  shorten <=15]{rr} & \ & \ & \ & \ & & \\
		B\mathscr{N} \ar[r, "\psi_{h_1}",swap] \ar[rrr, bend left =90, "\rho_{\mathscr{N}}"]& \tarrow[ "\rho_{\gamma_2} \circ  \rho_{\gamma_1}",  shorten >=15,  shorten <=15, swap]{ddd}	B \mathscr{N} \ar[r, "\psi_{h_2}",swap] \ar[rr, bend left=60, "\rho_{\mathscr{N}}"] \ar[uu, Rightarrow, "\rho_{h_1}",  shorten >=5] & \ar[u, Rightarrow, "\rho_{h_2}",  shorten >=0,swap] B \mathscr{N} \ar[r, "\rho_{\mathscr{N}}",swap] & \BiCat	& & 	B\mathscr{N} \ar[r, "\psi_{h_1}",swap] \ar[rrr, bend left =90, "\rho_{\mathscr{N}}"] \ar[rr, bend left=60, "\psi_{(h_1\otimes h_2)}"] &	B \mathscr{N} \ar[r, "\psi_{h_2}",swap]  \ar[u, Rightarrow,   shorten >=0,swap] \qarrow["\rho^{\otimes}_{\gamma_1,\gamma_2}", shorten <= 40]{ddllll} & \tarrow[ "\rho_{\gamma_2 \otimes \gamma_1} ",  shorten >=15,  shorten <=15]{ddd}  B \mathscr{N} \ar[r, "\rho"] \ar[uu, Rightarrow, "\psi_{h_1\otimes h_2}",  shorten >=5,swap] & \BiCat \\
		& \ & \ & \ & \ & \ & \ & \ & \\
		& \ & \ & \ & \ & \ & \ & \ & \\
		& \ & \ & \ & \ & \ & \ & \ & \\
		&\	&\ &   \tarrow[  shorten >=15,  shorten <=15]{rr} & \ & \ & \ & \ & & \\
		B\mathscr{N} \ar[r, "\psi_{h_1}",swap, ""{name=h1, above}] \ar[r, "\psi_{h_1'}", bend left=40, ""{name=h1', below}] \ar[rrr, bend left =90, "\rho_{\mathscr{N}}"]&	B \mathscr{N} \ar[r, "\psi_{h_2'}", bend left=40, ""{name=h2', below}] \ar[r, "\psi_{h_2}",swap, ""{name=h2,above}] \ar[rr, bend left=65, "\rho_{\mathscr{N}}"] \ar[uu, Rightarrow, "\rho_{h_1'}",  shorten >=5] & \ar[u, Rightarrow, "\rho_{h_2'}",  shorten >=0,swap] B \mathscr{N} \ar[r, "\rho_{\mathscr{N}}",swap] & \BiCat	& & 	B\mathscr{N} \ar[r, "\psi_{h_1}",swap] \ar[rrr, bend left =90, "\rho_{\mathscr{N}}"] \ar[rr, bend left=80, "\psi_{(h'_1\otimes h'_2)}"] \ar[rr, bend left=30, "\psi_{(h_1\otimes h_2)}"] &	B \mathscr{N} \ar[r, "\psi_{h_2}",swap]  \ar[u, Rightarrow,   shorten >=7,swap] \ar[uu, Rightarrow,   shorten >=15, shorten <=20,swap] &  B \mathscr{N} \ar[r, "\rho"] \ar[uu, Rightarrow, "\psi_{h_1\otimes h_2}",  shorten >=5,  shorten >=5,swap] & \BiCat
		\arrow[Rightarrow, from=h1, to=h1'] \arrow[Rightarrow, from=h2, to=h2']
	\end{tikzcd}
\end{equation}

satisfying a large list of coherence conditions. Since 
4-categories have a slightly shaky foundations we want to highlight that the data above can be used to explicitly construct the data for action of
$\mathscr{N}\rtimes_{\psi} \mathscr{H}$ and the coherence conditions will exactly
correspond to the requirement that this defines an action. 

We will now sketch how to construct the data starting with an exact sequence $B\mathscr{N}\to 
B\mathscr{G} \to B\mathscr{H}$ satisfying the same simplifying assumptions as above explicitly. 
Let $\rho_{\mathscr{G}}\colon B \mathscr{G} \to \BiCat$ be an  action of $\mathscr{G}$ on a bicategory. 
The corresponding representation of $\mathscr{N}$ is $\rho_{\mathscr{N}}=  \iota^* \rho_{\mathscr{G}}$.
To describe the other data we use the inverse $\iota^{-1}\colon \lambda/* \to 
B \mathscr{N}$ described above. Note that 
\begin{equation}
	\begin{tikzcd}
		B \mathscr{N} \ar[rr, "{\iota}"]  \ar[rd] &  & \ar[ld] \lambda / *	 \\ 
		& B \mathscr{G} &  
	\end{tikzcd}
\end{equation} 
commutes strictly and hence  
\begin{equation}
	\begin{tikzcd}
		\lambda/ * \ar[rr, "{\iota^{-1}}"]  \ar[rd] & \ar[d, Rightarrow, "\tilde{s}"] & \ar[ld] B \mathscr{N}	 \\ 
		& B \mathscr{G} &  
	\end{tikzcd}
\end{equation}
commutes up to a natural isomorphism $\tilde{s}$. The value of $\tilde{s}$ at an object $\lambda(*)\xrightarrow{h} *$ is $s(h)^{-1}$. Using $\tilde{s}$ we define for all $H\in \mathscr{H}$ the natural isomorphism
\begin{equation}
	\begin{tikzcd}
		B\mathscr{N} \ar[r,"\iota"] \ar[rrrrd, "\rho_{\mathscr{N}}", bend left=60] &  \lambda / * \ar[r, "H"] &  \lambda / * \ar[d] \ar[rd] & \ar[ld, "\tilde{s}" near end, Leftarrow, shorten <=30 ]	\\
		& & B \mathscr{N} \ar[r]  & B\mathcal{G} \ar[r, "\rho_{\mathscr{G}}"] & \BiCat \\ 	 
	\end{tikzcd}
\end{equation}
where the upper part commutes strictly. This defines $\rho_{H}$, which concretely is given by
$\cat{B} \xrightarrow{\rho_{\mathscr{G}}(s(H)^{-1})} \cat{B}$. 
For a 2-morphism $\gamma_{\mathscr{H}}\colon h \longrightarrow h'$ in $B \mathscr{N}$ part of the modification  
\begin{equation}
	\begin{tikzcd}
		& \ &  \\
		B\mathscr{N} \ar[r,"\psi_h",swap] \ar[rr,"{\rho_{\mathscr{N}}}", bend left=60] & \ar[u, Rightarrow, "\rho_h"] B \mathscr{N} \ar[r,"\rho_{\mathscr{N}}",swap] & \BiCat
	\end{tikzcd}
	\begin{tikzcd}
		\tarrow["\rho_\gamma"]{r}  & \
	\end{tikzcd}
	\begin{tikzcd}
		& \ & \ & \ \\
		B\mathscr{N} \ar[rr,"\psi_h",swap] \ar[rr,"\psi_{h'}",bend left=40] \ar[rrr,"{\rho_{\mathscr{N}}}", bend left=60] & \ar[u,Rightarrow, "\psi_{\gamma}" near start, shorten >=6] & \ar[ru, Rightarrow, "\rho_{h'}" near start, shorten >= 14] B \mathscr{N} \ar[r,"\rho_{\mathscr{N}}",swap] & \BiCat
	\end{tikzcd} 
\end{equation} is the natural transformation 
\begin{equation}
\begin{tikzcd}
B \ar[dr, "\rho(\Gamma(\gamma))", swap]\ar[rr, bend left = 0, "\rho(s(h)^{-1}) "]& \ar[d, Rightarrow] & B	\\ 
 & B \ar[ru, "\rho(s(h')^{-1})", swap] &
\end{tikzcd}
\end{equation}  
which is given by applying $\rho$ to the morphism $s(h)^{-1} \longrightarrow s(h'^{-1})\otimes \Gamma(\gamma)$ constructed from $s(\gamma)$. 
The remaining data can be constructed similarly.  

\paragraph{Decompositions of fixed point categories} 
In the final part of this section we give a formula for the decomposition
of categories of homotopy fixed points. We use Kan extension in a 3-categorical setting to derive our
results abstractly. It is not clear to us how well developed the theory of those is, even though we only 
need the most basic properties. For this reason 
we will also give a more explicit description of the results in Remark~\ref{Rem: Explicit form of action}. 

Let $1\longrightarrow \mathscr{N}\xrightarrow{\iota}  \mathscr{G} \xrightarrow{\lambda} \mathscr{H} \longrightarrow 1$  
be an exact sequence of 2-groups where we assume for simplicity that $\epsilon$ is the identity. Equivalently, this can 
be described by a fibre sequence $B\mathscr{N}\to B\mathscr{G} \to 
B\mathscr{H}$ in the $\infty$-category of topological spaces ($\infty$-groupoids) or 2-groupoids. Let $\rho \colon B\mathscr{G} \to \BiCat$ be an action of $\mathscr{G}$ on a bicategory $\mathcal{B}$. The bicategory of homotopy fixed points was 
defined as the limit $\mathcal{B}^\rho \coloneqq \lim 
\rho$. This limit computes the right Kan extension 
\begin{equation}
	\begin{tikzcd}
		B \mathscr{G} \ar[r, "\rho"] \ar[d, "*", swap] & \BiCat \\ 
		* \ar[ru, "\text{Ran}_*\rho", swap]
	\end{tikzcd}
\end{equation}   
Using the composition law for right Kan extensions we can compute the category of fixed points in two steps. 
First we consider the functor $\text{Ran}_\lambda \rho \colon B\mathscr{H} \to \BiCat$. This corresponds to a new bicategory
with an action of $\mathscr{H}$ such that its bicategory of 
$\mathscr{H}$-fixed points agrees with $\mathcal{B}^\rho$. 
Next we will identify this bicategory explicitly: By the limit 
formula for right Kan extensions we can identify the value 
of $\text{Ran}_\lambda \rho$ with the limit over the slice 
2-category $\lambda/*$. Since the sequence is exact the 
inclusion
$\iota \colon B\mathscr{N} \to \lambda / *$  
is an equivalence of bicategories. This in turn implies that $\text{Ran}_\lambda \rho  \simeq \lim_{B\mathscr{H}} \iota^* \rho \simeq \mathcal{B}^{\iota^* \rho}$. Combining the discussion so far we have proven the following proposition
\begin{proposition}\label{Prop: Fixed points}
	Let $1\longrightarrow \mathscr{N}\xrightarrow{\iota}  \mathscr{G} \xrightarrow{\lambda} \mathscr{H} \longrightarrow 1$  
	be an exact sequence of 2-groups and $\rho \colon B\mathscr{G} \to \BiCat $ an action of $\mathscr{G}$ on a bicategory $\mathcal{B}$. There is an induced $\mathscr{H}$-action $\rho_{\mathscr{H}}$ 
	on the bicategory of $\mathscr{N}$-fixed points $\mathcal{B}^{\iota^* \rho}$ such that its bicategory of fixed
	points is equivalent to the bicategory of $\rho $-fixed points.
\end{proposition}  
\begin{remark}\label{Rem: Explicit form of action}
	We describe the $\mathscr{H}$-action on $\mathcal{B}^\mathscr{N}$ explicitly. For this we first describe 
	the simpler action on $\lim_{ \lambda / *} \rho $. This is the 2-category of natural transformations of the 
	form
	\begin{equation}
		\begin{tikzcd}
			\lambda / * \ar[r] \ar[rr, bend right=60, "*", swap] & B\mathcal{G} \ar[r, "\rho"] & \BiCat \\ 	
			& \ar[u, Rightarrow] &
		\end{tikzcd} \ \ .
	\end{equation} 
	Such a natural transformation has an explicit description similar to the one given in Definition~\ref{Def: Hfixed point}. Its objects consists of 
	\begin{itemize}
		\item An object $F_h\in \mathcal{B}$ for every objects $*\xrightarrow{h} *$ in $ \lambda / *$ 
		\item For all 1-morphisms 
		\begin{equation}
			\begin{tikzcd}
				* \ar[dd,"g",swap] \ar[rd,"h"] & \\
				\ar[r, Rightarrow, "\omega" ] & * \\ 
				* \ar[ru,"h'", swap] &
			\end{tikzcd}
		\end{equation} a 1-morphism $F_{(g,\omega)}\colon \rho(g)[F_h] \to F_{h'}$. 
		\item 2-morphisms $F_{(g,\omega),(g',\omega')}$ implementing the compatibility with the composition 
		of 1-morphisms in $ \lambda / *$. 
		\item 	for all 2-morphisms $\Omega \colon (g',\omega') \longrightarrow (g, \omega)$ 2-isomorphisms 
		\begin{equation}
			\begin{tikzcd}
				\rho(g)[F_h] \ar[dd, "{\rho(\Omega)[F_h]} ", swap] \ar[rd,"{F_{(g,\omega)}}"] & \\ 
				\  & \ar[l,Rightarrow, "F_\Omega "near end,shorten <=0.4cm] F_{h'} \\ 
				\rho(g')[F_h] \ar[ru, "{F_{(g',\omega')}}",swap ] & 
			\end{tikzcd}
		\end{equation}
	\end{itemize} 
	satisfying similar relations to those in Definition~\ref{Def: Hfixed point}. One and two morphism have
	similar descriptions. Now the action of $\mathscr{H}$ on this category is induced by the functors 
	\begin{align}
		H\colon */\lambda & \longrightarrow \lambda / * \\
		*\xrightarrow{h}* & \longmapsto *\xrightarrow{h}*\xrightarrow{H}*
	\end{align} 
	for $H\in \mathscr{H}$. 
	
	To construct the action on the $\mathscr{N}$-fixed points we pick as always an inverse $\iota^{-1} \colon 
	\lambda / * \to B \mathscr{N}$. Recall, that we have the following natural equivalence
	\begin{equation}
		\begin{tikzcd}
			\lambda / * \ar[rr, "{\iota^{-1}}"]  \ar[rd] & \ar[d, Rightarrow, "\tilde{s}"] & \ar[ld]	B \mathscr{N} \\ 
			& B \mathscr{G} &  
		\end{tikzcd}
	\end{equation} 
	Now the equivalence between fixed point categories is 
	\begin{equation}
		\begin{tikzcd}
			\lambda / * \ar[r] \ar[rr, bend right=60, "*", swap] & B\mathcal{G} \ar[r, "\rho"] & \BiCat \\ 	
			& \ar[u, Rightarrow] &	
		\end{tikzcd} \longmapsto 
		\begin{tikzcd}
			B \mathcal{N} \ar[d] \ar[rd]	\\
			\lambda / * \ar[r] \ar[rr, bend right=60, "*", swap] & B\mathcal{G} \ar[r, "\rho"] & \BiCat \\ 	
			& \ar[u, Rightarrow] & 
		\end{tikzcd}
	\end{equation}   
	where the triangle on the right commutes strictly. The map in the other direction is 
	\begin{equation}
		\begin{tikzcd}
			B \mathscr{N} \ar[r] \ar[rr, bend right=60, "*", swap] & B\mathcal{G} \ar[r, "\rho"] & \BiCat \\ 	
			& \ar[u, Rightarrow] &	
		\end{tikzcd} \longmapsto 
		\begin{tikzcd}
			\lambda / * \ar[d] \ar[rd] & \ar[ld, "\tilde{s}" near end, Leftarrow, shorten <=30 ]	\\
			B \mathscr{N} \ar[r] \ar[rr, bend right=60, "*", swap] & B\mathcal{G} \ar[r, "\rho"] & \BiCat \\ 	
			& \ar[u, Rightarrow] & 
		\end{tikzcd}\label{Eq: map fixed points}
	\end{equation}   
	The $\mathscr{H}$-action on the bicategory of homotopy fixed points is the one transferred through these equivalences. Abstractly, 
	an element $H\in \mathscr{H}$ acts by 
	\begin{equation}\label{Eq: Action}
		\begin{tikzcd}
			B \mathscr{N} \ar[r] \ar[rr, bend right=60, "*", swap] & B\mathcal{G} \ar[r, "\rho"] & \BiCat \\ 	
			& \ar[u, Rightarrow] &	
		\end{tikzcd} \longmapsto 
		\begin{tikzcd}
			B\mathscr{N} \ar[r,"\iota"] \ar[rrrrd, "\rho_{\mathscr{N}}", bend left =60] &  \lambda / * \ar[r, "H^{-1}"] &  \lambda / * \ar[d] \ar[rd] & \ar[ld, "\tilde{s}" near end, Leftarrow, shorten <=30 ]	\\
			& & B \mathscr{N} \ar[r] \ar[rr, bend right=60, "*", swap] & B\mathcal{G} \ar[r, "\rho"] & \BiCat \\ 	
			& & & \ar[u, Rightarrow] & 
		\end{tikzcd}
	\end{equation}  
	This also makes the action of path $\gamma\colon H \to H'$ apparent.  
	
	Let us be more explicit on how the maps look. The image of an $\mathscr{N}$ fixed point $(F,F_n,F_{n,n'}, F_{\gamma_n})$ 
	under the map \ref{Eq: map fixed points} can be described as follows: Its underlying family of objects is $F_h=\rho(s(h)^{-1})[F]$, 
	the morphism $F_{g,\omega}\colon \rho(g) \rho(s(h)^{-1}) F \simeq \rho(g \otimes s(h)^{-1}) F \xrightarrow{\rho(s(g,\omega))} \rho(s(h')^{-1} \otimes \Gamma(g,\omega)) F \simeq \rho(s(h')^{-1}) \rho(\Gamma(g,\omega)) F   \xrightarrow{F_{\Gamma(g,\omega)}} \rho(s(h')^{-1})F $, and so one.
	Now the map~\ref{Eq: Action} just picks out the component $F_{H^{-1}}=\rho(H)[F]$ with the corresponding structure. For
	example the morphism $\rho(n).F_{H^{-1}} \to F_{H^{-1}}$ is the value of $F_{(n,\id)}$, where $(n,\id)$ is seen as a 1-morphism
	from $*\xrightarrow{H^{-1}}*$ to $*\xrightarrow{H^{-1}}*$. The action of path in $\mathscr{H}$ is now also easy 
	to read of. 
	
	We briefly comment on the coherence isomorphisms for this action. For this we note that
	\begin{align}
		H_\mathscr{N} \colon B \mathscr{N} \xrightarrow{\iota}  \lambda / * \xrightarrow{H}  \lambda / * \xrightarrow{\iota^{-1}} B \mathscr{N} \ \ 
	\end{align}
	is exactly the action of $\mathscr{H}$ on $B \mathscr{N}$ corresponding to the exact sequence
	which come with natural isomorphisms $H'_\mathscr{N} \circ H_\mathscr{N} \Longrightarrow (H'\otimes H)_\mathscr{N}$. 
	The natural isomorphism $\tilde{s}$ induces a natural transformation 
	\begin{equation}
		\begin{tikzcd}
			&\ & \\
			B \mathscr{N} \ar[r, "H_\mathscr{N}"] \ar[rr, bend left=60, "\rho"] & \ar[u, Rightarrow, "\rho_H",  shorten >=0,swap] B \mathscr{N} \ar[r, "\rho"] & \BiCat	
		\end{tikzcd}
	\end{equation}
	This is exactly part of the decomposition of the $\mathscr{G}$-action into $\mathscr{N}$ and
	$\mathscr{H}$ described above. Hence there are invertible modification $\omega$
	\begin{equation}
		\begin{tikzcd}
			& \ & \ & \ & \ & \ & \ & \ & \\
			&\	&\ &   \tarrow["\omega",  shorten >=15,  shorten <=15]{rr} & \ & \ & \ & \ & & \\
			B\mathscr{N} \ar[r, "H_\mathscr{N}'"] \ar[rrr, bend left =90, "\rho"]&	B \mathscr{N} \ar[r, "H_\mathscr{N}"] \ar[rr, bend left=60, "\rho"] \ar[uu, Rightarrow, "s",  shorten >=5] & \ar[u, Rightarrow, "s",  shorten >=0,swap] B \mathscr{N} \ar[r, "\rho"] & \BiCat	& & 	B\mathscr{N} \ar[r, "H_\mathscr{N}'"] \ar[rrr, bend left =90, "\rho"] \ar[rr, bend left=60, "(H\otimes H')_{\mathscr{N}}"] &	B \mathscr{N} \ar[r, "H_\mathscr{N}"]  \ar[u, Rightarrow,   shorten >=0,swap] &  B \mathscr{N} \ar[r, "\rho"] \ar[uu, Rightarrow, "s",  shorten >=5,swap] & \BiCat
		\end{tikzcd}
	\end{equation}
	which satisfies a natural coherences condition. 
	Using $\omega$ we can describe the coherence isomorphism concretely. It is given by
	\begin{equation}
		\begin{tikzcd}
			\	& \ & \ & \ & \ & \ & \ & \ & \ & \ \\
			&\	&\ &   &   \ & \ & \ & \ & \ & \ & \\
			B\mathscr{N} \ar[rrrr, bend right=60, "*",swap] \ar[r, "H_\mathscr{N}'"] \ar[rrrr, bend left =90, "\rho"]&	B \mathscr{N} \ar[r, "H_\mathscr{N}"] \ar[rrr, bend left=60, "\rho"] \ar[uu, Rightarrow, "s",  shorten >=5] & \ar[u, Rightarrow, "s",  shorten >=0,swap] B \mathscr{N} \ar[rr, bend right=40, "*",swap] \ar[rr, "\rho"] & \ & \BiCat \tarrow["\omega",  shorten >=1,  shorten <=1]{r}	& 	B\mathscr{N} \ar[r, "H_\mathscr{N}'"] \ar[rrrr, bend right=60, "*", swap] \ar[rrrr, bend left =90, "\rho"] \ar[rr, bend left=60, "(H\otimes H')_{\mathscr{N}}"] &	B \mathscr{N} \ar[r, "H_\mathscr{N}"]  \ar[u, Rightarrow,   shorten >=0,swap] &  B \mathscr{N} \ar[rr, "\rho"] \ar[rr, bend right =40, "*", swap] \ar[uu, Rightarrow, "s",  shorten <=5,swap] & \ &  \BiCat \\ 
			& \ & \ & \  \ar[u, Rightarrow, shorten <=5, "F"]  & \ & \ & \ & \ &  \ar[u, Rightarrow, shorten <=5, "F"] &
		\end{tikzcd}
	\end{equation}
	\begin{equation}
		=
		\begin{tikzcd}
			&\ & \ \\
			B \mathscr{N} \ar[rr, "(H \otimes H')_\mathscr{N}"] \ar[rrrr, bend left=35, "\rho"] \ar[rrrr, bend right=60, "*", swap] & & \ar[u, Rightarrow, "s",  shorten >=0,swap] B \mathscr{N} \ar[rr, "\rho"] \ar[rr,bend right=40, swap, "*"] & \ & \BiCat \\
			& & & \ar[u, Rightarrow, shorten <=5, "F"] & 	
		\end{tikzcd}
	\end{equation}
	where the last equality follows from the fact that there are no non-trivial natural transformation
	from $*$ to $*$. All the other coherence data can also be derived from the decomposition of the
	action.
\end{remark}

\section{Dualizability in symmetric monoidal bicategories}
\label{App:dualizable}

In this appendix, we will review dualizability in symmetric monoidal bicategories.
In the first subsection, we define adjoints of morphisms, duals of objects and fully dualizable objects.
We also discuss uniqueness of duals and the behavior of symmetric monoidal functors and natural transformations with respect to duals.
In the second subsection, we discuss the Serre automorphism and how it defines an $SO_2$-action on the core of the fully dualizable objects.
Next we provide the full $O_2$-action and its compatibility with symmetric monoidal actions of other $2$-groups.

\subsection{Duals and adjoints}
\label{Sec:duals}

We briefly list the necessary definitions on duals we will use in the next section and the main text.

\begin{definition}
Let $f: c_1 \to c_2$ be a $1$-morphism in a bicategory. 
A \emph{right adjoint} of $f$ is a $1$-morphism $f^R: c_2 \to c_1$ together with $2$-morphisms $\coev^R_f: \id_{c_1} \Rightarrow f^R \circ f$ and $\ev^R_f: f \circ f^R \Rightarrow \id_{c_2}$ satisfying the usual snake relations.
In that case we also call $f$ a \emph{left adjoint} of $f^R$.
\end{definition}

\begin{example}
The adjoint of a $1$-morphism in the bicategory of categories is an adjoint functor.
\end{example}

\begin{example}
Let $\mathcal{C}$ be a monoidal category seen as a bicategory $B \mathcal{C}$ with one object.
Then a morphism $c: * \to *$ in $B \mathcal{C}$ is the same as an object in $\mathcal{C}$.
A left adjoint of $c \in B \mathcal{C}$ is a left dual of $c \in \mathcal{C}$.\footnote{Some references call this a right dual instead.}
\end{example}

Adjoints have similar existence and uniqueness behavior to duals in monoidal categories.
For example, if a $1$-morphism has an adjoint, then two choices of adjoint are $2$-isomorphic and there is a unique such $2$-morphism preserving the unit and counit of the adjunction \cite[Proposition A.29]{schommerpriesthesis}.
Note that even though we will exclusively work in symmetric monoidal bicategories, the right and left adjoints of $1$-morphisms can be different.
Recall that a $1$-morphism $f: c_1 \to c_2$ is called a $1$-isomorphism or an equivalence if it admits a (pseudo-)inverse $g: c_2 \to c_1$, i.e. $f\circ g$ and $g \circ f$ are isomorphic to an identity $1$-morphism.
We will often use the following result. 
If $f$ is both an equivalence and admits a right adjoint, its unique up to isomorphism right adjoint $(f^R, \ev, \coev)$ realizes $f^R$ as the inverse of $f$, i.e. $\ev$ and $\coev$ are invertible. \cite[Proposition A.27]{schommerpriesthesis}.
If $f: c_1 \to c_2$ and $g: c_2 \to c_3$ admit right adjoints $f^R$ and $g^R$, then $f^R g^R$ is a right adjoint of $gf$ where the evaluation and coevaluation maps of $gf$ are induced by those of $f$ and $g$.

If $\phi: f \Longrightarrow g$ is a $2$-morphism between $1$-morphisms $c_1 \to c_2$ with specified right adjoints $f^R$ and $g^R$ and evaluation data, then there is an induced $2$-morphism $\phi^R: g^R \Longrightarrow f^R$.
Its existence and uniqueness are again similar to the existence and uniqueness of duals of morphisms between objects with specified duality data in a monoidal category.
More precisely, we can define $\phi^R$ by the explicit composition
\[
g^R = \id_{c_1} \circ g^R \overset{\coev^R_f}{\Longrightarrow} f^R \circ f \circ g^R \overset{\phi}{\Longrightarrow} f^R \circ g \circ g^R \overset{\ev^R_g}{\Longrightarrow} f^R \circ \id_{c_2} = f^R
\]
or equivalently as the unique $2$-morphism such that the square
\[
\begin{tikzcd}
f \circ g^R \arrow[r,"\phi", Rightarrow] \arrow[d,"\phi^R", Rightarrow] &  g \circ g^R \arrow[d,"{\ev_{g}^R}", Rightarrow] 
\\
f \circ f^R \arrow[r,"{\ev^R_{f}}", Rightarrow] & 1 
\end{tikzcd}
\]
commutes.

\begin{definition}
Let $\mathcal{C}$ be a symmetric monoidal bicategory.
A \emph{(left) dual} of an object $c$ consists of an object $c^\vee$, $1$-morphisms $\ev_c: c^\vee \otimes c \to 1$ and $\coev_c: 1 \to c \otimes c^\vee$ and so-called \emph{cusp $2$-isomorphisms} $\alpha$ and $\beta$ equating the the snake identities to the identity $1$-morphisms on $c^\vee$ and $c$ respectively. 
If $c$ admits a dual we call it \emph{1-dualizable}.
We call $c$ \emph{fully dualizable} or \emph{$2$-dualizable} if $\ev_c$ and $\coev_c$ both have right and left adjoints.
Let $\mathcal{C}^{\text{fd}} \subseteq \mathcal{C}$ denote the maximal sub $2$-groupoid on the fully dualizable objects.
Explicitely this means
\begin{itemize}
\item \textbf{Objects} are fully dualizable objects $c \in \mathcal{C}$;
\item \textbf{$1$-morphisms} between two fully dualizable objects are $1$-isomorphisms;
\item \textbf{$2$-morphisms} are $2$-isomorphisms between such $1$-isomorphisms.
\end{itemize}
\end{definition}

Note that an object of a symmetric monoidal bicategory is $1$-dualizable if and only it is dualizable in the homotopy $1$-category, but realizing two objects as each others duals in the bicategory is requires more data than realizing them as each others dual in the homotopy category.

\begin{remark}
There is also a notion of right dual.
Similarly to symmetric monoidal $1$-categories, left and right duals are canonically isomorphic using the symmetric braiding.
\end{remark}

\begin{remark} 
For $c$ to be fully dualizable, it suffices to require $\ev_c$ to have both adjoints.
It also suffices to require both $\ev_c$ and $\coev_c$ to have a left adjoint. \cite{piotr}
\end{remark}

\begin{remark}
For monoidal $1$-categories duals are unique in the sense that if $(c^\vee, \ev_c, \coev_c), (c'^\vee, \ev'_c, \coev'_c)$ are duality data  for the object $c$, then there is a unique isomorphism $s:c'^\vee \cong c^\vee$ intertwining the evaluation and coevaluation maps.
Explicitly $s$ is given by the composition
\[
c'^\vee \xrightarrow{\id \otimes \coev_c} c'^\vee \otimes c \otimes c^\vee \xrightarrow{\ev'_c \otimes \id} c^\vee
\]
with inverse
\[
c^\vee \xrightarrow{\id \otimes \coev'_c } c^\vee \otimes c \otimes c'^\vee \xrightarrow{\ev_c \otimes \id } c'^\vee.
\]
There is a similar more technical statement for $2$-categories for which we refer to \cite{piotr} for details. 
If $c$ is $1$-dualizable, then it always admits so-called coherent duality data, see \cite[Figure 3, Figure 4, Corollary 2.8]{piotr}.
Moreover, if $(c^\vee, \ev_c,\coev_c,\alpha,\beta)$ and $(c'^\vee, \ev'_c,\coev'_c,\alpha',\beta')$ are two coherent duality data, then there is a invertible $1$-morphism $s: c^\vee \to c'^\vee$ and two $2$-morphisms filling the diagrams
\[
\begin{tikzcd}
1 \arrow[r,"{\coev_c'}"] \arrow[rd,"{\coev_c}", swap] & c \otimes c'^\vee & \ & c'^\vee \otimes c \arrow[r,"{\ev'_c}"] & 1
\\
\ & c \otimes c^\vee \arrow[u,"1 \otimes s", swap] & \ & c^\vee \otimes c \arrow[u,"s \otimes 1"] \arrow[ru,"\ev_c", swap]
\end{tikzcd}
\]
satisfying two compatibility conditions with the $\alpha$s and $\beta$s.
Moreover, the choice of comparison data is unique up to a choice of $2$-morphism preserving the fillings of the two triangles above \cite[Lemma 2.18]{piotr}.
For example, $s$ could be given by the analogous expression as the one for monoidal $1$-categories above.
However, in this category level $s$ is only determined up to canonical $2$-isomorphism, so in practice we might want to choose a different $s$ and equip it with data preserving evaluations and coevaluations instead.
\end{remark}

A similar reasoning applies to the case of duals of $1$-morphisms in a bicategory; we could define them directly by a formula:
\begin{definition}
If $c_1,c_2$ are $1$-dualizable equipped with duality data and $f: c_1 \to c_2$  is a $1$-morphism, its \emph{dual} is defined to be the composition
\[
c_2^\vee \xrightarrow{\id \otimes \coev_{c_1}} c_2^\vee \otimes c_1 \otimes c_1^\vee \xrightarrow{\id \otimes f \otimes \id} c_2^\vee \otimes c_2 \otimes c_1^\vee \xrightarrow{\ev_{c_2} \otimes \id} c_1^\vee
\]
\end{definition}
This assembles into an equivalence of categories \cite[Lemma 2.5]{piotr}
\[
\Hom(c_1, c_2) \cong \Hom(c_2^\vee, c_1^\vee)
\]

However, it can be more natural to define them by the universal property, which in the $1$-categorical case uniquely specifies them, but in the $2$-categorical case only up to canonical $2$-isomorphism.
More specifically, if $c_1,c_2$ are $1$-dualizable equipped with duality data and $f: c_1 \to c_2$  is a $1$-morphism, we could define a dual of $f$ to consists of a $1$-morphism $f^\vee :c_2^\vee \to c_1^\vee$ together with $2$-isomorphisms filling the diagrams
\[
\begin{tikzcd}
c_2^\vee \otimes c_1 \arrow[r,"f"] \arrow[d,"g"] & c_2^\vee \otimes c_2 \arrow[d,"{\ev_{c_2}}"] & \ & c_2 \otimes c_1^\vee   & c_2 \otimes c_2^\vee \arrow[l,"f^\vee",swap]
\\
c_1^\vee \otimes c_1 \arrow[r,"{\ev_{c_1}}",swap] & 1 & \ & c_1 \otimes c^\vee_1 \arrow[u,"f"]  & \arrow[l,"{\coev_{c_1}}"] \arrow[u,"{\coev_{c_2}}",swap] 1
\end{tikzcd}
\]
satifying some natural coherence conditions.
This other definition of the dual is verified in \cite[Lemma 2.6]{janserre}.
However, since we will mostly be working with the inverse of the dual of $1$-morphisms, we will instead directly define $f^{\vee-1}:c_1^\vee \to c_2^\vee$ to be a dual inverse of $f$ when it comes equipped with filling of the same diagram but with the direction of $f^\vee$ reversed:
\[
\label{eq: dualinverse}
\begin{tikzcd}
\ & c_2^\vee \otimes c_2 \arrow[d,"{\ev_{c_2}}"] & \ & \   & c_2 \otimes c_2^\vee
\\
c_1^\vee \otimes c_1 \arrow[ur,"f^{\vee -1} \otimes f"] \arrow[r,"{\ev_{c_1}}",swap] & 1 & \ & c_1 \otimes c^\vee_1 \arrow[ur,"f \otimes f^{\vee -1}"]  & \arrow[l,"{\coev_{c_1}}"] \arrow[u,"{\coev_{c_2}}",swap] 1
\end{tikzcd}
\]
All statements we have seen of the form `choices of duals are essentially unique' can be succinctly summarized as follows. 
If $\mathcal{C}$ is a symmetric monoidal bicategory in which all objects admit duals, then there is a bicategory $\operatorname{CohDualPair} \mathcal{C}$ in which objects are cohorent dual pairs and $1$-morphisms are pairs consisting of a morphism and a choice of dual inverse.
The forgetful map from coherent dual pairs to the maximal sub 2-groupoid $\mathcal{C}^{\cong}$ of $\mathcal{C}$ specifying the underlying object
\[
\operatorname{CohDualPair} \mathcal{C} \to \mathcal{C}^{\cong}
\] 
is an equivalence of bicategories.
In practice often a canonical choice for duality data given an object of $c$ exists and otherwise it is useful to pick duality data.
Essentially all necessary choices involved in `picking duals coherently' can be obtained after choosing an inverse $s: \mathcal{C}^{\cong} \to \operatorname{CohDualPair} \mathcal{C}$ and the data specifying how $s$ is the inverse of the forgetful map.
We will from now on assume we have made this contractible choice.

We now briefly discuss how one would construct the $O_1$-action on the $2$-groupoid of dualizable objects coming from the cobordism hypothesis, as this does not seem to appear in the literature.
The functor underlying the action can be specified by picking for every object arbitrary duality data and for every invertible $1$-morphism an arbitrary choice of dual inverse as explained above.
This data can moreover be chosen functorially.
The desired canonical isomorphism $\hat{R}_c:= (R_{-,-})_{c}: c \cong c^{\vee\vee}$ in a symmetric monoidal bicategory specifying how the action squares to the identity is then a special case of uniqueness of duals.
Indeed, in a similar fashion to the $1$-categorical situation, one needs to realize that both are duals of $c^\vee$; $c$ is a right dual of $c^\vee$ instead, but this does not matter because of the symmetry.
By using the standard formula for the canonical morphism between two duals of the same object
we obtain the $1$-isomorphism $\hat{R}_c$ as
\[
c \xrightarrow{\id \otimes \coev_{c^\vee} } c \otimes c^\vee \otimes c^{\vee\vee} \xrightarrow{\sigma_{c,c^\vee} \otimes \id} c^\vee \otimes c \otimes c^{\vee\vee} \xrightarrow{\ev_c \otimes \id} c^{\vee\vee}.
\]
However, since $\hat{R}_c$ is only unique up to canonical $2$-isomorphism, it will be useful in practice to be able to pick it in a different way that preserves the duality data.
Note also that there is a $2$-isomorphism $\omega_c: \hat{R}_c^\vee \circ \hat{R}_{c^\vee} \cong \id_{c^\vee}$ uniquely given by the fact that both $1$-morphisms $\hat{R}_{c^\vee}^{-1}$ and $\hat{R}_c^\vee$ preserve the duality data and such isomorphisms are unique up to a canonical $2$-isomorphism.
Indeed, $\hat{R}_{c^\vee}$ is defined as the canonical $1$-isomorphism witnessing the fact that both $c^\vee$ and $c^{\vee\vee\vee}$ are duals of $c^{\vee \vee}$.

One convenient framework to pick the dual functor without making any choices is by working in the equivalent $2$-groupoid $\operatorname{CohDualPair} \mathcal{C}$.
Then the dual functor is most conveniently given on objects by
\[
(c,c^\vee,\ev,\coev,\alpha,\beta) \mapsto (c^\vee,c,\ev \circ \sigma_{c, c^\vee}, \sigma_{c^\vee,c} \circ \coev, \alpha', \beta')
\]
so that the double dual equals the identity by using the $2$-isomorphism $\sigma_{c^\vee,c} \circ \sigma_{c,c^\vee} \cong \id_{c \otimes c^\vee}$.
Because of the analogous well-understood situation in the homotopy $1$-category, we know that a choice of $\alpha'$ and $\beta'$ making the above a coherent dual pair exists and one should also be able to specify them explicitly.
We can then explicitly transport an $O_1$-action obtained in this way, along the equivalence by a choice of inverse $\mathcal{C} \to \operatorname{CohDualPair} \mathcal{C}$.
This gives a choice of dual $c^\vee$ for every object and a choice of inverse dual $f^{\vee-1}: c_1^\vee \to c_2^\vee$ for every invertible $1$-morphism $f: c_1 \to c_2$.
As such the section $s$ specifies for us a covariant functor $\mathcal{C} \to \mathcal{C}$ which is an inverse dual on $1$-morphisms.

We now sketch how symmetric monoidal functors preserve duals.
Let $F: \mathcal{C} \to \mathcal{C}$ be a symmetric monoidal functor between symmetric monoidal bicategories.
For simplicity assume all data is strictly unital and $\mathcal{C}$ has identity associator for the monoidal product.
Using that duals are unique, it suffices to make $F(c^\vee)$ into a dual of $F(c)$.
So let $(c, \coev_c,\ev_c, \alpha, \beta)$ be duality data for $c$.
The evaluation map will be the composition
\[
F(c^\vee) \otimes F(c) \to F(c^\vee \otimes c) \xrightarrow{F(\ev_c)} F(1) = 1
\]
and similar for the coevaluation.
The $2$-isomorphism specifying the snake identity is given by the filling of the diagram
\[
\begin{tikzcd}[column sep =80]
F(c) \arrow[dd, equals] \arrow[r,"{F(\coev_c) \otimes \id_{F(b)}}"] \arrow[rd,"{F(\coev_c \otimes \id_c)}", swap]& F(c \otimes c^\vee) \otimes F(c) \arrow[d] & \
\\
\ & F(c \otimes c^\vee \otimes c) \arrow[ld,"{F(\id_c \otimes \ev_c)}", swap] & F(c) \otimes F(c^\vee) \otimes F(c) \arrow[lu] \arrow[ld]
\\
F(c) & F(c) \otimes F(c^\vee \otimes c) \arrow[l,"{\id_{F(c)} \otimes F(\ev_c)}"] \arrow[u] & \
\end{tikzcd}
\]
The right part is filled by the associativity data of $F$ being monoidal, the left triangle is filled by $F(\beta)$ and the two remaining triangles are the naturality data of $F(c_1 \otimes c_2)\cong F(c_1) \otimes F(c_2)$ together with unitality.
There is an analogous filling of the other snake identity.
This shows $F(c^\vee)$ is a dual of $F(c)$. 

In particular, the $1$-isomorphism $F(c^\vee) \cong F(c)^\vee$ is given by a formula analogous to the $1$-categorical case
\[
F(c^\vee) \xrightarrow{\id \otimes \coev_{F(c)}} F(c^\vee) \otimes F(c) \otimes F(c)^\vee \cong F(c^\vee \otimes c) \otimes F(c)^\vee \xrightarrow{F(\ev_{c}) \otimes \id} F(c)^\vee
\]
The canonical $2$-isomorphism between the composition
\[
F(c) \xrightarrow{F(\phi_c)} F(c^{\vee\vee}) \cong F(c^\vee)^\vee \cong F(c)^{\vee\vee} \xrightarrow{\phi_{F(c)}^{-1}} F(c)
\]
and the identity is obtained by noting that we are comparing two $1$-isomorphisms that compare two choices of dual of $F(c)^\vee$ and such $1$-isomorphisms are unique up to canonical $2$-isomorphism.

\subsection{The Serre automorphism}
\label{Sec:serre}

In this section we will consider the Serre automorphism on fully dualizable objects of a symmetric monoidal bicategory $\mathcal{C}$.

\begin{definition}
If $c$ is fully dualizable in $\mathcal{C}$, its \emph{Serre automorphism} $S_c$ is defined to be the composition
\[
c \xrightarrow{\ev^R_c} c \otimes c \otimes c^\vee \xrightarrow{\sigma_{c,c}} c \otimes c \otimes c^\vee \xrightarrow{\ev_c} c.
\]
\end{definition}


\begin{remark}
The name `Serre automorphism' is motivated by the relationship with Serre duality, see \cite[Remark 4.2.4]{lurietft}.
\end{remark}

The following is \cite[Proposition 4.23]{janthesis}.

\begin{proposition}
\label{prop:serrenat}
Let $\mathcal{C}$ be a symmetric monoidal bicategory.
The Serre automorphism is a monoidal natural automorphism of the identity functor on the core of fully dualizable objects.
\end{proposition}

The Serre defines the $SO_2$-part of the $O_2$-action on $\mathcal{C}^{\operatorname{fd}}$ in the following sense.
An $SO_2$-action on a bicategory is specified by a single generating natural automorphism of the identity functor because $SO_2$ is a $B\Z$.
The $SO_2$-action on $\mathcal{C}^{\operatorname{fd}}$ coming from the cobordism hypothesis is the one induced by the Serre automorphism using the above proposition\cite{janserre} \cite[Example 2.4.14]{lurietft}.
In practice, this means that positive integers $n>0$ will act by
\[
\rho(n)c = \underbrace{S_c \circ \dots \circ S_c}_n.
\]
For negative integers $n<0$ we will have to pick an inverse $S_c^{-1}$ of $S_c$ (which is unique up to $2$-isomorphism) and define
\[
\rho(n)c = \underbrace{S_c^{-1} \circ \dots \circ S_c^{-1}}_{-n}.
\]
There is actually a canonical choice given by the composition
\[
c \xrightarrow{\ev^L_{c} \otimes \id} c^\vee \otimes c \otimes c \xrightarrow{\id \otimes \sigma_{c,c}} c^\vee \otimes c \otimes c \xrightarrow{\ev_c \otimes \id} c
\]
, see \cite[Section 1.3]{douglasSPsnyder}.
We will enhance this $SO_2$-action to an $O_2$-action in the next section.

\begin{proposition}
\label{prop:preserve serre}
Let $F: \mathcal{C} \to \mathcal{C}$ be a symmetric monoidal functor.
Assume every object is fully dualizable and let $S:\id_{\mathcal{C}} \Longrightarrow \id_{\mathcal{C}}$ denote the Serre isomorphism.
There is a natural modification $F(S_c) \Rrightarrow S_{F(c)}$
\end{proposition}
\begin{proof}
The $2$-isomorphism $F(S_c) \cong S_{F(c)}$ is uniquely specified by the fact that every symmetric monoidal functor preserves duals and adjoints of $1$-morphisms.
For example, recalling how the evaluation map is defined for realizing $F(c^\vee)$ as a dual of $F(c)$ and the uniqueness of duals isomorphism $F(c^\vee) \cong F(c)^\vee$ we get a filling of
\[
\begin{tikzcd}
1 \arrow[r,"{\ev^R_{F(c)}}"] \arrow[ddr,"{F(\ev_c^R)}", swap, bend right] & F(c)^\vee \otimes F(c) \arrow[d]
\\
\ & F(c^\vee) \otimes F(c) \arrow[d]
\\
\ & F(c^\vee \otimes c)
\end{tikzcd}
\]
by uniqueness of adjoints.
All relevant $2$-isomorphisms used in this definition define modifications.
\end{proof}


\subsection{The $O_2$-action on fully dualizable objects}
\label{Sec:O2-action}

If $\mathcal{C}$ is a bicategory in which every object is fully dualizable, the cobordism hypothesis says that there is a canonical $O_2$-action on $\mathcal{C}^{\operatorname{fd}}$, the full sub 2-groupoid on the fully dualizable objects of $\mathcal{C}$.
In this section we will spell out this action explicitly, the reader can compare with Section 16 and further of \cite{CSPdualizability}.
The $SO_2$-subgroup acts through the Serre automorphism of Section \ref{Sec:serre} and a reflection will act by the dual explained in Section \ref{Sec:duals}.
In the abstract setting of this section, we will therefore often assume $s: \mathcal{C}^{\cong} \to \operatorname{CohDualPair} \mathcal{C}$ is a specified choice of inverse of the forgetful map.

Because $O_2$ is a homotopy $1$-type, it is equal to its truncated 2-group $\pi_{\leq 1}(O_2)$.
We will use the skeletal $2$-group model, which we will now describe explicitly as the semidirect product $\Z_2 \rtimes B\Z$.
Namely, we have $\pi_0(O_2) \cong \Z/2, \pi_1(O_2) \cong \Z$ and any choice of reflection in the plane gives a splitting $O_2 \cong SO_2 \rtimes \Z/2$ of the exact sequence
\[
1 \to SO_2 \to O_2 \to O_1 \to 1
\]
so that the associator of the $2$-group is trivial.
In other words, the space $BO_2$ has trivial $k$-invariant.
After picking such a reflection, we get the following skeletal model for the 2-group. 
It has objects $\{+, -\}$ isomorphic to $\Z/2$ under $\otimes$. 
The morphisms are $\Hom(-,-) \cong \Hom(+,+) \cong \Z$ under composition and the other two morphism sets are empty.
We write morphisms as $n_+ \in \Hom(+,+)$ and $n_- \in \Hom(-,-)$ for integers $n$.
The monoidal structure on morphisms is defined by
\[
n_+ \otimes m_- = (n+m)_-, \quad m_- \otimes n_+ = (m-n)_-.
\]
This model can be explicitly realized as paths in $O_2 \subseteq \Aut S^1$ where $S^1 \subseteq \C$ is the unit circle as
\begin{align*}
n_+ \iff \gamma_{n_+}(t) z &=e^{int} z
\\
n_- \iff \gamma_{n_-}(t) z &=e^{int} \overline{z}
\end{align*} 
We now turn to how this 2-group acts on the core of fully dualizable objects.

Since $\rho(+) = \id$ by our assumptions on strictly preserving units, the first data we have to give is a functor $\rho(-): \mathcal{C}^{\operatorname{fd}} \to \mathcal{C}^{\operatorname{fd}}$ .
It is given by taking the dual on objects and morphisms
\begin{align*}
\rho(-) c = c^\vee \quad \rho(-)(f: c_1 \to c_2) = (f^{\vee})^{-1} \quad \rho(-)(\phi: f_1 \to f_2) = \phi^{\vee}
\end{align*}
Note that if we have a preferred choice of duals of $1$-morphisms, we have to make a choice of an inverse for every 1-morphisms to make the functor covariant, which is allowed because we are working on the core.
This also requires choosing between $(f^{\vee})^{-1}$ and $(f^{-1})^\vee$, which are however canonically $2$-isomorphic.
In this abstract setting it is therefore convenient to not pick duals on $1$-morphisms, but dual inverses directly as explained in Section \ref{Sec:duals}.
Unfortunately, when working with this $O_1$-action in practice we have to face the fact that that the duals of $1$-morphisms are usually more canonically defined than the dual inverse.
We will therefore also often rewrite diagrams in Section \ref{Sec:actions and hfps} in forms where it is not necessary to pick inverses whenever possible.

The data of the dual functor preserving composition of $1$-morphisms is easily derived by using the diagrams \ref{eq: dualinverse}.
Next we have to give the four natural isomorphisms $R_{g_1,g_2}$ for $g_i = \pm$.
The only nontrivial isomorphism is
\[
\hat{R} := R_{-,-}: \rho(- \otimes - ) = \rho(+) = \id \Longrightarrow \rho(-) \circ \rho(-).
\]
It is given by the natural isomorphism to the double dual $\hat{R}(c): c^{\vee \vee} \to c$ specifying uniqueness of duals, as discussed in Section \ref{Sec:duals}. 
In particular, if $f: c_1 \to c_2$ is a 1-morphism, $\hat{R}(f)$ is the canonical filling of the diagram
\[
\begin{tikzcd}
c_1 \arrow[r,"f"] \arrow[d, swap,"\hat{R}(c_1)"] & c_2 \arrow[d,"\hat{R}(c_2)"]
\\
c_1^{\vee \vee} \arrow[r,"f^{\vee \vee}"] & c_2^{\vee \vee}
\end{tikzcd}
\]
Heuristically this diagram is the analogue of the $1$-categorical fact that $f^{\vee \vee} = f$ under the identification $c^{\vee \vee} \cong c$.
The modifications $\omega_{g_1,g_2,g_3}$ are trivial except $\omega:= \omega_{-,-,-}$.
After removing the inverse on $\rho(-)(\hat{R}(c)) = \hat{R}(c)^{\vee -1}$, the result is a collection of 2-isomorphisms
\[
\omega_c: \hat{R}(c)^\vee \circ \hat{R}(c^\vee) \Rrightarrow \id_{c^\vee}
\]
specifying the agreement between the two identifications between $c^{\vee \vee \vee}$ and $c^{\vee}$.

We will now combine the above $O_1$-action with the $SO_2$-action from Section \ref{Sec:serre} to obtain a $O_2$-action.
We start by giving the collection of natural isomorphisms $\rho(n_-): \rho(-) \Longrightarrow \rho(-)$ obtained by whiskering Serre automorphisms with the dual functor.
For $n>0$ the action will be 
\[
\rho(n_-)c = \underbrace{S_{c^\vee} \circ \dots \circ S_{c^\vee}}_n
\]
while for $n< 0$ 
\[
\rho(n_-)c = \underbrace{S_{c^\vee}^{-1} \circ \dots \circ S_{c^\vee}^{-1}}_{-n}.
\]
Since the Serre automorphism gives a natural automorphism of the identity functor, these are natural too.
By our choice of definition, $\alpha$ can be taken trivial up to the choice of invertibility data $S_c \circ S_c^{-1} \to \id_c$.

The $R_{\gamma',\gamma}$ are given by inserting identifications $(S_c)^\vee \cong S_{c^\vee}$ whenever necessary. Before moving on we briefly explain how to obtain the isomorphism $(S_c)^\vee \cong S_{c^\vee}$: For this note that $(-)^\vee$ is a symmetric
monoidal functor $\cat{C}\longrightarrow \cat{C}^{\otimes \op}$ where $\cat{C}^{\otimes \op}$ is the symmetric monoidal bicategory build from $\cat{C}$ 
by reversing the direction of composition and tensor product. Next recall that by Proposition~\ref{prop:preserve serre} this implies
that there is a natural 2-isomorphism $(S_c)^\vee \cong S^{\otimes \op }_{c^\vee}$ where the Serre automorphism on the right is computed in $C^{\otimes \op}$. Hence to arrive at the desired 2-isomorphism it is left to show that $S_x^{\otimes \op}\cong S_x$ for an arbitrary element $x\in \cat{C}^{\otimes \op}$. It is straightforward to verify that
the dual of $x$ in $\cat{C}^{\otimes \op}$ is also $x^\vee$ with evaluation 
$\ev_x^{\otimes \op} = \sigma \circ \ev_x^R$ and coevaluation $\coev_x^{\otimes \op} = \coev_x^R \circ \sigma $. The claim now becomes a short straightforward computation using the graphical calculus for symmetric monoidal bicategories and the fact that the right adjoint of a morphism in $\cat{C}^{\otimes \op }$ is the left adjoint of the corresponding morphism in $\cat{C}$   
    
For example, if $\gamma = n_+$ and $\gamma' = m_-$, then the relevant diagram applied to an object $c$ becomes
\[
\begin{tikzcd}
c^\vee \arrow[r, equal] \arrow[d,"S_{c^\vee}^m \circ ((S^n_c)^\vee)^{-1}",swap]& c^\vee \arrow[d,"S^{m-n}_{c^\vee}"]
\\
c^\vee \arrow[r,equal] & c^\vee
\end{tikzcd}
\]
which we can indeed fill.
Note how the isomorphism $S_c^\vee \cong S_{c^\vee}$ implies that the Serre and the dual functor assemble into an $O_2$-action and not an $SO_2 \times O_1$-action, since the $O_1$-action is by inverse dual on $1$-morphisms.

Next we will explain the fact that if a $2$-group $\mathcal{G}$ acts symmetric monoidally on $\mathcal{C}$ we have an induced $O_2 \times \mathcal{G}$-action on fully dualizable objects.
On the level of data, this entails 
\begin{enumerate}
\item specifying the natural transformation between the symmetric monoidal functors $\rho(g)$ and the dual functor;
\item specifying the modification intertwining the double dual natural isomorphism $\phi_c$ with $\rho(g)$;
\item specifying the modification intertwining the Serre natural isomorphism with $\rho(g)$;
\item specifying the modification intertwining the dual functor with the symmetric monoidal natural transformations $\rho(\gamma)$ and $R_{g,g'}$.
\end{enumerate}
This data satisfies a large list of conditions, which holds by the cobordism hypothesis. 
We will not prove these conditions, since in this article we will assume the cobordism hypothesis holds.
However, we will use the above four points of data and so we proceed to describe them.
Point 1 was given at the end of Section \ref{Sec:duals} and point 3 in Section \ref{Sec:serre}.

It only remains to show how to obtain the modification specifying how symmetric monoidal natural transformations preserve duals.
Let $F_1,F_2$ be symmetric monoidal functors and $\phi: F_1 \Rightarrow F_2$ a symmetric monoidal natural transformation.
Assume we have chosen an inverse of $\operatorname{CohDualPair} \mathcal{C} \to \mathcal{C}$.
This gives coherent duality data on all objects so that $(.)^\vee$ is a covariant functor on the core.
It also choose the data of $F_1$ and $F_2$ preserving duals by looking at their extension to $\operatorname{CohDualPair} \mathcal{C}$. 
Then there is a canonical filling of the square
\[
\begin{tikzcd}
F_1 \circ (.)^\vee \arrow[d,"{\phi * \id_{(.)^\vee } }"] \arrow[r] 
& (.)^\vee \circ F_1 \arrow[d,"{\id_{(.)^\vee } * \phi}"]
\\
F_2 \circ (.)^\vee \arrow[r] 
& (.)^\vee \circ F_2
\end{tikzcd}
\]
Indeed, going through either direction of the diagram preserves duality data since $\phi$ is symmetric monoidal and the filling follows by essential uniqueness of $1$-isomorphisms specifying uniqueness of duals.

\section{Super algebras}\label{App: Super alg}

Let $\sVect$ denote the category in which objects are supervector spaces over the complex numbers, i.e. $\Z_2$-graded vector spaces $V = V_0 \oplus V_1$ and morphisms are linear maps that preserve the grading.
$\sVect$ is a symmetric monoidal category in which the braiding has the usual Koszul sign.
Let $\sVect^{\operatorname{fd}}$ denote the groupoid of finite-dimensional vector spaces and invertible even linear maps between them.
In the next subsection we will 
\begin{enumerate}
\item define a symmetric monoidal $\Z^B_2 \times B\Z^F_2$-action on $\sVect$ which we need in the main text to define one-dimensional topological field theories with spin-statistics and reflection structure;
\item show how the action extends to a $O_1 \times \Z^B_2 \times B\Z^F_2$-action on $\sVect^{\operatorname{fd}}$ by the cobordism hypothesis.
\end{enumerate}
In the subsequent sections we will then mimic this procedure for the bicategory of superalgebras:
\begin{enumerate}
\item In Section \ref{Sec:salgbasics} we will provide the basics of superalgebras and highlight some subtleties in defining $*$-algebras
\item In Section \ref{Sec:Z2xBZ2 action on sAlg} we will explicitly work out the $\Z^B_2 \times B\Z^F_2$-action on the bicategory of superalgebras.
The main tool is a technical lemma on how actions on the category of supervector spaces induce an action on superalgebras, of which the long but straightforward proof is included in Section \ref{Sec:induced action on sAlg}.
\item In Section \ref{Sec:salgdual} we provide explicit expressions for the sub $2$-groupoid of fully dualizable superalgebras and the $O_2$-action coming from the cobordism hypothesis, using the theory from Section \ref{App:dualizable}.
\item In Section \ref{Sec:total action} we unveil the complete picture of the $O_2 \times \Z^B_2 \times B\Z^F_2$-action used in the main text to study two-dimensional topological field theories with spin-statistics and reflection structure.
\end{enumerate}

\subsection{The $\Z_2^B \times B\Z_2^F$-action on supervector spaces}
\label{Sec:action on sVect}

Given $V \in \sVect$ its \emph{fermion parity operator} $(-1)^F_V$ is the even automorphism of order at most $2$ defined by the grading operator
\[
(-1)^F_V(v) = (-1)^{|v|} v
\]
on homogeneous vectors $v \in V$.
The \emph{complex conjugate vector space} $\overline{V}$ is defined to be equal to $V$ as an abelian group and we schematically write its elements as $\overline{v} \in \overline{V}$ for $v \in V$.
We define a scalar multiplication on $\overline{V}$ by
\[
\lambda \cdot \overline{v} := \overline{\overline{\lambda} \cdot v} \quad \lambda \in \C.
\]
We want to extend $V \mapsto \overline{V}$ to a $\Z_2^B$-action on $\sVect$ and $V \mapsto (-1)^F_V$ to a $B\Z_2^F$-action on $\sVect$.
Recall that $B\Z_2^F$ is the $2$-group with one object and two morphisms.
For the definition of an action of a $2$-group on a $1$-category, the reader can take Definition \ref{def:action} for actions on bicategories and restrict to the case where the bicategory only has identity $2$-morphisms.

To describe how $V \mapsto \overline{V}$ extends to a symmetric monoidal $\Z_2^B$-action $\tau$, there is still some freedom in the rest of the data we have to pick.
In particular, there are various possible `sign' choices to make, which are worthy of consideration for certain applications.
In this article, we always work with the canonical sign choices given by blindly following the Koszul sign rule, which can be briefly summarized as
\begin{align}\label{eq:signdata}
& f: V \to W\mapsto \overline{f}: \overline{V} \to \overline{W} \quad \overline{f}(\overline{v}) := \overline{f(v)}
\\
& \overline{V}_1 \otimes \overline{V}_2 \cong \overline{V_1 \otimes V_2}\quad  \overline{v}_1 \otimes \overline{v}_2 \mapsto \overline{v_1 \otimes v_2}
 \\
& V \cong \overline{\overline{V}} \quad v \mapsto \overline{\overline{v}}.
\end{align}
The first point tells us how $\tau(B)$ is a functor, the second makes $\tau(B)$ monoidal and the last gives the data of the natural transformation $R_{B,B}: \id_\sVect \to \tau(B) \circ \tau(B)$ telling us how $\tau(B)$ squares to one.
To make this into a $\Z_2^B \times B\Z_2^F$-action we now only have to check several conditions:

\begin{proposition}
The data described above defines a symmetric monoidal $\Z_2^B \times B\Z_2^F$-action on $\sVect$.
\end{proposition}
\begin{proof}
It is easy to check $\tau(B): V \mapsto \overline{V}$ is a functor.
For morphisms $f_1: V_1 \to W_1$ and $f_2: V_2 \to W_2$ the diagram
\[
\begin{tikzcd}
\overline{V_1 \otimes V_2}\arrow[d,"{\overline{f_1 \otimes f_2}}", swap] & \overline{V}_1 \otimes \overline{V}_2 \arrow[l] \arrow[d,"{\overline{f}_1 \otimes \overline{f}_2}"]
\\
\overline{W_1 \otimes W_2} & \overline{W}_1 \otimes \overline{W}_2 \arrow[l]
\end{tikzcd}
\]
clearly commutes.
There is a canonical isomorphism $\C \cong \overline{\C}$ given by complex conjugation.
The diagram
\[
\begin{tikzcd}[column sep=large]
\overline{V_1 \otimes V_2} \arrow[r,"{\overline{\sigma_{V_2, V_1}}}"] \arrow[d] & \overline{V_2 \otimes V_1} \arrow[d]
\\
\overline{V_1} \otimes \overline{V_2} \arrow[r,"{\sigma_{\overline{V}_2, \overline{V}_1}}",swap] & \overline{V_2} \otimes \overline{V_1}
\end{tikzcd}
\]
also clearly commutes.
So complex conjugation is a symmetric monoidal functor.

Because morphisms in $\sVect$ are all even, $V \mapsto (-1)^F_V$ defines a natural transformation $\tau((-1)^F): \id_{\sVect} \Longrightarrow \id_{\sVect}$.
Note that $(-1)^F_{\overline{\C}} =\id_{\overline{\C}} = \overline{(-1)^F_{\C}}$.
The natural transformation is monoidal by definition of the grading of a tensor product $(-1)^F_{V_1 \otimes V_2} := (-1)^F_{V_1} \otimes (-1)^F_{V_2}$.

It is easy to check $R_{B,B}$ is a monoidal natural transformation.
It is also easy to see that $R_{B,B}(\overline{V}) = \overline{R_{B,B}(V)}$.

The natural transformation $\overline{V} \to \overline{V}$ corresponding to the non-identity morphism in $\Z_2^B \times B \Z_2^F$ going from $B$ to itself is now determined by compatibility of the action with tensor products of morphisms in $\Z_2^B \times B\Z_2^F$.
Indeed, applying this compatibility to $(-1)^F$ and the constant path at $B$ implies that $\tau(B (-1)^F)$ is given by the horizontal composition $\id_{\tau(B)} \bullet \tau((-1)^F)$, so that
\[
\tau(B (-1)^F)[V] =\tau(B)[(-1)^F_V] = \overline{(-1)^F_V}.
\]
Because of the equation $B (-1)^F = (-1)^F B$ which holds in $\Z_2^B \times B \Z_2^F$, we obtain the required condition
\[
\overline{(-1)^F_V} = \tau(B (-1)^F)[V] = \tau((-1)^F B)[V] =\tau((-1)^F)[\overline{V}] = (-1)^F_{\overline{V}}
\]
which indeed is easy to check.
The other compatibility conditions for $\tau(\gamma)$ with respect to composition and tensor product for $\gamma$ one of the two morphisms of $\Z_2^B \times B\Z_2^F$ follow from the linear map $(-1)^F_V$ squaring to one.
\end{proof}

Next we restrict to the maximal subgroupoid on the fully dualizable objects, which is $\sVect^{\operatorname{fd}}$, the groupoid in which objects are finite-dimensional vector spaces and morphisms are even invertible linear maps.
By the cobordism hypothesis, it has an $O_1$-action which we proceed to describe.
Define the covariant functor $(-)^*: \sVect^{\operatorname{fd}} \to \sVect^{\operatorname{fd}}$ as being the dual on objects and the inverse of the dual on morphisms.
For concreteness, we specify duality data $V^* \otimes V \to \C$ as $f \otimes v \mapsto f(v)$ and the corresponding coevaluation data is $1 \mapsto \sum_i e_i \otimes \epsilon^i$, where $\{e_i\}$ is any basis of $V$ and $\{\epsilon^i\}$ the corresponding dual basis.
The natural transformation $\id_{\sVect^{\operatorname{fd}}} \Longrightarrow (-)^{**}$ is provided abstractly by the fact that both $V^{**}$ and $V$ are canonically dual to $V^*$ since $\sVect$ is symmetric.
It is easy to check that this natural isomorphism is given on $V \in \sVect^{\operatorname{fd}}$ by the evaluation of a functional at a vector 
\[
V \to V^{**} \quad v \mapsto \ev^V_v \quad \ev^V_v(f) = (-1)^{|f||v|} f(v) \quad v \in V, f\in V^*.
\]
We have now provided all the data of the $O_1$-action.

Finally we show how to combine the $O_1$-action and the $\Z_2^B \times B\Z_2^F$-action to a $O_1 \times \Z_2^B \times B\Z_2^F$ on $\sVect^{\operatorname{fd}}$.
This is a formal consequence of the cobordism hypothesis.
Indeed, recall that the groupoid of framed TFTs 
\[
\mathcal{C}^{\operatorname{fd}} \overset{\sim}{\leftarrow} \Fun_{\operatorname{sym}-\otimes}(\Bord^{\operatorname{fr}}_{1,0}, \mathcal{C})
\]
comes equipped with its $O_1$-action given by reversing the framing (or in 1-dimension equivalently the orientation) on the domain category.
Now, if $G$ acts symmetric monoidally on the target symmetric monoidal category $\mathcal{C}$, then it acts on topological field theories by postcomposition.
Given that the $O_1$-action is in the domain and the $G$-action is in the target, there is an induced $O_1 \times G$-action on $\mathcal{C}^{\operatorname{fd}}$.
For an explicit computation of the action, in practice, the main ingredient is the fact that the $G$-action is symmetric monoidal and so comes equipped with canonical data preserving the $O_1$-action.
Note that the $O_1 \times G$-action itself is also symmetric monoidal.
This will be relevant for understanding the symmetric monoidal structure on homotopy fixed points and so in particular for understanding the stacking of theories with spin-statistics and reflection structures.
However, this topic is outside of the scope of this paper and so we will work out the $O_1 \times G$-action as an action by ordinary functors.

We specialize to $\mathcal{C} = \sVect$ and the action of $G = \Z_2^B \times B\Z_2^F$ as above.
The only part of the $O_1 \times \Z_2^B \times B\Z_2^F$-action that is new is the commutation data between $O_1$ and $\Z_2^B$.
It is given by the canonical isomorphism $\overline{V^*} \cong \overline{V}^*$ given that both sides are the dual of $\overline{V}$.
Computing the isomorphism explicitly as the composition
\[
\overline{V^*} \xrightarrow{\id_{\overline{V^*}} \otimes \operatorname{coev}_{\overline{V}}} \overline{V^*} \otimes \overline{V} \otimes \overline{V}^* 
\cong \overline{V^* \otimes V} \otimes \overline{V}^* \xrightarrow{\overline{\operatorname{ev}_V} \otimes \id_{\overline{V}^*}} \overline{V}^*.
\]
 gives the canonical map
\[
\overline{f}(\overline{v}) = \overline{f(v)}
\]
For the curious reader, we remark that changing the monoidality data of the functor $\overline{(.)}$ as defined in \eqref{eq:signdata} by a sign $(-1)^{|v||w|}$ would also result in a sign change by $(-1)^{|v|}$ in the formula above.\footnote{This might seem unnatural to the mathematician, but this convention is sometimes used in the physics literature for dealing with Grassmann variables.}
The remaining conditions for obtaining a $O_1 \times \Z_2^B \times B\Z_2^F$-action are now easy to see by direct computation. 
Since for reflection structures, we will need to take fixed points for the action of the diagonal subgroup $\Z_2^R \subseteq O_1 \times \Z_2^B$ we will separately record the that $R = TB$ acts by $V \mapsto \overline{V}^*$ and the data of this action squaring to one
\[
V \cong \overline{\overline{V}^*}^*
\]
is a composition of the isomorphism exchanging the bar and the star with the two isomorphisms given by the double dual evaluation isomorphism and the double bar isomorphism given above.

\subsection{Basic definitions and the bicategory of super algebras}
\label{Sec:salgbasics}

Let $\mathbb{F}$ be a field either equal to $\R$ or $\C$.
A \emph{superalgebra} $A = A_0 \oplus A_1$ over $\mathbb{F}$ is an algebra over $\mathbb{F}$ with a $\Z/2$-grading that is respected by the multiplication. 
The usual definitions for normal algebras directly transfer to superalgebras, but one has to be careful about the Koszul sign rule. 
In other words, we consider superalgebras to be monoids in $\sVect$ with its interesting braiding.

Let $A$ and $B$ be super algebras over the complexes.
A \emph{(super)algebra} homomorphism $\phi: A\to B$ is an even algebra homomorphism.
We let $\sAlg_1$ denote the category of superalgebras with homomorphisms between them.
An \emph{$(A,B)$-bimodule $M$} is a $\Z_2$-graded $(A,B)$-bimodule such that the actions of $A$ and $B$ respect the grading.
Given an $(A,B)$-bimodule $M$ and an $(B,C)$-bimodule $N$ one can define an $(A,C)$-bimodule structure on the (graded) tensor product $M \otimes_B N$.
A \emph{homomorphism $\phi: M \to N$ of bimodules} is an even bimodule map.\footnote{We will not consider odd bimodule maps in this article, except in Section \ref{Sec:salgdual}}.
Superalgebras form a $1$-cateogory $\sAlg_1^\mathbb{F}$ in which objects are algebras and morphisms are algebra homomorphisms and a bicategory $\sAlg_\mathbb{F}$ in which objects are superalgebras, 1-morphisms $A \to B$ are $(B,A)$-bimodules and 2-morphisms are even bimodule maps.
We denote $\sAlg := \sAlg_\C$.
The composition of 1-morphisms is given by the relative tensor product.
Note that in our convention an $(B,A)$-bimodule $_B M_A$ has source $A$ and target $B$ and therefore composition is in the usual order $_C N_B \circ {}_B M_A := N \otimes_B M$.
With the monoidal structure $\otimes := \otimes_\mathbb{F}$ given by tensor product over $\mathbb{F}$ and the braiding with the appropriate Koszul sign
\[
A \otimes B \to B \otimes A \quad a\otimes b \mapsto (-1)^{|a| |b|} b \otimes a
\]
$\sAlg$ becomes a symmetric monoidal bicategory \cite{claudiathesis}.
A superalgebra is called a \emph{superdivision algebra} when all homogeneous elements are invertible.
Superdivision algebras need not to be division algebras in the ungraded sense.
Indeed, consider the first complex Clifford algebra
\[
A = \C l_1 = \frac{\C[e]}{(e^2-1)}
\]
as a superalgebra over $\C$ with odd generator $e$.
It is clearly a superdivision algebra. 
However, as an ungraded ring it is not even an integral domain, because $(1-e)(1+e) = 0$.

\begin{proposition}
\cite{wallgradedbrauer}
	There are ten superdivision algebras over $\R$:
	\[
	\R, \quad Cl_1, \quad Cl_2, \quad Cl_3, \quad \mathbb{H}, \quad Cl_{-3}, \quad Cl_{-2}, \quad Cl_{-1}, \quad \C, \quad\C l_1
	\]
\end{proposition}

There is a natural `contravariant $\Z_2$-action' $-^\op$ on $\sAlg$: let $A$ and $B$ be superalgebras. 
The \emph{opposite superalgebra} $A^{\op}$ is $A$ as a vector space but with multiplication $a^{\op} b^{\op} = (-1)^{|a||b|} (ba)^{\op}$.
If $M$ is an $(A,B)$-super bimodule, then $M^{\op}$ is the $(B^{\op},A^{\op})$-super bimodule defined by $b^{\op} \cdot m^{\op} = (-1)^{|b||m|} (m b)^{\op}$ and similar for the right $A^{\op}$-action.
If $\phi: M \to N$ is a super bimodule map, then $\phi^{\op}: M^{\op} \to N^{\op}$ is defined to be equal to $\phi$ as a linear map, i.e. $\phi^{\op}(m^{\op}) = \phi(m)^{\op}$. 
The following straightforward lemma implies that this construction extends to a functor
$\sAlg \to \sAlg^\op$
\begin{lemma}\label{Lem: op bimod}
	Let $M$ be an $(A,B)$-bimodule and $N$ a $(B,C)$-bimodule.
	Then there is an isomorphism
	\[
	\label{eq:optensor}
	M^{\op} \otimes_{B^{\op}} N^{\op} \cong (N \otimes_B M)^{\op}, \quad \phi(m^{\op} \otimes n^{\op}) = (-1)^{|m||n|} (n \otimes m)^{\op}
	\]
\end{lemma}
Clearly $A^{\op \op} = A$ and $M^{\op \op} = M$.

A useful way of constructing 1-morphisms in $\sAlg$ is to twist actions by super algebra homomorphisms as explained in the following lemma. 

\begin{lemma}
Let $\phi: A \to B$ be a super algebra homomorphism. Define the $(B,A)$-module $B_{\phi}$ by $B$ as a left $B$-module, but with right $A$-action
\[
b_\phi \cdot a = (b \phi(a))_\phi.
\]
This assembles into a symmetric monoidal functor $\sAlg_1 \to \sAlg$ from the $1$-category of superalgebras with algebra homomorphisms as morphisms seen as a $2$-category with trivial $2$-morphisms, to the Morita $2$-category of superalgebras.
\end{lemma}
\begin{proof}
If $A \overset{\phi}{\to} B \overset{\psi}{\to} C$ are maps of superalgebras, then there is a canonical isomorphism $C_\psi \otimes_B B_\phi \to C_{\psi \phi}$ of $(C,A)$-bimodules given by
\[
c_\psi \otimes b_\phi \mapsto c \psi(b)_{\psi \phi}
\]
It is clearly a well-defined left $C$-module map. It is also a right $A$-module map because
\[
(c_\psi \otimes b_\phi) a = c_\psi \otimes b \phi(a)_\phi 
\]
is mapped to 
\[
c \psi(b \phi(a))_{\psi \phi} = c \psi(b) \psi \phi(a)_{\psi \phi} = c \psi(b)_{\psi \phi} \cdot a.
\]
The diagram
\[
\begin{tikzcd}
D_{\chi} \otimes_C C_\psi \otimes_B B_\phi \arrow[r] \arrow[d] & D_\chi \otimes_C C_{\psi \phi} \arrow[d]
\\
D_{\chi \psi} \otimes_B B_\phi \arrow[r] & D_{\chi \psi \phi}
\end{tikzcd}
\]
clearly commutes.
The identity homomorphism $A \to A$ induces the $(A,A)$-bimodule $A_{\id} = A$ and the composition
\[
B_\phi \otimes_A A = B_\phi \otimes_A A_{\id} \to B_{\phi \circ \id}
\]
is equal to right multiplication in the $(B,A)$-module $B_\phi$.
There is a similar statement for left multiplication by $b \in B$.

The monoidal structure of the functor is induced by the canonical $(B_1 \otimes B_2, A_1 \otimes A_2)$-bimodule isomorphisms $(B_1)_{\phi_1} \otimes (B_2)_{\phi_2} \cong (B_1 \otimes B_2)_{\phi_1 \otimes \phi_2}$ given algebra homomorphisms $\phi_i: A_i \to B_i$ for $i=1,2$.
The symmetric structure is given by the invertible bimodules induced by the braiding algebra isomorphisms $A_1 \otimes A_2 \cong A_2 \otimes A_1$.
\end{proof}

\begin{remark}
If $\chi: B \to A$ is a algebra homomorphism, we can also define the $(B,A)$-module $_{\chi} A$, which is $A$ as a right $A$-module and the left $B$-action is given by composing with $\chi$.
However, we prefer to induce bimodules on the other side, because for this convention the tensor product now changes order under composition of two homomorphisms.
\end{remark}

\begin{remark}
If $\phi: A \to B$ is a superalgebra isomorphism, then there is an isomorphism of $(B,A)$-bimodules
\begin{equation}
\label{eq:left-right}
B_\phi \cong {}_{\phi^{-1}} A \quad b_\phi \mapsto {}_{\phi^{-1}} \phi^{-1}(b)
\end{equation}
\end{remark}

\begin{remark}
Since functors preserved invertible $1$-morphisms, an isomorphism $\phi: A \to B$ gives an invertible bimodule $B_\phi$.
Explicitly, the invertibility data is given by the bimodule isomorphism
\[
B_{\phi} \otimes_A A_{\phi^{-1}} \to B \quad b_{\phi} \otimes a_{\phi^{-1}} \mapsto b \phi(a)
\]
and a similar one on the other side.
\end{remark}

\begin{remark}
Given algebra homomorphisms $\phi_1,\phi_2: A \to B$, the induced bimodules $B_{\phi_1}$ and $B_{\phi_2}$ are isomorphic if and only if there exists an invertible $b\in B_{\operatorname{ev}}$ such that $\phi_1(a) b_0 = b_0 \phi_2(a)$ for all $a \in A$.
Indeed, if $T: B_{\phi_1} \to B_{\phi_2}$ is an even intertwiner, then $T(1_{\phi_1}) = (b_{0})_{\phi_2}$ for some even $b_0$.
Because $T$ is a left $B$-module map we derive $T(b_{\phi_1}) =  (b b_{0})_{\phi_2}$.
By invertibility of $T$ there exists a left inverse of $b_0$. Using the right module structure we can also conclude that $b_0$ also has a right inverse and hence is invertible.
Because $T$ is a right $A$-module map, we derive
\begin{equation}
b_{\phi_1} \cdot a = b \phi_1 (a)_{\phi_1} \mapsto b \phi_1(a) (b_{0})_{\phi_1}
\end{equation}
Since $(b b_{0})_{\phi_2} \cdot a = b b_0 \phi_2(a)_{\phi_2}$, we derive this works if and only if $\phi_1(a) b_0 = b_0 \phi_2(a)$.
\end{remark}

\begin{remark}
A useful isomorphism is the identification between the opposite algebra homomorphism and the opposite of the induced bimodule.
Namely, for $\phi: A \to B$ any algebra map, there is an isomorphism of $(A^{\op},B^{\op})$-bimodules
\begin{equation}
	\label{eq:op-iso}
	(B_\phi)^{\op} \cong {}_{\phi^{\op}} B^{\op} \quad (b_\phi)^{\op} \mapsto {}_{\phi^{\op}} b^{\op}.
\end{equation}
This is indeed a left $A^{\op}$-module map since
\begin{align*}
	a^{\op} \cdot b_{\phi}^{\op} = (-1)^{|a||b|} (b_{\phi} \cdot a)^{\op} = (-1)^{|a||b|} (b \phi(a)_{\phi})^{\op}
\end{align*}
is mapped to
\begin{align*}
	a^{\op} \cdot {}_{\phi^{\op}} b^{\op} = {}_{\phi^{\op}} \phi^{\op}(a^{\op}) \cdot b^{\op} = {}_{\phi^{\op}} \phi(a)^{\op} \cdot b^{\op} = {}_{\phi^{\op}} (-1)^{|a| |b|} (b \phi(a))^{\op}
\end{align*}
\end{remark}
We conclude the subsection with a proposition which is relevant when studying Morita equivalences between symmetric super Frobenius algebras.  
\begin{proposition}
	Let $A,B$ be super algebras and $M$ an invertible $(A,B)$-bimodule.
	Then $M$ induces a map
	\[
	f_M: \frac{B}{[B,B]} \to \frac{A}{[A,A]}
	\]
\end{proposition}
\begin{proof}
	Let $\eta: B \to N \otimes_A M$ and $\epsilon: M \otimes_B N \to A$ implement the invertibility of $M$.
	Let $b \in B$ and write
	\[
	\eta(b) = \sum_i n_i \otimes_A m_i.
	\]
	Define 
	\[
	f_M([b]) = \left[ \sum_i (-1)^{|n_i| |m_i|} \epsilon(m_i \otimes_B n_i) \right].
	\]
	This is well-defined, because first of all
	\begin{align*}
	\sum_i (-1)^{|m_i| |n_i a|} \epsilon(m_i \otimes_B n_i a) &= 
	\sum_i (-1)^{|m_i| |n_i| + |m_i||a|} \epsilon(m_i \otimes_B n_i) a 
	\\
	&\sim \sum_i (-1)^{|m_i| |n_i| + |m_i||a| + |a| |m_i \otimes n_i|} a \epsilon(m_i \otimes_B n_i) 
	\\
	&=
	\sum_i (-1)^{|am_i| |n_i|} \epsilon(a m_i \otimes_B n_i)
	\end{align*}
	Secondly this is well-defined modulo commutators: if $b' = b_1 b - (-1)^{|b_1||b|} b b_1$, then $\eta(b') = b_1 \eta(b) - (-1)^{|b_1||b|} \eta(b) b_1$ and so
	\begin{align*}
	f_M(b') = \sum_i (-1)^{|b_1 n_i||m_i|} m_i \otimes_B b_1 n_i - (-1)^{|b_1| |b| + |n_i| |m_i b_1|} m_i b_1 \otimes_B n_i = 0.
	\end{align*}
	Indeed, it is easy to see that the signs check out when one notices that $|b| = |m_i| + |n_i|$ for all $i$.
\end{proof}

\begin{example}
	Let $\mathcal{A} = \bigoplus_{g \in G} A_g$ be a strongly graded algebra. 
	Then we get induced maps $f_g: \frac{A}{[A,A]} \to \frac{A}{[A,A]}$ coming from the $(A,A)$-bimodule $A_g$.
	Take the inverse to be $N = A_{g^{-1}}$ and $\eta,\epsilon$ are induced by the multiplication maps, which are invertible by assumption.
	It follows by the associativity of $\mathcal{A}$ that these from an adjoint equivalence.
	Pick 
	\[
	\sum_i a_{g^{-1}}^i \otimes_A a^i_g \in A_{g^{-1}} \otimes_A A_g \quad \text{such that} \quad \sum_i a_{g^{-1}}^i a^i_g \in A_{g^{-1}} = 1.
	\]
	Then 
	\[
	f_g([a]) = \left[ \sum_i (-1)^{|a| |a_g^i| + |a_{g^{-1}}||a_g^i|} a_g^i a a_{g^{-1}}^i \right]
	\]
	Note that if all elements involved are homogeneous, the sign is only nontrivial when $a_g^i$ and $a_{g^{-1}}^i$ are odd and $a$ is even.
	
	As a subexample, let $A$ be any superalgebra and consider the associated $\Z_2^F = \{1,c\}$-graded `spin-statistics' algebra
	\[
	\mathcal{A} = A \oplus Ac
	\]
	with multiplication given by $c^2 = 1$ and $ac = (-1)^{|a|} ca$.
	Take $g = c$, then $\eta(a) = a c \otimes_A c$ and so
	\[
	f_g([a]) = [cac] = [(-1)^{|a|} a].
	\]
\end{example}

\subsection{$*$-algebras}
\label{Sec:*-algs}

In contrast to the common conventions for $\Z_2$-graded $C^*$-algebras, we work with a definition of $*$-algebra that uses the Koszul sign rule.

\begin{definition}
	A \emph{(super) $*$-algebra} is a complex superalgebra $A$ together with a complex antilinear even involution $a \mapsto a^*$ such that $(ab)^* = (-1)^{|a||b|} b^* a^*$.
	In case that $(ab)^* = b^* a^*$ instead, we call $A$ a \emph{$\Z_2$-graded $*$-algebra}.
\end{definition}

The notion of a super $*$-algebra is equivalent to the notion of a $\Z_2$-graded $*$-algebra.
Namely, given a complex-antilinear super $*$-algebra $(A,*)$ we can define the structure of a $\Z_2$-graded $*$-algebra on $A$ by $a^\dagger = a^*$ if $a$ is even and $a^\dagger = -i a^*$ if $a$ is odd.
Then $\dagger$ is a $\Z_2$-graded complex-antilinear $*$-structure on $A$.
In particular, since $\Z_2$-graded $C^*$-algebras are well-studied, we can use the above identification to get a theory of super $C^*$-algebras.
Directly transferring the definition of a $\Z_2$-graded $C^*$-algebra\cite{blackadar} to the corresponding super $*$-algebra gives the following fairly convoluted definition. 

\begin{definition}
	A super $*$-algebra $A$ is called a \emph{super $C^*$-algebra} if it comes equipped with a complete norm $\|.\|: A \to \C$ such that $\|ab\| \leq \| a\| \|b\|$ and for all $a = a_0 + a_1 \in A = A_0 \oplus A_1$ we have the $B^*$-identity
	\[
	\| a_0 a_0^* - i a_1 a_1^* - i a_0 a_1^* + a_1 a_0^*\| = \| (a_0 + a_1) (a_0 - i a_1^*) \| = \| a\|^2.
	\]
\end{definition}

		In particular, note how $\operatorname{Spec} a^* a \subseteq i \R_{\geq 0}$ if $A$ is a super $C^*$-algebra.
		
		We will also use $\Z_2$-graded $*$-algebras over $\R$, where we replace the condition that $*: A \to A$ is complex antilinear by it simply being real linear.
The reader should be warned that the analogous notion of super $*$-algebra over $\R$ is not equivalent to $\Z_2$-graded $*$-algebra over $\R$.
For example $Cl_{+1}$ does not admit any structure of a super $*$-algebra.
Therefore this notion should not be used to develop $\Z_2$-graded $C^*$-theory and so we will stick to $\Z_2$-graded $C^*$-algebras over $\R$ instead.

\begin{example}
	Let $(V, \langle.,. \rangle)$ be a super Hermitian space, see Definition \ref{def:hermvect}.
	We define the (super Hermitian) adjoint $T^*$ as
	\[
	\langle T^* v, w \rangle = (-1)^{|T||v|} \langle v, T w \rangle.
	\]
	If $T^\dagger$ denotes the usual ungraded Hermitian adjoint with respect to the $\Z_2$-graded Hilbert space structure $\langle\cdot |\cdot \rangle$ as in equation \ref{eq:Z2Hilb}, then 
	\[
	T^\dagger  = 
	\begin{cases}
	T^*  & T \text{ even,}
	\\
	- i T^*  & T \text{ odd,}
	\end{cases}
	\]
	so the super Hermitian adjoint is related to the ungraded Hermitian adjoint in the same way that super $*$-algebras are related to $\Z_2$-graded $*$-algebras.
		Indeed, if $T$ is odd, $v$ is odd and $w$ is even, then
	\begin{align}
	- \langle T^* v, w \rangle = \langle v, Tw \rangle = i \langle v | T w\rangle = i \langle T^\dagger v |w \rangle = i \langle T^\dagger v,w \rangle.
	\end{align}
	Similarly, if $T$ is odd, $v$ is even and $w$ is odd, then
	\begin{align*}
	\langle T^* v, w \rangle = \langle v, Tw \rangle = \langle v | Tw \rangle = \langle T^\dagger v |w \rangle = -i \langle T^\dagger v,w \rangle.
	\end{align*}
	Using complex antilinearity in the first slot, we get $T^\dagger = -iT^*$.
	The above $*$ makes the superalgebra $\End V$ over $\C$ into a complex antilinear super $*$-algebra.
In this positive case $\End V$ becomes a $C^*$-algebra (a negative definite pairing would also work).
\end{example}

\begin{example}
	We show that every $*$-structure on $M_n(\C)$ is induced by a Hermitian pairing as in the last example or by a skew-Hermitian pairing in the sense that
	\[
	\langle v,w \rangle = - \overline{\langle w,v \rangle}
	\]
	Let $*$ denote the standard Hermitian conjugate transpose $*$-structure coming from the standard Hilbert space structure on $\C^n$ and let $\star$ be any $*$-structure.
	Then $a \mapsto a^{\star *}$ is a complex-linear algebra automorphism.
	By Skolem-Noether, there exists an invertible $a_0 \in M_n(\C)$ unique up to $\C^\times$ such that $a^{\star } = a_0 a^* a_0^{-1}$.
	Moreover, we need to have
	\[
	a = a^{\star \star} = a_0 a_0^{-1*} a a_0^* a_0^{-1} \Longrightarrow a_0^* a_0^{-1} \in Z(M_n(\C)) = \C 
	\]
	so that $a_0^* = za_0$ for some $z \in \C^\times$.
	We use the $\C^\times$-undeterminacy in the definition of $a_0$ to require that $z$ is real.
	This is possible since
	\[
	(w a_0)^* = z \overline{w} a_0 = \frac{z \overline{w}^2}{|w|^2} wa_0.
	\]
	Using that $a^{**} = a$ we get that $z = \pm 1$.
	Now recall that every Hermitian form on $\C^n$ is of the form 
	\[
	B(v,w) = \langle v, a w \rangle
	\]
	where $\langle \cdot, \cdot \rangle$ is the standard inner product and $a \in M_n(\C)$ is self-adjoint and invertible.
	Similarly a skew-Hermitian form is of this form for $a$ skew-adjoint.
	The Hermitian form that induces the $*$-structure $\star$ is
	\[
	B(v,w) := \langle v, a_0^{-1} w \rangle
	\]
	since
	\[
	B(T v,w) = \langle v, T^* a_0^{-1} w \rangle = \langle v, a_0^{-1} T^\star w \rangle = B(v,T^\star w).
	\]
	Also note that in the case $a_0$ is self-adjoint, $\star$ is a $C^*$-structure if and only if $a_0$ is either positive definite or negative definite.
\end{example}

\begin{example}
	There are two complex-antilinear super $*$-structures on $\C l_1$ given by a choice of sign in $e^* = \pm i e$.
	They correspond to the two complex-linear $\Z_2$-graded $*$-structures $e^\dagger = \pm e$.
	Picking the sign to be negative results in a $*$-algebra that can never satisfy the $B^*$-identity for any norm $\|.\|$.
	Indeed, for $a = 1 +e$ we have
	\[
	\|1 + e\|^2 = \| (1 + e)(1- ie^*)\| = \|(1+e)(1-e)\| =\| 0 \| = 0.
	\]
	Also note that $e^* e = - i$ has spectrum $\{-i\}$ and so there is no $*$-superrepresentation of $\C l_1$ for which the resulting matrix $e^* e$ has spectrum contained in $i\R_{\geq 0}$.
	Therefore this $*$-algebra is not super $C^*$.
\end{example}	

Our Koszul sign convention for the definition of a super $*$-algebra $A$ also has consequences for the correct definition of $\overline{A}$.
To understand this, first recall from Section \ref{Sec:fermrep} that if $(V,h)$ is a super Hermitian vector space, we defined the Hermitian vector space structure on $\overline{V}$ to be $(-1)^F_{\overline{V}} \circ \overline{h}$.
This had two separate reasons:
\begin{enumerate}
\item if $(V,h)$ is a super Hilbert space,  $(\overline{V},\overline{h})$ is in general not a super Hilbert spaces, but $(\overline{V}, (-1)^F_{\overline{V}} \circ \overline{h})$ is;
\item in Section \ref{Sec:1D} we found that this is related to the action on super Hermitian vector spaces for which fixed points give one-dimensional TFTs with reflection structure.
\end{enumerate}
The situation for super $*$-algebras $A$ is similar.
There is a $\Z_2 \times \Z_2$-action $A \mapsto A^{\op}$ and $A \mapsto \overline{A}$ on $\sAlg_1$ (our sign conventions are explained in Section \ref{Sec:Z2xBZ2 action on sAlg}).
Seeing $A$ as a fixed point in the $1$-category $\sAlg_1$ of superalgebras under the action $A \mapsto \overline{A}^{\op}$ the induced remaining action $A \mapsto\overline{A}$ on $*$-algebras gives $\overline{A}$ a $*$-structure
\[
\overline{A} \to \overline{\overline{A}^{\op}} \cong \overline{\overline{A}}^{\op}
\]
which is simply $\overline{a}^* = \overline{a^*}$.
Now if $A$ is $C^*$, then for every odd element $a \in A$ we have that $a^* a$ has spectrum contained in the positive imaginary axis.
But then in $\overline{A}$
\[
\Spec \overline{a}^* \overline{a} =  \Spec \overline{a^* a} = \overline{\Spec a^* a}
\]
is contained in the negative imaginary axis.
Therefore just like for super Hilbert spaces, we change the $*$ by a $(-1)^F$ defining
\[
\overline{a}^* := (-1)^{|a|} \overline{a^*}
\]
for homogeneous $a \in A$.
In other words, we change the data of $R$ squaring to $1$ in this action with the $B\Z_2^F$-action $(-1)^F$.

The analogue of the second reason above applies to defining the complex conjugate of a stellar algebra as in Section \ref{Sec:stellar}.

\subsection{The $\Z_2^B \times B\Z_2^F$-action on superalgebras}
\label{Sec:Z2xBZ2 action on sAlg}

In this section we will discuss the $\Z_2^B \times B\Z_2^F$-action on $\sAlg$ induced by the $\Z_2^B \times B\Z_2^F$-action on $\sVect$ from Section \ref{Sec:action on sVect}.
This will be an application of the construction of a $\mathcal{G}$-action on $\sAlg$ coming from a symmetric monoidal $\mathcal{G}$-action on $\sVect$ worked out in Section \ref{Sec:induced action on sAlg}.
We will also specify how the action is symmetric monoidal because of the compatibility with the $O_2$-action as discussed in Section \ref{Sec:total action}.

We start by spelling out some of the most basic information contained in this action and slowly get more technical.
The $\Z_2^B$-action maps a superalgebra $A$ to its complex conjugate $\overline{A}$:

\begin{definition}
	\label{def:C-conj algebra}
The complex conjugate $\overline{A}$ of a superalgebra $A$ is defined to be the complex conjugate vector space of $A$ together with the multiplication $\overline{a} \cdot \overline{b} = \overline{ab}$.
\end{definition}

The $B\Z_2^F$-action maps $A$ to the automorphism of $A$ given by the spin-statistics bimodule:

\begin{definition}
	\label{def:spin-statistics bimodule}
	Given a superalgebra $A$, define the \emph{spin-statistics bimodule} $A_{(-1)^F}$ of $A$ to be the $(A,A)$-bimodule associated to the algebra automorphism $(-1)_A^F: A \to A$ given by $a \mapsto (-1)^{|a|} a$ on homogeneous $a \in A$.
	We write this bimodule explicitly as $A_{(-1)^F} = \{ a (-1)^F : a \in A\}$ so that we can think of $(-1)^F$ as a symbol satisfying $a (-1)^F = (-1)^{|a|} (-1)^F a$.
\end{definition}

Using Section \ref{Sec:action on sVect} and the results above, we will now spell out all the data of how this assembles into a $\Z_2^B \times B\Z_2^F$-action on $\sAlg$.
The $\Z_2^B$-action (which we again write by a bar) starts with the data of a functor $\rho(B): \sAlg \to \sAlg$ given by $A \mapsto \overline{A}$.
The action can be defined on morphisms concretely as follows: Given a $(A,B)$-bimodule $M$, $\overline{M}$ is the $(\overline{A},\overline{B})$-bimodule given by $\overline{a} \cdot \overline{m} = \overline{am}$ and similar for the right $\overline{B}$-action.
If $M'$ is another $(A,B)$-bimodule and $\phi: M \to M'$ is a bimodule map, the bimodule map $\overline{\phi}: \overline{M} \to \overline{M'}$ is defined by $\overline{\phi}(\overline{m}) = \overline{\phi(m)}$.
Given $M$ an $(A,B)$-bimodule and $N$ a $(B,C)$-bimodule, the data of $\rho(B)$ preserving composition of $1$-morphisms involves the isomorphism
\[
\overline{N \otimes_B M} \cong \overline{N} \otimes_{\overline{B}} \overline{M} \quad \overline{n \otimes_B m} \mapsto \overline{n} \otimes_B \overline{m}.
\]
We have now defined the functor $\rho(B)$.
Its monoidality data on algebras $A_1$ and $A_2$ is induced by the algebra isomorphism
\[
\overline{A_1 \otimes A_2} \cong \overline{A}_1 \otimes \overline{A}_2 \quad \overline{a_1 \otimes a_2} \mapsto \overline{a}_1 \otimes \overline{a}_2
\]
giving an invertible $(\overline{A_1 \otimes A_2}, \overline{A}_1 \otimes \overline{A}_2)$-bimodule $N_{A_1,A_2}$.
The naturality data for this isomorphism with respect to a $(B_1,A_1)$-bimodule $M_1$ and a $(B_2,A_2)$-bimodule $M_2$ is given by the obvious bimodule isomorphism
\[
 \overline{M_1 \otimes M_2} \otimes_{\overline{A_1 \otimes A_2}} N_{A_1,A_2}\cong  N_{A_1,A_2} \otimes_{\overline{A}_1 \otimes \overline{A}_2} (\overline{M}_1 \otimes \overline{M}_2) .
\]
We will from now on identify $\overline{A_1 \otimes A_2} \cong \overline{A}_1 \otimes \overline{A}_2$ directly. 
Then all other data to make $\rho(R)$ into a symmetric monoidal functor becomes the identity.

We move on to the data $R_{g_1,g_2}$ for objects $g_1,g_2$ in $\Z_2^B \times B \Z_2^F$.
The only interesting case is $g_1 = g_2 = B$ in which case the invertible bimodule is induced by the equality of algebras $A = \overline{\overline{A}}$.
The other data $\omega_{g,g',g''} $ needed to define an action are therefore also identities.

We now give $\rho(\gamma)$ for morphisms $\gamma$ in $\Z_2^B \times B\Z_2^F$.
$\rho((-1)^F)[A] = A_{(-1)^F}$ becomes a natural transformation using the $(B,A)$-bimodule isomorphisms
\[
M \otimes_A A_{(-1)^F} \to B_{(-1)^F} \otimes_B M \quad  m \otimes a (-1)^F \mapsto (-1)^{|m| + |a|} (-1)^F \otimes ma.
\]
The data of $\rho((-1)^F)$ being a monoidal natural transformation is corresponds to the canonical isomorphism of $(A \otimes B, A \otimes B)$-bimodules
\[
A_{(-1)^F} \otimes B_{(-1)^F} \cong (A \otimes B)_{(-1)^F} \quad a (-1)_A^F \otimes b(-1)_B^F \mapsto (a \otimes b) (-1)^F_{A \otimes B}.
\]
Note that in contrast to the naturality of $(-1)^F$ with respect to a bimodule $M$, there is no sign.
This is similarly true for $\rho((-1)^F R) := \rho((-1)^F) \circ \rho(R)$.

We turn to the $\alpha_{\gamma, \gamma'}$.
The data of $\alpha_{(-1)^F, (-1)^F}$ is determined by the isomorphisms $A_{(-1)^F} \otimes_A A_{(-1)^F} \cong A$ corresponding to the fact that the algebra homomorphism $(-1)^F$ squares to one.
Explicitly, this data looks like
\[
A_{(-1)^F} \otimes_A A_{(-1)^F} \to A \quad a_1 (-1)^F \otimes_A a_2 (-1)^F \mapsto(-1)^{|a_2|} a_1 a_2.
\]
The natural isomorphism $\alpha_{R (-1)^F, R (-1)^F}$ is defined analogously.

Next we move to the natural transformations specifying $\rho((-1)^F R) = \rho(R (-1)^F) \cong \rho(R) \bullet \rho((-1)^F)$, i.e. $R_{\id_{R}, (-1)^F}$ and $R_{(-1)^F, \id_R}$.
Note first that because of the equality of linear maps $(-1)^F_{\overline{V}} = \overline{(-1)^F_V}$, there is an equality of $(A,A)$-bimodules $\overline{A}_{(-1)^F_{\overline{A}}} = \overline{A_{(-1)^F_A}}$.
Now the $(\overline{B}, \overline{A})$-bimodule isomorphism
\[
\overline{M} \otimes_{\overline{A}} \overline{A}_{(-1)^F_{\overline{A}}}  \cong \overline{B}_{(-1)^F_{\overline{A}}}  \otimes_{\overline{B}} \overline{M}
\]
is similarly given by
\[
\overline{m} \otimes \overline{a} (-1)^F \mapsto (-1)^{|m| + |a|} (-1)^F \otimes \overline{ma}.
\]
The other $R_{\gamma_1, \gamma_2}$ are defined analogously.

\subsection{Dualizability in $\sAlg$}
\label{Sec:salgdual}

We apply the theory of dualizability from Appendix \ref{App:dualizable} to the bicategory $\sAlg$.

To define adjoints, let $M$ be an $(A,B)$-bimodule.
Let $\operatorname{HOM}_A(M,A)$ be the super vector space consisting of maps of left $A$-modules that are not necessarily even.
More precisely, homogeneous elements $f \in \operatorname{HOM}_A(M,A)$ satisfy the graded $A$-linearity condition
\[
f(am) = (-1)^{|a||f|} a f(m)
\]
Then $\operatorname{HOM}_A(M,A)$ becomes an $(B,A)$-bimodule in the obvious way using appropriate Koszul signs:
\[
(fa)(m) := f(am) \quad (bf)(m) := (-1)^{|b|(|f| + |m|)} f(mb).
\]
Analogously we can define $\operatorname{HOM}_B(M,B)$ as a $(B,A)$-bimodule using right $B$-module maps.
There are evaluation bimodule maps
\begin{align*}
&\ev: M \otimes_B \operatorname{HOM}_A(M,A) \to A \quad \ev(m \otimes f) = (-1)^{|m||f|} f(m)
\\
&\ev: \operatorname{HOM}_B(M,B) \otimes_A M \to B \quad \ev(g \otimes m) = g(m)
\end{align*}
To get these to be nondegenerate pairings, we need a finiteness condition:

\begin{proposition}[Proposition 1.6 of \cite{doldpuppe}]
\label{ex:adjoints in Alg}
Define the $(B,A)$-bimodules 
\[
M^L := \operatorname{HOM}_B(M,B)  \quad M^R := \operatorname{HOM}_A(M,A) 
\]
Using the above pairings, $M^L$ is a left adjoint of $M$ if and only if $M$ is finitely generated and projective over $A$ and $M^R$ is a right adjoint of $M$ if and only if $M$ is finitely generated and projective over $B$. 
\end{proposition}
An adjoint equivalence in the bicategory $\sAlg$ is also called a \emph{Morita context}.

Let $\phi:M_1 \Longrightarrow M_2$ be a homorphism of $(A,B)$-bimodules that are finitely generated as $B$-modules.
Then it is easy to check that the induced map $\phi^R: M^R_2 \Longrightarrow M_1^R$ is given by
\[
\phi^R(f)(m_1) = f(\phi(m_1)).
\]

Every object in $\sAlg$ is $1$-dualizable.\footnote{This is a special case of the general fact that every $E_n$-algebra is $n$-dualizable~\cite{ClaudiaOwen}.}
We provide the data to make the opposite algebra into the dual of an algebra $A \in \sAlg$.
The evaluation map is the $1$-morphism $\operatorname{ev}_A: A \otimes A^{\op} \to \C$ given by $A$ itself seen as a $(\C, A \otimes A^{\op})$-bimodule.
The coevaluation map $\operatorname{coev}_A: \C \to A^{\op} \otimes A$ is $A$ seen as a $( A^{\op} \otimes A, \C)$-bimodule.
Let $M$ be an $(B,A)$-bimodule. 
Then $M^{\op}$ becomes a dual of $M$ as follows.
We have to fill the diagram
\[
\begin{tikzcd}
B^{\op} \otimes A \arrow[r,"{\id \otimes M}"] \arrow[d,"{M^{\op} \otimes \id_{A}}",swap] & B^{\op} \otimes B \arrow[d,"{\ev_B}"]
\\
A^{\op} \otimes A \arrow[r,"{\ev_A}"] & \C 
\end{tikzcd}
\]
The two bimodules we have to show an isomorphism between satisfy
\begin{align*}
b_1 \otimes b_2^{\op} \otimes m &= (-1)^{|b_1||b_2|} 1 \otimes 1 \otimes b_2 b_1 m \in B \otimes_{B^{\op} \otimes B} (B^{\op} \otimes M) 
\\
 a_1 \otimes m^{\op} \otimes a_2 &= (-1)^{|a_1||m|} 1 \otimes (m a_1 a_2)^{\op} \otimes 1 \in A \otimes_{A^{\op} \otimes A} (M^{\op} \otimes A)
\end{align*}
Using this it is easy to verify that $1 \otimes 1 \otimes m \mapsto 1 \otimes m^{\op} \otimes 1$ induces an isomorphism of left $B^{\op} \otimes A$-modules, so we filled the square above.
It is straightforward to check that the relevant higher condition holds.
We emphasize that the dual of $M$ always exists and is not related to the left and right adjoint of $M$.

A general superalgebra $A$ is not fully dualizable.
Indeed, for $\ev_A$ to admit right and left adjoints, we need $A$ to be finite-dimensional and semisimple.
We obtain an explicit description of the maximal sub $2$-groupoid $\sAlg^{\text{fd}}$ on the fully dualizable objects:
\begin{itemize}
\item \textbf{Objects} are finite-dimensional semisimple superalgebras;
\item \textbf{$1$-morphisms} are invertible bimodules;
\item \textbf{$2$-morphisms} are invertible even bimodule homomorphisms.
\end{itemize}

We proceed to give the $O_2$-action on $\sAlg^{\text{fd}}$, where we continue to use the skeletal model of $O_2$ and the notation from Section \ref{Sec:O2-action}.
We provide our choices of duals and dualizability data, following \cite[Section 2.1]{gunningham2016spin}.

Let $A \in \sAlg^{\operatorname{fd}}$ be finite-dimensional and semisimple. The object $-$ of $O_2$ acts by sending a super algebra to the algebra $A^{\op}$ (which is the dual as explained above), an invertible bimodule $M$ to ${M^{\op}}^{-1}$ and an invertible bimodule map to the canonical map between the inverses of the opposites it induces. 
This is a functor by Lemma~\ref{Lem: op bimod} which is also easily seen to be symmetric monoidal.
Note that we have to make a contractible choice of an inverse for every invertible bimodule to uniquely specify this functor.
We specify the coherence isomorphisms making this a $\Z_2$-action:         
The invertible $(A = A^{\op \op}, A)$-bimodule $\hat{R}_A $ to be the identity.
Then $\omega: \hat{R}^{\op}_A \otimes_{A^{\op \op \op}} \hat{R}_{A^{\op}} \to A^{\op}$ is the multiplication map. 

We now turn to the $SO_2$ part of the action.
Looking at Proposition \ref{ex:adjoints in Alg}, a right adjoint of $\operatorname{ev}_A$ is the $(A^{\op} \otimes A,\C)$-bimodule $\Hom_{\C}(A, \C) = A^*$.
Hence we have $S_A = A^*$ with $(A,A)$-bimodule structure given by the appropriate Koszul signs
	\[
	(fa_1)(a_2) = f(a_1 a_2), \quad (a_1 f)(a_2) = (-1)^{|a_1|(|f|+ |a_2|)} f(a_2 a_1).
	\]
Similarly, the right adjoint of $\operatorname{ev}_A$ is the $(A^{\op} \otimes A,\C)$-bimodule $\Hom_{A^{\op} \otimes A}(A, \C)$.
So as in \cite[Definition 2.4]{gunningham2016spin}, we find the Serre automorphism to be given by the $(A,A)$-bimodule $A^*$ while the inverse of the Serre is given by $\Hom_{A^{\op} \otimes A}(A, A^{\op} \otimes A)$.

The naturality data of the Serre automorphism in the context of superalgebras can be worked out explicitly as follows, compare \cite[Proposition 2.8]{janserre} and \cite[Proposition 3.2]{LorantNils}.
Let $A,B \in \Alg^{\operatorname{fd}}$ and let $M$ be an invertible $(A,B)$-bimodule.
Since $M$ is invertible, we can make $M$ into a Morita context by picking its inverse $(B,A)$-bimodule $N$ to be adjoint to $M$ so that the invertibility data
\[
\eta: N \otimes_A M \to B \quad \epsilon: M \otimes_B N \to A
\]
satisfyies the snake relations.
We will use the following convenient notation.
		Given $n \in N, m \in M$ and $f \in A^*$, define $n f m$ of $B^*$ as
	\[
	(m f n)(b) := (-1)^{|m|(|f| + |n| + |b|)} f(\eta(n \otimes_B b m)).
	\]
	Note that this notation is compatible with the notation for the $(A,A)$-bimodule structure on $A^*$.
We realize $N^{\op}$ as a dual inverse of $M$ in an analogous fashion to how we realized $M^{\op}$ as a dual of $M$.
Explicitly, the filling of the diagram
\[
\begin{tikzcd}
\ & A^{\op} \otimes A \arrow[d,"\ev_A"]
\\
B^{\op} \otimes B \arrow[ru,"{N^{\op} \otimes M}"] \arrow[r,"\ev_B",swap] & \C
\end{tikzcd}
\]
amounts to the isomorphism of $(\C,B^{\op} \otimes B)$-bimodules
\[
\psi: A \otimes_{A^{\op} \otimes A} (N^{\op} \otimes M) \cong B
\]
given by
\[
a \otimes (n^{\op} \otimes m) \mapsto (-1)^{|a||n|} \eta(n a \otimes m).
\]
This satisfies the desired condition because the condition holds for realizing $M^{\op}$ as the dual of $M$. 
Next we need to take the right adjoint of the above diagram to get the map of $(B^{\op} \otimes B,\C)$-bimodules
\[
\psi^R: B^* \to \left(\operatorname{HOM}_{A^{\op}} (N^{\op},A^{\op}) \otimes \operatorname{HOM}_A( M,A)\right) \otimes_{A^{\op} \otimes A} A^*
\]
compare \cite[Lemma 2.7]{janserre}.
It maps $f \in B^*$ to the element of $A^*$ mapping $a$ to
\[
a \otimes n^{\op} \otimes b \mapsto (-1)^{|a||n|} f(\eta(n a \otimes m))
\]
We are now in shape to give the Serre naturality data, which by definition is a filling of
\[
\begin{tikzcd}[column sep = 1.8cm]
A  \arrow[r,"{\ev^R_A \otimes A}"] & A^{\op} \otimes A \otimes A \arrow[r]  & A \otimes A^{\op} \otimes A \ \arrow[r,"{A \otimes \ev_A}"] & A 
\\
B \arrow[u,"M",swap] \arrow[r,"{\ev^R_B \otimes B}",swap] & B^{\op} \otimes B \otimes B \arrow[r] \arrow[u,"{M^{\op -1} \otimes M \otimes M}"] & B \otimes B^{\op} \otimes B \arrow[r,"{B \otimes \ev_B}",swap] \arrow[u,"{M \otimes M^{\op -1} \otimes M}"] & B \arrow[u,"M"]
\end{tikzcd}
\]
We can do this by the braiding naturality and the two diagrams expressing naturality of $\ev$ and its adjoint above.
The result is the following proposition, which was communicated to us by Chris Schommer-Pries~\cite{chrismathoverflow}. 



\begin{proposition}
The naturality of the natural transformation $A \mapsto A^*$ is expressed by the bimodule map $S_M: A^* \otimes_A M \to M \otimes_B B^*$ explicitly given as  
	\[
	S_M(f \otimes_A m) = \sum_j m_j \otimes_B n_j f m
	\]
	where 
	\[
	\sum_j \eta(m_j \otimes n_j) = 1.
	\]
\end{proposition}

\begin{example}\label{ex:Serre nat unitary}
	Take $A = B$, let $G$ be a finite group and let $\mathcal{A}$ be a strongly $G$-graded algebra with $A_1 = A$.
	Take $M = A_g$.
	Then we can take $N = A_{g^{-1}}$ and $\eta: A_g \otimes_A A_{g^{-1}} \to A$ can be taken to be the multiplication map $\phi_{g,g^{-1}}$.
For $a \in A^*$ and $a_g \in A_g$, the Serre naturality map can be written 
		\[
	S_{A_g}(f \otimes a_g) = \sum_j m_j \otimes n_j f a_g.
	\]
	where $m_j \in A_g$ and $n_j \in A_{g^{-1}}$ are chosen so that
	\[
	\sum_j m_j n_j = 1.
	\]
	Explicitly, given $ a_{g^{-1}} \in A_{g^{-1}}$, we have
	\[
	(a_{g^{-1}}  f a_g)(a) := (-1)^{|a_{g^{-1}}|(|f| + |a_g| + |a|)} f(a_g a a_{g^{-1}}).
	\]
	In case $g =1$, this agrees with the bimodule action of $A$ on $A^*$.
\end{example}

\begin{example}
\label{ex:*-alg Serre nat}
	Take $B = \overline{A}^{\op}$ and suppose $A$ has a super $*$-algebra structure $\overline{A}^{\op} \to A$.
	Take $M = A_*$ to be the associated stellar module. 
	Because $*$ is an involution, $_* \overline{A}^{\op}$ is a canonical inverse of $A_*$ and
	the Serre naturality map $S_{A_*}: A^* \otimes_A A_* \to A_* \otimes_{\overline{A}^{\op}} \overline{A}^{\op *}$ simplifies to 
	\[
	S_{A_*}(f \otimes_A a_*) = 1_* \otimes_{A^{\op}} (\overline{b}^{\op} \mapsto f(ab^*)).
	\]
We can write this more succinctly as $S_{A_*}(f \otimes_A 1_*) = 1_* \otimes_{\overline{A}^{\op}} f^*$, where
	$f^*: \overline{A}^{\op} \to \C$ is defined as $f^*(\overline{b}^{\op}) = f(b^*)$.
\end{example}

\begin{example}
	Let $(B, \epsilon)$ be a symmetric Frobenius algebra and $M$ an invertible $(B,A)$-bimodule.
	The induced Frobenius structure on $A$ is the composition
	\[
	A \cong M^{-1} \otimes_B B \otimes_B M \xrightarrow{\id_{M^{-1}} \otimes \epsilon \otimes \id_{M}} M^{-1} \otimes_B B^* \otimes_B M  \xrightarrow{S_M} A^*
	\]
	where we used Serre naturality at the end.
	
	As a subexample, take $A = B$ and $M = \Pi B$.
	Then we can take $M = M^{-1} = \Pi A$ where $M^{-1} \otimes_A M \to A$ is multiplication.
	The Serre naturality maps $1 \otimes f \otimes 1 \in M^{-1} \otimes B^* \otimes M$ to 
	\[
	A \ni a \sim a \otimes 1 \in M^{-1} \otimes M \mapsto f(-a)
	\]
	We see that the new Frobenius structure on $A$ is $- \epsilon$.
\end{example}	

\begin{example}
\label{ex: (-1)^F Serre nat}
We compute Serre naturality for the $(A,A)$-bimodule $A_{(-1)^F}$ and show that the resulting filling of the square
\[
\begin{tikzcd}
A \arrow[r,"A^*"] \arrow[d, "{A_{(-1)^F}}",swap] & A \arrow[d, "{A_{(-1)^F}}"]
\\
A \arrow[r,"A^*",swap] & A
\end{tikzcd}
\]
is equal as a bimodule map to the filling of the square coming from the naturality of the $B\Z_2^F$-action on superalgebras.\footnote{This already follows by the cobordism hypothesis: $A_{(-1)^F}$ is a symmetric monoidal natural transformation and so preserves the $SO_2$-action.}
For this, we will use that $A_{(-1)^F}$ is its own inverse under the map 
\[
A_{(-1)^F} \otimes_A A_{(-1)^F} \to A \qquad a_1 (-1)^F \otimes a_2 (-1)^F \mapsto (-1)^{|a_2|} a_1 a_2.
\]
The Serre naturality map can now equivalently described as a map
\[
A_{(-1)^F} \otimes_A A^* \otimes_A A_{(-1)^F} \to A^*
\]
It suffices to compute it on elements of the form $(-1)^F \otimes f \otimes (-1)^F$.
Evaluating the result on $b \in A$ gives
\begin{align*}
f(((-1)^F b (-1)^F )\cdot ((-1)^F \cdot (-1)^F)) = (-1)^{|b|} f(b)
\end{align*}
We obtain $(-1)^F \otimes f \otimes (-1)^F \mapsto (-1)^{|f|} f$ which is the naturality data of $(-1)^F$.
\end{example}

We warn the reader of the non-Koszul sign in point 3 of the following lemma, which is the origin of the ungraded symmetry of the Frobenius structure.

\begin{lemma}
\label{Lem:ungrsym}
Recall that for $f \in A^*, a_g \in A_g$ and $a_{g^{-1}} \in A_{g^{-1}}$ we have the element of $\overline{A}^*$ defined by
\[
(a_{g^{-1}} f \overline{a_g} )(\overline{a}) := (-1)^{|a_{g^{-1}}| (|f| + |a_g| + |a|)} f(a_g \overline{a} \overline{a_{g^{-1}}}).
\]
The following are equivalent for an ungraded symmetric Frobenius structure on $A$.
\begin{enumerate}
\item $\overline{\lambda(a_g a_{g^{-1}})} = \lambda(a_{g^{-1}} a_g)$ for all $a_g \in A_g$ and $a_{g^{-1}} \in A_{g^{-1}}$
\item $a_g \lambda \overline{a_{g^{-1}}} = (-1)^{|a_{g}|} (a_g \overline{a_{g^{-1}}}) \overline{\lambda}$ for all $a_g \in A_g$ and $a_{g^{-1}} \in A_{g^{-1}}$.
\item 
\[
S_{A_g}(\lambda \otimes_A a_g) = (-1)^{|a_g|} a_g \otimes_{\overline{A}} \overline{\lambda}
\]
for all $a_g \in A_g$
\end{enumerate}
Here we have implicitly used the isomorphism $\overline{A}^* \cong \overline{A^*}$ defined by $\overline{f}(\overline{a}) = \overline{f(a)}$.
\end{lemma}
\begin{proof}
By definition
\begin{align*}
(a_{g^{-1}} \lambda a_g)(a) &= (-1)^{|a_{g^{-1}}|(|a_g| + |a|)} \lambda(a_g a a_{g^{-1}})
\\
((a_{g^{-1}} a_g) \overline{\lambda})(a) &= (-1)^{|a|(|a_{g^{-1}}| +|a_g|)} \overline{\lambda(a a_{g^{-1}} a_g)}.
\end{align*}
Since $\lambda$ is even, we can assume without loss of generality that $|a_{g^{-1}}| + |a_g| + |a| = 0$.
Working out the signs, we see that $a_g \lambda a_{g^{-1}} = (-1)^{|a_{g}|} (a_g a_{g^{-1}}) \overline{\lambda}$ is equivalent to
\[
\overline{\lambda(a_g a a_{g^{-1}})} = \lambda(a a_{g^{-1}} a_g)
\]
for all $a \in A$.
Since $a a_{g^{-1}} \in A_{g^{-1}}$ and we can specialize to $a = 1$, this is equivalent to
\[
\overline{\lambda(a_g a_{g^{-1}})} = \lambda(a_{g^{-1}} a_g)
\]
for all $a_g \in A_g$ and $a_{g^{-1}} \in A_{g^{-1}}$.
So we have shown that point 1 is equivalent to point 2.

We now spell out Serre naturality, generalizing Example \ref{ex:Serre nat unitary} to the anti-unitary case.
Let $m_j \in A_g$ and $n_j \in \overline{A_{g^{-1}}}$ be finite collections of elements such that
\[
\sum_j m_j n_j = 1 \in A.
\]
Then for $f \in A^*$ we have
\[
S_{A_g}(f \otimes_A a_g) = \sum_j m_j \otimes_{\overline{A}} n_j f a_g
\]
So it suffices to show that point 2 of the lemma is equivalent to
\[
\label{eq:completelyrandomequation}
\sum_j m_j \otimes_A n_j \lambda a_g = (-1)^{|a_g|} a_g \otimes_{\overline{A}} \overline{\lambda} \quad \forall a_g \in A_g .
\]
Assuming point 2 implies this equation by the computation
\[
\sum_j m_j \otimes_A n_j \lambda a_g = (-1)^{|a_g|}\sum_j m_j \otimes_A (n_j a_g) \overline{\lambda} = (-1)^{|a_g|} \sum_j m_j (n_j a_g) \otimes_A \overline{\lambda} = (-1)^{|a_g|} a_g \otimes_A \overline{\lambda}
\]
Conversely, assume that the Equation \ref{eq:completelyrandomequation} holds for all $a_g$ and let $a_{g^{-1}} \in A_{g^{-1}}$.
Since $A_{g^{-1}}$ is an invertible bimodule and the multiplication map $A_{g^{-1}} \otimes_A A_g \to A$ is an isomorphism, we can tensor our equation from the left with $a_{g^{-1}}$ to obtain that
\[
(-1)^{|a_g|} a_{g^{-1}} a_g  \overline{\lambda} = \sum_j (a_{g^{-1}} m_j) (n_j \lambda a_g ) = a_{g^{-1}}  \lambda a_g
\]
\end{proof}

The Serre and the dual in the bicategory of superalgebras combine as follows, specializing the isomorphism $S_c^* \cong S_{c^*}$ to the case at hand.

\begin{lemma}
	Let $A$ be a finite-dimensional superalgebra. Then the map
	\[
	\phi:(A^*)^{\op} \to (A^{\op})^*, \quad \phi(f^{\op})(a^{\op}) = f(a)
	\]
	is an $(A^{\op},A^{\op})$-bimodule isomorphism.
	Note that in $(A^*)^{\op}$ we take the opposite of a bimodule, while in $(A^{\op})^*$ we take the opposite of an algebra.
\end{lemma}
\begin{proof}
	If $b^{\op} \in B^{\op}$, then $b^{\op} f^{\op} = (-1)^{|b| |f|}(fb)^{\op}$ is mapped by $\phi$ to
	\[
	\phi(b^{\op} f^{\op})(a^{\op}) = (-1)^{|b| |f|}(fb)(a) = (-1)^{|b| |f|}f(ba).
	\]
	On the other hand we can compute
	\begin{align*}
	(b^{\op} \phi(f))(a^{\op}) = (-1)^{|b||\phi(f)| + |b| |a|} \phi(f)(a^{\op} b^{\op}) = (-1)^{|b||\phi(f)|} \phi(f)((ba)^{\op})=(-1)^{|b||f|} f(ba)
	\end{align*}
	and so $\phi$ is a left $A^{\op}$-module map.
	Checking it is a right module map is analogous.
\end{proof}

\subsection{The $\Z_2^B \times B\Z_2^F \times\O_2$-action on $\sAlg^{\operatorname{fd}}$}
\label{Sec:total action}

Next we go into the compatibility between the $O_2$-action on $\sAlg^{\operatorname{fd}}$ and the symmetric monoidal $\Z_2^B \times B \Z_2^F$-action on $\sAlg$ from Section \ref{Sec:Z2xBZ2 action on sAlg}.
Note that the $\Z_2^B \times B \Z_2^F$-action restricts to $\sAlg^{\operatorname{fd}}$ since it is symmetric monoidal and so preserves duals and inverses. 
Similarly, because of the symmetric monoidality, the actions combine to a $(\Z_2^B \times B \Z_2^F) \times O_2$-action.  
We will now provide the data of the two actions commuting explicitly.

First of all, note that the functors $\overline{(.)}$ and $(.)^{\op}$ strictly commute.
In particular, the equality $\overline{A}^{\op} = \overline{A^{\op}}$ preserves the canonical duality data on both sides.
This is a consequence of our choice of monoidal data on the complex conjugation functor on $\sVect$ and hence on $\sAlg$.
Therefore, after choosing inverses for all invertible bimodules and remembering the induced canonical isomorphisms $\overline{M}^{\op -1} \cong \overline{M^{\op -1}}$, this provides the data of the functors $\rho(-)$ and $\rho(R)$ commuting.
We have thus specified $R_{-,R}$ and $R_{R,-}$.
Since also $\overline{\overline{A}} = A$ and $A^{\op \op} = A$ strictly, the remaining higher compatibilities $\omega_{g_1,g_2,g_3}$ are identities too.

We turn to the compatibility data between $O_1$ and $B\Z_2^F$.
As a special case of Equation~\eqref{eq:left-right} and~\eqref{eq:op-iso}, we have
\[
A_{(-1)^F} \cong {}_{(-1)^F} A \quad a (-1)^F \mapsto (-1)^F (-1)^{|a|} a
\]
and $(A_{(-1)^F})^{\op} \cong {}_{(-1)^F} (A^{\op})$.
Also recall that $A_{(-1)^F}$ is canonically a self-inverse using the fact that the algebra homomorphism $(-1)F$ squares to one.
Combining these isomorphisms yields the modification specifying how $(-1)^F$ preserves duals as explained in Section \ref{Sec:O2-action}.
In other words, we have specified the $R_{\gamma_1,\gamma_2}$-type data of the action coming from the equation $(-) (-1)^F = (-1)^F (-)$ of $1$-morphisms in the $2$-group $O_1 \times \Z_2^F$.

The only compatibility data between the $\Z_2^R \times B\Z_2^F$-action and the Serre is the isomorphism $\overline{A^*} \cong \overline{A}^*$ which follows from Proposition \ref{prop:preserve serre}. 
It is given by $\overline{f}(\overline{a}) = \overline{f(a)}$.

%

%
%
%
%
%

	 \subsection{Induced actions on the bicategory of superalgebras}
 \label{Sec:induced action on sAlg}
 
In this section we show that a symmetric monoidal action of a $2$-group on $\sVect$ canonically induces a symmetric action on $\sAlg$. 
The main motivation is to get a symmetric monoidal $\Z_2^B \times B \Z_2^F$-action on $\sAlg$ induced by the action of Section \ref{Sec:action on sVect}. 
The proofs in this section are long but straightforward and only included for completeness.
We are not assuming any dualizability and don't restrict to the core.

From a more conceptual perspective, we will start by at least partially constructing a functor $\Aut_{sym-mon} \sVect \to \Aut_{sym-mon} \sAlg$ of monoidal $2$-categories.
Then we get an action of $\mathcal{G}$ on $\sAlg$ by composing
\[
\mathcal{G} \to \Aut_{sym-mon} \sVect \to \Aut_{sym-mon} \sAlg.
\]
Therefore we start with several lemmas on defining this functor on objects and morphisms of $\Aut_{sym-mon} \sVect$, i.e. how to lift symmetric monoidal functors and monoidal natural transformations from $\sVect$ to $\sAlg$.
We will often make use of the symmetric monoidal functor $\sAlg_1 \to \sAlg$.

\begin{remark}
Let $\mathcal{G}$ be a $2$-group with a symmetric monoidal action on a symmetric monoidal $1$-category $\mathcal{C}$.
Suppose that $\mathcal{C}$ satisfies the properties necessary to make a well-defined Morita $2$-category of $E_1$-algebras in $\mathcal{C}$.
Then the results in this section should generalize to obtain a $\mathcal{G}$-action on this.
\end{remark}

\begin{lemma}
Let $F: \sVect \to \sVect$ be a monoidal equivalence. 
Then it induces a functor $\hat{F}: \sAlg \to \sAlg$ which on objects maps the monoid
\[
(A, \eta_A: \C \to A, \mu_A: A \otimes A \to A)
\]
 in $\sVect$ to
\begin{align*}
\hat{F} (A) := (F[A], F[\eta_A]: \C \to F[\C] \to F[A], F[A] \otimes F[A] \cong F[A \otimes A] \xrightarrow{F[\mu_A]} F[A])
\end{align*}
where we used the monoidality data and unitality of $F$.
When $F_1, F_2$ are two monoidal functors, then $\widehat{F_2 \circ F_1} = \hat{F}_1 \circ \hat{F}_2$.
Moreover, if $F$ is symmetric monoidal, then it induces a symmetric monoidal structure on $\hat{F}$.
\end{lemma}
\begin{proof}
Without loss of generality we can assume $F$ is strictly unital and associative.
To verify associativity of the algebra $\hat{F}(A)$, consider the diagram
\[
\begin{tikzcd}
\ &
F[A] \otimes F[A] \otimes F[A] \arrow[dl] \arrow[dr] &
\
\\
F[A \otimes A] \otimes F[A] \arrow[d,"{F[\mu_A] \otimes \id_{F[A]}}", swap] \arrow[dr] &
\ &
F[A] \otimes F[A \otimes A] \arrow[d, "{ \id_{F[A]} \otimes F[\mu_A]}"] \arrow[dl]
\\
F[A] \otimes F[A] \arrow[d] &
F[A \otimes A \otimes A] \arrow[dl, "{F[\mu_A \otimes \id_{A}]}"] \arrow[dr, "{F[\id_A \otimes \mu_A]}", swap] &
F[A] \otimes F[A] \arrow[d]
\\
F[A \otimes A] \arrow[dr, "{F[\mu_A]}"] &
\ &
F[A \otimes A] \arrow[dl, "{F[\mu_A]}"]
\\
\ &
F[A] &
\
\end{tikzcd}
\]
The upper square commutes by the condition that $F$ is monoidal.
The left and right square commute by naturality of the isomorphism $F[V \otimes W] \cong F[V] \otimes F[W]$.
For the left square for example, we take $V = A \otimes A, W = A$ and the morphism to be $\mu_A \otimes \id_A$.
The lower square commutes by associativity of $A$.
It is easy to see that $\hat{F}[A]$ is unital.

Next we define a functor $\hat{F}: Bim(A,B) \to Bim(\hat{F}[A], \hat{F}[B])$.
Similarly as for algebras, if $(M, l_M: B \otimes M \to M, r_M: M \otimes A \to M)$ is a $(B,A)$-bimodule, we define $\hat{F}[M]$ by applying $F$ on all vector spaces and maps, applying the monoidality data whereever necessary.
The three associativity conditions equating the several ways to compose three elements holds for $\hat{F}[M]$ by a similar diagram as the one above.

If $\phi: M \to M'$ is a morphism of $(B,A)$-bimodules, we define $\hat{F}[\phi]$ by $F[\phi]$.
This is a left $\hat{F}(B)$-module map because the diagram
\[
\begin{tikzcd}
F(B) \otimes F(M) \arrow[r] \arrow[d, "{\id_{F(A)} \otimes F(\phi)}"]&
F(B \otimes M) \arrow[r, "{l_M}"] \arrow[d,"{F(\id_A \otimes \phi)}"]&
F(M) \arrow[d,"{F(\phi)}"]
\\
F(B) \otimes F(M') \arrow[r] &
F(B \otimes M') \arrow[r, "{l_{M'}}"] &
F(M')
\end{tikzcd}
\]
commutes in vector spaces. 
Indeed, the right square commutes because $\phi$ is a left $B$-module map and the left square commutes because $F[V \otimes W] \cong F[V] \otimes F[W]$ is natural.
Similarly, $\hat{F}(\phi)$ is a right $\hat{F}(A)$-module map.
The map $\hat{F}: Bim(A,B) \to Bim(\hat{F}[A], \hat{F}[B])$ maps $\id_M$ to $\id_{\hat{F}[M]}$ and is a functor because $F$ is.

Next we have to give the data specifying that $\hat{F}$ preserves composition of $1$-morphisms.
Explicitly, we have to define a natural transformation
\[
\hat{F}[N \otimes_B M] \cong \hat{F}[N] \otimes_{\hat{F}[B]} \hat{F}[M]
\]
of functors
\[
Bim(B,C) \times Bim(A,B) \to Bim(\hat{F}[A], \hat{F}[B]).
\]
We claim the monoidality data of $F$ factors uniquely to the desired isomorphism
\[
\begin{tikzcd}
\hat{F}[N \otimes_B M] 
&
 \hat{F}[N] \otimes_{\hat{F}[B]} \hat{F}[M] \arrow[l,dashed] 
\\
\hat{F}[N \otimes_\C M]  \arrow[u]
&
\hat{F}[N] \otimes_{\C} \hat{F}[M] \arrow[l] \arrow[u]
\end{tikzcd}
\]
For showing that the dashed arrow exists, it suffices to show that the composition from the lower right to the upper left corner through the left is $B$-balanced.
For this, we have to show the diagram
\[
\begin{tikzcd}[column sep=50]
F[N] \otimes F[B] \otimes F[M] \arrow[r] \arrow[d]
&
F[N] \otimes F[B \otimes M] \arrow[d] \arrow[r, "{\id_{F[N]} \otimes F[l_M]}"]
&
F[N] \otimes F[M] \arrow[d]
\\
F[N \otimes B] \otimes F[M] \arrow[r] \arrow[d,"{F[r_N] \otimes \id_{F[M]}}"]
&
F[N \otimes B \otimes M] \arrow[r,"{F[\id_{N} \otimes l_M]}"] \arrow[d,"{F[r_N \otimes \id_M]}"]
&
F[N \otimes M] \arrow[d]
\\
F[N] \otimes F[M]  \arrow[r]
&
F[N \otimes M] \arrow[r]
& 
F[N \otimes_B M]
\end{tikzcd}
\]
commutes.
The upper left square commutes because $F$ is monoidal, the upper right and lower left square commute because $F[V \otimes W] \cong F[V \otimes W]$ is natural and the lower right square commutes by definition of the tensor product over $B$.
We now show that the map
 \[
\hat{F}[N \otimes_B M] \to \hat{F}[N] \otimes_{\hat{F}[B]} \hat{F}[M]
\]
is an isomorphism.
It is surjective as the left and lower map in the above square are surjective.
For injectivity we use that $F$ is an equivalence. 
Let $F^{-1}$ be a choice of monoidal inverse of $F$ with corresponding monoidal natural isomorphism $F^{-1} F \Rightarrow \id_{\sVect}$.
To show that $\hat{F}[N \otimes_B M]$ is a tensor product of $\hat{F}[N]$ over $\hat{F}[B]$ with $\hat{F}[M]$ it suffices to verify the universal property.
Given how the map that we want to be an isomorphism is defined, we have to show that for every $\hat{F}[B]$-balanced map $f: \hat{F}[M] \otimes \hat{F}[N] \to V$ into any supervector space $V$ there exists a unique lift $\tilde{f}: \hat{F}[M \otimes_B N] \to V$ with respect to the composition
\[
\hat{F}[M] \otimes \hat{F}[N] \cong \hat{F}[M \otimes N] \to \hat{F}[M \otimes_B N].
\]
Define $f':N \otimes M \to F^{-1}[V]$ by the composition
\[
N \otimes M \cong F^{-1} F[N \otimes M] \cong F^{-1} [F[N] \otimes F[M]] \xrightarrow{F^{-1}(f)} F^{-1}[V].
\]
If we can show that $f'$ is $B$-balanced, then there exists a unique lift $\tilde{f'}:N \otimes_B M \to F^{-1}[V]$. 
Using that $F$ is fully faithful, we can then apply $F$ to $f'$ and obtain that there is a unique lift in the diagram
\[
\begin{tikzcd}
\ &
F[N \otimes M] \arrow[ld,equals] \arrow[r] \arrow[d] &
F[N \otimes_B M] 
\\
F[N \otimes M] \arrow[d] \arrow[r] &
F F^{-1} F[N \otimes M] \arrow[d] &
\
\\
F[N] \otimes F[M] \arrow[d] \arrow[r] &
F F^{-1} [F[N] \otimes F[M]] \arrow[d] &
\
\\
V \arrow[r] &
F F^{-1} [V] \arrow[uuur, dashed, bend right] & 
\
\end{tikzcd}
\]
The left vertical composition is $f$ and the two left squares commute because $\id_{\sVect} \Rightarrow F F^{-1}$ is a monoidal natural transformation.
Therefore there is a desired lift $\tilde{f}$ in that case and so it suffices to show that $f'$ is $B$-balanced.
This means we have to show the following diagram commutes
\[
\begin{tikzcd}
N \otimes B \otimes M \arrow[r,"l_M"] \arrow[d,"r_N"] \arrow[dr] &
N \otimes M \arrow[rr] & 
\ &
F^{-1} F[N \otimes M]
\\
N \otimes M  \arrow[dd]&
F^{-1} F[N \otimes B \otimes M] \arrow[ddl,"{r_N}", bend right=8, swap] \arrow[rru,"{l_M}"]&
F^{-1}[F[N] \otimes F[B \otimes M]] \arrow[l] \arrow[rd,"l_M"]&
\
\\
\ &
F^{-1}[F[N \otimes B] \otimes F[M]] \arrow[u] \arrow[dr,"{r_N}"] &
F^{-1}[F[N] \otimes F[B] \otimes F[M]] \arrow[l] \arrow[u] \arrow[d,"{r_{\hat{F}[N]}}"] \arrow[r,"{l_{\hat{F}[M]}}"]  &
F^{-1}[F[N] \otimes F[M]] \arrow[d] \arrow[uu]
\\
F^{-1} F[N \otimes M] &
\ & 
F^{-1}[F[N] \otimes F[M]]  \arrow[ll] \arrow[r]&
F^{-1}[V]
\end{tikzcd}
\]
where we used a natural isomorphism $\id \Rightarrow F^{-1} F$.
The lower right square commutes because $f$ is $\hat{F}[B]$-balanced. The triangles north and west of that square commute by definition of the bimodule structure on $\hat{F}[M]$.
The middle square commutes by monoidal associativity of $F$.
The remaining quadrilaterals commute by naturality of $\id \Rightarrow F^{-1} F$ and the iso $F[V] \otimes F[W] \cong F[V \otimes W]$.

Now let $\psi: N \to N'$ be a map of $(C,B)$-bimodules and $\phi: M \to M'$ a map of $(B,A)$-bimodules.
Then the naturality square
\[
\begin{tikzcd}
\hat{F}[N \otimes_B M] \arrow[d, "{\hat{F}[\psi \otimes_B \phi]}"] & \arrow[l] \hat{F}[N] \otimes_{\hat{F}[B]} \hat{F}[M] \arrow[d,"{\hat{F}[\psi] \otimes_{\hat{F}[B]} \hat{F}[\phi]}"]
\\
\hat{F}[N' \otimes_B M']  & \hat{F}[N'] \otimes_{\hat{F}[B]} \hat{F}[M'] \arrow[l] 
\end{tikzcd}
\]
commutes because it is induced by the square in vector spaces
\[
\begin{tikzcd}
F[N \otimes M] \arrow[d, "{F[\psi \otimes \phi]}"] & \arrow[l] F[N] \otimes F[M] \arrow[d,"{F[\psi] \otimes F[\phi]}"]
\\
F[N' \otimes M']  & F[N'] \otimes F[M'] \arrow[l] 
\end{tikzcd}
\]
which commutes by naturality of $F[V \otimes W] \cong F[V] \otimes F[W]$.

Note that the algebra $\hat{F}[A]$ seen as a $(\hat{F}[A], \hat{F}[A])$-bimodule is equal to the functor $\hat{F}: Bim(A,A) \to Bim(\hat{F}[A], \hat{F}[A])$ applied to the $(A,A)$-bimodule $A$.
This shows that $\hat{F}$ preserves the identity $1$-morphisms.
$\hat{F}$ satisfies the associativity axiom for composition of three $1$-morphisms, because $F$ does for the tensor products over $\C$ that these compositions are induced by.
Similarly, for an $(A,B)$-bimodule the diagram
\[
\begin{tikzcd}
\hat{F}[A] \otimes_{\hat{F}[A]} \hat{F}[M] \arrow[r] \arrow[d] & \hat{F}[A \otimes_A M] \arrow[ld]
\\
\hat{F}[M] & \
\end{tikzcd}
\]
is commutative because it is for tensoring over $\C$ and the same for the other side.
We have shown that $\hat{F}$ is a functor between $2$-categories.
Looking at the constructions, we have the equality $\widehat{F_2 \circ F_1} = \hat{F}_1 \circ \hat{F}_2$ on the nose.

Now assume $F$ is symmetric.
We start by defining the natural isomorphism $\hat{\phi}_{A_1,A_2}:\hat{F}[A_1 \otimes A_2] \cong \hat{F}[A_1] \otimes \hat{F}[A_2]$ between functors $\sAlg \times \sAlg \to \sAlg$.
For an object $(A,\mu_A, \eta_A)$ the monoidal data $\phi_{A_1, A_2}: F[A_1 \otimes A_2] \cong F[A_1] \otimes F[A_2]$ is an isomorphism of supervector spaces.
We have to verify that it is an algebra map so that we can take the natural isomorphism to be the induced invertible bimodule on the object $A$.
Recall that the tensor product of two algebras has its multiplication defined as
\[
A_1 \otimes A_2 \otimes A_1 \otimes A_2 \xrightarrow{\id \otimes \sigma_{A_2, A_1} \otimes \id} A_1 \otimes A_1 \otimes A_2 \otimes A_2 \xrightarrow{\mu_{A_1} \otimes \mu_{A_2}} A_1 \otimes A_2.
\]
where $\sigma$ denotes the braiding.
Therefore we consider the diagram
\[
\begin{tikzcd}
F[A_1 \otimes A_2] \otimes F[A_1 \otimes A_2]  \arrow[rr] \arrow[d]
& \
& F[A_1] \otimes F[A_2] \otimes F[A_1] \otimes F[A_2] \arrow[d,"{\id \otimes \sigma_{F[A_2] \otimes F[A_1]} \otimes \id}"] \arrow[dl]
\\
F[A_1 \otimes A_2 \otimes A_1 \otimes A_2]  \arrow[d,"{F(\id \otimes \sigma_{A_2, A_1} \otimes \id)}"] 
& F[A_1] \otimes F[A_2 \otimes A_1] \otimes F[A_2]  \arrow[d,"{\id \otimes F[\sigma_{A_2, A_1}] \otimes \id}", swap] \arrow[l]
& F[A_1] \otimes F[A_1] \otimes F[A_2] \otimes F[A_2] \arrow[d] \arrow[dl]
\\
F[A_1 \otimes A_1 \otimes A_2 \otimes A_2] \arrow[d,"{F[\mu_{A_1} \otimes \mu_{A_2}]}"] 
& F[A_1] \otimes F[A_1 \otimes A_2] \otimes F[A_2]  \arrow[l]
& F[A_1 \otimes A_1] \otimes F[A_2 \otimes A_2] \arrow[d,"{F[\mu_{A_1}] \otimes F[\mu_{A_2}]}"]\arrow[ll, bend left=12, shorten >=.3cm]
\\
F[A_1 \otimes A_2] \arrow[rr] 
& \
& F[A_1] \otimes F[A_2]
\end{tikzcd}
\]
Going from the upper left to the lower right corner through the left is the multiplication in $\hat{F}[A_1 \otimes A_2]$ followed by the isomorphism to $\hat{F}[A_1] \otimes \hat{F}[A_2]$, while going through the right amounts to first doing the isomorphism and then multiplication in $\hat{F}[A_1] \otimes \hat{F}[A_2]$.
The upper left corner commutes by various applications of the associativity condition a monoidal functor satisfies.
The left square commutes by naturality of $\phi$.
The upper right parallelogram commutes because $F$ is symmetric.
The remaining lower part commutes by naturality of $\phi$ and the associativity condition.

We now set 
\[
\hat{\phi}_{A_1,A_2} := \hat{F}[A_1 \otimes A_2]_{\phi_{A_1,A_2}}
\]
which is an invertible $1$-morphism from $\hat{F}[A_1] \otimes \hat{F}[A_2]$ to $\hat{F}[A_1 \otimes A_2]$.
Let $M_1$ be a $(B_1, A_1)$-bimodule and $M_2$ a $(B_2, A_2)$-bimodule.
For the naturality data of $\hat{\phi}$, we have to specify a natural isomorphism
\[
\hat{F}[M_1 \otimes M_2] \otimes_{\hat{F}[A_1 \otimes A_2]} \hat{F}[A_1 \otimes A_2]_{\phi_{A_1,A_2}} \cong \hat{F}[B_1 \otimes B_2]_{\phi_{B_1,B_2}} \otimes_{\hat{F}[B_1] \otimes \hat{F}[B_2]} (\hat{F}[M_1] \otimes \hat{F}[M_2])
\]
Using the bimodule actions, the right hand side is canonically isomorphic as a vector space to $F[M_1] \otimes F[M_2]$ while the left hand side is canonically isomorphic to $F[M_1 \otimes M_2]$.
This becomes a bimodule isomorphism if we equip $F[M_1 \otimes M_2]$ with the left $\hat{F}[B_1 \otimes B_2]$ action using the definition of $\hat{F}$ on modules from before, but with a right $\hat{F}[A_1] \otimes \hat{F}[A_2]$-action given by first applying $\phi_{A_1,A_2}$ before using the right $\hat{F}[A_1 \otimes A_2]$-action.
Similarly we equip $F[M_1] \otimes F[M_2]$ with its canonical right $\hat{F}[A_1] \otimes \hat{F}[A_2]$-action, but compose the left $\hat{F}[B_1] \otimes \hat{F}[B_2]$-action with $\phi_{B_1,B_2}^{-1}$ to make it into a $\hat{F}[B_1 \otimes B_2]$-action.
Under these identifications, we define $\hat{\phi}_{M_1,M_2} := \phi_{M_1,M_2}^{-1}$. \footnote{The reason for the awkward inverse in this formula is a convention issue. Note that the usual convention for the direction of the $2$-morphisms filling the naturality square of a natural transformation $F \Rightarrow G$ in $2$-categories goes in the direction $G \Rightarrow F$. We also use this convention, see the second bullet of Defininition \ref{Def: Hfixed point}.}

We now show that this is a bimodule map.
Recalling how the definition of the tensor product of modules is defined, the diagram that has to commute for the preservation of left multiplication is
\[
\begin{tikzcd}
F[A_1 \otimes A_2] \otimes F[M_1 \otimes M_2] \arrow[d] &
F[A_1 \otimes A_2] \otimes F[M_1] \otimes F[M_2] \arrow[l]
\\
F[A_1 \otimes A_2 \otimes M_1 \otimes M_2] \arrow[dd,"{F[\id \otimes \sigma_{A_2, M_1} \otimes \id]}"] & 
F[A_1] \otimes F[A_2] \otimes F[M_1] \otimes F[M_2] \arrow[u] \arrow[d,"{\id \otimes \sigma_{F[A_2], F[M_1]} \otimes \id}"]
\\
\ &
F[A_1] \otimes F[M_1] \otimes F[A_2] \otimes F[M_2] \arrow[d]
\\
F[A_1 \otimes M_1 \otimes A_2 \otimes M_2] \arrow[d,"{F[l_{M_1} \otimes l_{M_2}]}"] & 
F[A_1 \otimes M_1] \otimes F[A_2 \otimes M_2] \arrow[l] \arrow[d,"{F[l_{M_1}] \otimes F[l_{M_2}]}"]
\\
F[M_1 \otimes M_2] &
F[M_1 ] \otimes F[M_2] \arrow[l]
\end{tikzcd}
\]
This amounts to applying symmetry of $F$ and associativity with respect to $\otimes$ in the upper rectangle and naturality of $\phi$ in the lower rectangle.
Compatibility for the right $\hat{F}[B_1] \otimes \hat{F}[B_2]$-action is analogous.

Next we have to show that the above naturality data defines a natural isomorphism between functors
\[
\catf{Bim}(A_1,B_1) \times \catf{Bim}(A_2, B_2) \to \catf{Bim}(\hat{F}[A_1] \otimes \hat{F}[A_2], \hat{F}[B_1 \otimes B_2])
\]
For showing this is a natural isomorphism, let $\phi_1: M_1 \to M_1'$ be a $(B_1,A_1)$-bimodule map and $\phi_2: M_2 \to M_2'$ a $(B_2, A_2)$-bimodule map.
We have to show that the diagram
\[
\begin{tikzcd}
\hat{F}[M_1 \otimes M_2] \otimes_{\hat{F}[A_1 \otimes A_2]} \hat{F}[A_1 \otimes A_2]_{\phi_{A_1,A_2}} \arrow[r,"{\hat{\phi}_{M_1,M_2}}"] \arrow[d,"{\hat{F}[\phi_1 \otimes \phi_2] \otimes \id}"]
& \hat{F}[B_1 \otimes B_2]_{\phi_{B_1,B_2}} \otimes_{\hat{F}[B_1] \otimes \hat{F}[B_2]} (\hat{F}[M_1] \otimes \hat{F}[M_2]) \arrow[d,"{\id \otimes \hat{F}[\phi_1] \otimes \hat{F}[\phi_2]}"]
\\
\hat{F}[M_1' \otimes M_2'] \otimes_{\hat{F}[A_1 \otimes A_2]} \hat{F}[A_1 \otimes A_2]_{\phi_{A_1,A_2}} \arrow[r,"{\hat{\phi}_{M_1',M_2'}}"]
& \hat{F}[B_1 \otimes B_2]_{\phi_{B_1,B_2}} \otimes_{\hat{F}[B_1] \otimes \hat{F}[B_2]} (\hat{F}[M_1'] \otimes \hat{F}[M_2'])
\end{tikzcd}
\]
commutes.
This follows because the underlying diagram in supervector spaces
\[
\begin{tikzcd}
F[M_1 \otimes M_2] \arrow[d,"{F[\phi_1 \otimes \phi_2]}"] & F[M_1] \otimes F[M_2] \arrow[d,"{F[\phi_1] \otimes F[\phi_2]}"] \arrow[l,"{\phi_{M_1,M_2}}"]
\\
F[M_1' \otimes M_2']  & F[M_1'] \otimes F[M_2'] \arrow[l,"{\phi_{M_1',M_2'}}"] 
\end{tikzcd}
\]
commutes.
Finally we have to show that
\[
\begin{tikzcd}
\hat{F}[A_1 \otimes A_2] \otimes_{\hat{F}[A_1 \otimes A_2]} \hat{F}[A_1 \otimes A_2]_{\phi_{A_1, A_2}}  \arrow[d,"{l_{\hat{F}[A_1 \otimes A_2]_{\phi_{A_1, A_2}} }}"]
& \hat{F}[A_1 \otimes A_2]_{\phi_{A_1, A_2}}  \otimes_{\hat{F}[A_1] \otimes \hat{F}[A_2]} (\hat{F}[A_1] \otimes \hat{F}[A_2]) \arrow[dl,"{r_{\hat{F}[A_1 \otimes A_2]_{\phi_{A_1, A_2}} }}"] \arrow[l,"{\hat{\phi_{M_1, M_2}}}"]
\\
\hat{F}[A_1 \otimes A_2]_{\phi_{A_1, A_2}} &
\end{tikzcd}
\]
commutes, where $M_1$ is the $(A_1,A_1)$-bimodule $A_1$ and $M_2$ is the $(A_2,A_2)$-bimodule $A_2$ (we did not plug in $A_1$ as this would lead to a conflict of notation with $\hat{\phi}_{A_1, A_2}$ for $A_1,A_2$ algebras).
This is immediate after writing out the definitions given that only right multiplication is altered by $\phi_{A_1, A_2}$.
We have now shown that $\hat{\phi}$ defines a natural isomorphism $F[A_1 \otimes A_2] \cong F[A_1] \otimes F[A_2]$.
Because we can assume strict unitality of $F$ and corresponding monoidal data $F[A \otimes \C] \cong F[A] \otimes F[\C] = F[A]$, we also get strict unitality of $\hat{F}$.
The associator for $\hat{F}$ is the identity using that $F$ is associative on the nose.
We have now shown that $\hat{F}$ is monoidal.

To show that $\hat{F}$ is symmetric, take $\hat{F}[A_1 \otimes A_2] \cong \hat{F}[A_2 \otimes A_1]$ to be the bimodule induced by the algebra isomorphism $\sigma_{A_1,A_2}$.
The symmetry data requires filling the diagram of $1$-morphisms
\[
\begin{tikzcd}
\ &
\hat{F}[A_1 \otimes A_2] \arrow[ld] &
\
\\
\hat{F}[A_1] \otimes \hat{F}[A_2] &
\ &
\hat{F}[A_2 \otimes A_1] \arrow[lu] \arrow[ld]
\\
\ &
\hat{F}[A_2]  \otimes \hat{F}[A_1] \arrow[lu] &
\
\end{tikzcd}
\]
The diagram is induced by a commutative diagram of homomorphisms of algebras and so can be filled by isomorphisms of the form $C_g \otimes_B B_f \cong C_{gf}$.
For the same reason, the fact that $\sigma_{A_2, A_1} \sigma_{A_1,A_2} = \id_{A_1 \otimes A_2}$ gives a $2$-isomorphism from the composition
\[
\hat{F}[A_1 \otimes A_2] \to \hat{F}[A_2 \otimes A_1] \to \hat{F}[A_1 \otimes A_2]
\]
to the identity bimodule.
Braiding three elements in different order as in \cite[Figure 2.3]{schommerpriesthesis} also has a canonical filling because of the corresponding condition for the braiding of $\sVect$.
Finally, the remaining conditions \cite[BHA1, BHA2]{mccrudden2000balanced} on being a symmetric monoidal functor follow because all $2$-morphisms involved are induced by strict equatlities in the $1$-category $\sAlg_1$.
\end{proof}

\begin{remark}
Note that we needed $F$ to be monoidal to induce a functor $\hat{F}$ at all.
Also if $F$ is not symmetric, $\hat{F}$ will not be monoidal.
But if $F$ is symmetric monoidal, then $\hat{F}$ is even symmetric monoidal.
\end{remark}

\begin{lemma}
Let $\alpha: F \Rightarrow G$ be a monoidal natural isomorphism between monoidal functors on $\sVect$.
Then $\alpha$ induces a natural transformation on $\sAlg$ where on objects it is given by mapping $A$ to the $(\hat{G} (A), \hat{F} (A))$-bimodule induced by the algebra homomorphism $\hat{F}(A) \to \hat{G}(A)$ given by $\alpha(A): F(A) \to G(A)$ as a map of vector spaces.
If $F$ and $G$ are symmetric, then $\alpha$ is canonically symmetric monoidal.
\end{lemma}
\begin{proof}
For $A$ an algebra, $\alpha[A]: F[A] \to G[A]$ is a linear isomorphism.
Because $\alpha$ is monoidal and natural, this defines an algebra homomorphism $\hat{F}[A] \to \hat{G}[A]$.
Define
\[
\hat{\alpha}[A] := \hat{G}[A]_{\alpha[A]}
\]
to be the $(\hat{G}[A], \hat{F}[A])$-bimodule induced by this algebra homomorphism.
Because $\alpha[A]$ is an invertible linear map, this bimodule is invertible.
Being a natural isomorphism between functors between $2$-categories, we still need to provide additional naturality data for $\hat{\alpha}$ given a $(B,A)$-bimodule $M$.
For this we have to provide an isomorphism
\[
\hat{G}[M] \otimes_{\hat{G}[A]} \hat{G}[A]_{\alpha[A]} \cong \hat{G}[B]_{\alpha[B]} \otimes_{\hat{F}[B]} \hat{F}[M] 
\]
of $(\hat{G}[B], \hat{F}[A])$-bimodules.
The situation is similar as in the proof of naturality of the monoidality data of $\hat{F}$ in the last lemma.
The right hand side is canonically isomorphic to $F[M]$ with its canonical right $\hat{F}[A]$-action, but we compose the left $\hat{F}[B]$-action with $\alpha[A]^{-1}$ to make it into a $\hat{G}[B]$-action. 
The left hand side is $G[M]$ instead has an interesting right $\hat{F}[A]$-action given by first applying $\alpha[A]$ before using the right $\hat{G}[A]$-action.
Then we define $\hat{\alpha}[M] := \alpha[M]^{-1}$. 

We now have to show that this is a bimodule map for the bimodule structures described above.
For showing it is a right $\hat{F}[A]$-module map, consider the diagram
\[
\begin{tikzcd}
F[M] \otimes F[A] \arrow[d,"{\alpha[M] \otimes \id}"] \arrow[r]
&
F[M \otimes A]  \arrow[dd,"{\alpha[M \otimes A]}"] \arrow[r, "{F[r_M]}"]
&
F[M] \arrow[dd,"{\alpha[M]}"]
\\
G[M] \otimes F[A] \arrow[d,"{\id \otimes \alpha[A]}"]
&
\
&
\
\\
G[M] \otimes G[A] \arrow[r]
& 
G[M \otimes A] \arrow[r,"{G[r_M]}"]
&
G[M]
\end{tikzcd}
\]
The left rectangle commutes by monoidality of $\alpha$ and the right rectangle commutes by naturality of $\alpha$.
Starting in the left upper corner and moving through the diagram along the left compositions corresponds to first applying the function $\alpha[M]$ and then multiplying from the right with an element of $F[A]$, while the other direction corresponds to first multiplying from the right with $F[A]$.
A similar diagram shows that $\hat{\alpha}[M]$ is a left $\hat{G}[B]$-module map.

Next we have to show naturality of this isomorphism between functors of $1$-categories
\[
Bim(A,B) \to Bim(\hat{F}[A], \hat{G}[B])
\]
So let $\phi: M \to M'$ be a $(B,A)$-bimodule map.
The diagram that we have to show commutes is
\[
\begin{tikzcd}
\hat{G}[M] \otimes_{\hat{G}[A]} \hat{G}[A]_{\alpha[A]} \arrow[r,"{\hat{\alpha}[M]}"]  \arrow[d,"{\hat{G}[\phi] \otimes \id}"]
&
\hat{G}[B]_{\alpha[B]} \otimes_{\hat{F}[B]} \hat{F}[M] \arrow[d,"{\id \otimes \hat{F}[\phi]}"]
\\
\hat{G}[M'] \otimes_{\hat{G}[A]} \hat{G}[A]_{\alpha[A]} \arrow[r,"{\hat{\alpha}[M']}"] 
&
\hat{G}[B]_{\alpha[B]} \otimes_{\hat{F}[B]} \hat{F}[M'] 
\end{tikzcd}
\]
This diagram clearly commutes in vector spaces by naturality of $\alpha$ where we can equivalently describe it as the diagram
\[
\begin{tikzcd}
G[M] \arrow[r,"{\alpha[M]}"] \arrow[d,"{G[\phi]}"] & F[M] \arrow[d,"{F[\phi]}"]
\\
G[M'] \arrow[r,"{\alpha[M']}"] & F[M']
\end{tikzcd}
\]
To show that $\hat{\alpha}$ is a natural transformation we still have to show that if $N$ is a $(C,B)$-bimodule, then
\[
\begin{tikzcd}
\hat{G}[N \otimes_B M] \otimes_{\hat{G}[A]} \hat{G}[A]_{\alpha[A]} \arrow[dd,"{\hat{\alpha}[N \otimes_B M]}"]
&
\hat{G}[N] \otimes_{\hat{G}[B]} \hat{G}[M] \otimes_{\hat{G}[A]} \hat{G}[A]_{\alpha[A]} \arrow[l] \arrow[d,"{\id \otimes \hat{\alpha}[M]}"]
\\
\ 
&
\hat{G}[N] \otimes_{\hat{G}[B]} \hat{G}[B]_{\alpha[B]} \otimes_{\hat{F}[B]} \hat{F}[M] \arrow[d,"{\hat{\alpha}[M] \otimes \id}"]
\\
\hat{G}[C]_{\alpha[C]} \otimes_{\hat{F}[C]} \hat{G}[N \otimes_B M]
&
\hat{G}[C]_{\alpha[C]} \otimes_{\hat{F}[C]} \hat{F}[N] \otimes_{\hat{F}[B]} \hat{F}[M] \arrow[l]
\end{tikzcd}
\]
commutes.
Recalling the identifications above, this follows because it is induced by the diagram
\[
\begin{tikzcd}
G[N \otimes M]
&
G[N] \otimes G[M] \arrow[l]
\\
\
&
G[N] \otimes F[M] \arrow[u,"{\id \otimes \alpha[M]}", swap]
\\
F[N\otimes M]  \arrow[uu,"{\alpha[M \otimes N]}"]
&
F[N] \otimes F[M] \arrow[u,"{\alpha[N] \otimes \id}", swap] \arrow[l]
\end{tikzcd}
\]
in vector spaces which commutes by monoidality of $\alpha$.
The final diagram we have to check for $\hat{\alpha}$ being a natural transformation is the following unitality condition:
\[
\begin{tikzcd}
\hat{G}[A] \otimes_{\hat{G}[A]} \hat{G}[A]_{\alpha[A]} \arrow[r,"{l_{\hat{G}[A]_{\alpha[A]}}}"] \arrow[dr,"{\hat{\alpha}[\hat{G}[A]]}", swap]
&
\hat{G}[A]_{\alpha[A]}
\\
\
&
\hat{G}[A]_{\alpha[A]} \otimes_{\hat{F}[A]} \hat{F}[A] \arrow[u,"{r_{\hat{G}[A]_{\alpha[A]}}}", swap]
\end{tikzcd}
\]
Here the arrow going diagonally down and right is the natural transformation $\hat{\alpha}$ applied to the $(\hat{G}[A], \hat{G}[A])$-bimodule $\hat{G}[A]$.
Recalling that only the right action of $\hat{G}[A]_{\alpha[A]}$ is twisted by $\alpha[A]$, this diagram commutes.
So $\hat{\alpha}$ is a natural transformation.

Now assume $\hat{F}$ and $\hat{G}$ are symmetric monoidal.
\end{proof}

\begin{remark}
Note that we needed $\alpha$ to be monoidal to induce a natural transformation $\hat{\alpha}$ at all. 
But in that case it is automatically symmetric monoidal if $F,G$ are.
\end{remark}

\begin{theorem}
Let a $2$-group $\mathcal{G}$ act symmetric monoidaly on $\sVect$.
There is a canonical induced action by $\mathcal{G}$ on $\sAlg$.
Moreover, all functors and natural transformations involved in the definition are symmetric monoidal.
\end{theorem}
\begin{proof}
Recall that the data of a symmetric action on a $1$-category $\mathcal{C}$ consists of
\begin{itemize}
\item symmetric monoidal functors $\tau(g): \mathcal{C} \to \mathcal{C}$, strictly unital without loss of generality;
\item (pointed) monoidal natural transformations $T_{g_2,g_1}: \tau(g_2 g_1) \Longrightarrow \tau(g_2) \circ \tau(g_1)$;
\item monoidal natural transformations $\tau(\gamma): \tau(g) \Longrightarrow \tau(g')$ for every path $\gamma: g \to g'$ in $\mathcal{G}$.
\end{itemize}
satisfying the conditions that 
\begin{itemize}
\item $\tau(\gamma' \circ \gamma) = \tau(\gamma') \circ \tau(\gamma)$;
\item $T_{\bullet, \bullet}$ is `associative' in its two arguments;
\item The diagram
\begin{equation}
\label{eq: T-condition}
\begin{tikzcd}
\tau(g_2) \circ \tau(g_1) \arrow[d,"{\tau(\gamma_2) \bullet \tau(\gamma_1)}", Rightarrow] 
&
 \tau(g_2 g_1) \arrow[l, "{T_{g_2,g_1}}", Rightarrow, swap] \arrow[d,Rightarrow, "{\tau(\gamma_2 \otimes \gamma_1)}"]
\\
\tau(g_2') \circ \tau(g_1') 
&
 \tau(g_2' g_1') \arrow[l,"{T_{g_2',g_1'}}", Rightarrow]
\end{tikzcd}
\end{equation}
commutes.
\end{itemize}

We want to define an action $(\rho, R, \alpha, \omega)$ of $\mathcal{G}$ on $\sAlg$.
Let $\gamma: g \to g'$ be a path in $\mathcal{G}$.
Since $\tau(g)$ are a symmetric monoidal functors and $\tau(\gamma): \tau(g) \Rightarrow \tau(g')$ monoidal natural transformations, the last two lemmas yield canonically induced symmetric monoidal functors $\rho(g) := \widehat{\tau(g)}: \sAlg \to \sAlg$ and symmetric monoidal natural transformations $\rho(\gamma) := \widehat{\tau(\gamma)}: \rho(g) \Rightarrow \rho(g')$.
Moreover, the monoidal natural transformations $T_{g_2,g_1}: \tau(g_2 g_1) \Rightarrow \tau(g_2) \tau(g_1)$ induce symmetric monoidal natural transformations
\[
R_{g_2, g_1}: \rho(g_2 g_1) = \widehat{\tau(g_2 g_1)} \Rightarrow \widehat{\tau(g_2) \tau(g_1) } = \widehat{\tau(g_2)} \widehat{\tau(g_1) }.
\]
By Diagram \ref{eq: T-condition}, there is a diagram of algebra homomorphisms
\[
\begin{tikzcd}
\rho(g' g)[A] \arrow[r,"{T_{g'_1,g_1}[A]}"] \arrow[d,"{\tau(\gamma' \otimes \gamma)[A]}", swap]&
\rho(g_1') \rho(g_1) [A] \arrow[d,"{\tau(\gamma') \bullet \tau(\gamma)[A]}"]
\\
\rho(g_2' g_2)[A] \arrow[r,"{T_{g_2',g_2}[A]}"]&
\tau(g_2') \tau(g_2)[A]
\end{tikzcd}
\]
This implies that the diagram of induced bimodules
\[
\begin{tikzcd}
\rho(g' g)[A] \arrow[r,"{R_{g'_1,g_1}[A]}"] \arrow[d,"{\rho(\gamma' \otimes \gamma)[A]}", swap]&
\rho(g_1') \rho(g_1) [A] \arrow[d,"{\rho(\gamma') \bullet \rho(\gamma)[A]}"]
\\
\rho(g_2' g_2)[A] \arrow[r,"{R_{g_2',g_2}[A]}"]&
\rho(g_2') \rho(g_2)[A]
\end{tikzcd}
\]
can be filled by the identity $2$-morphism.
Moreover, on $2$-morphisms the two natural transformations given by either side are also equal, as they are induced by the same algebra homomorphism.
Hence we can set $R_{\gamma', \gamma} = \id$ to be the identity modification and the condition on being a modification is satisfied. 
Similarly, by the associativity for $T$ in its two arguments, the diagram defining $\omega$ has a strict filler.
Finally we can set $\alpha_{\gamma', \gamma} = \id$ because of the strict equality $\tau(\gamma') \tau(\gamma) = \tau(\gamma' \gamma)$.
Because so much of the data is trivial, all other conditions (pentagon for $\omega$, associativity of $\alpha$ and naturality condition of $\rho( - \otimes -) \Rightarrow \rho(-) \circ \rho(-)$) are now trivial.
By the past lemmas all functors and natural transformations involved are indeed symmetric monoidal.
It is easy to check all relevant types of units are preserved.
\end{proof}
	\small	
\bibliographystyle{alpha}
\bibliography{biblio}

\begin{thebibliography}{KTTW15}

\bibitem[AF17]{AFCH}
David Ayala and John Francis.
\newblock The cobordism hypothesis.
\newblock 2017.

\bibitem[ALW19]{Aasen:2017ubm}
David Aasen, Ethan Lake, and Kevin Walker.
\newblock {Fermion condensation and super pivotal categories}.
\newblock {\em J. Math. Phys.}, 60(12):121901, 2019.

\bibitem[Ati88]{atiyah1988topological}
Michael~F Atiyah.
\newblock Topological quantum field theory.
\newblock {\em Publications Math{\'e}matiques de l'IH{\'E}S}, 68:175--186,
  1988.

\bibitem[BD95]{baezdolan}
John~C. Baez and James Dolan.
\newblock Higher-dimensional algebra and topological quantum field theory.
\newblock {\em Journal of mathematical physics}, 36(11):6073--6105, 1995.

\bibitem[BEP15]{BerwickEvans2015SmoothOT}
Daniel Berwick-Evans and Dmitri Pavlov.
\newblock Smooth one-dimensional topological field theories are vector bundles
  with connection.
\newblock {\em arXiv: Algebraic Topology}, 2015.

\bibitem[Bla98]{blackadar}
Bruce Blackadar.
\newblock {\em K-theory for operator algebras}, volume~5.
\newblock Cambridge University Press, 1998.

\bibitem[CS19]{calaquescheimbauer}
Damien Calaque and Claudia Scheimbauer.
\newblock A note on the ($\infty$, n)--category of cobordisms.
\newblock {\em Algebraic \& Geometric Topology}, 19(2):533--655, 2019.

\bibitem[CS21]{LorantNils}
Nils Carqueville and L{\'o}r{\'a}nt Szegedy.
\newblock Fully extended $r$-spin tqfts.
\newblock {\em arXiv preprint arXiv:2107.02046}, 2021.

\bibitem[Dav11]{oritthesis}
Orit Davidovich.
\newblock {\em State sums in two dimensional fully extended topological field
  theories}.
\newblock PhD thesis, 2011.

\bibitem[DP83]{doldpuppe}
Albrecht Dold and Dieter Puppe.
\newblock Duality, trace and transfer.
\newblock {\em Trudy Matematicheskogo Instituta imeni VA Steklova}, 154:81--97,
  1983.

\bibitem[DSPS20]{douglasSPsnyder}
Christopher Douglas, Christopher Schommer-Pries, and Noah Snyder.
\newblock {\em Dualizable tensor categories}, volume 268.
\newblock American Mathematical Society, 2020.

\bibitem[FH21]{freedhopkins}
Daniel~S Freed and Michael~J Hopkins.
\newblock Reflection positivity and invertible topological phases.
\newblock {\em Geometry \& Topology}, 25(3):1165--1330, 2021.

\bibitem[GJF19]{gaiottoSPT}
Davide Gaiotto and Theo Johnson-Freyd.
\newblock Symmetry protected topological phases and generalized cohomology.
\newblock {\em Journal of High Energy Physics}, 2019(5):1--36, 2019.

\bibitem[GP21]{Grady2021TheGC}
Daniel Grady and Dmitri Pavlov.
\newblock The geometric cobordism hypothesis.
\newblock 2021.

\bibitem[GPS95]{streettricat}
Robert Gordon, Anthony~John Power, and Ross Street.
\newblock {\em Coherence for tricategories}, volume 558.
\newblock American Mathematical Soc., 1995.

\bibitem[GS18]{ClaudiaOwen}
Owen Gwilliam and Claudia Scheimbauer.
\newblock Duals and adjoints in higher morita categories.
\newblock {\em arXiv preprint arXiv:1804.10924}, 2018.

\bibitem[Gun16]{gunningham2016spin}
Sam Gunningham.
\newblock Spin hurwitz numbers and topological quantum field theory.
\newblock {\em Geometry \& Topology}, 20(4):1859--1907, 2016.

\bibitem[GW09]{gu2009tensor}
Zheng-Cheng Gu and Xiao-Gang Wen.
\newblock Tensor-entanglement-filtering renormalization approach and
  symmetry-protected topological order.
\newblock {\em Physical Review B}, 80(15):155131, 2009.

\bibitem[Har12]{Harpaz2012TheCH}
Yonatan Harpaz.
\newblock The cobordism hypothesis in dimension 1.
\newblock 2012.

\bibitem[Hes17]{janthesis}
Jan Hesse.
\newblock {\em Group Actions on Bicategories and Topological Quantum Field
  Theories}.
\newblock PhD thesis, Staats-und Universit{\"a}tsbibliothek Hamburg Carl von
  Ossietzky, 2017.

\bibitem[hs]{chrismathoverflow}
Luuk~Stehouwer (https://mathoverflow.net/users/122457/luuk stehouwer).
\newblock An explicit expression for the naturality of the serre automorphism
  in the bicategory of algebras.
\newblock MathOverflow.
\newblock URL:https://mathoverflow.net/q/373685 (version: 2020-10-09).

\bibitem[HSV16]{HSV}
Jan Hesse, Christoph Schweigert, and Alessandro Valentino.
\newblock Frobenius algebras and homotopy fixed points of group actions on
  bicategories.
\newblock {\em Theory and Applications of Categories, Vol. 32, 2017, No. 18, pp
  652-681}, 2016.

\bibitem[HV19]{janserre}
Jan Hesse and Alessandro Valentino.
\newblock The serre automorphism via homotopy actions and the cobordism
  hypothesis for oriented manifolds.
\newblock {\em Cahiers de topologie et géométrie différentielle
  catégoriques}, Volume LX:194--236, 2019.

\bibitem[Jaf18]{jaffe2018reflection}
Arthur Jaffe.
\newblock Reflection positivity then and now.
\newblock {\em arXiv preprint arXiv:1802.07880}, 2018.

\bibitem[JF17]{theospinstatistics}
Theo Johnson-Freyd.
\newblock Spin, statistics, orientations, unitarity.
\newblock {\em Algebraic \& Geometric Topology}, 17(2):917--956, 2017.

\bibitem[JY21]{Johnson2021}
Niles Johnson and Donald Yau.
\newblock {\em {2-Dimensional Categories}}.
\newblock Oxford University Press, 01 2021.

\bibitem[Kap14]{Kapustin:2014tfa}
Anton Kapustin.
\newblock {Symmetry Protected Topological Phases, Anomalies, and Cobordisms:
  Beyond Group Cohomology}.
\newblock 2014.

\bibitem[Kit09]{kitaevperiodic}
Alexei Kitaev.
\newblock Periodic table for topological insulators and superconductors.
\newblock In {\em AIP conference proceedings}, volume 1134, pages 22--30.
  American Institute of Physics, 2009.

\bibitem[KT17]{Kapustin:2015uma}
Anton Kapustin and Alex Turzillo.
\newblock {Equivariant Topological Quantum Field Theory and Symmetry Protected
  Topological Phases}.
\newblock {\em JHEP}, 03:006, 2017.

\bibitem[KTTW15]{Kapustin:2014dxa}
Anton Kapustin, Ryan Thorngren, Alex Turzillo, and Zitao Wang.
\newblock {Fermionic Symmetry Protected Topological Phases and Cobordisms}.
\newblock {\em JHEP}, 12:052, 2015.

\bibitem[LS20]{Ludewig2020AFF}
Matthias Ludewig and Augusto Stoffel.
\newblock A framework for geometric field theories and their classification in
  dimension one.
\newblock {\em Symmetry, Integrability and Geometry: Methods and Applications},
  2020.

\bibitem[Lur09]{lurietft}
Jacob Lurie.
\newblock On the classification of topological field theories.
\newblock In {\em Current developments in mathematics, 2008}, pages 129--280.
  International Press of Boston, 2009.

\bibitem[McC00]{mccrudden2000balanced}
Paddy McCrudden.
\newblock Balanced coalgebroids.
\newblock {\em Theory and Applications of categories}, 7(6):71--147, 2000.

\bibitem[MW22]{MW}
Lukas M{\"u}ller and Lukas Woike.
\newblock Cyclic framed little disks algebras, grothendieck–verdier duality
  and handlebody group representations.
\newblock {\em The Quarterly Journal of Mathematics}, 2022.

\bibitem[Pst22]{piotr}
Piotr Pstragowski.
\newblock On dualizable objects in monoidal bicategories.
\newblock {\em Theory and Applications of Categories}, 38(9):257--310, 2022.

\bibitem[Sch14]{claudiathesis}
Claudia Scheimbauer.
\newblock {\em Factorization homology as a fully extended topological field
  theory}.
\newblock PhD thesis, ETH Zurich, 2014.

\bibitem[Soz19]{Sozer}
Kursat Sozer.
\newblock Two-dimensional extended homotopy field theories.
\newblock 2019.

\bibitem[SP09]{schommerpriesthesis}
Christopher Schommer-Pries.
\newblock {\em The classification of two-dimensional extended topological field
  theories}.
\newblock PhD thesis, Berkeley, 2009.

\bibitem[SP14]{CSPdualizability}
Christopher Schommer-Pries.
\newblock Dualizability in low-dimensional higher category theory.
\newblock {\em Topology and Field Theories}, 613:111, 2014.

\bibitem[ST11]{Stolz2011SupersymmetricFT}
Stephan Stolz and Peter Teichner.
\newblock Supersymmetric field theories and generalized cohomology.
\newblock {\em arXiv: Algebraic Topology}, pages 279--340, 2011.

\bibitem[SV22]{Sozer2022MonoidalCG}
Kursat Sozer and Alexis Virelizier.
\newblock Monoidal categories graded by crossed modules and 3-dimensional
  hqfts.
\newblock 2022.

\bibitem[TMAK20]{tanaka2020lectures}
Hiro~Lee Tanaka, Lukas M{\"u}ller, Araminta Amabel, and Artem Kalmykov.
\newblock {\em Lectures on Factorization Homology, $\infty$-Categories, and
  Topological Field Theories}.
\newblock SpringerBriefs in Mathematical Physics. Springer International
  Publishing, 2020.

\bibitem[Tur16]{turaev2016quantum}
Vladimir Turaev.
\newblock Quantum invariants of knots and 3-manifolds.
\newblock In {\em Quantum Invariants of Knots and 3-Manifolds}. de Gruyter,
  2016.

\bibitem[TV12]{turaev20123}
Vladimir Turaev and Alexis Virelizier.
\newblock On 3-dimensional homotopy quantum field theory, i.
\newblock {\em International Journal of Mathematics}, 23(09):1250094, 2012.

\bibitem[Wal64]{wallgradedbrauer}
Charles Terence~Clegg Wall.
\newblock Graded brauer groups.
\newblock 1964.

\bibitem[Yon19]{Yonekura:2018ufj}
Kazuya Yonekura.
\newblock {On the cobordism classification of symmetry protected topological
  phases}.
\newblock {\em Commun. Math. Phys.}, 368(3):1121--1173, 2019.

\end{thebibliography}
	

\end{document}